\documentclass[
letterpaper, 
11pt, 
onecolumn, 
openany, 
]{memoir}

\usepackage[utf8]{inputenc} 
\usepackage[T1]{fontenc}    %
\usepackage[english]{babel} 
\usepackage[final]{microtype} 
\usepackage{xcolor}
\usepackage{enumitem}

\usepackage{amsmath,amssymb,mathtools} 

\usepackage{graphicx} 

\usepackage{natbib}

\usepackage{latexsym}
\usepackage{amsmath}
\usepackage{amstext}
\usepackage{amsfonts}
\usepackage{amsopn}

\usepackage{pifont}

\usepackage{ifthen}

\usepackage{theorem}

\newtheorem{theorem}                            {Theorem}[chapter]
\newtheorem{lemma}              [theorem]       {Lemma}
\newtheorem{fact}               [theorem]       {Fact}

\newtheorem{corollary}          [theorem]       {Corollary}
\newtheorem{proposition}               [theorem]       {Proposition}
{\theorembodyfont{\rmfamily} }
{\theorembodyfont{\rmfamily} \newtheorem{remark}                [theorem]
{Remark}}
{\theorembodyfont{\rmfamily} }
{\theorembodyfont{\rmfamily} }
{\theorembodyfont{\rmfamily} }
{\theorembodyfont{\rmfamily} }
{\theorembodyfont{\rmfamily} }
{\theorembodyfont{\rmfamily} }
{\theorembodyfont{\rmfamily} }
\theoremstyle{break}
{\theorembodyfont{\rmfamily} }

\newenvironment{proof}{\noindent {\em {Proof:}}}{$\blacksquare$\vskip
\belowdisplayskip}

\newenvironment{prevproof}[2]{\noindent {\em {Proof of
{#1}~\ref{#2}:}}}{$\blacksquare$\vskip \belowdisplayskip}

\newcommand{\prob}[2][]{\textbf{Pr}\ifthenelse{\not\equal{}{#1}}{_{#1}}{}\!\left[#2\right]}
\newcommand{\expect}[2][]{\textbf{E}\ifthenelse{\not\equal{}{#1}}{_{#1}}{}\!\left[#2\right]}

\DeclareMathOperator{\argmin}{argmin}

\newcommand{\ip}[2]{
{\langle {#1} , {#2} \rangle}
}
\newcommand{\n}[1]{
{\| {#1} \|}
}

\usepackage{titlesec}
\usepackage{titletoc}

\newcommand{\secmark}{}
\newcommand{\marktotoc}[1]{\renewcommand{\secmark}{#1}}

\titleformat{\section}
{\normalfont\bfseries\centering}
{\makebox[1.5em][l]{\llap{\secmark}\thesection}}
{0.4em}
{}
\titlecontents{section}[3.7em]{}{\contentslabel[\llap{\secmark}\thecontentslabel]{2.3em}}{\hspace*{-2.3em}}{\titlerule*{}\contentspage}

\newpagestyle{front}{%
  \headrule
  \sethead[\thepage][][\sectiontitle]{\sectiontitle}{}{\thepage}
  \setfoot[][][]{}{}{}
}

\newpagestyle{main}{%
  \headrule
  \sethead[\thepage][][\chaptertitle]{\ifthesection{\llap{\secmark}\thesection\quad\sectiontitle}{\sectiontitle}}{}{\thepage}
  \setfoot[][][]{}{}{}
}

\newpagestyle{back}{%
  \headrule
  \sethead[\thepage][][\subsectiontitle]{\sectiontitle}{}{\thepage}
  \setfoot[][][]{}{}{}
}

\newcommand\invisiblesection[1]{%
  \refstepcounter{section}%
  \addcontentsline{toc}{section}{#1}%
  \sectionmark{#1}}%

\setlrmarginsandblock{0.15\paperwidth}{*}{1} 
\setulmarginsandblock{0.2\paperwidth}{*}{1}  
\checkandfixthelayout

\maxsecnumdepth{subsection} 
\setsecnumdepth{subsection}

\makeatletter %
\makechapterstyle{standard}{
  \setlength{\beforechapskip}{0\baselineskip}
  \setlength{\midchapskip}{1\baselineskip}
  \setlength{\afterchapskip}{8\baselineskip}
  
  \renewcommand{\chapnamefont}{\centering\normalfont\Large}
  \renewcommand{\printchaptername}{\chapnamefont \@chapapp}

}
\makeatother

\addto\captionsenglish{}

\chapterstyle{chappell}

\setsecheadstyle{\normalfont\large\bfseries}
\setsubsecheadstyle{\normalfont\normalsize\bfseries}
\setparaheadstyle{\normalfont\normalsize\bfseries}
\setparaindent{0pt}\setafterparaskip{0pt}

\newsubfloat{figure}

\makeatletter                  
\renewcommand\fps@figure{htbp} %
\renewcommand\fps@table{htbp}  %
\makeatother                   %

\captiondelim{\space } 
\captionnamefont{\small\bfseries} 
\captiontitlefont{\small\normalfont} 

\changecaptionwidth          
\captionwidth{1\textwidth} %

\makepagestyle{standard} 

\makeatletter                 
\makeevenfoot{standard}{}{}{} %
\makeoddfoot{standard}{}{}{}  %
\makeevenhead{standard}{\thepage}{}{\normalfont\qquad\small\leftmark}
\makeoddhead{standard}{\small\rightmark}{}{\thepage}
 \makeheadrule{standard}{\textwidth}{\normalrulethickness}
\makeatother                  %

\makeatletter
\makepsmarks{standard}{
\createmark{chapter}{both}{nonumber}{}{}
\createmark{section}{right}{shownumber}{}{ \quad }
\createplainmark{toc}{both}{\contentsname}
\createplainmark{lof}{both}{\listfigurename}
\createplainmark{lot}{both}{\listtablename}
\createplainmark{bib}{both}{\bibname}
\createplainmark{index}{both}{\indexname}
\createplainmark{glossary}{both}{\glossaryname}
}
\makeatother                               %

\makepagestyle{chap} 
\makeatletter
\makeevenfoot{chap}{}{\small\bfseries\thepage}{} 
\makeoddfoot{chap}{}{\small\bfseries\thepage}{}  %
\makeevenhead{chap}{}{}{}   %
\makeoddhead{chap}{}{}{}    %
\makeatother

\nouppercaseheads
\aliaspagestyle{chapter}{chap} %

\setsecheadstyle{\bfseries\centering}
\setsubsecheadstyle{\bfseries\centering}
\setsubsubsecheadstyle{\bfseries\centering}

\maxtocdepth{subsection} 
\settocdepth{subsection}

\AtEndDocument{\addtocontents{toc}{\par}} 

\usepackage{hyperref}   
\hypersetup{
pdfborder={0 0 0},      
pdfauthor={I am the Author} 
}
\usepackage{memhfixc}   %

\usepackage{xspace}
\newcommand{\var}[2][]{\textbf{Var}\ifthenelse{\not\equal{}{#1}}{_{#1}}{}\!\left[#2\right]}
\def\H{\mathcal{H}}
\def\P{\mathcal{P}}
\newcommand\eps{\epsilon}
\newcommand\sse{\subseteq}
\newcommand\sm{\setminus}
\newcommand\RR{\mathbb{R}}
\def\bfx{\mathbf{x}}
\def\bfy{\mathbf{y}}
\def\bfz{\mathbf{z}}
\def\bfr{\mathbf{r}}
\def\bfa{\mathbf{a}}
\def\bfb{\mathbf{b}}
\def\bfe{\mathbf{e}}
\def\bfp{\mathbf{p}}
\def\bfq{\mathbf{q}}
\def\bfv{\mathbf{v}}
\def\bfd{\mathbf{d}}
\def\bfaz{\mathbf{a(z)}}
\def\joint{(\bfx,\bfy)}
\newcommand\disj{\textsc{Disjointness}\xspace}
\def\inputs{(\bfx,\bfy)}
\def\inputsp{(\bfx',\bfy')}
\def\Inputs{X \times Y}
\def\pmsqrt{(1 \pm \tfrac{1}{\sqrt{n}})}
\newcommand\kldisj{\textsc{$(k,\ell)$-Disjointness}\xspace}
\newcommand\streaming{\textsc{Streaming}\xspace}
\renewcommand\index{\textsc{Index}\xspace}
\newcommand\gh{\textsc{Gap-Hamming}\xspace}
\newcommand\eq{\textsc{Equality}\xspace}
\newcommand\ght{\textsc{Gap-Hamming($t$)}\xspace}
\def\A{\mathbf{A}}
\def\Ax{\mathbf{Ax}}
\def\Ay{\mathbf{Ay}}
\def\Az{\mathbf{Az}}
\def\yh{\hat{\mathbf{y}}}
\def\wh{\hat{\mathbf{w}}}
\def\ones{\mathbf{1}}
\def\res{\text{res}}
\def\enc{\text{enc}}
\def\lnk{k \log \tfrac{n}{k}}
\def\lxn{\log |X| \log n}
\newcommand\aindex{\textsc{Augmented Index}\xspace}
\newcommand\gt{\textsc{Greater-Than}\xspace}
\newcommand\cis{\textsc{Clique-Independent Set}\xspace}
\def\zo{\{0,1\}}
\newcommand\nui{\textsc{$\neg$Unique-Intersection}\xspace}
\newcommand\noneq{\textsc{$\neg$Equality}\xspace}
\def\Cx{\mathbf{Cx}}
\def\Cv{\mathbf{Cv}}
\def\Dy{\mathbf{Dy}}
\def\xP{\bfx \in P}
\newcommand\fvp{\textsc{Face-Vertex($P$)}\xspace}
\newcommand\fvc{\textsc{Face-Vertex(\cor)}\xspace}
\def\ax{\mathbf{a^Tx}}
\def\av{\mathbf{a^Tv}}
\def\ab{\bfa,b}
\def\cor{\textsf{COR}\xspace}
\def\supp{\ensuremath{\text{supp}}}
\def\rect{\ensuremath{\text{rect}}}
\def\ndcc{nondeterministic communication complexity\xspace}
\def\C{\mathbf{C}}
\def\D{\mathbf{D}}
\def\zero{\mathbf{0}}
\def\axb{\ax \le b}

\newcommand\udisj{\textsc{Unique-Disjointness}\xspace}
\def\inputss{(\bfx'a,\bfy'b)}
\newcommand\epsgh{\textsc{$\eps$-Gap Hamming}\xspace}
\newcommand\qd{\textsc{Query-Database}\xspace}
\newcommand\epsdisj{\textsc{$(\tfrac{1}{\eps^2},n)$-Disjointness}\xspace}
\newcommand\edisj{\textsc{$(\tfrac{1}{\eps},n)$-Disjointness}\xspace}
\newcommand\mdisj{\textsc{Multi-Disjointness}\xspace}
\def\minputs{(\bfx_1,\ldots,\bfx_k)}
\newcommand{\wm}[1][k]{\textsc{Welfare-Maximization}($#1$)}

\newcommand{\action}{a}
\newcommand{\actions}{{\mathbf \action}}

\newcommand{\actioni}[1][i]{\action_{#1}}
\newcommand{\Action}{A}
\newcommand{\Actioni}[1][i]{\Action_{#1}}

\newcommand{\sigmasmi}{{\mathbf \sigma}_{-i}}

\newcommand{\V}{{\mathcal{V}}}

\def\nes{{Nash equilibria}\xspace}
\def\ne{{Nash equilibrium}\xspace}
\def\ene{{$\eps$-Nash equilibrium}\xspace}
\def\enes{{$\eps$-Nash equilibria}\xspace}
\def\poly{\ensuremath{\mbox{poly}}}
\def\xmi{\mathbf{x}_{-i}}
\def\xmj{\mathbf{x}_{-j}}
\def\FF{\mathbb{F}}
\def\ZZ{\mathbb{Z}}

\author{\LARGE Tim Roughgarden}
\title{\Huge Communication Complexity\\ (for Algorithm
  Designers)\vspace{.5in}}
\date{}

\begin{document}

\frontmatter
\addtocounter{page}{-1}

\thispagestyle{empty}

\maketitle

\newpage

\thispagestyle{empty}

\begin{vplace}
\begin{center}
\copyright~Tim          
  Roughgarden~2015
\end{center}
\end{vplace}

\newpage

\chapter{Preface}

The best algorithm designers prove both possibility and
impossibility results --- both upper and lower bounds. For example,
every serious computer scientist knows a collection of canonical
NP-complete problems and how to reduce them to other problems of
interest. Communication complexity offers a clean theory that is
extremely useful for proving lower bounds for lots of different
fundamental problems.  Many of the most significant algorithmic
consequences  of the theory follow from its most elementary
aspects.

This document collects the lecture notes from my course
``Communication Complexity (for Algorithm Designers),'' taught at
Stanford in the winter quarter of 2015.
The two primary goals of the course are: 
\begin{itemize}

\item [(1)] Learn several canonical problems in communication
  complexity that are
  useful for proving lower bounds for algorithms (\disj, \index, \gh,
  etc.).

\item [(2)] 
Learn how to reduce lower bounds for fundamental algorithmic
problems to communication complexity lower bounds. 

\end{itemize}
Along the way, we'll also: 
\begin{itemize}

\item [(3)] Get exposure to lots of cool computational models and some famous
results about them --- data streams and linear sketches, compressive
sensing, space-query time trade-offs in data structures,
sublinear-time algorithms, and the extension complexity of linear
programs. 

\item [(4)] Scratch the surface of techniques for proving
  communication complexity lower bounds (fooling sets, 
  corruption arguments, etc.).

\end{itemize}
Readers are assumed to be familiar with undergraduate-level
algorithms, as well as the statements of standard large deviation
inequalities (Markov, Chebyshev, and Chernoff-Hoeffding).

The course begins in Lectures~\ref{cha:data-stre-algor}--\ref{cha:lower-bounds-compr}
with the simple case of one-way communication protocols --- where only a
single message is sent --- and their relevance to algorithm design.  
Each of these lectures depends on the previous one.
Many of the ``greatest hits'' of communication complexity
applications, including lower bounds for small-space streaming
algorithms and compressive sensing, are already implied by lower bounds
for one-way protocols.  Reasoning about one-way protocols also provides
a gentle warm-up to the standard model of general two-party
communication protocols, which is the subject of
Lecture~\ref{cha:boot-camp-comm}.  
Lectures~\ref{cha:lower-bounds-extens}--\ref{cha:lower-bounds-prop}
translate communication complexity lower bounds into lower bounds in
several disparate problem domains:
the extension complexity
of polytopes, data structure design, algorithmic game theory, and
property testing.  Each of these final four lectures depends only on Lecture~\ref{cha:boot-camp-comm}.

The course Web page
(\url{http://theory.stanford.edu/~tim/w15/w15.html}) contains links
to relevant large deviation inequalities, links to many of the papers
cited in these notes, and a partial list of exercises.
Lecture notes and videos on several other topics in theoretical
computer science are available from my Stanford home page.

I am grateful to the Stanford students who took the course, for their
many excellent questions: 
Josh Alman,
Dylan Cable,
Brynmor Chapman,
Michael Kim,
Arjun Puranik,
Okke Schrijvers,
Nolan Skochdopole, 
Dan Stubbs,
Joshua Wang,
Huacheng Yu,
Lin Zhai, and several auditors whose names I've forgotten.
I am also indebted to Alex Andoni, Parikshit Gopalan, Ankur Moitra,
and C.\ Seshadhri for their advice on some of these lectures.
The writing of these notes was supported in part by NSF award
CCF-1215965.

I always appreciate suggestions and corrections from readers.

\vspace{.3in}

\begin{flushleft}
Tim Roughgarden\\
474 Gates Building, 353 Serra Mall\\
Stanford, CA 94305\\
Email: \texttt{tim@cs.stanford.edu}\\
WWW: \url{http://theory.stanford.edu/~tim/}
\end{flushleft}


\clearpage

\tableofcontents*
\clearpage

\listoffigures

\clearpage

\mainmatter

\chapter{Data Streams: Algorithms and Lower Bounds}
\label{cha:data-stre-algor}

\section{Preamble}

This class is mostly about impossibility results --- lower bounds on
what can be accomplished by algorithms.  However, our perspective will
be unapologetically that of an algorithm designer.\footnote{Already in
  this lecture, over half our discussion will be about algorithms and
  upper bounds!}
We'll learn lower bound technology on a ``need-to-know'' basis,
guided by fundamental algorithmic problems that we care about
(perhaps theoretically, perhaps practically).  That said, we will wind
up learning quite a bit of complexity theory --- specifically,
communication complexity --- along the way.  We hope this viewpoint
makes this course and these notes complementary to the numerous
excellent courses, books (\cite{J12} and~\cite{KN96}), and surveys
(e.g.~\cite{lee,lovasz,pitassi,razborov})
on communication complexity.\footnote{See~\cite{patrascu} for a
  series of four blog posts on data structures that share some spirit
  with our approach.}  The theme of communication complexity 
lower bounds also
provides a convenient excuse to take a guided tour of numerous models,
problems, and algorithms that are central to modern research in the
theory of algorithms but missing from many algorithms textbooks:
streaming algorithms, space-time trade-offs in data structures,
compressive sensing, sublinear algorithms, extended formulations for
linear programs, and more.

Why should an algorithm designer care about lower bounds?
The best mathematical researchers can work on an open problem
simultaneously from both sides.  Even if you have a strong prior
belief about whether a given mathematical statement is true or false,
failing to prove one direction usefully informs your efforts to prove
the other.  (And for most us, the prior belief is wrong surprisingly often!)
In algorithm design, working on both sides means striving
simultaneously for better algorithms and for better lower bounds.  For
example, a good undergraduate algorithms course teaches you both how
to design polynomial-time algorithms and how to prove that a problem
is $NP$-complete --- since when you encounter a new computational
problem in your research or workplace, both are distinct possibilities.
There are many other algorithmic problems where
the fundamental difficulty is not the amount of time required, but rather
concerns communication or information transmission.  The goal of this
course is to equip you with the basic tools of communication
complexity --- its canonical hard problems, the canonical reductions
from computation in various models to the design of low-communication
protocols, and a little bit about its lower bound techniques --- in
the service of becoming a better algorithm designer.

This lecture and the next study the {\em data stream} model of
computation.  There are some nice upper bounds in this model (see
Sections~\ref{s:f2} and~\ref{s:f0}), and the model also naturally
motivates a 
severe but useful restriction of the general communication complexity
setup (Section~\ref{s:1way}).  We'll cover many computational models
in the course, so whenever you get sick of one, don't worry, a new one
is coming up around the corner.

\section{The Data Stream Model}\label{s:streaming}

The data stream model is motivated by applications in which the input
is best thought of as a firehose --- packets arriving to a network
switch at a torrential rate, or data being generated by a telescope at
a rate of one exobyte per day.  In these applications, there's no
hope of storing all the data, but we'd still like to remember useful
summary statistics about what we've seen.

Alternatively, for example in database applications, it could be that
data is not thrown away but resides on a big, slow disk.  Rather than
incurring random access costs to the data, one would like to process
it sequentially once (or a few times), perhaps overnight, remembering
the salient features of the data in a limited main memory for
real-time use.  The daily transactions of Amazon or Walmart, for
example, could fall into this category.

Formally, suppose data elements belong to a known universe $U =
\{1,2,\ldots,n\}$.  The input is a stream $x_1,x_2,\ldots,x_m \in U$
of elements that arrive one-by-one.  Our algorithms will not assume
advance knowledge of $m$, while our lower bounds will hold even if $m$
is known a priori.  With space $\approx m \log_2 n$, it is possible to
store all of the data.  The central question in data stream algorithms
is: what is possible, and what is impossible, using a one-pass
algorithm and much less than $m \log n$ space?  Ideally, the space
usage should be sublinear or even logarithmic in $n$ and
$m$.  We're not going to worry about the computation time used
  by the algorithm (though our positive results in this
  model have low computational complexity, anyway).

Many of you will be familiar with a streaming or one-pass algorithm
from the following common interview question.  Suppose an array $A$,
with length $m$, is promised to have a {\em majority element} --- an
element that occurs strictly more than $m/2$ times.  A simple one-pass
algorithm, which maintains only the current candidate majority element
and a counter for it --- so $O(\log n + \log m)$ bits --- solves this
problem.  (If you haven't seen this algorithm before, see the
Exercises.)  This can be viewed as an exemplary small-space streaming
algorithm.\footnote{Interestingly, the promise that a majority element
  exists is crucial.  A consequence of the next lecture is that there is
  no small-space streaming algorithm to check whether or not a
  majority element exists!}

\section{Frequency Moments}

Next we introduce {\em the} canonical problems in the field of data
stream algorithms: computing the {\em frequency moments} of a stream.
These were studied in the paper that kickstarted
the field~\citep{AMS96}, and the data stream algorithms community has
been obsessed with them ever since.

Fix a data stream $x_1,x_2,\ldots,x_m \in U$.  For an element $j \in
U$, let $f_j \in \{0,1,2,\ldots,m\}$ denote the number of times that
$j$ occurs in the stream.  For a non-negative integer $k$, the {\em
  $k$th frequency moment} of the stream is
\begin{equation}\label{eq:fk}
F_k := \sum_{j \in U} f_j^k.
\end{equation}
Note that the bigger $k$ is, the more the sum in~\eqref{eq:fk} is
dominated by the largest frequencies.  It is therefore natural to
define
\[
F_{\infty} = \max_{j \in U} f_j
\]
as the largest frequency of any element of $U$.

Let's get some sense of these frequency moments.  $F_1$ is
boring --- since each data stream element contributes to exactly one
frequency $f_j$, $F_1 = \sum_{j \in U} f_j$ is simply the stream
length~$m$.  $F_0$ is the number of distinct elements in the stream
(we're interpreting $0^0 = 0$) --- it's easy to imagine wanting to
compute this quantity, for example a network switch might want to know
how many different TCP flows are currently going through it.
$F_{\infty}$ is the largest frequency, and again it's easy to imagine
wanting to know this 
--- 
for example to detect a denial-of-service attack at a network switch,
or identify the most popular product on Amazon yesterday.  Note that
computing $F_{\infty}$ is related to the aforementioned problem of
detecting a majority element.  Finally, $F_2 = \sum_{j \in U} f_2^2$ is
  sometimes called the ``skew'' of the data --- it is a measure of how
  non-uniform the data stream is.  In a database context, it arises
  naturally as the size of a ``self-join'' --- the table you get when
  you join a relation with itself on a particular attribute, with the
  $f_j$'s being the frequencies of various values of this attribute.
  Having estimates of self-join (and more generally join) sizes at the
  ready is useful for query optimization, for example.  We'll talk about
  $F_2$ extensively in the next section.\footnote{The problem of
    computing $F_2$ and the solution we give for it are also quite
    well connected to other important concepts, such as compressive
    sensing and dimensionality reduction.}

It is trivial to compute all of the frequency moments in $O(m \log n)$
space, just by storing the $x_i$'s, or in $O(n \log m)$, space, just
by computing and storing the $f_j$'s (a $\log m$-bit counter for each
of the $n$ universe elements).  Similarly, $F_1$ is trivial to compute
in $O(\log m)$ space (via a counter), and $F_0$ in $O(n)$ space (via a
bit vector).  Can we do better?

Intuitively, it might appear difficult to improve over the trivial
solution.  For $F_0$, for example, it seems like you have to know which
elements you've already seen (to avoid double-counting them), and
there's an exponential (in $n$) number of different possibilities
for what you might have seen in the past.  As we'll see, this is good
intuition for deterministic algorithms, and for (possibly randomized)
exact algorithms.  Thus, the following positive result is arguably
surprising, and very cool.\footnote{The Alon-Matias-Szegedy
  paper~\citep{AMS96} ignited the field of streaming algorithms as a
  hot area, and for this reason won the 2005 G\"odel Prize (a ``test
  of time''-type award in theoretical computer science).  The paper
  includes a number of other upper and lower bounds as well, some of
  which we'll 
  cover shortly.}

\begin{theorem}[\citealt{AMS96}]\label{t:ams}
Both $F_0$ and $F_2$ can be approximated, to within a $(1 \pm \eps)$
factor with probability at least $1-\delta$, in space
\begin{equation}\label{eq:ams}
O\left((\eps^{-2} (\log n + \log m) \log \tfrac{1}{\delta} \right).
\end{equation}
\end{theorem}

Theorem~\ref{t:ams} refers to two different algorithms, one for $F_0$
and one for $F_2$.  We cover the latter in detail below.
Section~\ref{s:f0} describes the high-order bit of the $F_0$
algorithm, which is a modification of the earlier algorithm
of~\cite{FM83}, with the details in the exercises.
Both algorithms are
randomized, and are approximately correct (to within $(1\pm\eps)$) most
of the time (except with probability $\delta$).  Also, the $\log m$
factor in~\eqref{eq:ams} is not needed for the $F_0$ algorithm, as
you might expect.  Some further optimization are possible; see
Section~\ref{ss:opt}.

The first reason to talk about Theorem~\ref{t:ams} is that it's a
great result in the field of algorithms --- if you only remember one
streaming algorithm, the one below might as well be the
one.\footnote{Either algorithm, for estimating $F_0$ or for $F_2$,
  could serve this purpose.  We present the $F_2$ estimation algorithm
  in detail, because the analysis is slightly slicker and more
  canonical.}
You should never tire of seeing clever algorithms that radically
outperform the ``obvious solution'' to a well-motivated problem.
And Theorem~\ref{t:ams} should serve as inspiration to any algorithm
designer --- even when at first blush there is no non-trivial algorithm for
problem in sight, the right clever insight can unlock a good
solution. 

On the other hand, there unfortunately are some important problems
out there with no non-trivial solutions.  And
it's important for the algorithm designer to know which ones they are
--- the less effort wasted on trying to find something that doesn't
exist, the more energy is available for formulating the motivating
problem in a more tractable way, weakening the desired guarantee,
restricting the problem instances,
and otherwise finding new ways to make algorithmic progress. 
A second interpretation of Theorem~\ref{t:ams} is that it illuminates
why such lower bounds can be so hard to prove.  A  lower bound is
responsible for showing that every algorithm, even fiendishly
clever ones like those employed for Theorem~\ref{t:ams}, cannot make
significant inroads on the problem.

\section{Estimating $F_2$: The Key Ideas}\label{s:f2}

In this section we give a nearly complete proof of Theorem~\ref{t:ams}
for the case of $F_2 = \sum_{j \in U} f^2_2$ estimation (a few details
are left to the Exercises).

\subsection{The Basic Estimator}\label{ss:basic}

The high-level idea is very natural, especially once you start
thinking about randomized algorithms.
\begin{enumerate}

\item Define a randomized unbiased estimator of $F_2$, which can be
  computed in one pass.  Small space seems to force a
  streaming algorithm to lose information, but maybe it's
  possible to produce a result that's correct ``on average.''

\item Aggregate many independent copies of the estimator, computed in
  parallel, to get an 
  estimate that is very accurate with high probability.

\end{enumerate}
This is very hand-wavy, of course --- does it have any hope of
working?  It's hard to answer that question without actually doing
some proper computations, so let's proceed to the estimator devised
in~\cite{AMS96}. 

\vspace{.1in}
\noindent
\textbf{The Basic Estimator:}\footnote{This is sometimes called the
  ``tug-of-war'' estimator.}
\begin{enumerate}

\item Let $h:U \rightarrow \{ \pm 1 \}$ be a function that associates
  each universe element with a random sign.  On a first reading, to
  focus on the main ideas, you
  should assume that $h$ is a totally random function.  Later we'll
  see that relatively lightweight hash functions are good enough
  (Section~\ref{ss:hash}), which enables a small-space implementation
  of the basic ideas.

\item Initialize $Z = 0$.

\item Every time a data stream element $j \in U$ appears,
add $h(j)$ to $Z$.  That is, increment $Z$ if $h(j) =
  +1$ and decrement $Z$ if $h(j) = -1$.\footnote{This is the ``tug of
    war,'' between elements $j$ with $h(j) = +1$ and those with $h(j)
    = -1$.}

\item Return the estimate $X = Z^2$.

\end{enumerate}

\begin{remark}\label{rem:remember}
A crucial point: since the function $h$ is fixed once and for all
before the data stream arrives, an element $j \in U$ is treated
consistently every time it shows up.  That is, $Z$ is either
incremented every time $j$ shows up or is decremented every time $j$
shows up.  In the end, element $j$ contributes $h(j)f_j$ to the final
value of $Z$.
\end{remark}

First we need to prove that the basic estimator is indeed unbiased.
\begin{lemma}\label{l:exp}
For every data stream,
\[
\expect[h]{X} = F_2.
\]
\end{lemma}

\begin{proof}
We have
\begin{eqnarray}\nonumber
\expect{X} & = & \expect{Z^2}\\ \nonumber
& = & \expect{ \left( \sum_{j \in U} h(j)f_j \right)^2 }\\ \nonumber
& = & \expect{ \sum_{j \in U} \underbrace{h(j)^2}_{=1}f_j^2 
+ 2 \sum_{j < \ell} h(j)h(\ell)f_jf_{\ell} }\\ \label{eq:exp1}
& = & \underbrace{\sum_{j \in U} f_j^2}_{=F_2}
+ 2 \sum_{j < \ell} f_jf_{\ell}
\underbrace{\expect[h]{h(j)h(\ell)}}_{=0}\\ \label{eq:exp2}
& = & F_2,
\end{eqnarray}
where in~\eqref{eq:exp1} we use linearity of expectation and the fact
that $h(j) \in \{ \pm 1\}$ for every $j$, and in~\eqref{eq:exp2} we
use the fact that, for every distinct $j,\ell$, all four sign patterns
for $(h(j),h(\ell))$ are equally likely.
\end{proof}
Note the reason for both incrementing and decrementing in the running
sum $Z$ --- it ensures that the ``cross terms'' $h(j)h(\ell)f_jf_{\ell}$
in our basic estimator $X = Z^2$ cancel out in expectation.
Also note, for future reference, that the only time we used the
assumption that $h$ is a totally random function was
in~\eqref{eq:exp2}, and we only used the property that all four sign patterns
for $(h(j),h(\ell))$ are equally likely --- that $h$ is ``pairwise
independent.''

Lemma~\ref{l:exp} is not enough for our goal, since 
the basic estimator $X$ is not guaranteed to be close to its
expectation with high probability.  A natural idea is to take
the average of many independent copies of the basic estimator.  That
is, we'll use $t$ independent functions $h_1,h_2,\ldots,h_t:U
\rightarrow \{ \pm 1 \}$ to define estimates $X_1,\ldots,X_t$.
On the arrival of a new data stream element, we update all
$t$ of the counters $Z_1,\ldots,Z_t$ appropriately, with some getting
incremented and others decremented.
Our final estimate will be
\[
Y = \frac{1}{t} \sum_{i=1}^t X_i.
\]
Since the $X_i$'s are unbiased estimators, so is $Y$ (i.e.,
$\expect[h_1,\ldots,h_t]{Y} = F_2$).  To see how quickly the variance
decreases with the number~$t$ of copies, note that
\[
\var{Y} = \var{\frac{1}{t} \sum_{i=1}^t X_i} 
= \frac{1}{t^2} \sum_{i=1}^t \var{X_i} = \frac{\var{X}}{t},
\]
where $X$ denotes a single copy of the basic estimator.  That is,
averaging reduces the variance by a factor equal to the number of
copies.  Unsurprisingly, the number of copies $t$
(and in the end, the space) that we need to get the performance
guarantee that we want is governed by the variance of the basic
estimator.  So there's really no choice but to roll up our sleeves and
compute it.

\begin{lemma}\label{l:var}
For every data stream,
\[
\var[h]{X} \le 2F_2^2.
\]
\end{lemma}

Lemma~\ref{l:var} states the standard deviation of the basic estimator
is in the same ballpark as its expectation.  That might sound
ominous, but it's actually great news --- a constant (depending on
$\eps$ and $\delta$ only) number of copies is good enough for our
purposes.  Before proving Lemma~\ref{l:var}, let's see why.

\begin{corollary}\label{cor:f2}
For every data stream, with $t = \frac{2}{\eps^2\delta}$, 
\[
\prob[h_1,\ldots,h_t]{Y \in (1 \pm \eps) \cdot F_2} \ge 1-\delta.
\]
\end{corollary}

\begin{proof}
Recall that {\em Chebyshev's inequality} is the inequality you want 
when bounding the deviation of a random variable from
its mean parameterized by the number of standard deviations.
Formally, it states that for every random variable $Y$ with finite
first and second moments, and every $c > 0$,
\begin{equation}\label{eq:chebyshev}
\prob{|Y - \expect{Y}| > c} \le \frac{\var{Y}}{c^2}.
\end{equation}
Note that~\eqref{eq:chebyshev} is non-trivial (i.e., probability less
than 1) once $c$ exceeds $Y$'s standard deviation, and the probability
goes down quadratically with the number of standard deviations.  It's
a simple inequality to prove; see the separate notes on tail
inequalities for details.

We are interested in the case where $Y$ is
the average of $t$ basic estimators, with variance
as in Lemma~\ref{l:var}.
Since we want to guarantee a $(1 \pm \eps)$-approximation,
the deviation $c$ of interest to us is $\eps F_2$.  We also want the
right-hand side of~\eqref{eq:chebyshev} to be equal to $\delta$.
Using Lemma~\ref{l:var} and solving, we
get $t = 2/\eps^2\delta$.\footnote{The dependence on
  $\tfrac{1}{\delta}$ can be decreased to logarithmic; see
  Section~\ref{ss:opt}.} 
\end{proof}

We now stop procrastinating and prove Lemma~\ref{l:var}.

\vspace{.1in}
\noindent
\begin{prevproof}{Lemma}{l:var}
Recall that
\begin{equation}\label{eq:var1}
\var{X} = \expect{X^2} - \left( \underbrace{\expect{X}}_{=F_2 \text{
    by Lemma~\ref{l:exp}}} \right)^2.
\end{equation}
Zooming in on the $\expect{X^2}$ term, recall that $X$ is already defined
as the square of the running sum $Z$, so $X^2$ is $Z^4$.  Thus,
\begin{equation}\label{eq:var2}
\expect{X^2} = \expect{ \left( \sum_{j \in U} h(j)f_j \right)^4}.
\end{equation}

Expanding the right-hand side of~\eqref{eq:var2} yields $|U|^4$ terms,
each of the form\\
$h(j_1)h(j_2)h(j_3)h(j_4)f_{j_1}f_{j_2}f_{j_3}f_{j_4}$.  (Remember:
the $h$-values are random, the $f$-values are fixed.)
This might seem unwieldy.  But, just as in the
computation~\eqref{eq:exp2} in the proof of Lemma~\ref{l:exp}, most of
these are zero in expectation.  For example, suppose $j_1,j_2,j_3,j_4$
are distinct.  Condition on the $h$-values of the first three.  Since
$h(j_4)$ is equally likely to be +1 or -1, the conditional expected
value (averaged over the two cases) of the corresponding term is~0.
Since this holds for all possible values of $h(j_1),h(j_2),h(j_3)$,
the unconditional expectation of this term is also~0.  This same
argument applies to any term in which some element $j \in U$ appears
an odd number of times.  Thus, when the dust settles, we have
\begin{eqnarray}\label{eq:var3}
\expect{X^2} & = & \expect{ \sum_{j \in U}
  \underbrace{h(j)^4}_{=1} f_j^4
+ 6 \sum_{j < \ell} \underbrace{h(j)^2}_{=1}
\underbrace{h(\ell)^2}_{=1}f_j^2f_{\ell}^2}\\ \nonumber
& = & \sum_{j \in U} f^4_j + 6 \sum_{j < \ell} f_j^2f_{\ell}^2,
\end{eqnarray}
where the ``6'' appears because 
a given $h(j)^2h(\ell)^2f_j^2f_{\ell}^2$ term with $j < \ell$
arises in $\binom{4}{2} = 6$ different ways.

Expanding terms, we see that
\[
F_2^2 = \sum_{j \in U} f^4_j + 2 \sum_{j < \ell} f_j^2f_{\ell}^2
\]
and hence 
\[
\expect{X^2} \le 3F_2^2.
\]
Recalling~\eqref{eq:var1} proves that $\var{X} \le 2F_2^2$, as
claimed.
\end{prevproof}

Looking back over the proof of Lemma~\ref{l:var}, we again see
that we only used the fact that~$h$ is random in a limited way.
In~\eqref{eq:var3} we used that, for every set of four
distinct universe elements, their 16 possible sign patterns (under
$h$) were equally likely.  (This implies the required property that,
if $j$ appears in a term an odd number of times, then even after
conditioning on the $h$-values of all other universe elements in the
term, $h(j)$ is equally likely to be +1 or -1.)  That is, we only used
the ``4-wise independence'' of the function $h$.

\subsection[4-Wise                             
  Independent Hash Functions]{Small-Space Implementation via 4-Wise
  Independent Hash Functions}\label{ss:hash}

Let's make sure we're clear on the final algorithm.
\begin{enumerate}

\item Choose functions $h_1,\ldots,h_t:U \rightarrow \{ \pm 1 \}$,
  where $t = \tfrac{2}{\eps^2\delta}$.

\item Initialize $Z_i = 0$ for $i=1,2,\ldots,t$.

\item When a new data stream element $j \in U$ arrives, add $h_i(j)$
  to $Z_i$ for every $i=1,2,\ldots,t$.

\item Return the average of the $Z_i^2$'s.

\end{enumerate}
Last section proved that, if the $h_i$'s are chosen uniformly at
random from all functions, then the output of this algorithm lies in $(1
\pm \eps)F_2$ with probability at least $1-\delta$.

How much space is required to implement this algorithm?  There's
clearly a factor of $\tfrac{2}{\eps^2\delta}$, since we're effectively
running this many streaming algorithms in parallel, and each needs its
own scratch space.  How much space does each of these need?  To
maintain a counter $Z_i$, which always lies between $-m$ and $m$, we
only need $O(\log m)$ bits.  But it's easy to forget that we have to
also store the function $h_i$.  Recall from Remark~\ref{rem:remember}
the reason: we need to treat an element $j \in U$ consistently every
time it shows up in the data stream.  Thus, once we choose a sign $h_i(j)$
for $j$ we need to remember it forevermore.  Implemented naively, with
$h_i$ a totally random function, we would need to remember one bit for
each of the possibly $\Omega(n)$ elements that we've seen in the
past, which is a dealbreaker.

Fortunately, as we noted after the proofs of Lemmas~\ref{l:exp}
and~\ref{l:var}, our entire analysis has relied only on
4-wise independence --- that when we look at an arbitrary 4-tuple of
universe elements, the projection of $h$ on their 16 possible sign
patterns is uniform.  (Exercise: go back through this section in
detail and double-check this claim.)  And happily, there are small
families of simple hash functions that possess this property.

\begin{fact}\label{fact:4wise}
For every universe $U$ with $n = |U|$, there is a family $\H$ of
4-wise independent hash functions (from $U$ to $\{ \pm 1 \}$) with
size polynomial in~$n$.
\end{fact}

Fact~\ref{fact:4wise} and our previous observations imply that, to
enjoy an approximation of $(1 \pm \eps)$ with probability at least
$1-\delta$, our streaming algorithm can get away with choosing the
functions $h_1,\ldots,h_t$ uniformly and independently from~$\H$.

If you've never seen a construction of a $k$-wise independent family of
hash functions with $k > 2$, check out the Exercises for details.  The
main message is to realize that you shouldn't be scared of them --- heavy
machinery is not required.  For example, it suffices to map the
elements of $U$ injectively into a suitable finite field (of size
roughly $|U|$), and then
let $\H$ be the set of all cubic polynomials (with all operations
occurring
in this field).  The final output is then +1 if the polynomial's output
(viewed as an integer) is even, and -1 otherwise.
Such a hash function is easily specified with $O(\log n)$ bits (just
list its four coefficients), and can also be evaluated in $O(\log n)$
space (which is not hard, but we won't say more about it here).

Putting it all together, we get a space bound of
\begin{equation}\label{eq:delta}
O\left(\underbrace{\frac{1}{\eps^2\delta}}_{\text{\# of copies}} \cdot
\left( \underbrace{\log m}_{\text{counter}} + \underbrace{\log
  n}_{\text{hash function}} \right) \right).
\end{equation}

\subsection{Further Optimizations}\label{ss:opt}

The bound in~\eqref{eq:delta} is worse than that claimed in
Theorem~\ref{t:ams}, with a dependence on $\tfrac{1}{\delta}$ instead
of $\log \tfrac{1}{\delta}$.  A simple trick yields the better bound.
In Section~\ref{ss:basic}, we averaged $t$ copies of the basic
estimator to accomplish two conceptually different things: to improve
the approximation ratio to $(1 \pm \eps)$, for which we suffered an
$\tfrac{1}{\eps^2}$ factor, and to improve the success probability to
$1-\delta$, for which we suffered an additional $\tfrac{1}{\delta}$.
It is more efficient to implement these improvements one
at a time, rather than in one shot.  The smarter implementation first
uses $\approx \tfrac{1}{\eps^{2}}$ copies to obtain an approximation of
$(1\pm\eps)$ with probably at least $\tfrac{2}{3}$ (say).  To boost the
success probability from $\tfrac{2}{3}$ to $1-\delta$, it is enough to
run $\approx \log \tfrac{1}{\delta}$ different copies of this solution, and
then take the {\em median} of their $\approx \log
\tfrac{1}{\delta}$ different estimates.  Since we expect at least
two-thirds of these estimates to lie in the interval $(1 \pm
\eps)F_2$, it is very likely that the median of them lies in this
interval.  The details are easily made precise using a Chernoff bound
argument; see the Exercises for details.

Second, believe it or not, the $\log m$ term in Theorem~\ref{t:ams}
can be improved to $\log \log m$.  The reason is that we don't need to
count the $Z_i$'s exactly, only approximately and with high
probability.  This relaxed counting problem can be solved using {\em
  Morris's algorithm}, which can be implemented as a streaming
algorithm that uses $O(\eps^{-2} \log \log m \log
\tfrac{1}{\delta})$ space.  See the Exercises for further details.

\section{Estimating $F_0$: The High-Order Bit}\label{s:f0}

Recall that $F_0$ denotes the number of distinct elements present in a
data stream.
The high-level idea of the $F_0$ estimator 
is the same as
for the $F_2$ estimator above.  
The steps are to define a basic estimator that is essentially
unbiased, and then reduce the variance by taking averages and
medians.  (Really making this work takes an additional idea;
see~\cite{B+02} and the Exercises.)

The basic estimator for $F_0$ --- originally
from~\cite{FM83} and developed further in~\cite{AMS96} and~\cite{B+02}
--- is as simple as but quite different
from that used to estimate $F_2$.  The first step is to choose a
random permutation $h$ of $U$.\footnote{Or rather, a simple hash
  function with the salient properties of a random permutation.}
Then, just remember (using $O(\log n)$ space) the minimum value of
$h(x)$ that ever shows up in the data stream.

Why use the minimum?  One intuition comes from the suggestive match
between the idempotence of $F_0$ and of the minimum --- adding 
duplicate copies of an element to the input has no effect on
the answer.  

Given the minimum $h(x)$-value in the data stream, how do we
extract from it an estimate of $F_0$, the number of distinct elements?
For intuition, think about the uniform distribution on $[0,1]$ (Figure~\ref{f:uniform}).
Obviously, the expected value of one draw from the distribution
is~$\tfrac{1}{2}$.  For two i.i.d.\ draws, simple calculations show
that 
the expected minimum and maximum are $\tfrac{1}{3}$ and
$\tfrac{2}{3}$, respectively.  More generally, the expected order
statistics of $k$ i.i.d.\ draws split the interval into $k+1$ segments of
equal length.  In particular, the expected minimum is
$\tfrac{1}{k+1}$.  In other words, if you are told that some number of
i.i.d.\ draws were taken from the uniform distribution on $[0,1]$ and
the smallest draw was $c$, you might guess that there were roughly
$1/c$ draws. 

\begin{figure}
\centering
\includegraphics[width=.6\textwidth]{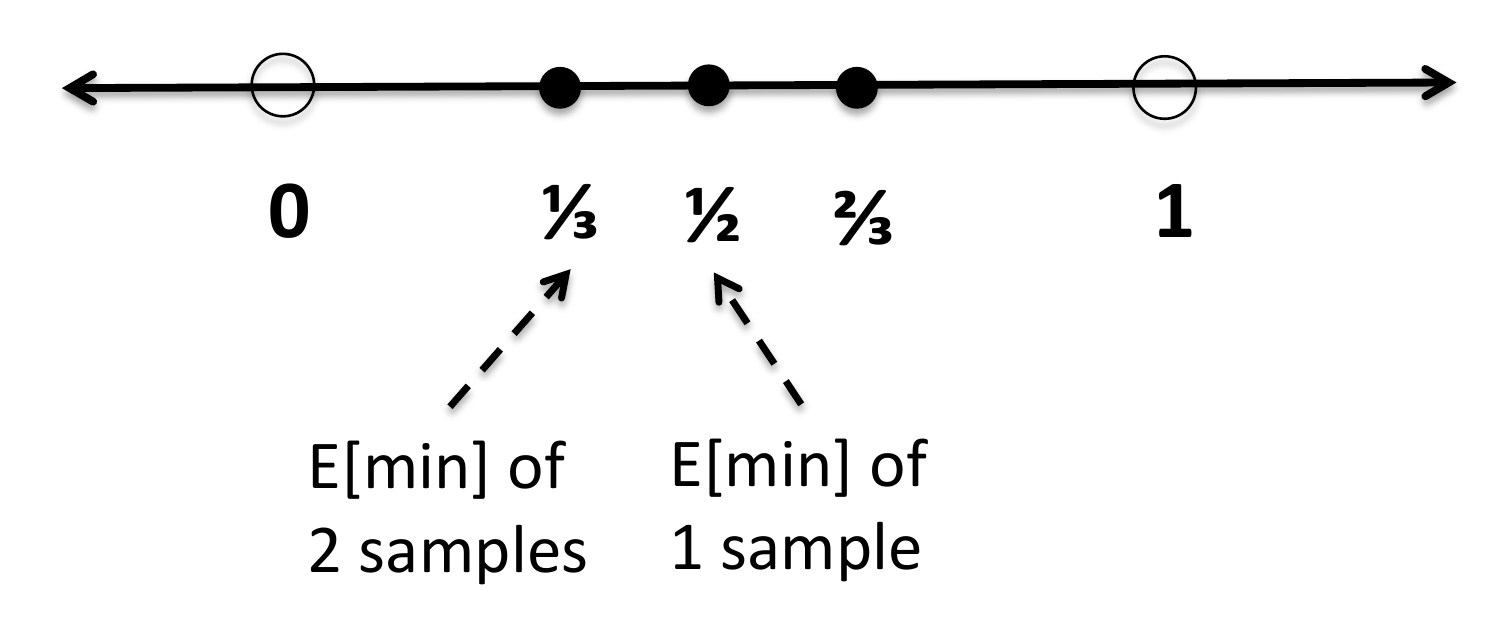}
\caption[Expected order statistics of i.i.d.\ samples from the
   uniform distribution]{The expected order statistics of i.i.d.\ samples from the
uniform distribution on the unit interval are spaced out evenly over
the interval.}\label{f:uniform}
\end{figure}

Translating this idea to our basic $F_0$ estimator, if there are $k$
distinct elements in a data stream $x_1,\ldots,x_m$, then there are
$k$ different (random) 
hash values $h(x_i)$, and we expect the smallest of these to be roughly
$|U|/k$.  This leads to the basic estimator $X = |U|/(\min_{i=1}^m
h(x_i))$.  Using averaging and an extra idea to reduce the variance,
and medians to boost the success probability, this leads to the bound
claimed in Theorem~\ref{t:ams} (without the $\log m$ term).  The
details are outlined in the exercises.

\begin{remark}
You'd be right to ask if this high-level approach to probabilistic and
approximate estimation applies to all of the frequency moments,
not just $F_0$ and $F_2$.  The approach can indeed be used to
estimate $F_k$ for all $k$.  However, the
variance of the basic estimators will be different for different frequency
moments.  For example, as $k$ grows, the statistic $F_k$ becomes quite
sensitive to small changes in the input, resulting in
probabilistic estimators with large variance, necessitating a large
number of independent copies to obtain a good approximation.
More generally, no frequency moment $F_k$ with
$k \not\in \{0,1,2\}$ can be computed using only a logarithmic amount
of space (more details to come).
\end{remark}

\section{Can We Do Better?}\label{s:better}

Theorem~\ref{t:ams} is a fantastic result.  But a good algorithm
designer is never satisfied, and always wants more.  So what are the
weaknesses of the upper bounds that we've proved so far?
\begin{enumerate}

\item We only have interesting positive results for $F_0$ and $F_2$
  (and maybe $F_1$, if you want to count that).  What about for $k > 2$
  and $k = \infty$?

\item Our $F_0$ and $F_2$ algorithms only approximate the
  corresponding frequency moment.  Can we compute it exactly, possibly
  using a randomized algorithm?

\item Our $F_0$ and $F_2$ algorithms are randomized, and with
  probability $\delta$ fail to
  provide a good approximation.  (Also, they
  are Monte Carlo algorithms, in that we can't tell when they fail.)
  Can we compute $F_0$ or $F_2$ deterministically, at least
  approximately?

\item Our $F_0$ and $F_2$ algorithms use $\Omega(\log n)$ space.  Can
  we reduce the dependency of the space on the universe
  size?\footnote{This might seem like a long shot, but you never know.
    Recall our comment about reducing the space dependency on $m$ from
    $O(\log m)$ to $O(\log \log m)$ via probabilistic approximate
    counters.}

\item Our $F_0$ and $F_2$ algorithms use $\Omega(\eps^{-2})$ space.
Can the dependence on $\eps^{-1}$ be improved? 
The $\eps^{-2}$ dependence can be painful in practice, where you might want
  to take $\eps = .01$, resulting in an extra factor of 10,000 in the
  space bound.  An improvement to $\approx \eps^{-1}$, for example,
  would be really nice.

\end{enumerate}
Unfortunately, we can't do better --- the rest of this lecture and the
next (and the exercises) explain why {\em all} of these compromises
are necessary for positive results.  This is kind of amazing, and it's
also pretty amazing that we can prove it without
overly heavy machinery.  Try to think of other basic computational
problems where, in a couple hours of lecture and with minimal
background, you can explain complete proofs of both a non-trivial
upper bound and an unconditional (independent of $P$ vs.\ $NP$, etc.)
matching lower bound.\footnote{OK, comparison-based sorting, sure.
  And we'll see a couple others later in this course.  But I don't know of
  that many examples!}

\section{One-Way Communication Complexity}\label{s:1way}

We next describe a simple and clean formalism that is extremely useful
for proving lower bounds on the space required by streaming algorithms
to perform various tasks.
The model will be a quite restricted form of the general communication
model that we study later --- and this is good for us, because the
restriction makes it easier to prove lower bounds.  Happily,
even lower bounds for this restricted model typically translate to
lower bounds for streaming algorithms.

In general, communication complexity is a sweet spot.  It is a general
enough concept to capture the essential hardness lurking in many
different models of computation, as we'll see throughout the course.
At the same time, it is possible to prove numerous different lower
bounds in the model --- some of these require a lot of work, but many
of the most important ones are easier that you might have
guessed.  These lower bounds are ``unconditional'' --- they are simply
true, and don't depend on any unproven (if widely believed)
conjectures like $P \neq NP$.  Finally, because the 
model is so clean and free of distractions, it naturally guides one
toward the development of the ``right'' mathematical techniques needed
for proving new lower bounds. 

In (two-party) communication complexity, there are two parties, Alice
and Bob.  Alice has an input $\bfx \in \{0,1\}^a$, Bob an input $\bfy
\in \{0,1\}^b$.  Neither one has any idea what the other's input is.
Alice and Bob want to cooperate to compute a Boolean function (i.e., a
predicate) $f:\{0,1\}^a
\times \{0,1\}^b \rightarrow \{0,1\}$ that is defined on their joint input.
We'll discuss several examples of such functions shortly.

For this lecture and the next, we can get away with restricting
attention to {\em one-way communication protocols}.  All that is
allowed here is the following:
\begin{enumerate}

\item Alice sends Bob a message $\bfz$, which is a function of her
  input $\bfx$
  only.

\item Bob declares the output $f(\bfx,\bfy)$, as a function of Alice's
  message $\bfz$ and his input $\bfy$ only.

\end{enumerate}
Since we're interested in both deterministic and randomized
algorithms, we'll discuss both deterministic and randomized one-way
communication protocols.

The {\em one-way communication complexity} of a Boolean function $f$
is the minimum worst-case number of bits used by any one-way protocol
that correctly decides the function.  (Or for randomized protocols,
that correctly decides it with probability at least $2/3$.)
That is, it is
\[
\min_{\P} \max_{\bfx,\bfy} \{ \text{length (in bits) of Alice's message $\bfz$
when Alice's input is $\bfx$} \},
\]
where the minimum ranges over all correct protocols.

Note that the one-way communication complexity of a function $f$ is
always at most $a$, since Alice can just send her entire $a$-bit input
$\bfx$ to Bob, who can then certainly correctly compute~$f$.  The
question is to understand when one can do better.  This will depend on
the specific function~$f$.  For example, if $f$ is the
parity function (i.e., decide whether the total number of 1s in
$(\bfx,\bfy)$ is even or odd), then the one-way communication complexity
of $f$ is~1 (Alice just sends the parity of $\bfx$ to Bob, who's then in
a position to figure out the parity of $(\bfx,\bfy)$).

\section{Connection to Streaming Algorithms}

If you care about streaming algorithms, then you should also care
about one-way communication complexity.  Why?  Because of the
unreasonable effectiveness of the following two-step plan to proving
lower bounds on the space usage of streaming algorithms.
\begin{enumerate}

\item Small-space streaming algorithms imply low-communication one-way
  protocols.

\item The latter don't exist.

\end{enumerate}
Both steps of this plan are quite doable in many cases.  

Does the connection in the first step above surprise you?  It's the
best kind of statement --- genius and near-trivial at the same time.
We'll be formal about it shortly, but it's worth remembering a cartoon
meta-version of the connection, illustrated in
Figure~\ref{f:streamlb}.  
Consider a problem that can be solved using a streaming algorithm $S$
that uses space only $s$.  How can we use it to define a
low-communication protocol?  The idea is for Alice and Bob to treat
their inputs as a stream $(\bfx,\bfy)$, with all of $\bfx$ arriving
before all of $\bfy$.  Alice can feed $\bfx$ into $S$ without
communicating with Bob (she knows $\bfx$ and $S$).  After processing
$\bfx$, $S$'s state is completely summarized by the $s$ bits in its
memory.  Alice sends these bits to Bob.  Bob can then simply restart
the streaming algorithm $S$ seeded with this initial memory, and then
feed his input $\bfy$ to the algorithm.  The algorithm $S$ winds up
computing some function of $(\bfx,\bfy)$, and Alice only needs to
communicate $s$ bits to Bob to make it happen.
The communication cost of the induced protocol is
exactly the same as the space used by the streaming algorithm.

\begin{figure}
\centering
\includegraphics[width=.8\textwidth]{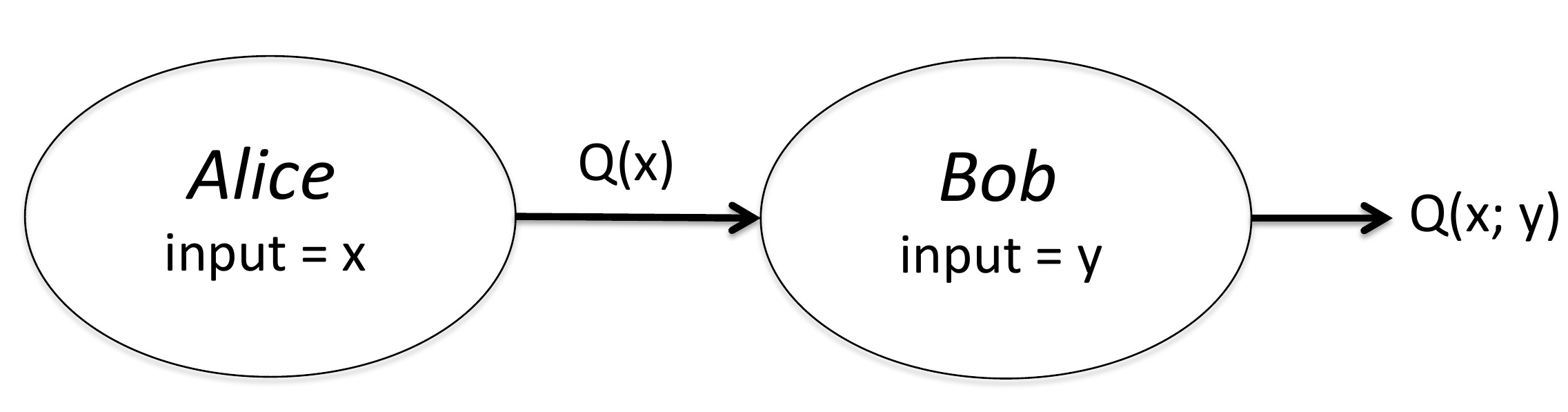}
\caption[A small-space streaming algorithm induces a
low-communication one-way protocol]{Why a small-space streaming algorithm induces a
low-communication one-way protocol.  Alice runs the streaming
algorithm on her input, sends the memory contents of the algorithm to
Bob, and Bob resumes the execution of the algorithm where Alice left
off on his input.}\label{f:streamlb}
\end{figure}

\section{The Disjointness Problem}

To execute the two-step plan above to prove lower bounds
on the space usage of streaming algorithms, we need to come up with a
Boolean function that (i) can be reduced to a streaming problem that
we care about and (ii) does not admit a low-communication one-way
protocol.

\subsection{Disjointness Is Hard for One-Way Communication}

If you only remember one problem that is hard for communication protocols,
it should be the \disj problem.  This is the
canonical hard problem in communication complexity, analogous
to satisfiability (SAT) in the theory of
$NP$-completeness.  We'll see more reductions from the
\disj problem than from any other in this course.

In an instance of \disj, both Alice and Bob hold $n$-bit
vectors $\bfx$ and $\bfy$.  We interpret these as characteristic
vectors of two subsets of the universe $\{1,2,\ldots,n\}$, with the
subsets corresponding to the ``1'' coordinates.  We then define
the Boolean function $DISJ$ in the obvious way, with $DISJ(\bfx,\bfy)
= 0$ if there is an index $i \in \{1,2,\ldots,n\}$ with $x_i = y_i =
1$, and $DISJ(\bfx,\bfy) = 1$ otherwise.

To warm up, let's start with an easy result.
\begin{proposition}\label{prop:disj_det}
Every deterministic one-way communication protocol that computes the
function $DISJ$ uses at least $n$ bits of communication in the worst
case.
\end{proposition}
That is, the trivial protocol is optimal among deterministic
protocols.\footnote{We'll see later that the 
  communication complexity remains $n$ even when we allow general
  communication protocols.}  The proof follows pretty
straightforwardly from the Pigeonhole Principle --- you might want to
think it through before reading the proof below.

Formally, consider any one-way communication protocol where Alice always
sends at most $n-1$ bits.  This means that, ranging over the $2^n$
possible inputs $\bfx$ that Alice might have, she only sends $2^{n-1}$
distinct messages.  By the Pigeonhole Principle, there are distinct
messages $\bfx^1$ and 
$\bfx^2$ where Alice sends the same message $\bfz$ to Bob.  Poor Bob,
then, has to compute $DISJ(\bfx,\bfy)$ knowing only $\bfz$ and $\bfy$
and not knowing $\bfx$ --- $\bfx$ could be $\bfx^1$, or it could be
$\bfx^2$.  Letting $i$ denote an index in which $\bfx^1$ and $\bfx^2$
differ (there must be one), Bob is really in trouble if his input
$\bfy$ happens to be the $i$th basis vector (all zeroes except
$y_i=1$).  For then, whatever Bob says upon receiving the message
$\bfz$, he will be wrong for exactly one of the cases $\bfx = \bfx^1$
or $\bfx = \bfx^2$.  We conclude that the protocol is not correct.

A stronger, and more useful, lower bound also holds.
\begin{theorem}\label{t:disj}
Every randomized one-way protocol\footnote{There are different flavors of
  randomized protocols, such as ``public-coin'' vs.\ ``private-coin''
  versions.  These distinctions won't matter until next lecture, and
  we elaborate on them then.}
that, for every input $(\bfx,\bfy)$, correctly decides the
function $DISJ$ with probability at least $\tfrac{2}{3}$,
uses $\Omega(n)$ communication in the worst case.
\end{theorem}
The probability in Theorem~\ref{t:disj} is over the coin flips
performed by the protocol (there is no randomness in the input, which
is ``worst-case'').
There's nothing special about the constant $\tfrac{2}{3}$ in the statement
of Theorem~\ref{t:disj} --- it can be replaced by any constant strictly
larger than $\tfrac{1}{2}$.

Theorem~\ref{t:disj} is certainly harder to prove than
Proposition~\ref{prop:disj_det}, but it's not too bad --- we'll kick
off next lecture with a proof.\footnote{A more difficult and important
  result 
  is that the communication complexity of \disj remains $\Omega(n)$ even if we
  allow arbitrary (not necessarily one-way) randomized protocols.
We'll use this stronger result several times later in the course.  We'll
also briefly discuss the proof in Section~\ref{ss:disj_proof} of
Lecture~\ref{cha:boot-camp-comm}.}
For the rest of this lecture, we'll take
Theorem~\ref{t:disj} on faith and use it to derive lower bounds on the
space needed by streaming algorithms.

\subsection[Space Lower Bound for $F_{\infty}$]{Space Lower Bound for $F_{\infty}$ (even with
  Randomization and Approximation)}

Recall from Section~\ref{s:better} that the first weakness of
Theorem~\ref{t:ams} is that it applies only to $F_0$ and $F_2$ (and
$F_1$ is easy).  The next result shows that, assuming
Theorem~\ref{t:disj}, there is no sublinear-space algorithm for
computing $F_{\infty}$, even probabilistically and approximately.
\begin{theorem}[\citealt{AMS96}]\label{t:infty}
Every randomized streaming algorithm that, for every data stream of
length $m$, 
computes $F_{\infty}$ to within a $(1 \pm .2)$ factor with
probability at least $2/3$ uses space $\Omega(\min\{m,n\})$.
\end{theorem}

Theorem~\ref{t:infty} rules out, in a strong sense, extending our
upper bounds for $F_0,F_1,F_2$ to all $F_k$. 
Thus, the different frequency moments vary widely in tractability
in the streaming model.\footnote{For finite
  $k$ strictly 
larger than 2, the optimal space of a randomized $(1 \pm
\eps)$-approximate streaming algorithm turns out to be roughly
$\Theta(n^{1-1/2k})$~\citep{B+02b,CKS03,IW05}.  See the exercises for a bit more
about these problems.}

\vspace{.1in}
\noindent
\begin{prevproof}{Theorem}{t:infty}
The proof simply implements the cartoon in Figure~\ref{f:streamlb},
with the problems of computing $F_{\infty}$ (in the streaming model)
and \disj (in the one-way communication model).  In more
detail, let $S$ be a space-$s$ streaming algorithm that for every
data stream, with probability at least $2/3$, outputs an estimate in
$(1 \pm .2)F_{\infty}$.  Now consider the following one-way
communication protocol $\P$ for solving the \disj problem
(given an input $\joint$):
\begin{enumerate}

\item Alice feeds into $S$ the indices~$i$ for which $x_i = 1$; the
  order can be arbitrary.  Since Alice knows $S$ and $\bfx$, this step
  requires no communication.

\item Alice sends $S$'s current memory state $\sigma$ to Bob.  Since $S$ uses
  space $s$, this can be communicated using $s$ bits.

\item Bob resumes the streaming algorithm $S$ with the memory state
  $\sigma$, and feeds into $S$ the indices $i$ for which $y_i = 1$ (in
  arbitrary order).

\item Bob declares ``disjoint'' if and only if $S$'s final answer is
  at most $4/3$.

\end{enumerate}

To analyze this reduction, observe that the frequency of an index $i
\in \{1,2,\ldots,n\}$ in the data stream induced by $\joint$ is 0 if
$x_i=y_i=0$, 1 if exactly one of $x_i,y_i$ is 1, and 2 if
$x_i=y_i=2$.  Thus, $F_{\infty}$ of this data stream is 2 if $\joint$
is a ``no'' instance of Disjointness, and is at most 1 otherwise.
By assumption, for every ``yes'' (respectively, ``no'') input
$\joint$, with probability at least $2/3$ the algorithm $S$ outputs an
estimate that is at most 1.2 (respectively, at least 2/1.2); in this
case, the protocol $\P$ correctly decides the input $\joint$.  
Since $\P$ is a one-way protocol using $s$ bits of communication, 
Theorem~\ref{t:disj} implies that $s = \Omega(n)$.  Since the data
stream length $m$ is $n$, this reduction also rules out $o(m)$-space
streaming algorithms for the problem.
\end{prevproof}

\begin{remark}[The Heavy Hitters Problem]
Theorem~\ref{t:infty} implies that computing the maximum frequency is a
hard problem in the streaming model, at least for worst-case inputs.
As mentioned, the problem is nevertheless practically quite important,
so it's important to make progress on it despite this lower bound.
For example, consider the
following relaxed version, known as the ``heavy
hitters'' problem: for a parameter $k$, if there are any elements with
frequency bigger than $m/k$, then find one or all such elements.  When
$k$ is constant, there are good solutions to this problem: the
exercises outline the ``Mishra-Gries'' algorithm, and the ``Count-Min
Sketch'' and its variants also give good
solutions~\citep{CCF04,CM05}.\footnote{This does not contradict
  Theorem~\ref{t:disj} --- in the 
hard instances of $F_{\infty}$ produced by that proof, all frequencies
are in $\{0,1,2\}$ and hence there are no heavy hitters.}
The heavy hitters problem captures many of the applications that
motivated the problem of computing $F_{\infty}$.
\end{remark}

\subsection[Space Lower Bound for Exact Computation of
$F_0$ and $F_2$]{Space Lower Bound for Randomized Exact Computation of
$F_0$ and $F_2$}

In Section~\ref{s:better} we also criticized our
positive results for $F_0$ and $F_2$ --- to achieve them, we had to
make two compromises, allowing approximation and a non-zero
chance of failure.  The reduction in the proof of
Theorem~\ref{t:infty} also implies that merely allowing randomization
is not enough.

\begin{theorem}[\citealt{AMS96}]\label{t:rand}
For every non-negative integer $k \neq 1$,
every randomized streaming algorithm that, for every data stream,
computes $F_{\infty}$ exactly with probability at least $2/3$ uses
space $\Omega(\min\{n,m\})$. 
\end{theorem}

The proof of Theorem~\ref{t:rand} is almost identical to that of
Theorem~\ref{t:disj}.  The reason 
the proof of Theorem~\ref{t:disj} rules out
approximation (even with randomization) is because 
$F_{\infty}$ differs by a factor of~2 in the two different cases
(``yes'' and ``no'' instances of \disj).  
For finite~$k$, the correct value of $F_k$ will be at least slightly
different in the two 
cases, which is enough to rule out a randomized algorithm
that is exact at least two-thirds of the time.\footnote{Actually,
  this is not quite true (why?).  But if Bob also knows the number of
  1's in Alice's input (which Alice can communicate in $\log_2 n$
  bits, a drop in the bucket), then the exact computation of $F_k$
  allows Bob to distinguish ``yes'' and ``no'' inputs of \disj
  (for any $k \neq 1$).}

The upshot of Theorem~\ref{t:rand} is that, even for $F_0$ and $F_2$,
approximation is essential to obtain a sublinear-space algorithm.
It turns out that randomization is also essential --- every
deterministic streaming algorithm that always outputs a $(1 \pm
\eps)$-estimate of $F_k$ (for any $k \neq 1$) uses linear
space~\cite{AMS96}.  
The argument is not overly difficult --- 
see the Exercises for the details.

\section{Looking Backward and Forward}

Assuming that randomized one-way communication protocols require
$\Omega(n)$ communication to solve the \disj problem
(Theorem~\ref{t:disj}), we proved 
that some frequency moments (in particular, $F_{\infty}$) cannot be
computed in sublinear space, even allowing randomization and
approximation.
Also, both randomization and approximation are
essential for our sublinear-space streaming algorithms for $F_0$ and
$F_2$. 

The next action items are:
\begin{enumerate}

\item Prove Theorem~\ref{t:disj}.

\item Revisit the five compromises we made to obtain positive results
  (Section~\ref{s:better}).  We've showed senses in which the first
  three compromises are necessary.  Next lecture we'll see why the
  last two are needed, as well.

\end{enumerate}

\chapter[Lower Bounds for One-Way Communication]{Lower Bounds for One-Way Communication: Disjointness, Index,                    
and Gap-Hamming}
\label{cha:lower-bounds-one}

\section{The Story So Far}

Recall from last lecture the simple but useful model of one-way
communication complexity.  Alice has an input $\bfx \in \{0,1\}^a$,
Bob has an input $\bfy \in \{0,1\}^b$, and the goal is to compute a
Boolean function $f:\{0,1\}^a \times \{0,1\}^b \rightarrow \{0,1\}$ of the joint input
$\inputs$.  The players communicate as in Figure~\ref{f:1way}: Alice
sends a message $\bfz$ to Bob as a function of $\bfx$ only (she
doesn't know Bob's input $\bfy$), and Bob has to decide the
function~$f$ knowing only $\bfz$ and $\bfy$ (he doesn't know Alice's
input $\bfx$).  The {\em one-way communication complexity} of $f$ is
the smallest number of bits communicated (in the worst case over
$\inputs$) of any protocol that computes $f$.  We'll sometimes consider
deterministic protocols but are interested mostly in randomized
protocols, which we'll define more formally shortly.

\begin{figure}
\centering
\includegraphics[width=.8\textwidth]{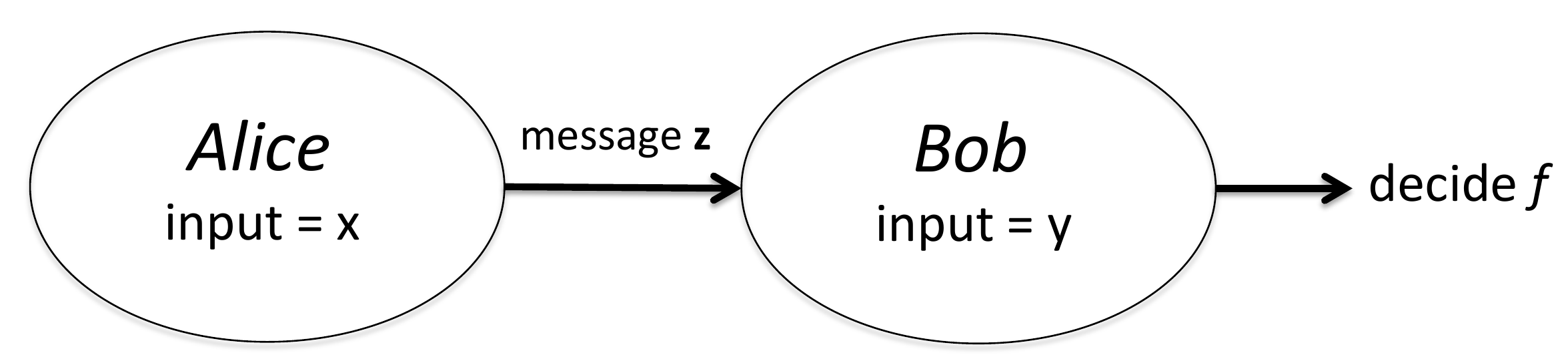}
\caption[A one-way communication protocol]{A one-way communication protocol.  Alice sends a message to
  Bob that depends only on her input; Bob makes a decision based on
  his input and Alice's message.}\label{f:1way}
\end{figure}

We motivated the one-way communication model through applications to
streaming algorithms.  Recall the data stream model, where a data
stream $x_1,\ldots,x_m \in U$ of elements from a universe of $n = |U|$
elements arrive one by one.  The assumption is that there is insufficient
space to store all of the data, but we'd still like to compute useful
statistics of it via a one-pass computation.  Last lecture, we showed
that very cool and non-trivial positive results are possible in this
model.  We presented a slick and low-space ($O(\eps^{-2} (\log n +
\log m) \log \tfrac{1}{\delta})$) streaming algorithm that, with
probability at least $1-\delta$, computes a $(1 \pm
\eps)$-approximation of $F_2 = \sum_{j \in U} f^2_j$, the skew of the
data.  (Recall that $f_j \in \{0,1,2,\ldots,m\}$ is the number of
times that $j$ appears in the stream.)   We also mentioned the main
idea (details in the homework) for an analogous low-space streaming
algorithm that estimates $F_0$, the number of distinct elements in a
data stream.

Low-space streaming algorithms $S$ induce low-communication one-way
protocols $P$, with the communication used by $P$ equal to the space
used by $S$.  Such reductions typically have the following form.
Alice converts her input $\bfx$ to a data stream and feeds it into
the assumed space-$s$ streaming algorithm $S$.  She then sends the
memory of $S$ (after processing $\bfx$) to Bob; this requires only $s$
bits of communication.  Bob then resumes $S$'s execution at the
point that Alice left off, and feeds a suitable representation of his
input $\bfy$ into $S$.  When $S$ terminates, it has computed some kind
of useful function of $\inputs$ with only $s$ bits of communication.
The point is that lower bounds for one-way communication protocols ---
which, as we'll see, we can actually prove in many cases --- imply
lower bounds on the space needed by streaming algorithms.

Last lecture we used without proof the following
result (Theorem~\ref{t:disj}).\footnote{Though we did prove it for the special case of
  deterministic protocols, using a simple Pigeonhole Principle argument.}
\begin{theorem}
The one-way communication complexity of the \disj problem is
$\Omega(n)$, even for randomized protocols.
\end{theorem}
We'll be more precise about the randomized protocols that we consider
in the next section.
Recall that an input of \disj is defined by
$\bfx,\bfy \in \{0,1\}^n$, which we view as characteristic vectors
of two subsets of $\{1,2,\ldots,n\}$, and the output should be ``0''
is there is an index $i$ with $x_i = y_i = 1$ and ``1'' otherwise.

We used Theorem~\ref{t:disj} to prove a few lower bounds on the space
required by streaming algorithms.  A simple reduction showed that
every streaming algorithm that computes $F_{\infty}$, the maximum
  frequency, even approximately and with probability~$2/3$, needs
  linear (i.e., $\Omega(\min\{n,m\})$) space.  This is in sharp
  contrast to our algorithms for approximating $F_0$ and $F_2$, which
  required only logarithmic space.
The same reduction proves that, for $F_0$ and $F_2$, exact computation
requires linear space, even if randomization is allowed.
A different simple argument (see the homework) shows that
randomization is also essential for our positive results: every
deterministic streaming algorithm that approximates $F_0$ or $F_2$ up
to a small constant factor requires linear space.

In today's lecture we'll prove Theorem~\ref{t:disj}, introduce and
prove lower bounds for a couple of other problems that are hard for
one-way communication, and prove via reductions some further space
lower bounds for streaming algorithms.

\section{Randomized Protocols}\label{s:rand}

There are many different flavors of randomized communication protocols.
Before proving any lower bounds, we need to be crystal clear about
exactly which protocols we're talking about.  The good news is that,
for algorithmic applications, we can almost always focus on a
particular type of randomized protocols.  By default, we adopt the
following four assumptions and rules of thumb.  The common theme
behind them is we want to allow as permissible a class of randomized
protocols as possible, to maximize the strength of our lower bounds
and the consequent algorithmic applications.

\textbf{Public coins.}  First, unless otherwise noted, we consider
{\em public-coin} protocols.  This means that, before Alice and Bob ever
show up, a deity writes an infinite sequence of perfectly random bits
on a blackboard visible to both Alice and Bob.  Alice and Bob can
freely use as many of these random bits as they want --- it doesn't
contribute to the communication cost of the protocol.

The {\em private coins} model might seem
more natural to the algorithm designer --- here, Alice and Bob just
flip their own random coins 
as needed.  Coins flipped by one player are unknown to the other
player unless they are explicitly communicated.\footnote{Observe that
  the one-way communication protocols induced by streaming algorithms
  are private-coin protocols --- random coins flipped during the first
  half of the data stream are only available to the second half if
  they are explicitly stored in memory.}
Note that every
private-coins protocol can be simulated with no loss by a public-coins
protocol: for example, Alice uses the shared random bits 1, 3, 5,
etc.\ as needed, while Bob used the random bits 2, 4, 6, etc.

It turns out that while public-coin protocols are strictly more powerful
than private-coin protocols, 
for the purposes of this course, the two models have essentially the
same behavior.  In any case, our lower bounds will generally apply to
public-coin (and hence also private-coin) protocols.

A second convenient fact about public-coin randomized protocols is
that they are equivalent to distributions over deterministic
protocols.  Once the random bits on the blackboard have been fixed,
the protocol proceeds deterministically.  Conversely, every
distribution over deterministic protocols (with rational
probabilities) can be implemented via a public-coin protocol --- just use
the public coins to choose from the distribution.

\textbf{Two-sided error.} We consider randomized algorithms that are
allowed to error with some probability on every input $\inputs$,
whether $f\inputs = 0$ or $f\inputs = 1$.  A stronger requirement would
be one-sided error --- here there are two flavors, one that forbids
false positives (but allows false negatives) and one the forbids false
negatives (but allows false positives).  Clearly, lower bounds that
apply to protocols with two-sided error are at least as strong as
those for protocols with one-sided error --- indeed, the latter lower
bounds are often much easier to prove (at least for one of the two
sides).   Note that the one-way
protocols induces by the streaming algorithms in the last lecture are
randomized protocols with two-sided error.  There are other problems
for which the natural randomized solutions have only
one-sided error.\footnote{One can also consider ``zero-error''
  randomized protocols, which always output the correct answer but use
  a random amount of communication.  We won't need to discuss such
  protocols in this course.}

\textbf{Arbitrary constant error probability.}  A simple but important
fact is that all constant error probabilities $\eps \in
(0,\tfrac{1}{2})$ yield the same communication complexity, up to a
constant factor.  The reason is simple: the success probability of a
protocol can be boosted through amplification (i.e., repeated
trials).\footnote{We mentioned a similar ``median of means'' idea
last lecture
(developed further in the homework), when we discussed how to reduce
  the $\tfrac{1}{\delta}$ factor in the space usage of our streaming
  algorithms to a factor of$\log \tfrac{1}{\delta}$.}
In more detail, suppose $P$ uses $k$ bits on communication and has
success at least 51\% on every input.  Imagine repeating $P$ 10000
times.  To preserve one-way-ness of the protocol, all of the repeated
trials need to happen in parallel, with the public coins providing the
necessary 10000 independent random strings.  Alice sends 10000
messages to Bob, Bob imagines answering each one --- some answers will
be ``1,'' others ``0'' --- and concludes by reporting the majority
vote of the 10000 answers.  In expectation 5100 of the trials give
the correct answer, and the probability that more than 5000
of them are correct is big (at least 90\%, say).  In general, a
constant number of trials, followed by a majority vote, 
boosts the success probability of a protocol from any constant
bigger than $\tfrac{1}{2}$ to any other constant less than $1$.  These
repeated trials increase the amount of communication by only a
constant factor.  See the exercises and the separate notes on Chernoff
bounds for further details.

This argument justifies being sloppy about the exact (constant) error
of a two-sided protocol.  For upper bounds, we'll be content to achieve
error 49\% --- it can be reduced to an arbitrarily small constant
with a constant blow-up in communication.  For lower bounds, we'll be
content to rule out protocols with error \%1 --- the same
communication lower bounds hold, modulo a constant factor, even for
protocols with error 49\%.

\textbf{Worst-case communication.}  When we speak of the communication
used by a randomized protocol, we take the worst case over inputs
$\inputs$ {\em and} over the coin flips of the protocol.  So if a
protocol uses communication at most $k$, then Alice always sends at
most $k$ bits to Bob.

This definition seems to go against our
guiding rule of being as permissive as possible.  Why not measure only
the expected communication used by a protocol, with respect to its
coin flips?  This objection is conceptually justified but technically
moot --- for protocols that can err, passing to the technically more
convenient worst-case measure can only increase the communication
complexity of a problem by a constant factor.

To see this, consider a protocol $R$ that, for every input $\inputs$,
has two-sided error at most $1/3$ (say) and uses at most $k$ bits of
communication on average over its coin flips.  
This
protocol uses at most $10k$ bits of communication at least 90\% of the
time --- if it used more than $10k$ bits more than 10\% of the time,
its expected communication cost would be more than $k$.  Now consider
the following protocol $R'$: simulate $R$ for up to $10k$ steps; if
$R$ fails to terminate, then abort and output an arbitrary answer.
The protocol $R'$ always sends at most $10k$ bits of communication and
has error at most that of $R$, plus 10\% (here, $\approx 43\%$).  This
error probability of $R'$ can be reduced back down (to
$\tfrac{1}{3}$, or whatever) through repeated trials, as before.

In light of these four standing assumptions and rules, we can restate
Theorem~\ref{t:disj} as follows.
\begin{theorem}\label{t:disj2}
Every public-coin randomized one-way protocol for \disj
that has two-sided error at most a constant $\eps \in (0,\tfrac{1}{2})$
uses $\Omega(\min\{n,m\})$ communication in the worst case (over
inputs and coin flips).
\end{theorem}

Now that we are clear on the formal statement of our lower bound, how
do we prove it?

\section{Distributional Complexity}

Randomized protocols are much more of a pain to reason about than
deterministic protocols.  For example, recall 
our Pigeonhole Principle-based argument last lecture for deterministic
protocols: if Alice holds an $n$-bit input and always sends at most
$n-1$ bits, then there are distinct inputs $\bfx,\bfx'$ such that
Alice sends the same message $\bfz$.  (For \disj,
this ambiguity left Bob in a lurch.)  In a randomized protocol where
Alice always sends at most $n-1$ bits, Alice can use a different
distribution over $(n-1)$-bit messages for each of her $2^n$ inputs
$\bfx$, and the naive argument breaks down.  While Pigeonhole
Proof-type arguments can sometimes be pushed through for randomized
protocols, this section introduces a different approach.

{\em Distributional complexity} is the main methodology by which one
proves lower bounds on the communication complexity of randomized
algorithms.  The point is to reduce the goal to proving lower
bounds for {\em deterministic protocols only}, with respect to a
suitably chosen input distribution.

\begin{lemma}[\citealt{Y83}]\label{l:yao}
Let $D$ be a distribution over the space of inputs $\inputs$ to a
communication problem, and $\eps \in (0,\tfrac{1}{2})$.  Suppose that
every deterministic one-way protocol $P$ with
\[
\prob[\inputs \sim D]{\text{$P$ wrong on $\inputs$}} \le \eps
\]
has communication cost at least $k$.  Then every (public-coin)
randomized one-way protocol $R$ with (two-sided) error at most $\eps$ on every
input has communication cost at least $k$.
\end{lemma}

In the hypothesis of Lemma~\ref{l:yao}, all of the randomness is in
the input --- $P$ is deterministic, $\inputs$ is random.  In the
conclusion, all of the randomness is in the protocol $R$ --- the input
is arbitrary but fixed, while the protocol can flip coins.
Not only is Lemma~\ref{l:yao} extremely useful, but it is easy to
prove.

\vspace{.1in}
\noindent
\begin{prevproof}{Lemma}{l:yao}
Let $R$ be a randomized protocol with communication cost less than
$k$.  Recall that such an $R$ can be written as a distribution over
deterministic protocols, call them $P_1,P_2,\ldots,P_s$.  Recalling
that the communication cost of a randomized protocol is defined as the
worst-case communication (over both inputs and coin flips), each
deterministic protocol $P_i$ always uses less than $k$ bits of
communication.  By assumption,
\[
\prob[\inputs \sim D]{\text{$P_i$ wrong on $\inputs$}} > \eps
\]
for $i=1,2,\ldots,s$.  Averaging over the $P_i$'s, we have
\[
\prob[\inputs \sim D; R]{\text{$R$ wrong on $\inputs$}} > \eps.
\]
Since the maximum of a set of numbers is at least is average, there
exists an input $\inputs$ such that
\[
\prob[R]{\text{$R$ wrong on $\inputs$}} > \eps,
\]
which completes the proof.
\end{prevproof}

The converse of Lemma~\ref{l:yao} also holds --- whatever the true
randomized communication complexity of a problem, there exists a bad
distribution $D$ over inputs that proves it~\citep{Y83}.  The proof is
by strong linear programming duality or, equivalently, von Neumann's
Minimax Theorem for zero-sum games (see the exercises for details).
Thus, the distributional methodology is ``complete'' for proving lower
bounds --- one ``only'' needs to find the right distribution~$D$ over
inputs.  In general this is a bit of a dark art, though in today's
application $D$ will just be the uniform distribution.

\section{The \textsc{Index} Problem}

We prove Theorem~\ref{t:disj2} in two steps.  The first step is to
prove a linear lower bound on the randomized communication complexity
of a problem called \index,  which is widely useful for proving
one-way communication complexity lower bounds.
The second step, which is easy, reduces \index to \disj.

In an instance of \index, Alice gets an $n$-bit string $\bfx \in
\{0,1\}^n$ and Bob gets an integer $i \in \{1,2,\ldots,n\}$, encoded
in binary using $\approx \log_2 n$ bits.  The goal is simply to compute
$x_i$, the $i$th bit of Alice's input.

Intuitively, since Alice has no idea which of her bits Bob is
interested in, she has to send Bob her entire input.  This intuition
is easy to make precise for deterministic protocols, by a Pigeonhole
Principle argument.  The intuition also holds for randomized
protocols, but the proof takes more work.

\begin{theorem}[\citealt{KNR95}]\label{t:index} The randomized one-way
  communication   complexity of \index is $\Omega(n)$.
\end{theorem}

With a general communication protocol, where Bob can also send
information to Alice, \index is trivial to solve using only
$\approx \log_2 n$ bits of information --- Bob just sends $i$ to
Alice.  Thus \index nicely captures the difficulty of
designing non-trivial one-way communication protocols, above and
beyond the lower bounds that already apply to general protocols.

Theorem~\ref{t:index} easily implies Theorem~\ref{t:disj2}.

\vspace{.1in}
\noindent
\begin{prevproof}{Theorem}{t:disj2}
We show that \disj reduces to \index.
Given an input $(\bfx,i)$ of \index, Alice forms the input $\bfx' =
\bfx$ while Bob forms the input $\bfy' = e_i$; here $e_i$ is the
standard basis vector, with a ``1'' in the $i$th coordinate and ``0''s
in all other coordinates.  Then, $(\bfx',\bfy')$ is a ``yes'' instance
of \disj if and only if $x_i = 0$.  Thus, every one-way
protocol for \index induces one for \disj, with the same
communication cost and error probability.
\end{prevproof}

We now prove Theorem~\ref{t:index}.  While some computations are
required, the proof is conceptually pretty straightforward.

\vspace{.1in}
\noindent
\begin{prevproof}{Theorem}{t:index}
We apply the distributional complexity methodology.  This requires
positing a distribution $D$ over inputs.  Sometimes this takes
creativity.  Here, the first thing you'd try --- the uniform
distribution $D$, where $\bfx$ and $i$
are chosen independently and uniformly at random --- works.

Let $c$ be a sufficiently small constant (like .1 or less) and assume that $n$
is sufficiently large (like 300 or more).  We'll show that every
deterministic 
one-way protocol that uses at most $cn$ bits of communication has error
(w.r.t.\ $D$) at least $\tfrac{1}{8}$.  By Lemma~\ref{l:yao}, this
implies that every randomized protocol has error at least
$\tfrac{1}{8}$ on some input.  Recalling the discussion about error
probabilities in Section~\ref{s:rand}, this implies that for every
error $\eps' > 0$, there is a constant $c' > 0$ such that every
randomized protocol that uses at most $c'n$ bits of communication has
error bigger than $\eps'$.

Fix a deterministic one-way protocol $P$ that uses at most $cn$ bits
of communication.  Since $P$ is deterministic, there are only $2^{cn}$
distinct messages $\bfz$ that Alice ever sends to Bob (ranging over
the $2^n$ possible inputs $\bfx$).  We need to formalize the intuition
that Bob typically (over $\bfx$) doesn't learn very much about $\bfx$,
and hence typically (over $i$) doesn't know what $x_i$ is.

Suppose Bob gets a message $\bfz$ from Alice, and his input is $i$.
Since $P$ is deterministic, Bob has to announce a bit, ``0'' or ``1,'' as a
function of $\bfz$ and $i$ only.  (Recall Figure~\ref{f:1way}).
Holding $\bfz$ fixed and considering
Bob's answers for each of his possible inputs $i=1,2,\ldots,n$, we get an
$n$-bit vector --- Bob's {\em answer vector} $\bfaz$ when he receives
message $\bfz$ from Alice.  Since there are at most $2^{cn}$ possible
messages $\bfz$, there are at most $2^{cn}$ possible answer vectors
$\bfaz$.  

Answer vectors are a convenient way to express the error of the
protocol $P$, with respect to the randomness in Bob's input.  Fix
Alice's input $\bfx$, which results in the message $\bfz$.  The
protocol is correct if Bob holds an input $i$ with $\bfaz_i = x_i$,
and incorrect otherwise.  Since Bob's index $i$ is chosen uniformly at
random, and independently of $\bfx$, we have
\begin{equation}\label{eq:hdist}
\prob[i]{\text{$P$ is incorrect} \,|\, \bfx,\bfz} = \frac{d_H(\bfx,\bfaz)}{n},
\end{equation}
where $d_H(\bfx,\bfaz)$ denotes the Hamming distance between the
vectors $\bfx$ and $\bfaz$ (i.e., the number of coordinates in which
they differ).  Our goal is to show that, with constant probability
over the choice of $\bfx$, the expression~\eqref{eq:hdist} is bounded
below by a constant.

Let $A = \{ \bfa(\bfz(\bfx)) \,:\, \bfx \in \{0,1\}^n \}$ denote the
set of all answer vectors used by the protocol $P$.  Recall that
$|A| \le 2^{cn}$.  Call Alice's input $\bfx$ {\em good} if there
exists an answer vector $\bfa \in A$ with $d_H(\bfx,\bfa) <
\tfrac{n}{4}$, and {\em bad} otherwise.  Geometrically, you should
think of each answer vector $\bfa$ as the center of a ball of radius
$\tfrac{n}{4}$ in the Hamming cube --- the set $\{0,1\}^n$ equipped
with the Hamming metric.  See Figure~\ref{f:balls}.
The next claim states that, because there
aren't too many balls (only $2^{cn}$ for a small constant
$c$) and their radii aren't too big (only $\tfrac{n}{4}$), the union of
all of the balls is less than half of the Hamming cube.\footnote{More
  generally, the following is good intuition about the Hamming cube
  for large $n$: as you blow up a ball of radius $r$ around a point,
  the ball includes very few points until $r$ is almost equal to
  $n/2$; the ball includes roughly half the points for $r \approx
  n/2$; and for $r$ even modestly larger than $r$, the ball contains
  almost all of the points.}

\begin{figure}
\centering
\includegraphics[width=.7\textwidth]{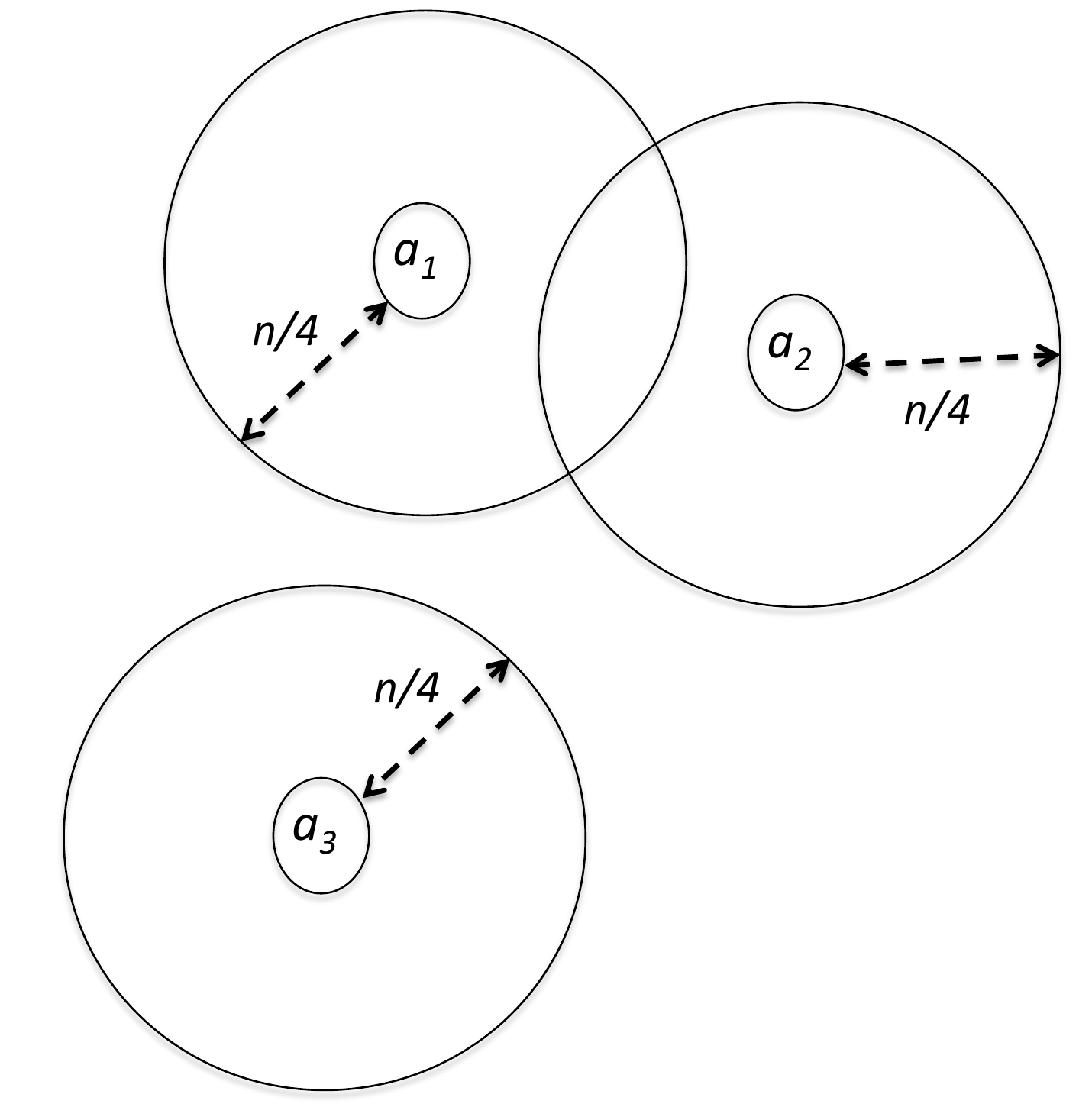}
\caption[Balls of radius $n/4$ in the Hamming metric]{Balls of radius $n/4$ in the Hamming metric, centered at the
  answer vectors used by the protocol~$P$.}\label{f:balls}
\end{figure}

\vspace{.1in}
\noindent
\textbf{Claim: } Provided $c$ is sufficiently small and $n$ is
sufficiently large, there are at least $2^{n-1}$ bad inputs $\bfx$.

\vspace{.1in}
Before proving the claim, let's see why it implies the theorem.  We
can write
\begin{eqnarray*}
\prob[\inputs \sim D]{\text{$D$ wrong on $\inputs$}}
& = &
\underbrace{\prob{\text{$\bfx$ is good}} \cdot \prob{\text{$D$ wrong
      on $\inputs$} \,|\,
  \text{$\bfx$ is good}}}_{\ge 0}\\
&& +
\underbrace{\prob{\text{$\bfx$ is bad}}}_{\ge 1/2 \text{ by Claim}} \cdot
\prob{\text{$D$ wrong on $\inputs$} \,|\,   \text{$\bfx$ is bad}}.
\end{eqnarray*}
Recalling~\eqref{eq:hdist} and the definition of a bad input $\bfx$,
we have
\begin{eqnarray*}
\prob[\inputs]{\text{$D$ wrong on $\inputs$} \,|\, \text{$\bfx$ is bad}}
& = &
\expect[\bfx]{\frac{d_H(\bfx,\bfa(\bfz(\bfx)))}{n}
   \,\left|\right.\, \text{$\bfx$ is bad}}\\
& \ge & 
\expect[\bfx]{\underbrace{\min_{\bfa \in A}
      \frac{d_H(\bfx,\bfa)}{n}}_{\ge 1/4 \text{ since $\bfx$ is bad}}
   \,\left|\right.\, \text{$\bfx$ is bad}}\\
& \ge & \frac{1}{4}.
\end{eqnarray*}
We conclude that the protocol $P$ errs on the distribution $D$ with
probability at last $\tfrac{1}{8}$, which implies the theorem.  We
conclude by proving the claim.

\vspace{.1in}
\noindent
\textbf{Proof of Claim: } 
Fix some answer vector $\bfa \in A$.  The number of inputs $\bfx$ with
Hamming distance at most $\tfrac{n}{4}$ from $\bfa$ is
\begin{equation}\label{eq:ball}
\underbrace{1}_{\bfa} + \underbrace{\binom{n}{1}}_{d_H(\bfx,\bfa) = 1} +
\underbrace{\binom{n}{2}}_{d_H(\bfx,\bfa) = 2} + \cdots + \underbrace{\binom{n}{n/4}}_{d_H(\bfx,\bfa)
  = n/2}.
\end{equation}
Recalling the inequality
\[
\binom{n}{k} \le \left( \frac{en}{k} \right)^k,
\]
which follows easily from Stirling's approximation of the factorial
function (see the exercises), we can crudely bound~\eqref{eq:ball}
above by
\[
n(4e)^{n/4} = n2^{\log_2(4e) \tfrac{n}{4}} \le n2^{.861n}.
\]
The total number of good inputs $\bfx$ --- the union of all the
balls --- is at most $|A|2^{.861n} \le 2^{(.861+c)n}$, which is at
most $2^{n-1}$ for 
$c$ sufficiently small (say .1) and $n$ sufficiently large (at least
300, say).
\end{prevproof}

\section{Where We're Going}

Theorem~\ref{t:index} completes our first approach to proving lower
bounds on the space required by streaming algorithms to compute
certain statistics.  To review, we proved from scratch that \index
is hard for one-way communication protocols
(Theorem~\ref{t:index}), reduced \index to \disj to extend the
lower bound to the latter problem (Theorem~\ref{t:disj2}), and
reduced \disj to various streaming computations 
(last lecture).  See also Figure~\ref{f:plan1}.  
Specifically, we
showed that linear space is necessary to compute the highest frequency
in a data stream ($F_{\infty}$), even when randomization and approximation
are allowed, and that linear space is necessary to compute exactly
$F_0$ or $F_2$ by a randomized streaming algorithm
with success probability~$2/3$.

\begin{figure}[h]
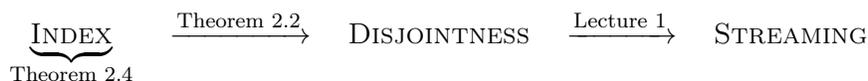

\centering
$\underbrace{\text{\index}}_{\text{Theorem~\ref{t:index}}}
\quad
\xrightarrow{\text{Theorem~\ref{t:disj2}}}
\quad
\text{\disj}
\quad
\xrightarrow{\text{Lecture~\ref{cha:data-stre-algor}}}
\quad
\text{\streaming}$
\caption[Proof structure of linear space     
lower bounds for streaming algorithms]{Review of the proof structure of linear (in $\min\{n,m\}$) space lower bounds for streaming algorithms.  Lower bounds travel from left to right.}\label{f:plan1}
\end{figure}

We next focus on the dependence on the approximation parameter $\eps$
required by a streaming algorithm to compute a $(1 \pm
\eps)$-approximation of a frequency moment.  Recall that the
streaming algorithms that we've seen for $F_0$ and $F_2$ have
quadratic dependence on $\eps^{-1}$. 
Thus an approximation of 1\% would require a blowup of 10,000 in
the space.  Obviously, it would be useful to have algorithms with 
a smaller dependence on $\eps^{-1}$.  
We next prove that space quadratic in $\eps^{-1}$ is necessary,
even allowing randomization and even for $F_0$ and $F_2$, to achieve a
$(1 \pm \eps)$-approximation.

Happily, we'll prove this via reductions, and won't need
to prove from scratch any new communication lower bounds.
We'll follow the
path in Figure~\ref{f:plan2}.  
First we introduce a new problem, also
very useful for proving lower bounds, called the \gh
problem.  
Second, we give a quite clever reduction from \index to \gh.
Finally, it is straightforward to show that one-way protocols for \gh
with sublinear communication induce streaming algorithms that can
compute a $(1 \pm \eps)$-approximation of $F_0$ or $F_2$ in
$o(\eps^{-2})$ space.


\begin{figure}[h]
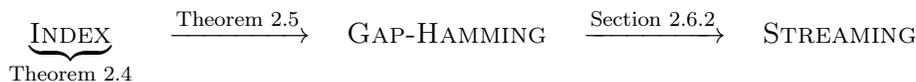

\centering
$\underbrace{\text{\index}}_{\text{Theorem~\ref{t:index}}}
\quad
\xrightarrow{\text{Theorem~\ref{t:gh}}}
\quad
\text{\gh}
\quad
\xrightarrow{\text{Section~\ref{ss:ghf0}}}
\quad
\text{\streaming}$
\caption[Proof plan for $\Omega(\eps^{-2})$ space lower bounds]{Proof plan for $\Omega(\eps^{-2})$ space lower bounds for
  (randomized) streaming algorithms that approximate $F_0$ or $F_2$ up
  to a $1 \pm
  \eps$ factor.  Lower bounds travel from left to right.}\label{f:plan2}
\end{figure}

\section{The \gh Problem}

Our current goal is to prove that every streaming algorithm that computes
a $(1 \pm \eps)$-approximation of $F_0$ or $F_2$ needs
$\Omega(\eps^{-2})$ space.  Note that we're not going to prove this
when $\eps \ll 1/\sqrt{n}$, since we can always compute a frequency
moment exactly in linear or near-linear space.  So the extreme case of
what we're trying to prove is that a $\pmsqrt$-approximation requires
$\Omega(n)$ space.  This special case already requires all of the
ideas needed to prove a lower bound of $\Omega(\eps^{-2})$ for all
larger $\eps$ as well. 

\subsection{Why \disj Doesn't Work}

Our goal is also to prove this lower bound through reductions, rather
than from scratch.  We don't know too many hard problems yet, and we'll
need a new one.  To motivate it, let's see why 
\disj is not good enough for our purposes.  

Suppose we have a streaming algorithm $S$ that gives a
$\pmsqrt$-approximation to $F_0$ --- how could we use it to solve
\disj?  The obvious idea is to follow the reduction
used last lecture for $F_{\infty}$.  Alice converts her input $\bfx$
of \disj and converts it to a stream, feeds this stream into
$S$, sends the final memory state of $S$ to Bob, 
and Bob converts his input $\bfy$ of \disj into a stream and resumes
$S$'s computation 
on it.  With healthy probability, $S$ returns a
$\pmsqrt$-approximation of $F_0$ of the stream induced by $\inputs$.
But is this good for anything?

Suppose $\inputs$ is a ``yes'' instance to Disjointness.  Then, $F_0$
of the corresponding stream is $|\bfx|+|\bfy|$, where $|\cdot|$
denotes the number of 1's in a bit vector.  If $\inputs$
is a ``no'' instance of Disjointness, then $F_0$ is somewhere between
$\max\{|\bfx|,|\bfy|\}$ and $|\bfx|+|\bfy|-1$.  A particularly hard
case is when $|\bfx|=|\bfy|=n/2$ and $\bfx,\bfy$ are either disjoint
or overlap in exactly one element --- $F_0$ is then either $n$ or $n-1$.
In this case, a $\pmsqrt$-approximation of $F_0$ translates to
additive error $\sqrt{n}$, which is nowhere near enough resolution to
distinguish between ``yes'' and ``no'' instances of \disj.

\subsection{Reducing \gh to $F_0$ Estimation}\label{ss:ghf0}

A $\pmsqrt$-approximation of $F_0$ is insufficient to
solve \disj --- but perhaps there is some other
hard problem that it does solve?  
The answer is yes, and the problem
is estimating the Hamming distance between two vectors $\bfx,\bfy$ ---
the number of coordinates in which $\bfx,\bfy$ differ.

To see the connection between $F_0$ and Hamming distance, consider
$\bfx,\bfy \in \{0,1\}^n$ and the usual data stream (with elements in
$U=\{1,2,\ldots,n\}$) induced by them.  As usual, we can interpret
$\bfx,\bfy$ as characteristic vectors of subsets $A,B$ of $U$
(Figure~\ref{f:venn}).  Observe that the Hamming distance $d_H\inputs$
is the just the size of the symmetric difference,
$|A \sm B| + |B \sm A|$.  Observe also that $F_0 = |A \cup B|$,
so $|A \sm B| = F_0 - |B|$ and $|B \sm A| = F_0 - |A|$, and hence
$d_H\inputs = 2F_0 - |\bfx|-|\bfy|$.  Finally, Bob knows $|\bfy|$, and
Alice can send $|\bfx|$ to Bob using $\log_2 n$ bits.

\begin{figure}
\centering
\includegraphics[width=.6\textwidth]{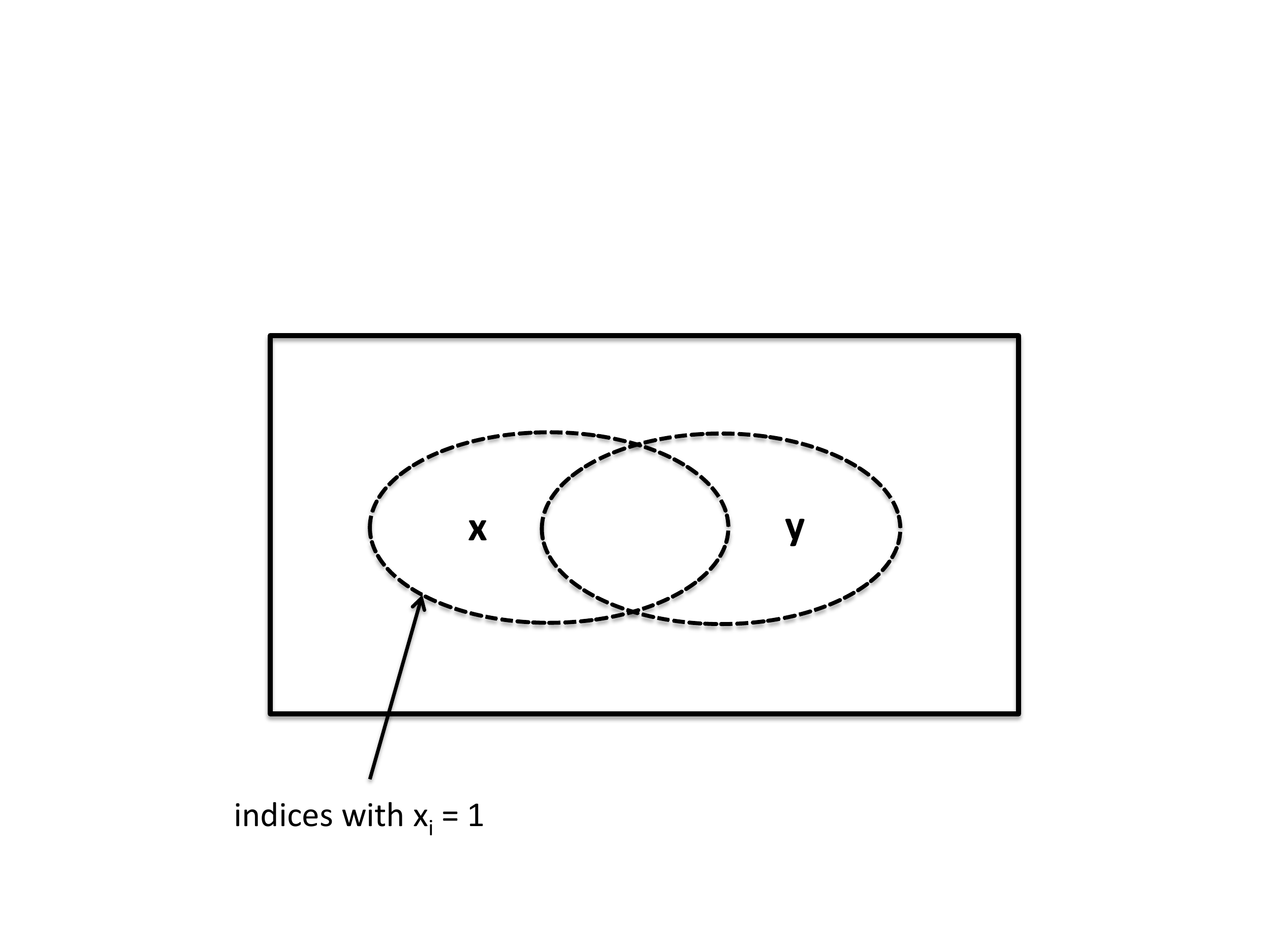}
\caption[Hamming distance and symmetric difference]{The Hamming distance between two bit vectors equals the size
  of the symmetric difference of the corresponding subsets of
  1-coordinates.}\label{f:venn}
\end{figure}

The point is that a one-way protocol that computes $F_0$ with
communication $c$ yields a one-way protocol that computes $d_H\inputs$
with communication $c + \log_2 n$.  More generally, a
$\pmsqrt$-approximation of $F_0$ yields a protocol that estimates
$d_H\inputs$ up to $2F_0/\sqrt{n} \le 2\sqrt{n}$ additive error, with
$\log_2 n$ extra communication.

This reduction from Hamming distance estimation to $F_0$ estimation is
only useful to us if the former problem has large communication
complexity.  It's technically convenient to convert Hamming distance
estimation into a decision problem.  We do this using a ``promise
problem'' --- intuitively, a problem where we only care about a
protocol's correctness when the input satisfies some conditions (a
``promise'').  Formally, for a parameter $t$,
we say that a protocol correctly solves \ght
if it outputs ``1'' whenever $d_H\inputs < t - c\sqrt{n}$ and
outputs ``0'' whenever $d_H\inputs > t + c\sqrt{n}$, where $c$ is a
sufficiently small constant.
Note that the
protocol can output whatever it wants, without penalty, on inputs for
which $d_H\inputs = t \pm c\sqrt{n}$.

Our reduction above shows that, for every $t$, \ght
reduces to the $(1 \pm \tfrac{c}{\sqrt{n}})$-approximation of
$F_0$.  Does it 
matter how we pick $t$?  Remember we still need to prove that the
\ght problem does not admit low-communication one-way
protocols.  If we pick $t=0$, then the problem becomes a special case
of the \eq problem (where $f\inputs = 1$ if and only $\bfx
= \bfy$).  We'll see next lecture that the one-way randomized
communication complexity of \eq
is shockingly low --- only $O(1)$ for
public-coin protocols.  Picking $t=n$ has the same issue.  Picking $t
= \tfrac{n}{2}$ seems more promising.  For example, it's easy to
certify a ``no'' instance of \eq --- just exhibit an
index where $\bfx$ and $\bfy$ differ.  How would you succinctly certify
that $d_H\inputs$ is either at least $\tfrac{n}{2}+\sqrt{n}$ or at most
$\tfrac{n}{2} - \sqrt{n}$?  For more intuition, think about two
vectors $\bfx,\bfy \in \{0,1\}^n$ chosen uniformly at random.  The
expected Hamming distance between them is $\frac{n}{2}$, with a
standard deviation of $\approx \sqrt{n}$.  Thus deciding an instance
of \textsc{Gap-Hamming}($\tfrac{n}{2}$)
has the flavor of learning an unpredictable fact
about two random strings, and it seems difficult to do this without
learning detailed information about the particular strings at hand.

\section[Lower Bound for \gh]]{Lower Bound on the One-Way
  Communication Complexity of \gh}\label{s:ghlb}

This section dispenses with the hand-waving and formally proves that
every protocol that solves \gh --- with
$t=\tfrac{n}{2}$ and $c$ sufficiently small --- requires linear
communication.   
\begin{theorem}[\citealt{JKS08,W04,W07}]\label{t:gh} The randomized one-way
  communication   complexity of \gh is
  $\Omega(n)$.
\end{theorem}

\begin{proof}
The proof is a randomized reduction from \index, and is more clever than the
other reductions that we've seen so far.  Consider an input to 
\index, where Alice holds an $n$-bit string $\bfx$ and Bob holds an
index $i \in \{1,2,\ldots,n\}$.  We assume, without loss of
generality, that $n$ is odd and sufficiently large.

Alice and Bob generate, without any communication, an input
$\inputsp$ to \gh.  They do this one bit at a time, using
the publicly available randomness.  
To generate the first bit of the \gh input, Alice and Bob
interpret the first $n$ public coins as a random string $\bfr$.  Bob
forms the bit $b = r_i$, the $i$th bit of the random string.
Intuitively, Bob says ``I'm going to pretend that $\bfr$ is actually
Alice's input, and report the corresponding answer $r_i$.''
Meanwhile, Alice checks whether $d_H(\bfx,\bfr) < \tfrac{n}{2}$ or
$d_H(\bfx,\bfr) > \tfrac{n}{2}$.  (Since $n$ is odd, one of these
holds.)  In the former case, Alice forms the bit $a=1$ to indicate
that $\bfr$ is a decent proxy for her input $\bfx$.  Otherwise, she
forms the bit $a=0$ to indicate that $1-\bfr$ would have been a better
approximation of reality (i.e., of $\bfx$).

The key and clever point of the proof is that $a$ and $b$ are
correlated --- positively if $x_i = 1$ and negatively if $x_i = 0$,
where $\bfx$ and $i$ are the given input to \index.  To see this,
condition on the $n-1$ bits of $\bfr$ other than $i$.  There are two
cases.  In the first case, $\bfx$ and $\bfr$ agree on strictly less
than or strictly greater than $(n-1)/2$ of the bits so-far.  In this
case, $a$ is already determined (to 0 or 1, respectively).  Thus, in
this case, $\prob{a=b} = \prob{a=r_i} = \tfrac{1}{2}$, using that
$r_i$ is independent of all the other bits.  In the second case,
amongst the $n-1$ bits of $\bfr$ other than $r_i$, exactly half of
them agree with $\bfx$.  In this case, $a=1$ if and only if $x_i =
r_i$.  Hence, if $x_i = 1$, then $a$ and $b$ always agree (if $r_i =
1$ then $a = b = 1$, if $r_i = 0$ then $a = b = 0$).  If $x_i = 0$,
then $a$ and $b$ always disagree (if $r_i=1$, then $a=0$ and $b=1$, if
$r_i = 0$, then $a=1$ and $b=0$).  

The probability of the second case is the probability of getting
$(n-1)/2$ ``heads'' out of $n-1$ coin flips, which is
$\binom{n-1}{(n-1)/2}$.  Applying Stirling's approximation of the
factorial function shows that this probability is bigger than you
might have expected, namely $\approx \tfrac{c'}{\sqrt{n}}$ for a
  constant $c'$ (see Exercises for details).
We therefore have
\begin{eqnarray*}
\prob{a=b} & = &
\underbrace{\prob{\text{Case 1}}}_{1-\tfrac{c'}{\sqrt{n}}} \cdot
\underbrace{\prob{a=b \,|\, \text{Case 1}}}_{=\tfrac{1}{2}}
+
\underbrace{\prob{\text{Case 2}}}_{\tfrac{c'}{\sqrt{n}}} \cdot \underbrace{\prob{a=b \,|\,
  \text{Case 2}}}_{\text{1 or 0}}\\
& = & \left\{ 
\begin{array}{cl}
\frac{1}{2} - \frac{c'}{\sqrt{n}} & \text{if $x_i = 1$}\\
\frac{1}{2} + \frac{c'}{\sqrt{n}} & \text{if $x_i = 0$.}
\end{array}
\right.
\end{eqnarray*}

This is pretty amazing when you think about it --- Alice and Bob have
no knowledge of each other's inputs and yet, with shared randomness
but no explicit communication, can generate bits correlated with
$x_i$!\footnote{This 
  would clearly not be possible with a private-coin protocol.
But we'll see later than the (additive) difference between the
private-coin and public-coin communication complexity of a problem is
$O(\log n)$, so a linear communication lower bound for one type
automatically carries over to the other type.}

The randomized reduction from \index to \gh now proceeds as one
would expect.  Alice and Bob repeat the bit-generating experiment
above $m$ independent times to generate $m$-bit inputs $\bfx'$ and
$\bfy'$ of \gh.  Here $m = qn$ for a sufficiently large
constant $q$.
The expected Hamming distance between $\bfx'$ and $\bfy'$ is at most
$\tfrac{m}{2} - c'\sqrt{m}$ (if $x_i = 1$) or at least
$\tfrac{m}{2} + c'\sqrt{m}$ (if $x_i = 0$).  A routine application of
the Chernoff bound (see Exercises) implies that, for a sufficiently
small constant $c$ and large constant $q$, with probability at least
$\tfrac{8}{9}$ (say), 
$d_H(\bfx',\bfy') < \tfrac{m}{2} - c\sqrt{m}$ (if $x_i = 1$) and
$d_H(\bfx',\bfy') > \tfrac{m}{2} + c\sqrt{m}$ (if $x_i = 0$).
When this event holds, Alice and Bob can correctly compute the answer
to the original input $(\bfx,i)$ to \index by simply invoking any
protocol $P$ for \gh on the input $(\bfx',\bfy')$.  The communication
cost is that of $P$ on inputs of length $m = \Theta(n)$.  The error is
at most the combined error of the randomized reduction and of the
protocol $P$ --- whenever the reduction and $P$ both proceed as
intended, the correct answer to the \index input $(\bfx,i)$ is computed.

Summarizing, our randomized reduction implies that, if there is a
(public-coin) randomized protocol for \gh with (two-sided)
error $\tfrac{1}{3}$ and sublinear communication, then there is a
randomized protocol for \index with error $\tfrac{4}{9}$.  Since we've
ruled out the latter, the former does not exist.
\end{proof}

Combining Theorem~\ref{t:gh} with our reduction from \gh to
estimating $F_{\infty}$, we've proved the following.

\begin{theorem}\label{t:f0}
There is a constant $c > 0$ such that the following statement holds:
There is no sublinear-space randomized streaming algorithm that, for
every data stream, computes $F_{0}$ to within a $1 \pm
\tfrac{c}{\sqrt{n}}$ factor with probability at least $2/3$.
\end{theorem}

A variation on the same reduction proves the same lower bound for
approximating $F_2$; see the Exercises.

Our original goal was to prove that the $(1 \pm
\eps)$-approximate computation of $F_0$ requires space
$\Omega(\eps^{-2})$, when $\eps \ge \tfrac{1}{\sqrt{n}}$.
Theorem~\ref{t:f0} proves this in the special case where $\eps =
\Theta(\tfrac{1}{\sqrt{n}})$.  This can be extended to larger $\eps$ by
a simple ``padding'' trick.  Fix your favorite values of $n$ and $\eps \ge
\tfrac{1}{\sqrt{n}}$ and modify the proof of Theorem~\ref{t:f0} as
follows.  Reduce from \gh on inputs of length $m =
\Theta(\eps^{-2})$.  Given an input $\inputs$ of \gh, form
$\inputsp$ by appending $n-m$ zeroes to $\bfx$ and $\bfy$.  A streaming
algorithm with space $s$ that estimates $F_0$ on the induced data
stream up to a $(1 \pm \eps)$ factor induces a randomized protocol
that solves this special case of \gh with communication $s$.
Theorem~\ref{t:gh} implies that every randomized
protocol for the latter problem uses communication
$\Omega(\eps^{-2})$, so this lower bound carries over to the space
used by the streaming algorithm.

\chapter{Lower Bounds for Compressive Sensing}
\label{cha:lower-bounds-compr}

\section[Randomized Communication Complexity of the             
  Equality Function]{An Appetizer: Randomized Communication Complexity of \eq}

We begin with an appetizer before starting the lecture proper --- an
example that demonstrates that randomized one-way communication
protocols can sometimes exhibit surprising power.

It won't surprise you that the \eq function ---
with $f\inputs = 1$ if and only if $\bfx = \bfy$ --- is a central problem
in communication complexity.  It's easy to prove, by the Pigeonhole
Principle, that its deterministic
one-way communication complexity is $n$, where $n$ is the length of
the inputs $\bfx$ and $\bfy$.\footnote{We'll see later that this lower
  bound applies to general deterministic protocols, not just to
  one-way protocols.} What about its randomized communication
complexity?  Recall from last lecture that by default, our randomized
protocols can use public coins\footnote{Recall the public-coin model:
  when Alice and Bob show up there is already an infinite stream of
  random bits written on a blackboard, which both of them can see.
  Using shared randomness does not count toward the communication
  cost of the protocol.} and can have two-sided error $\eps$,
where $\eps$ is any constant less than $\tfrac{1}{2}$.

\begin{theorem}[\citealt{Y79}]\label{t:eq}
The (public-coin) randomized one-way communication complexity of
\eq is $O(1)$.
\end{theorem}

Thus, the randomized communication complexity of a problem can be
radically smaller than its deterministic communication complexity.
A similar statement follows from our upper and lower bound
results for estimating the frequency moments $F_0$ and $F_2$ using
small-space streaming algorithms, but Theorem~\ref{t:eq} illustrates
this point in a starker and clearer way.  

Theorem~\ref{t:eq} provides a
cautionary tale: sometimes we expect a problem to be hard for a class
of protocols and are proved wrong by a clever protocol; other times,
clever protocols don't provide non-trivial solutions to a problem but
still make 
proving strong lower bounds technically difficult.  Theorem~\ref{t:eq}
also suggests that, if we want to prove strong communication lower
bounds for randomized protocols via a reduction, there might not be
too many natural problems out there to reduce from.\footnote{Recall our
  discussion about \gh last lecture: for the
  problem to be hard, it is important to choose the midpoint $t$ to
  be $\tfrac{n}{2}$.  With $t$ too close to 0 or $n$, the problem is a
  special case of \eq and is therefore easy for randomized
  protocols.}

\vspace{.1in}

\begin{prevproof}{Theorem}{t:eq}
The protocol is as follows.
\begin{enumerate}

\item Alice and Bob interpret the first $2n$ public coins as random
  strings $\bfr_1,\bfr_2 \in \{0,1\}^n$.  This requires no
  communication.

\item Alice sends the two random inner products 
$\ip{\bfx}{\bfr_1}
  \bmod 2$
and $\ip{\bfx}{\bfr_2} \bmod 2$ to Bob.  This requires two bits of
communication.

\item Bob reports ``1'' if and only if his random inner products match
  those of Alice: $\ip{\bfy}{\bfr_i} = \ip{\bfx}{\bfr_i} \bmod 2$ for
  $i=1,2$.  Note that Bob has all of the information needed to
  perform this computation.

\end{enumerate}
We claim that the error of this protocol is at most 25\% on every
input.  The protocol's error is one-sided: when $\bfx=\bfy$ the
protocol always accepts, so there are no false negatives.  Suppose
that $\bfx \neq \bfy$.  We use the Principle of Deferred Decisions to
argue that, for each $i=1,2$, the inner products 
$\ip{\bfy}{\bfr_i}$ and $\ip{\bfx}{\bfr_i}$ are different (mod~2) with
probability exactly 50\%.  To see this, pick an index $i$ where $x_i
\neq y_i$ and condition on all of the bits of a random string except
for the $i$th one.  Let $a$ and $b$ denote the values of the inner
products-so-far of $\bfx$ and $\bfy$ (modulo 2) with the random
string.  If the $i$th random bit is a 0, then the final inner products are
also $a$ and $b$.  If the $i$th random bit is a 1, then one inner product
stays the same while the other flips its value (since exactly one of
$x_i,y_i$ is a 1).  Thus, whether $a = b$ or $a \neq b$, exactly one
of the two random bit values (50\% probability) results in the final two
inner products having different values (modulo 2).
The probability that two unequal strings have equal inner products
(modulo 2) in two independent experiments is 25\%.
\end{prevproof}

The proof of Theorem~\ref{t:eq} gives a 2-bit protocol with (1-sided)
error 25\%.  As usual, executing many parallel
copies of the protocol reduces the error to an arbitrarily small
constant, with a constant blow-up in the communication complexity.

The protocol used to prove Theorem~\ref{t:eq} makes crucial use of
public coins.  We'll see later that the private-coin one-way
randomized communication complexity is $\Theta(\log n)$, 
which is worse than public-coin protocols but still radically better
than deterministic protocols.
More generally, next lecture we'll prove Newman's theorem, which states
that the private-coin randomized communication complexity of a problem 
is at most $O(\log n)$ more than its public-coin randomized communication
complexity. 

The protocol in the proof of Theorem~\ref{t:eq} effectively gives each
of the two strings $\bfx,\bfy$ a 2-bit ``sketch'' or ``fingerprint''
such that the property of distinctness is approximately preserved.
Clearly, this is the same basic idea as hashing.  This is a useful
idea in both theory and practice, and we'll use it again shortly.

\begin{remark}
The computational model studied in communication
complexity is potentially very powerful --- for example, Alice and Bob
have unlimited computational power --- and the primary point of the
model is to prove lower bounds.  Thus, whenever you see an {\em upper
  bound} result in communication complexity, like Theorem~\ref{t:eq},
it's worth asking what the point of the result is.  In many cases, a
positive result is really more of a ``negative negative result,''
intended to prove the tightness of a lower bound rather than offer a
practical solution to a real problem.  In other cases, 
the main point is demonstrate separations between different notions of
communication complexity or between different problems.  For example,
Theorem~\ref{t:eq} shows that the one-way deterministic and randomized
communication complexity of a problem can be radically different, even
if we insist on one-sided error.  It also shows that the randomized
communication complexity of \eq is very different than that of the
problems we studied last lecture: \disj, \index, and \gh.

Theorem~\ref{t:eq} also uses a quite reasonable protocol,
which is not far from a practical solution to probabilistic
equality-testing. 
In some applications, the public coins can be replaced by a hash
function that is published 
in advance; in other applications, one party can choose a random hash
function that can be specified with a reasonable number of bits and
communicate it to other parties.  
\end{remark}

\section{Sparse Recovery}

\subsection{The Basic Setup}

The field of sparse recovery has been a ridiculously hot area for the
past ten years, in applied mathematics, machine learning, and
theoretical computer science.  We'll study sparse recovery in the
standard setup of ``compressive sensing'' (also called ``compressed
sensing'').  There is an unknown ``signal'' --- i.e., a real-valued
vector $\bfx \in \RR^n$ --- that we want to learn.  The bad news is
that we're only allowed to access the signal through ``linear
measurements;'' the good news is that we have the freedom to choose
whatever measurements we want.  Mathematically, we want to design a
matrix $\A \in \RR^{m \times n}$, with $m$ as small as possible, such
that we can recover the unknown signal $\bfx$ from the linear
measurements $\Ax$ (whatever $\bfx$ may be).

As currently stated, this is a boring problem.  It is clear that $n$
measurements are sufficient -- just take $\A = I$, or any other
invertible $n \times n$ matrix.  It is also clear that $n$
measurements are necessary: if $m < n$, then there is a entire
subspace of dimension $n-m$ of vectors that have image $\Ax$ under
$\A$, and we have no way to know which one of them is the actual
unknown signal. 

The problem becomes interesting when we also assume that the unknown
signal $\bfx$ is ``sparse,'' in senses we define shortly.  The hope is
that under the additional promise that $\bfx$ is sparse, we can get away
with much fewer than $n$ measurements (Figure~\ref{f:cs}).

\begin{figure}
\centering
\includegraphics[width=.8\textwidth]{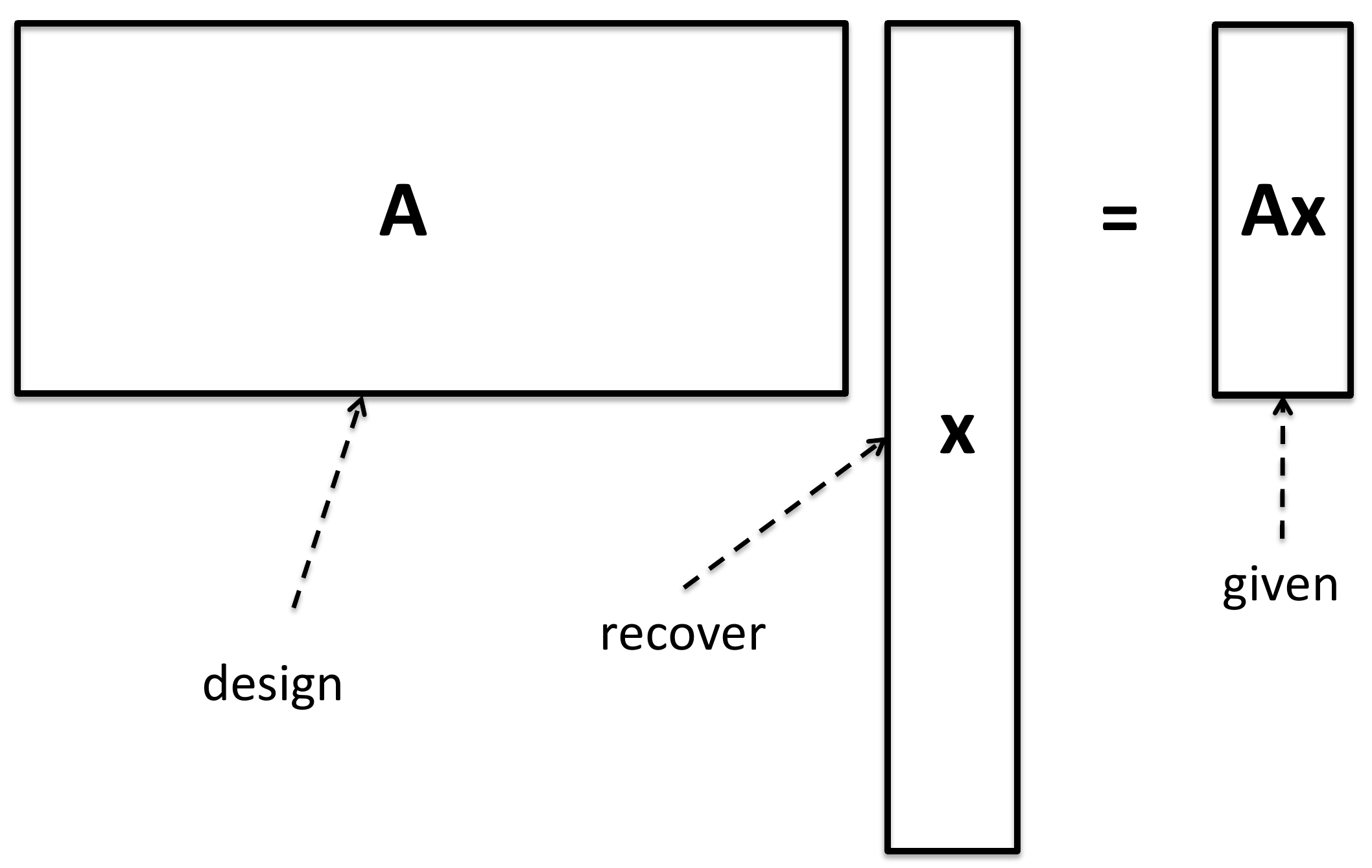}
\caption[Compressive sensing]{The basic compressive sensing setup.
The goal is to design a matrix $\A$ such that an unknown sparse signal $\bfx$
can be recovered from the linear measurements $\Ax$.}\label{f:cs}
\end{figure}

\subsection{A Toy Version}\label{ss:toy}

To develop intuition for this problem, let's explore a toy version.
Suppose you are promised that $\bfx$ is a 0-1 vector with exactly $k$
1's (and hence $n-k$ 0's).  Throughout this lecture, $k$ is the
parameter that measures the sparsity of the unknown signal $\bfx$ --- it
could be anything, but you might want to keep $k \approx \sqrt{n}$ in
mind as a canonical parameter value.  Let $X$ denote the set of all
such $k$-sparse 0-1 vectors, and note that $|X| = \binom{n}{k}$.

Here's a solution to the sparse recovery problem under that guarantee
that $\bfx \in X$.  Take $m = 3 \log_2 |X|$, and choose each of the $m$ rows
of the sensing matrix $\A$ independently and uniformly at random from
$\{0,1\}^n$.  By the fingerprinting argument used in our randomized
protocol for \eq (Theorem~\ref{t:eq}), for fixed distinct $\bfx,\bfx'
\in X$, we have
\[
\prob[\A]{\Ax = \Ax' \bmod 2} = \frac{1}{2^m} = \frac{1}{|X|^3}.
\]
Of course, the probability that $\Ax = \Ax'$ (not modulo~2) is only
less.  Taking a Union Bound over the at most $|X|^2$ different distinct
pairs $\bfx,\bfx' \in X$, we have
\[
\prob[\A]{\text{there exists } \bfx \neq \bfx' \text{ s.t. } \Ax = \Ax'}
\le \frac{1}{|X|}.
\]
Thus, there is a matrix $\A$ that maps all $\bfx \in X$ to distinct
$m$-vectors.  (Indeed, a random $\A$ works with high probability.)
Thus, given $\Ax$, one can recover $\bfx$ --- if nothing else, one can
just compute $\Ax'$ for every $\bfx' \in X$ until a match is
found.\footnote{We won't focus on computational efficiency in this
  lecture, but positive results in compressive sensing generally also
  have computationally efficient recovery algorithms.  Our lower
  bounds will hold even for recovery algorithms with unbounded
  computational power.}

The point is that $m = \Theta(\log |X|)$ measurements are sufficient
to recover exactly $k$-sparse 0-1 vectors.  Recalling that
$|X| = \binom{n}{k}$ and that
\[
\left( \frac{n}{k} \right)^k \underbrace{\le}_{\text{easy}}
\binom{n}{k} \underbrace{\le}_{\text{Stirling's approx.}} \left(
\frac{en}{k} \right)^k,
\]
we see that $m = \Theta(k \log \tfrac{n}{k})$ rows
suffice.\footnote{If you read through the compressive sensing
  literature, you'll be plagued by ubiquitous ``$k \log \tfrac{n}{k}$''
  terms --- remember this is just $\approx \log \binom{n}{k}$.}

This exercise shows that, at least for the special case of exactly
sparse 0-1 signals, we can indeed achieve recovery with far fewer than
$n$ measurements.  For example, if $k \approx \sqrt{n}$, we need only
$O(\sqrt{n} \log n)$ measurements.

Now that we have a proof of concept that recovering an unknown sparse
vector is an interesting problem, we'd like to do better in two
senses.  First, we want to move beyond the toy version of the problem
to the ``real'' version of the problem.  Second, we have to wonder
whether even fewer measurements suffice.

\subsection{Motivating Applications}

To motivate the real version of the problem, we mention a couple of
canonical applications of compressive sensing.
One buzzword you can
look up and read more about is the ``single-pixel camera.''  The standard
approach to taking pictures is to first take a high-resolution picture
in the ``standard basis'' --- e.g., a light intensity for each pixel
--- and then to compress the picture later (via software).  Because
real-world images are typically sparse in a suitable basis, they can
be compressed a lot.  The compressive sensing approach asks, then why
not just capture the image directly in a compressed form --- in a
representation where its sparsity shines through?  For example, one can
store random linear combinations of light intensities (implemented via
suitable mirrors) rather than the light intensities themselves.  This
idea leads to a reduction in the number of pixels needed to capture an
image at a given resolution.
Another application of compressive sensing is in MRI.  Here, decreasing
the number of measurements decreases the time necessary for a scan.
Since a patient needs to stay motionless during a scan --- in some cases, not even
breathing --- shorter scan times can be a pretty big deal.

\subsection{The Real Problem}

If the unknown signal $\bfx$ is an image, say, there's no way it's an
exactly sparse 0-1 vector.  We need to consider more general unknown
signals $\bfx$ that are real-valued and only ``approximately sparse.''  
To measure approximate sparsity, with respect to a choice of $k$,
we define the {\em residual} $\res(\bfx)$ of a vector $\bfx \in \RR^n$ as
the contribution to $\bfx$'s $\ell_1$ norm by its $n-k$ coordinates with
smallest magnitudes.  Recall that the $\ell_1$ norm is just $\|\bfx\|_1
= \sum_{i=1}^n |x_i|$.  If we imagine sorting the coordinates by
$|x_i|$, then the residual of $\bfx$ is just the sum of the $|x_i|$'s of
the final $n-k$ terms.  If $\bfx$ has exactly $k$-sparse, it has at least
$n-k$ zeros and hence $\res(\bfx) = 0$.

The goal is to design a sensing matrix $\A$ with a small number $m$ of
rows such that an unknown approximately sparse vector $\bfx$ can be
recovered from $\Ax$.  Or rather, given that $\bfx$ is only
approximately sparse, we want to recover a close approximation of
$\bfx$.

The formal guarantee we'll seek for the matrix $\A$ is the following:
for every $\bfx \in \RR^n$, we can compute from $\Ax$ a vector $\bfx'$
such that 
\begin{equation}\label{eq:rec}
\n{\bfx'-\bfx}_1 \le c \cdot \res(\bfx).
\end{equation}
Here $c$ is a constant, like~2.  The guarantee~\eqref{eq:rec} is very
compelling.  First, if $\bfx$ is exactly $k$-sparse, then $\res(\bfx) = 0$
and so~\eqref{eq:rec} demands exact recovery of $\bfx$.  The guarantee
is parameterized by how close $\bfx$ is to being sparse --- the
recovered vector $\bfx'$ should lie in a ball (in the $\ell_1$ norm)
around $\bfx$, and the further $\bfx$ is from being $k$-sparse, the bigger
this ball is.  Intuitively, the radius of this ball has to depend on
something like $\res(\bfx)$.  For example, 
suppose that $\bfx'$ is exactly $k$-sparse (with $\res(\bfx) = 0$).  The
guarantee~\eqref{eq:rec} forces the algorithm to return $\bfx'$ for
every unknown signal $\bfx$ with $\Ax = \Ax'$.  Recall that when $m < n$,
there is an $(n-m)$-dimensional subspace of such signals $\bfx$.
In the extreme case where there is such an $\bfx$ with $\bfx - \res(\bfx) =
\bfx'$ -- i.e., where $\bfx$ is $\bfx'$ with a little noise added to its
zero coordinates --- the recovery algorithm is forced to return a
solution $\bfx'$ with  $\n{\bfx'-\bfx}_1 = \res(\bfx)$.

Now that we've defined the real version of the problem, is there an
interesting solution?  Happily, the real version can be solved as well
as the toy version.
\begin{fact}[\citealt{CRT06,D06}]\label{fact:cs}
With high probability, 
a random $m \times n$ matrix $\A$ with $\Theta(k \log \tfrac{n}{k})$
rows admits a recovery algorithm that satisfies the guarantee
in~\eqref{eq:rec} for every $\bfx \in \RR^n$.
\end{fact}
Fact~\ref{fact:cs} is non-trivial and well outside the scope of this
course. The fact is true for several different distributions over
matrices, including the case where each matrix entry is an independent
standard Gaussian.  Also, there are computationally efficient recovery
algorithms that achieve the guarantee.\footnote{See e.g.~\cite{moitra} for an introduction to such positive results about
  sparse recovery.
Lecture \#9 of the
 instructor's CS264 course also gives a brief overview.}

\section{A Lower Bound for Sparse Recovery}

\subsection{Context}

At the end of Section~\ref{ss:toy}, when we asked ``can we do
better?,'' we
meant this in two senses.  First, can we extend the positive results
for the toy problem to a more general problem?  Fact~\ref{fact:cs}
provides a resounding affirmative answer.  Second, can we get away
with an even smaller value of $m$ --- even fewer measurements?  In
this section we prove a relatively recent result of~\cite{D+10}, who showed 
that the answer is ``no.''\footnote{Such a lower bound was previously
  known for various special cases --- for particular classes of
  matrices, for particular families of recovery algorithms, etc.}
Amazingly, they proved this fundamental
result via a reduction to a lower bound in communication complexity
(for \index).  
\begin{theorem}[\citealt{D+10}]\label{t:cs}
If an $m \times n$ matrix $\A$ admits a recovery algorithm $R$ that,
for every $\bfx \in \RR^n$, computes from $\Ax$ a vector $\bfx'$ that
satisfies~\eqref{eq:rec}, then $m = \Omega(k \log \tfrac{n}{k})$.
\end{theorem}
The lower bound is information-theoretic,
and therefore applies also to recovery algorithms with unlimited
computational power.

Note that there is an easy lower bound of $m \ge k$.\footnote{For
  example, consider just the subset of vectors $\bfx$ that are zero in
  the final $n-k$ coordinates.  The guarantee~\eqref{eq:rec} demands
  exact recovery of all such vectors.  Thus, we're back to the boring
  version of the problem mentioned at the beginning of the lecture,
  and $\A$ has to have rank at least $k$.}  
Theorem~\ref{t:cs} offers a non-trivial improvement when $k=o(n)$;
in this case, we're asking whether
or not it's possible to shave off the log factor in
Fact~\ref{fact:cs}.  Given how fundamental the problem is, and how
much a non-trivial improvement could potentially matter in
practice, this question is well worth asking.

\subsection{Proof of Theorem~\ref{t:cs}: First Attempt}\label{ss:first}

Recall that we first proved our upper bound of $m=O(\lnk)$ in the toy
setting of Section~\ref{ss:toy}, and then stated in Fact~\ref{fact:cs}
that it can be extended to the general version of the problem.  
Let's first try to prove a matching lower bound on $m$ that applies
even in the toy setting.  Recall that $X$ denotes the set of all 0-1
vectors that have exactly $k$ 1's, and that $|X| = \binom{n}{k}$.

For vectors $\bfx \in X$, the guarantee~\eqref{eq:rec} demands exact
recovery.  Thus, the sensing matrix $\A$ that we pick has to satisfy
$\Ax \neq \Ax'$ for all distinct $x,x' \in X$.  That is, $\Ax$ encodes
$\bfx$ for all $\bfx \in X$.  But $X$ has $\binom{n}{k}$ members, so the
worst-case encoding length of $\Ax$ has to be at least $\log_2
\binom{n}{k} = \Theta(\lnk)$.  So are we done?

The issue is that we want a lower bound on the number of rows $m$ of
$\A$, not on the worst-case length of $\Ax$ in bits.  What is the
relationship between these two quantities?  Note that, even if $\A$ is
a 0-1 matrix, then each entry of $\Ax$ is generally of magnitude
$\Theta(k)$, requiring $\Theta(\log k)$ bits to write down.  For
example, when $k$ is polynomial in $n$ (like our running choice $k =
\sqrt{n}$), then $\Ax$ generally requires $\Omega(m \log n)$ bits to
describe, even when $\A$ is a 0-1 matrix.  Thus our lower bound of
$\Theta(\lnk)$ on the length of $\Ax$ does not yield a lower bound on
$m$ better than the trivial lower bound of $k$.\footnote{This argument
  does imply that $m=\Omega(\lnk)$ is we only we report $\Ax$
  modulo~2 (or some other constant), since in this case the length of
  $\Ax$ is $\Theta(m)$.}

The argument above was doomed to fail.  The reason is that, if you
only care about recovering exactly $k$-sparse vectors $\bfx$ --- rather
than the more general and robust guarantee in~\eqref{eq:rec} --- then
$m=2k$ suffices!  One proof is via ``Prony's Method,'' which uses the
fact that a $k$-sparse vector $\bfx$ can be recovered exactly from its first
$2k$ Fourier coefficients (see e.g.~\cite{moitra}).\footnote{This
  method uses a sensing matrix $\A$ for which $\Ax$ will generally
  have $\Omega(m \log n) = \Omega(k \log n)$ bits, so this does not
  contradict out lower bound on the necessary length of $\Ax$.}
Our argument above only invoked the
requirement~\eqref{eq:rec} for $k$-sparse vectors $\bfx$, and such an
argument cannot prove a lower bound of the form $m = \Omega(\lnk)$.

The first take-away from this exercise is that, to prove the lower
bound that we want, we need to use the fact that the matrix
$\A$ satisfies the guarantee~\eqref{eq:rec} also for non-$k$-sparse
vectors $\bfx$.  The second take-away is that we need a
smarter argument --- a straightforward application of the Pigeonhole
Principle is not going to cut it.

\subsection{Proof of Theorem~\ref{t:cs}}

\subsubsection{A Communication Complexity Perspective}

We can interpret the failed proof attempt in Section~\ref{ss:first} as
an attempted reduction from a ``promise version'' of \index.
Recall that in this 
communication problem, Alice has an input $\bfx \in \{0,1\}^n$, Bob
has an index $i \in \{1,2,\ldots,n\}$, specified using $\approx \log_2
n$ bits, and the goal is to compute $x_i$ using a one-way protocol.
We showed last lecture that the deterministic communication complexity
of this problem is $n$ (via an easy counting argument) and its
randomized communication complexity is $\Omega(n)$ (via a harder
counting argument).  

The previous proof attempt can be rephrased as follows.
Consider a matrix $\A$ that permits exact recovery for all $\bfx \in X$.
This induces a one-way protocol for solving \index whenever
Alice's input $\bfx$ lies in $X$ --- Alice simply sends $\Ax$ to Bob,
Bob recovers $\bfx$, and Bob can then solve the problem, whatever
his index $i$ might be.  The communication cost of this protocol is
exactly the length of $\Ax$, in bits.  The deterministic one-way
communication complexity of this promise version of \index
is $\lnk$, by the usual counting argument, so this lower bound applies
to the length of $\Ax$.

\subsubsection{The Plan}

How can we push this idea further?  We begin by assuming the existence
of an $m \times n$ matrix~$\A$, a recovery algorithm $R$,
and a constant $c \ge 1$ such that,
for every $\bfx \in \RR^n$, $R$ computes from $\Ax$ a vector
$\bfx' \in \RR^n$ that satisfies
\begin{equation}\label{eq:req2}
\n{\bfx'-\bfx}_1 \le c \cdot \res(\bfx).
\end{equation}
Our goal is to show that if $m << \lnk$, then we can solve \index
with sublinear communication.

For simplicity, we assume that the recovery algorithm $R$ is
deterministic.  The lower bound continues to hold for randomized
recovery algorithms that have success probability at least
$\tfrac{2}{3}$, but the 
proof requires more work; see Section~\ref{ss:rand}.

\subsubsection{An Aside on Bit Complexity}

We can assume that the sensing matrix $\A$ has reasonable bit
complexity.  Roughly:
(i) we can assume without
loss of generality that $\A$ has orthonormal rows (by a change of basis
argument); (ii) dropping all but the $O(\log n)$ highest-order bits of
every entry has negligible effect on the recovery algorithm $R$.
We leave the details to the Exercises.

When every entry of the matrix $\A$ and of a vector $\bfx$ can be
described in $O(\log n)$ bits --- equivalently, by scaling, is a
polynomially-bounded integer --- the same is true of $\Ax$.  In this
case, $\Ax$ has length $O(m \log n)$.

\subsubsection{Redefining $X$}

It will be useful later to redefine the set $X$.  Previously, $X$ was
all 0-1 $n$-vectors with exactly $k$ 1's.  
Now it will be a subset of such vectors, subject to the
constraint that
\[
d_H(\bfx,\bfx') \ge .1k
\]
for every distinct pair $\bfx,\bfx'$ of vectors in the set.  Recall the
$d_H(\cdot,\cdot)$ denotes the Hamming distance between two vectors
--- the number of coordinates in which they differ.  Two distinct 0-1
vectors with $k$ 1's each have Hamming distance between~2 and~$2k$, so
we're restricting to a subset of such vectors that have mostly
disjoint supports.  The following lemma says that 
there exist sets of such vectors with size
not too much smaller than the number of all 0-1 vectors with $k$ 1's;
we leave the proof to the Exercises.
\begin{lemma}\label{l:size}
Suppose $k \le .01n$.
There is a set $X$ of 0-1 $n$-bits vectors such that each $\bfx \in X$
has $k$ 1s, each distinct $\bfx,\bfx' \in X$ satisfy $d_H(\bfx,\bfx') \ge
.1k$, and
$\log_2 |X| = \Omega(\lnk)$.
\end{lemma}
Lemma~\ref{l:size} is reminiscent of the fact
that there are large error-correcting codes with large
distance. 
One way to prove Lemma~\ref{l:size} is via the
probabilistic method --- by showing that a suitable randomized
experiment yields a set with the desired size with positive probability.

Intuitively, our proof attempt in Section~\ref{ss:first} used that,
because each $\bfx \in X$ is exactly $k$-sparse, it can be recovered
exactly from $\Ax$ and thus there is no way to get confused between
distinct elements of $X$ from their images.  In the following proof,
we'll instead need to recover ``noisy'' versions of the $\bfx$'s ---
recall that our proof cannot rely only on the fact that the matrix $\A$
performs exact recovery of exactly sparse vectors.
This means we
might get confused between two different 0-1 vectors $\bfx,\bfx'$ that
have small (but non-zero) Hamming distance.  The above redefinition of
$X$, which is essentially costless by Lemma~\ref{l:size}, fixes the
issue by restricting to vectors that all look very different from one
another.

\subsubsection{The Reduction}

To obtain a lower bound of $m=\Omega(\lnk)$ rather than $m=\Omega(k)$,
we need a more sophisticated reduction than in
Section~\ref{ss:first}.\footnote{We assume from now on that $k =
  o(n)$, since otherwise there is nothing to prove.}
The parameters offer some strong clues as to what the reduction should
look like.  If we use a protocol where Alice sends $\Ay$ to Bob, where
$\A$ is the assumed sensing matrix with a small number $m$ of rows and
$\bfy$ is some vector of Alice's choosing --- perhaps a noisy version of
some $\bfx \in X$ --- then the communication cost will be $O(m \log
n)$.  We want to prove that $m = \Omega(\lnk) = \Omega(\log |X|)$
(using Lemma~\ref{l:size}).  This means we need a communication lower
bound of $\Omega(\lxn)$.  Since the communication lower
bound for \index is linear, this suggests considering
inputs to \index where Alice's input has length $\lxn$, and
Bob is given an index $i \in \{1,2,\ldots,\lxn\}$.

To describe the reduction formally, let $\enc:X \rightarrow
\{0,1\}^{\log_2 |X|}$ be a binary encoding of the vectors of
$X$.  (We can assume that $|X|$ is a power of~2.)
Here is the communication protocol for solving \index.
\begin{itemize}

\item [(1)] 
Alice interprets her $(\lxn)$-bit input as $\log n$ blocks
  of $\log |X|$ bits each.  For each $j=1,2,\ldots,\log n$,
she interprets the bits of the $j$th block
  as $\enc(\bfx_j)$ for some $\bfx_j \in X$.  See Figure~\ref{f:blocks}.

\item [(2)] Alice computes a suitable linear combination of the
  $\bfx_j$'s:
\[
\bfy = \sum_{i=1}^{\log n} \alpha^j\bfx_j,
\]
with the details provided below.  Each entry of $\bfy$ will be a
polynomially-bounded integer.

\item [(3)] Alice sends $\Ay$ to Bob.  This uses $O(m \log n)$ bits of
  communication.

\item [(4)] Bob uses $\Ay$ and the assumed recovery algorithm $R$ to
  recover all of $\bfx_1,\ldots,\bfx_{\log n}$.  (Details TBA.)

\item [(5)] Bob identifies the block $j$ of $\log |X|$ bits that
  contains his given index $i \in \{1,2,\ldots,\lxn\}$, and outputs
  the relevant bit of $\enc(\bfx_j)$.

\end{itemize}
If we can implement steps~(2) and~(4), then we're done: this five-step
deterministic protocol would solve \index on $\lxn$-bit
inputs using $O(m \log n)$ communication.  Since the communication
complexity of \index is linear, we conclude that $m =
\Omega(\log |X|) = \Omega(\lnk)$.\footnote{Thus even the deterministic
  communication lower bound for \textsc{Index}, which is near-trivial to prove
  (Lectures~\ref{cha:data-stre-algor}--\ref{cha:lower-bounds-one}),
  has very interesting implications.  See
  Section~\ref{ss:rand} for the stronger implications provided by the
  randomized communication complexity lower bound.}

\begin{figure}
\centering
\includegraphics[width=.8\textwidth]{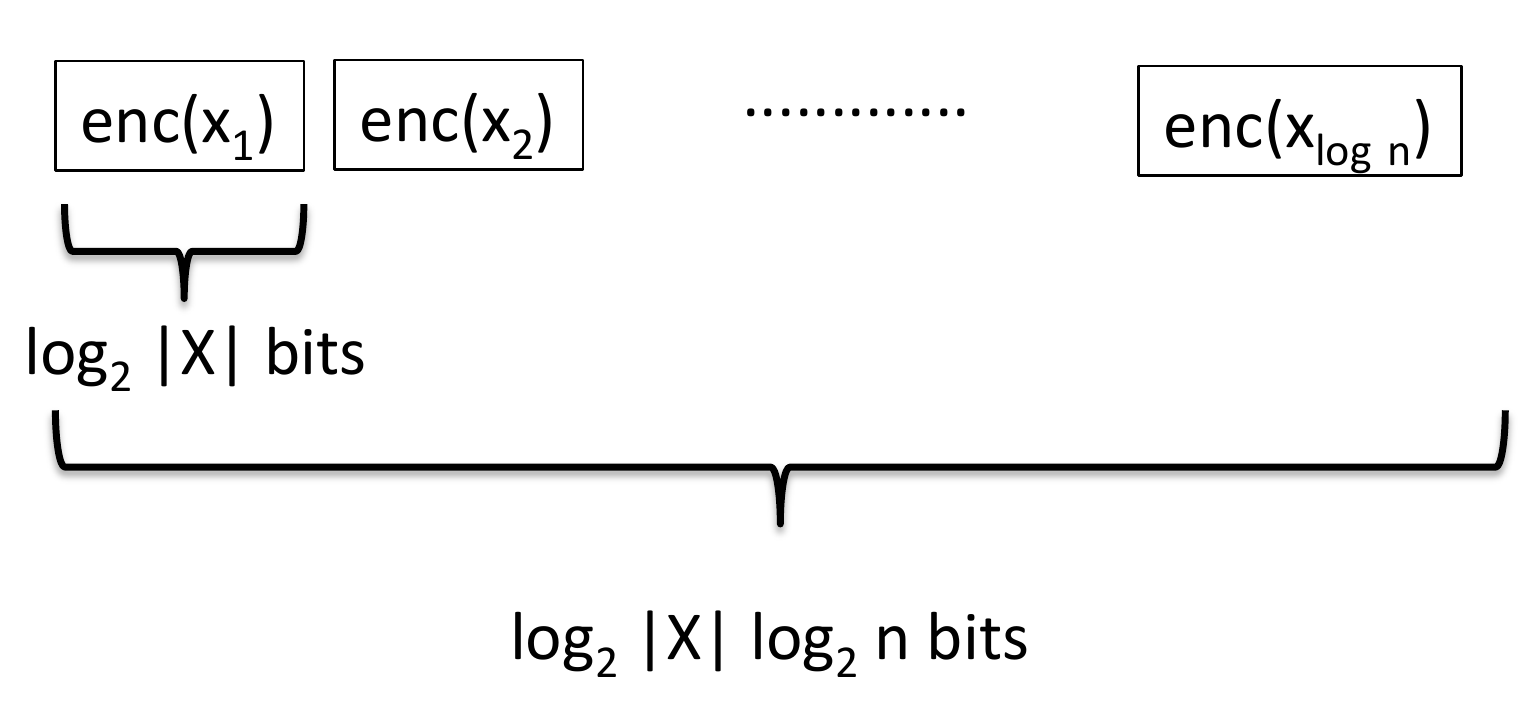}
\caption[How Alice interprets her $\lxn$-bit input]{In the reduction, Alice interprets her $\lxn$-bit input as
$\log n$ blocks, with each block encoding some vector $x_j \in X$.}
\label{f:blocks}
\end{figure}

For step~(2), let $\alpha \ge 2$ be a sufficiently large constant, depending
on the constant~$c$ in the recovery guarantee~\eqref{eq:req2} that the
matrix $\A$ and algorithm $R$ satisfy.  Then
\begin{equation}\label{eq:y}
\bfy = \sum_{i=1}^{\log n} \alpha^j \bfx_j.
\end{equation}
Recall that each $\bfx_j$ is a 0-1 $n$-vector with exactly $k$ 1's, so
$\bfy$ is just a superposition of $\log n$ scaled copies of such
vectors.  For example, the $\ell_1$ norm of $\bfy$ is simply $k
\sum_{j=1}^{\log n} \alpha^j$.
In particular, since $\alpha$ is a constant, the entries of
$\bfy$ are polynomially bounded non-negative integers, as promised
earlier, and $\Ay$ can be described using $O(m \log n)$ bits.

\subsubsection{Bob's Recovery Algorithm}

The interesting part of the protocol and analysis is step~(4), where
Bob wants to recover the vectors $\bfx_1,\ldots,\bfx_{\log n}$ encoded by
Alice's input using only knowledge of $\Ay$ and black-box access to
the recovery algorithm $R$.  To get started, we explain how Bob can
recover the last vector $\bfx_{\log n}$, which suffices to solve 
\index in the lucky case where Bob's index $i$ is one of the
last $\log_2 |X|$ positions.  Intuitively, this is the easiest case,
since $\bfx_{\log n}$ is by far (for large $\alpha$) the largest
contributor to the vector $\bfy$ Alice computes in~\eqref{eq:y}.  With
$\bfy = \alpha^{\log n} \bfx_{\log n} + \text{ noise}$, we might hope that
the recovery algorithm can extract $\alpha^{\log n} \bfx_{\log n}$ from
$\Ay$.

Bob's first step is the only thing he can do: invoke the recovery
algorithm $R$ on the message $\Ay$ from Alice.  By assumption, $R$
returns a vector $\yh$ satisfying
\[
\n{\yh-\bfy}_1 \le c \cdot \res(\bfy),
\]
where $\res(\bfy)$ is the contribution to $\bfy$'s $\ell_1$ norm by its
smallest $n-k$ coordinates.  

Bob then computes $\yh$'s nearest neighbor $\bfx^*$ in a scaled version of $X$
under the $\ell_1$ norm --- $\argmin_{\bfx \in X} \n{\yh - \alpha^{\log
    n}\bfx}_1$.  Bob can do this computation by brute force.

The key claim is that the computed vector $\bfx^*$ is indeed $\bfx^{\log
  n}$.  This follows from a geometric argument, pictured in
Figure~\ref{f:triangle}.  Briefly: (i) because $\alpha$ is large, $\bfy$
and $\alpha^{\log n}\bfx_{\log n}$ are close to each other; (ii) since
$\bfy$ is approximately $k$-sparse (by~(i)) and since $R$ satisfies the
approximate recovery guarantee in~\eqref{eq:req2}, $\yh$ and $\bfy$ are
close to each other and hence $\alpha^{\log n}\bfx_{\log n}$ and $\yh$
are close; and (iii) since distinct vectors in $X$ have large Hamming
distance, for every $\bfx \in X$ other than $\bfx_{\log n}$, $\alpha^{\log
  n}\bfx$ is far from $\bfx_{\log n}$ and hence also far from $\yh$.  We
conclude that $\alpha^{\log n}\bfx_{\log n}$ is closer to $\yh$ than any
other scaled vector from $X$.

\begin{figure}
\centering
\includegraphics[width=.8\textwidth]{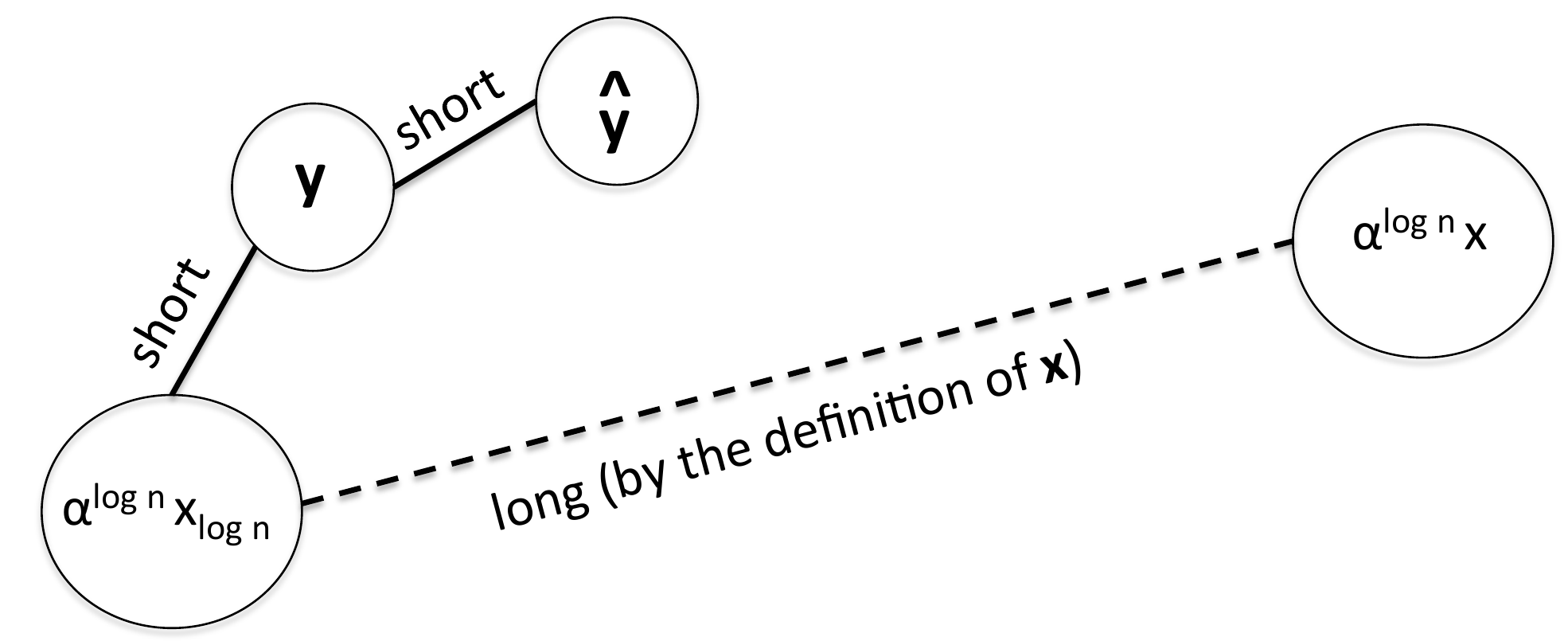}
\caption[The triangle inequality implies that Bob's recovery is correct]{The triangle inequality implies that the vector $\bfx^*$ computed
  by Bob from $\yh$ must be $\bfx_{\log n}$.}
\label{f:triangle}
\end{figure}

We now supply the details.
\begin{itemize}

\item [(i)] Recall that $\bfy$ is the superposition (i.e., sum) of
  scaled versions of the vectors $\bfx_1,\ldots,\bfx_{\log n}$.
$\bfy - \alpha^{\log n}\bfx_{\log n}$ is just $\bfy$ with the last
  contributor omitted.  Thus, 
\[
\n{\bfy - \alpha^{\log n}\bfx_{\log n}}_1 =
  \sum_{j=1}^{\log n -1} \alpha^j k.
\]
Assuming that $\alpha \ge \max\{2,200c\}$, where $c$ is the constant
that $\A$ and $R$ satisfy in~\eqref{eq:req2}, we can bound 
from above the geometric sum and derive
\begin{equation}\label{eq:close}
\n{\bfy - \alpha^{\log n}\bfx_{\log n}}_1 \le \frac{.01}{c}k\alpha^{\log n}.
\end{equation}

\item [(ii)] By considering the $n-k$ coordinates of $\bfy$ other than
  the $k$ that are non-zero in $\bfx_{\log n}$, we can upper bound the
  residual $\res(\bfy)$ by the contributions to the $\ell_1$ weight by
  $\alpha\bfx_1,\ldots,\alpha^{\log n-1}\bfx_{\log n - 1}$.  The
  sum of these contributions is $k \sum_{j=1}^{\log n -1} \alpha^j$,
  which we already bounded above in~\eqref{eq:close}.  In light of the
  recovery guarantee~\eqref{eq:req2}, we have
\begin{equation}\label{eq:close2}
\n{\yh-\bfy}_1 \le .01k\alpha^{\log n}.
\end{equation}

\item [(iii)] Let $\bfx,\bfx' \in X$ be
  distinct.  By the definition of $X$, $d_H(\bfx,\bfx') \ge .1k$.
Since $\bfx,\bfx'$ are both 0-1 vectors,
\begin{equation}\label{eq:far}
\n{\alpha^{\log n}\bfx - \alpha^{\log n}\bfx'}_1 \ge .1k\alpha^{\log n}.
\end{equation}

\end{itemize}
Combining~\eqref{eq:close} and~\eqref{eq:close2} with the triangle
inequality, we have
\begin{equation}\label{eq:close3}
\n{\yh-\alpha^{\log n}\bfx_{\log n}}_1 \le .02k\alpha^{\log n}.
\end{equation}
Meanwhile, for every other $\bfx \in X$, combining~\eqref{eq:far}
and~\eqref{eq:close3} with the triangle inequality gives
\begin{equation}\label{eq:far2}
\n{\yh-\alpha^{\log n}\bfx}_1 \ge .08k\alpha^{\log n}.
\end{equation}
Inequalities~\eqref{eq:close3} and~\eqref{eq:far2} imply that Bob's
nearest-neighbor computation will indeed recover $\bfx_{\log n}$.

You'd be right to regard the analysis so far with skepticism.  The same
reason that Bob can recover $\bfx_{\log n}$ --- because $\bfy$ is just a scaled
version of $\bfx_{\log n}$, plus some noise --- suggests that Bob should
not be able to recover the other $\bfx_j$'s, and hence unable to solve
\index for indices $i$ outside of the last block of Alice's
input.

The key observation is that, after recovering $\bfx_{\log n}$, Bob can
``subtract it out'' without any further communication from Alice and
then recover $x_{\log n-1}$.  Iterating this argument allows Bob to
reconstruct all of $\bfx_1,\ldots,\bfx_{\log n}$ and hence solve \index,
no matter what his index $i$ is.

In more detail, suppose Bob has already reconstructed
$\bfx_{j+1},\ldots,\bfx_{\log n}$.  He's then in a position to form
\begin{equation}\label{eq:z}
\bfz = \alpha^{j+1}\bfx_{j+1} + \cdots + \alpha^{\log n}\bfx_{\log n}.
\end{equation}
Then, $\bfy - \bfz$ equals the first $j$ contributors to $\bfy$ ---
subtracting $\bfz$ undoes the last $\log n - j$ of them --- and is therefore
just a scaled version of $\bfx_j$, plus some relatively small noise (the
scaled $\bfx_{\ell}$'s with $\ell < j$).  This raises the hope that Bob
can recover a scaled version of $\bfx_j$ from $\A(\bfy-\bfz)$.  How can
Bob get his hands on the latter vector?  (He doesn't know $\bfy$, and we don't
want Alice to send any more bits to Bob.)  The trick is to use the
linearity of $\A$ --- Bob knows $\A$ and $\bfz$ and hence can compute
$\Az$, Alice has already sent him $\Ay$, so Bob just computes $\Ay -
\Az = \A(\bfy-\bfz)$!

After computing $\A(\bfy-\bfz)$, Bob invokes the recovery algorithm $R$ to
obtain a vector $\wh \in \RR^n$ that satisfies 
\[
\n{\wh - (\bfy-\bfz)}_1 \le c \cdot \res(\bfy-\bfz),
\]
and computes (by brute force) the vector $\bfx \in X$ minimizing $\n{\wh
  - \alpha^j\bfx }_1$.  The minimizing vector is $\bfx_j$ --- the reader
should check that the proof of this is word-for-word the same as our
recovery proof for $\bfx_{\log  n}$, with every ``$\log n$'' replaced
by ``$j$,'' every ``$\bfy$'' replaced by ``$\bfy-\bfz$,'' and every
``$\yh$'' replaced by ``$\wh$.''

This completes the implementation of step~(4) of the protocol, and
hence of the reduction from \index to the design of a sensing matrix
$\A$ and recovery algorithm~$R$ that satisfy~\eqref{eq:req2}.  We
conclude that $\A$ must have $m=\Omega(\lnk)$, which completes the
proof of Theorem~\ref{t:cs}.

\subsection{Lower Bounds for Randomized Recovery}\label{ss:rand}

We proved the lower bound in Theorem~\ref{t:cs} only for fixed
matrices $\A$ and deterministic recovery algorithms $R$.
This is arguably the most relevant case, but it's also worth asking
whether or not better positive results (i.e., fewer rows) are possible
for a randomized recovery requirement, where recovery can fail with
constant probability.  
Superficially, the randomization can come from two sources:
first, one can use a distribution over matrices $\A$; second, the
recovery algorithm (given $\Ax$) can be randomized.  
Since we're not worrying about computational efficiency, we can assume
without loss of generality that $R$ is deterministic --- a randomized
recovery algorithm can be derandomized just be enumerating over all of
its possible executions and taking the majority vote.

Formally, the
relaxed requirement for a positive result is: there exists a
constant $c \geq 1$, a distribution $D$ over $m \times n$ matrices
$\A$, and a (deterministic) recovery algorithm $R$ such that,
for every $\bfx \in \RR^n$, with probability at least $2/3$ (over the
choice of $\A$), $R$ returns a
vector $\bfx' \in \RR^n$ that satisfies
\[
\n{\bfx'-\bfx}_1 \le c \cdot \res(\bfx).
\]

The lower bound in Theorem~\ref{t:cs} applies even to such randomized
solutions.  The obvious idea is to follow the proof of
Theorem~\ref{t:cs} to show that a randomized recovery guarantee yields
a randomized protocol for \index.  Since even randomized protocols for
the latter problem require linear communication, this would imply the
desired lower bound.  

The first attempt at modifying the proof of Theorem~\ref{t:cs} 
has Alice and Bob using the public coins to pick a sensing matrix
$\A$ at random from the assumed distribution --- thus, $\A$ is known
to both Alice and Bob with no communication.
Given the choice of
$\A$, Alice sends $\Ay$ to Bob as before and Bob runs the assumed
recovery algorithm $R$.  With probability at least 2/3, the result is
a vector $\yh$ from which Bob can recover $\bfx_{\log n}$.  The issue is
that Bob has to run $R$ on $\log n$ different inputs, once to recover
each of $\bfx_1,\ldots,\bfx_{\log n}$, and there is a failure probability
of $\tfrac{1}{3}$ each time.

The obvious fix is to reduce the failure probability by independent
trials.  So Alice and Bob use the public coins to pick
$\ell=\Theta(\log \log n)$ matrices $\A^1,\ldots,\A^{\ell}$
independently from the assumed distribution.  Alice sends
$\A^1\bfy,\ldots,\A^{\ell}\bfy$ to Bob, and Bob runs the recovery
algorithm $R$ on each of them and computes the corresponding vectors
$\bfx_1,\ldots,\bfx_{\ell} \in X$.
Except with probability $O(1/\log n)$, a majority of the $\bfx_j$'s will
be the vector $\bfx_{\log n}$.  By a Union Bound over the $\log n$
iterations of Bob's recovery algorithm, Bob successfully reconstructs
each of $\bfx_1,\ldots,\bfx_{\log n}$ with probability at least $2/3$,
completing the randomized protocol for \index.  This
protocol has communication cost $O(m \log n \log \log n)$ and the
lower bound for \index is $\Omega(\log |X| \log n)$, which yields a
lower bound of $m = \Omega(\lnk / \log \log n)$ for randomized
recovery.

The reduction in the proof of Theorem~\ref{t:cs} can be modified in a
different way to avoid the $\log \log n$ factor and prove the same
lower bound of $m=\Omega(\lnk)$ that we established for deterministic
recovery.  The trick is to modify the problem being reduced from
(previously \index) so that it becomes easier --- and therefore
solvable assuming only a randomized recovery guarantee --- subject to
its randomized communication complexity remaining linear.  

The modified
problem is called \aindex --- it's a contrived problem
but has proved technically convenient in several applications.
Alice gets an input $\bfx \in \{0,1\}^{\ell}$.  Bob gets an index $i \in
\{1,2,\ldots,\ell\}$ {\em and also the subsequent bits
  $x_{i+1},\ldots,x_{\ell}$ of Alice's input}.
This problem is obviously only easier than \index, but it's easy to
show that its deterministic one-way communication complexity is
$\ell$ (see the Exercises).  With some work, it can be shown that its
randomized one-way communication complexity is~$\Omega(\ell)$
(see~\cite{B+04,D+10,M+98}).

The reduction in Theorem~\ref{t:cs} is easily adapted to show that a
randomized approximate sparse recovery guarantee with matrices with
$m$ rows yields a randomized one-way communication protocol for
\aindex with $\lxn$-bit inputs with communication cost $O(m \log n)$ (and
hence $m = \Omega(\log |X|) = \Omega(\lnk)$).  We interpret Alice's
input in the same way as before.  Alice and Bob use the public coins
to a pick a matrix $\A$ from the assumed distribution and Alice sends
$\Ay$ to Bob.  Bob is given an index $i \in
\{1,2,\ldots,\lxn\}$ that belongs to some block $j$.  Bob is also
given all bits of Alice's input after the $i$th one.  These bits
include $\enc(\bfx_{j+1}),\ldots,\enc(\bfx_{\log n})$, so Bob can simply
compute $\bfz$ as in~\eqref{eq:z} (with no error) and invoke the
recovery algorithm $R$ (once).  Whenever $\A$ 
is such that the guarantee~\eqref{eq:req2} holds for $\bfy-\bfz$, 
Bob will successfully reconstruct $\bfx_j$ and therefore compute the
correct answer to \aindex.

\subsection{Digression}

One could ask if communication complexity is ``really needed'' for
this proof of Theorem~\ref{t:cs}.  Churlish observers might complain
that, due to the nature of communication complexity lower bounds (like
those last lecture), this
proof of Theorem~\ref{t:cs}  is ``just counting.''\footnote{Critics
  said (and perhaps still say) the same thing about the probabilistic
  method~\citep{AS08}.}  While not wrong, this attitude is
counterproductive.  The fact is that adopting the language and mindset
of communication complexity has permitted researchers to prove results
that had previously eluded many smart people --- in this case, the ends
justifies the means.

The biggest advantage of using the language of communication complexity
is that one naturally thinks in terms of reductions between different
lower bounds.\footnote{Thinking in terms of reductions seems to be a
  ``competitive advantage'' of theoretical computer scientists ---
  there are several examples where a reduction-based approach yielded
  new progress on old and seemingly intractable mathematical problems.}
Reductions can repurpose a single counting argument,
like our lower bound for \index, for lots of different problems.
Many of the more important lower bounds derived from communication
complexity, including today's main result, involve quite clever
reductions, and it would be difficult to devise from scratch the
corresponding counting arguments.

\chapter{Boot Camp on Communication Complexity}
\label{cha:boot-camp-comm}

\section{Preamble}

This lecture covers the most important basic facts about deterministic
and randomized communication protocols in the general two-party model,
as defined by \cite{Y79}.
Some version of this lecture would normally be the first lecture in a
course on communication complexity.  How come it's the fourth one here? 

The first three lectures were about one-way communication complexity
--- communication protocols where there is only one message, from
Alice to Bob ---
and its applications.  One reason we started with the one-way model is
that several of the ``greatest hits'' of algorithmic lower bounds via
communication complexity, such as space lower bounds for streaming
algorithms and row lower bounds for compressive sensing matrices,
already follow from communication lower bounds for one-way protocols.
A second reason is that considering only one-way protocols is a gentle
introduction to what communication protocols look like.  There
are already some non-trivial one-way protocols, like our randomized
protocol for \eq.  On the other hand, proving lower
bounds for one-way protocols is much easier than proving them for
general protocols, so it's also a good introduction to lower bound
proofs.

The rest of our algorithmic applications require
stronger lower bounds that apply to more than just one-way
protocols.  This lecture gives a ``boot camp'' on the basic model.
 We won't say much about applications in this
lecture, but the final five lectures all focus on applications.
We won't prove any hard results today, and focus instead on
definitions and vocabulary, examples, and some easy results. 
One point of today's lecture is to get a feel for what's involved in
proving a communication lower bound for general protocols.  It
generally boils down to a conceptually simple, if sometimes
mathematically challenging, combinatorial problem --- proving that a
large number of ``rectangles'' of a certain type are need to cover a
matrix.\footnote{There are many other methods for proving
  communication lower bounds, some quite deep and exotic (see
  e.g.~\cite{lee}), but all of our algorithmic applications can
  ultimately be derived from combinatorial covering-type arguments.
For example,
  we're not even going to mention the famous ``rank lower bound.''
For your edification, some other lower bound methods are discussed in
the Exercises.}

\section{Deterministic Protocols}\label{s:det}

\subsection{Protocols}

We are still in the two-party model, where Alice has an input $\bfx \in X$
unknown to Bob, and Bob has an input $\bfy \in Y$ unknown to Alice.  
(Most commonly, $X=Y=\{0,1\}^n$.)
A deterministic communication protocol specifies, as function of the
messages sent so far, whose turn it is to speak.  A protocol always
specifies when the communication has ended and, in each end state,
the value of the computed bit.  Alice and Bob can coordinate in
advance to decide upon the protocol, and both are assumed to
cooperative fully.  The only constraint faced by the players is that
what a player says can depend only on what the player knows --- his or
her own input, and the history of all messages sent so far.

Like with one-way protocols, we define the cost of a protocol as the
maximum number of bits it ever sends, ranging over all inputs.
The communication complexity of a function is then the minimum
communication cost of a protocol that correctly computes
it.

The key feature of general communication protocols absent from the
special case of one-way protocols is {\em interaction} between the two
players.  Intuitively, interaction should allow the players to
communicate much more efficiently.  Let's see this in a concrete example.

\subsection{Example: \cis}

The following problem might seem contrived, but it is fairly central
in communication complexity.  There is a graph $G=(V,E)$ with $|V|=n$
that is known to both players.  Alice's private input is a clique $C$
of $G$ --- a subset of vertices such that $(u,v) \in E$ for every 
distinct $u,v
\in C$.  Bob's private input is an independent set $I$ of $G$ --- a
subset of vertices such that $(u,v) \not\in E$ for every 
distinct $u,v \in I$.
(There is no requirement that $C$ or $I$ is maximal.)  Observe that
$C$ and $I$ are either disjoint, or they intersect in a single vertex
(Figure~\ref{f:cis}).  The players' goal is to figure out which of
these two possibilities is the case.  Thus, this problem is a special
case of \disj, where players' sets are not arbitrary but
rather a clique and an independent set from a known graph.

\begin{figure}
\centering
\includegraphics[width=.8\textwidth]{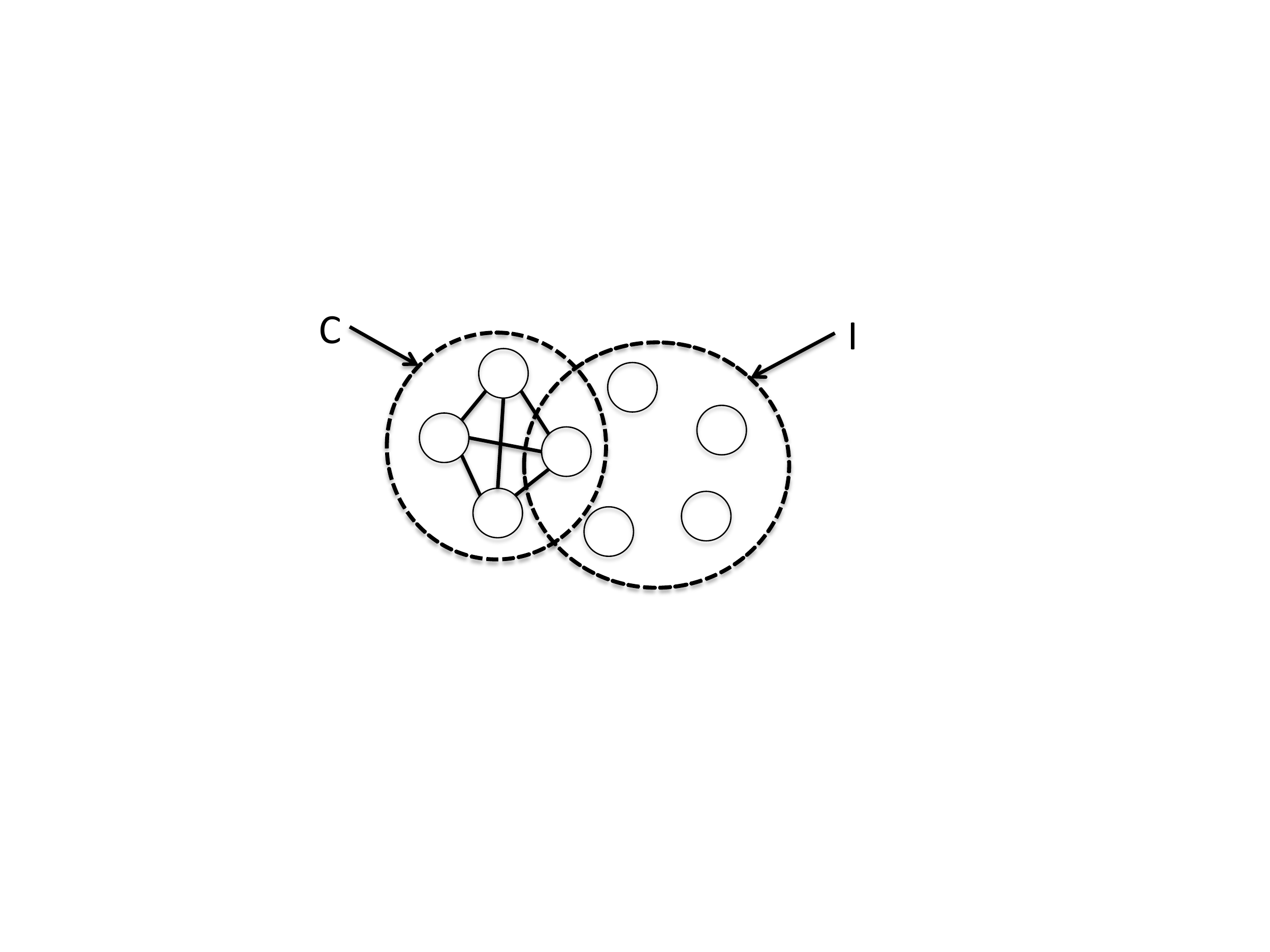}
\caption[A clique and an independent set]{A clique~$C$ and an independent set~$I$ overlap in zero or
  one vertices.}\label{f:cis}
\end{figure}

The naive communication protocol for solving the problem using
$\Theta(n)$ bits  --- Alice can send the characteristic vector of $C$
to Bob, or Bob the characteristic vector of $I$ to Alice, and then the
other player computes the correct answer.  Since the number of cliques
and independent sets of a graph is generally exponential in the number
$n$ of vertices, this protocol cannot be made significantly more
communication-efficient via a smarter encoding.  An easy reduction
from \index shows that one-way protocols, including
randomized protocols, require $\Omega(n)$ communication (exercise).

The players can do much better by interacting.  Here is the protocol.
\begin{enumerate}

\item If there is a vertex $v \in C$ with $\deg(v) < \tfrac{n}{2}$,
  then Alice sends the name of an arbitrary such vertex to Bob
  ($\approx \log_2 n$ bits).  

\begin{enumerate}

\item Bob announces whether or not $v \in I$ (1 bit).  If so, the
  protocol terminates with conclusion ``not disjoint.''

\item Otherwise, Alice
  and Bob recurse on the subgraph $H$ induced by $v$ and its neighbors.

[Note: $H$ contains at most half the nodes of $G$.  It contains all of
  $C$ and, if $I$ intersects $C$, it contains the vertex in their
  intersection.  $C$ and $I$ intersect in $G$ if and only if their
  projections to $H$ intersect in $H$.]

\end{enumerate}

\item Otherwise, Alice sends a ``NULL'' message to Bob ($\approx
  \log_2 n$ bits).

\item If there is a vertex $v \in I$ with $\deg(v) \ge \tfrac{n}{2}$,
  then Bob sends the name of an arbitrary such vertex to Bob ($\approx
  \log_2 n$ bits).

\begin{enumerate}

\item Alice announces whether or not $v \in C$ (1 bit).
If so, the protocol terminates (``not disjoint'').

\item If not, Alice
  and Bob recurse on the subgraph $H$ induced by $v$ and its
  non-neighbors. 

[Note: $H$ contains at most half the nodes of $G$.  It contains all of
  $I$ and, if $C$ intersects $I$, it contains the vertex in their
  intersection.  Thus the function's answer in $H$ is the same as that
  in $G$.]

\end{enumerate}

\item Otherwise, Bob terminates the protocol and declares
  ``disjoint.''

[Disjointness is obvious since, at this point in the protocol, we know
  that $\deg(v) < \tfrac{n}{2}$ for all $v \in C$ and
  $\deg(v) \ge \tfrac{n}{2}$ for all $v \in I$.]

\end{enumerate}

Since each iteration of the protocol uses $O(\log n)$ bits
of communication and cuts the number of vertices of the graph in half (or
terminates), the total communication is $O(\log^2 n)$.  As previously
noted, such a result is impossible without interaction between the
players.

\subsection{Trees and Matrices}

The \cis problem clearly demonstrates that we
need new lower bound techniques to handle general communication
protocols --- the straightforward Pigeonhole Principle arguments that
worked for one-way protocols are not going to be good enough.
At first blush this might seem intimidating --- communication
protocols can do all sorts of crazy things, so how can we reason about
them in a principled way?  How can we connect properties of a protocol
to the function that it computes?  Happily, we can
quickly build up some powerful machinery for answering these questions.

First, we observe that deterministic communication protocols are
really just binary trees.  We'll almost never use
this fact directly, but it should build some confidence that protocols
are familiar and elementary mathematical objects.  

The connection is easiest to see by example; see Figure~\ref{f:tree}.
Consider the following protocol for solving \eq
with $n=2$ (i.e., $f\inputs = 1$ if and only if $\bfx=\bfy$).
Alice begins by sending her first bit.  
If Bob's first bit is different, he terminates the protocol and
announces ``not equal.''
If Bob's first bit is the same, then he transmits the same bit back.
In this case, Alice then sends her second bit.  At this point, Bob
knows Alice's whole input and can therefore compute the correct
answer.

\begin{figure}
\centering
\includegraphics[width=.8\textwidth]{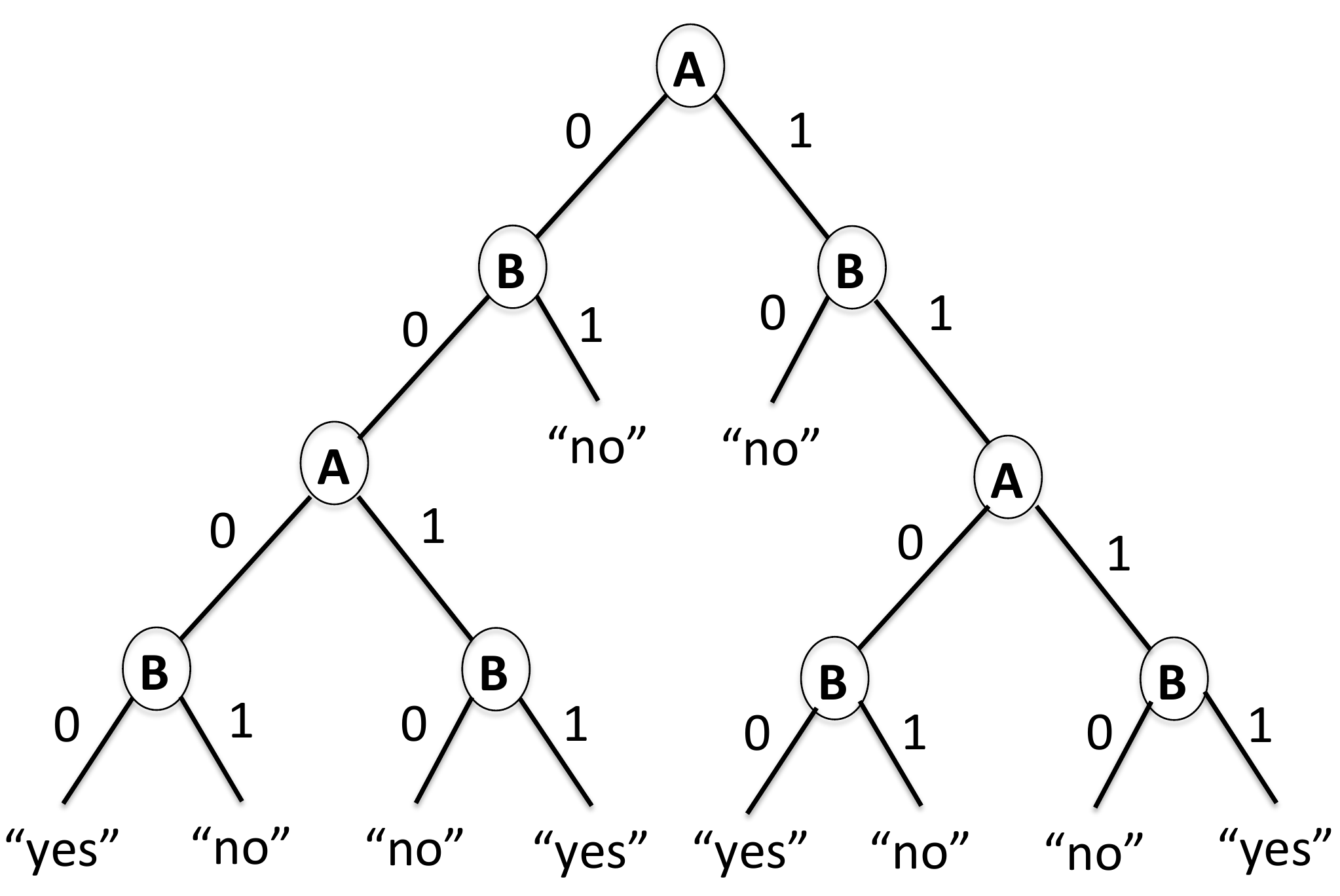}
\caption[The binary tree induced by a protocol for \eq]{The binary tree induced by a communication protocol for \eq
  with $n=2$.}\label{f:tree}
\end{figure}

In Figure~\ref{f:tree}, each node corresponds to a possible state of
the protocol, and is labeled with the player whose turn it is to
speak.  Thus the labels alternate with the levels, with the root
belonging to Alice.\footnote{In general, players need not alternate
  turns in a communication protocol.}  There are 10 leaves,
representing the possible end states of the protocol.  There are two
leaves for the case where Alice and Bob have different first bits and
the protocol terminates early, and eight leaves for the remaining
cases where Alice and Bob have the same first bit.  Note that the
possible transcripts of the protocol are in one-to-one correspondence
with the root-leaf nodes of the tree --- we use leaves and transcripts
interchangeably below.

We can view the leaves as a partition $\{ Z(\ell)\}$ of the input
space $X \times Y$, with $Z(\ell)$ the inputs $\inputs$ such that the
protocol terminates in the leaf $\ell$.
In our example, there
are~10 leaves for the~16 possible inputs $\inputs$, so
different inputs can generate the same transcript --- more on this
shortly.  

Next note that we can represent a function (from $\inputs$ to
$\{0,1\}$) a matrix.  In contrast to the visualization exercise
above, we'll use this matrix representation {\em all the time}.
The rows are labeled with the set $X$ of possible inputs of Alice, the
columns with the set $Y$ of possible inputs of Bob.  Entry $\inputs$
of the matrix is $f\inputs$.  Keep in mind that this matrix is fully
known to both Alice and Bob when they agree on a protocol.

For example, suppose that $X=Y=\{0,1\}^2$, resulting in $4 \times 4$
matrices.  \eq then corresponds to the identity matrix:
\begin{equation}\label{eq:eq}
\bordermatrix{\text{}& 00 & 01 & 10 & 11\cr
                00 & 1 &  0  & 0 & 0\cr
                01 & 0  &  1 & 0 & 0\cr
                10 & 0 & 0 & 1 & 0\cr
                11 & 0 & 0 & 0 & 1}
\end{equation}
If we define the \gt function as~1 whenever $\bfx$ is at least $\bfy$
(where $\bfx$ and $\bfy$ are interpreted as non-negative integers, written in
binary), then we just fill in the lower triangle with 1s:
\[
\bordermatrix{\text{}& 00 & 01 & 10 & 11\cr
                00 & 1 &  0  & 0 & 0\cr
                01 & 1  &  1 & 0 & 0\cr
                10 & 1 & 1 & 1 & 0\cr
                11 & 1 & 1 & 1 & 1}
\]
We also write out the matrix for \disj, which is somewhat
more inscrutable:
\[
\bordermatrix{\text{}& 00 & 01 & 10 & 11\cr
                00 & 1 &  1  & 1 & 1\cr
                01 & 1  &  0 & 1 & 0\cr
                10 & 1 & 1 & 0 & 0\cr
                11 & 1 & 0 & 0 & 0}
\]

\subsection{Protocols and Rectangles}

How can we reason about the behavior of a protocol?  Just visualizing
them as trees is not directly useful.  We know that simple Pigeonhole
Principle-based arguments are not strong enough, but it still feels
like we want some kind of counting argument.

To see what might be true, let's run the 2-bit \eq protocol
depicted in Figure~\ref{f:tree} and track its progress using the
matrix in~\eqref{eq:eq}.  Put yourself in the shoes of an outside
observer, who knows neither $\bfx$ nor $\bfy$, and makes inferences about
$\inputs$ as the protocol proceeds.  When the protocol terminates,
we'll have carved up the matrix into 10 pieces, one for each leaf of
protocol tree --- the protocol transcript reveals the leaf to an
outside observer, but nothing more.

Before the protocol beings, all 16 inputs are fair game.
After Alice sends her first bit, the outside observer can narrow down
the possible inputs into a set of~8 --- the top~8 if Alice sent a 0,
the bottom~8 if she sent a~1.  The next bit sent gives away whether or not
Bob's first bit is a 0 or 1, so the outsider observer learns which
quadrant the input lies in.  
Interestingly, in the northeastern and
southwestern quadrants, all of the entries are 0.  In these cases, even
though ambiguity remains about exactly what the input $\inputs$ is,
the function's value $f\inputs$ has been determined (it is~0, whatever
the input).  It's no coincidence that these two regions correspond to
the two leaves of the protocol in Figure~\ref{f:tree} that stop early,
with the correct answer.  If the protocol continues further, then
Alice's second bit splits the northwestern and southeastern quadrants
into two, and Bob's final bit splits them again, now into singleton
regions.  In these cases, an outside observer learns the entire input
$\inputs$ from the protocol's transcript.\footnote{It's also
  interesting to do an analogous thought experiment from the
  perspective of one of the players.  For example, consider Bob's
  perspective when the input is (00,01).  Initially Bob knows that the
  input lies in the second column but is unsure of the row.  After
  Alice's first message, Bob knows that the input is in the second column
  and one of the first two rows.  Bob still cannot be sure about the
  correct answer, so the protocol proceeds.}

What have we learned?  We already knew that every protocol induces a
partition of the input space $X \times Y$, with one set for each leaf
or, equivalently, for each distinct transcript.
At least for the particular protocol that we just studied, each of the
sets has a particularly nice submatrix form (Figure~\ref{f:eq}).
This is true in general, in the following sense.

\begin{figure}
\centering
\includegraphics[width=.25\textwidth]{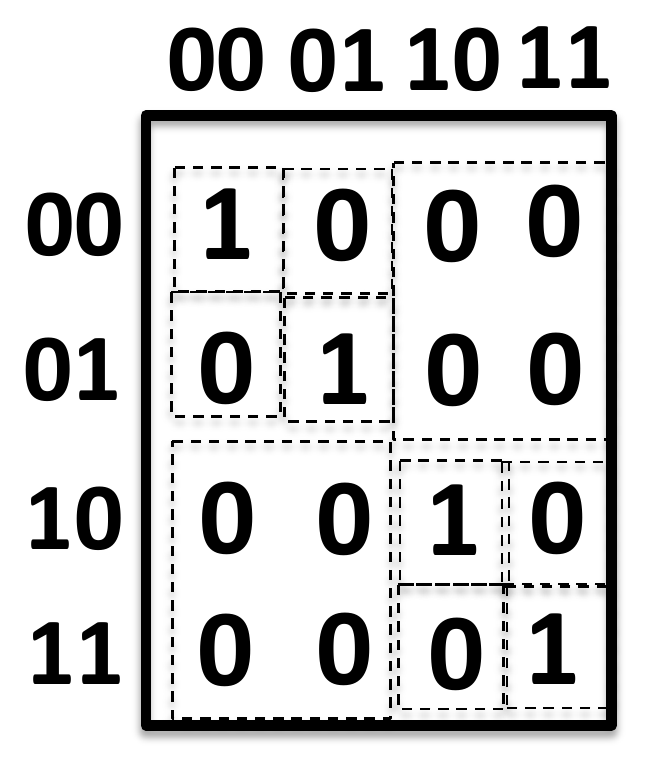}
\caption[Partition of the input space $X \times Y$]{The partition of the input space $X \times Y$ according to
  the~10 different transcripts that can be generated by the \eq
  protocol.}\label{f:eq}
\end{figure}

\begin{lemma}[Rectangles]\label{l:rect}
For every transcript $\bfz$ of a deterministic protocol $P$, the set of
inputs $\inputs$ that generate $\bfz$ are a {\em 
  rectangle}, of the form $A \times B$ for $A \sse X$ and $B \sse Y$.
\end{lemma}

A rectangle just means a subset of the input space $X
\times Y$ that can be written as a product.  For example, the set $\{
(00,00),(11,11) \}$ is {\em not} a rectangle, while the
set $\{(00,00),(11,00),(00,11),(11,11) \}$ is.  In general, a subset $S
\sse X \times Y$ is a rectangle if and only if it is
closed under ``mix and match,'' meaning that whenever $(\bfx_1,\bfy_1)$
and $(\bfx_2,\bfy_2)$ are in $S$, so are $(\bfx_1,\bfy_2)$ and $(\bfx_2,\bfy_1)$
(see the Exercises).

Don't be misled by our example (Figure~\ref{f:eq}), where the
rectangles induced by our protocol happen to be
``contiguous.''  For example, if we keep the protocol the same but
switch the order in which we write down the rows and columns
corresponding to 01 and 10, we get an analogous decomposition in which
the two large rectangles are not contiguous.  
In general, you shouldn't even think of $X$ and $Y$ as ordered sets.
Rectangles are sometimes called {\em combinatorial} rectangles to
distinguish them from ``geometric'' rectangles and to emphasize this
point.

Lemma~\ref{l:rect} is extremely important, though its proof is
straightforward --- we just follow the protocol like in our example
above.
Intuitively, each step of a protocol allows an outside observer to
narrow down the possibilities for $\bfx$ while leaving the possibilities
for $\bfy$ unchanged (if Alice speaks) or vice versa (if Bob speaks).

\vspace{.1in}
\begin{prevproof}{Lemma}{l:rect}
Fix a deterministic protocol $P$.
We proceed by induction on the number of bits exchanged.  
For the base case, all inputs $X \times Y$ begin with the empty
transcript.  For the inductive step, consider an arbitrary $t$-bit
transcript-so-far $\bfz$ generated by $P$, with $t \ge 1$.  Assume that
Alice was the most recent player to speak; the other case is
analogous.  Let $\bfz'$ denote $\bfz$ with the final bit $b \in \{0,1\}$
lopped off.  By the inductive hypothesis,
the set of inputs that generate $\bfz'$ has the form $A \times B$.
Let $A_b \sse A$ denote the inputs $\bfx \in A$ such that, in the
protocol $P$, Alice sends the bit~$b$ given the transcript $\bfz'$.
(Recall that the message sent by a player is a function only of his or
her private input and the history of the protocol so far.)
Then the set of inputs that generate $\bfz$ are $A_b \times B$,
completing the inductive step.
\end{prevproof}

Note that Lemma~\ref{l:rect} makes no reference to a function $f$ ---
it holds for any deterministic protocol, whether or not it computes a
function that we care about.  In Figure~\ref{f:eq}, we can clearly see
an additional property of all of the rectangles --- with respect to
the matrix in~\eqref{eq:eq}, every rectangle is {\em  monochromatic},
meaning all of its entries have the same value.  This is true for any
protocol that correctly computes a function $f$.

\begin{lemma}\label{l:mono}
If a deterministic protocol $P$ computes a function $f$, then every
rectangle induced by $P$ is monochromatic in the matrix
$M(f)$. 
\end{lemma}

\begin{proof}
Consider an arbitrary combinatorial rectangle $A \times B$ induced by
$P$, with all inputs in $A \times B$ inducing the same transcript.
The output of $P$ is constant on $A \times B$.  Since $P$ correctly
computes $f$, $f$ is also constant on $A \times B$.
\end{proof}

Amazingly, the minimal work we've invested so far already yields a
powerful technique for lower bounding the deterministic communication
complexity of functions.

\begin{theorem}\label{t:part}
Let $f$ be a function such that every partition of $M(f)$ into
monochromatic rectangles requires at least $t$ rectangles.  Then the
deterministic communication complexity of $f$ is at least $\log_2 t$.
\end{theorem}

\begin{proof}
A deterministic protocol with communication cost $c$ can only generate
$2^c$ distinct transcripts --- equivalently, its (binary) protocol
tree can only have $2^c$ leaves.  If such a protocol computes the
function $f$, then by Lemmas~\ref{l:rect} and~\ref{l:mono} it
partitions $M(f)$ into at most $2^c$ monochromatic rectangles.  By
assumption, $2^c \ge t$ and hence $c \ge \log_2 t$.
\end{proof}

Rather than applying Theorem~\ref{t:part} directly, we'll almost always
be able to prove a stronger and simpler
condition.  To partition a matrix, one needs to cover all of its
entries with disjoint sets.  The disjointness condition is annoying.
So by a {\em covering} of a 0-1 matrix, we
mean a collection of subsets of entries whose union includes all of its
elements --- overlaps between these sets are
allowed.  See Figure~\ref{f:cover}.

\begin{figure}
\centering
\includegraphics[width=.25\textwidth]{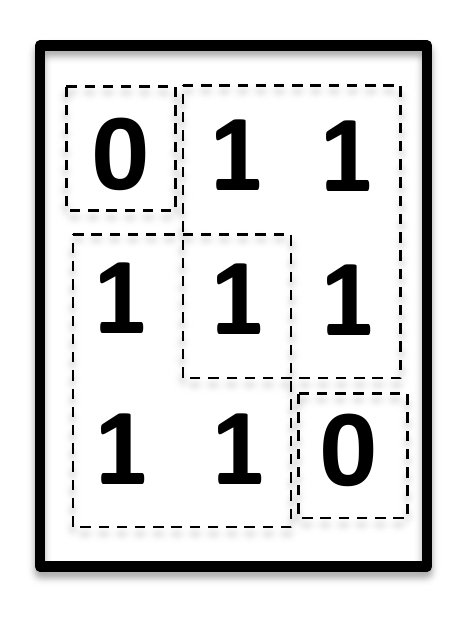}
\caption[A covering by four monochromatic rectangles that is not a
         partition]{A covering by four monochromatic rectangles that is not a
  partition.}\label{f:cover}
\end{figure}

\begin{corollary}\label{cor:cover}
Let $f$ be a function such that every covering of $M(f)$ by
monochromatic rectangles requires at least $t$ rectangles.  Then the
deterministic communication complexity of $f$ is at least $\log_2 t$.
\end{corollary}

Communication complexity lower bounds proved using covers ---
including all of those proved in Section~\ref{ss:lbs} ---
automatically apply also to more general 
``nondeterministic'' communication protocols, as well as
randomized protocols with 1-sided
error.  We'll discuss this more next lecture, when it will be relevant.

\subsection{Lower Bounds for \eq and \disj}\label{ss:lbs}

Armed with Corollary~\ref{cor:cover}, we can quickly prove
communication lower bounds for some functions of interest.  For
example, recall that when $f$ is the \eq function, the matrix
$M(f)$ is the identity.  The key observation about this matrix is:
{\em a monochromatic rectangle that includes a ``1'' contains only one
  element}.  The reason is simple: such a rectangle is not allowed to
contain any 0's since it is monochromatic, and if it included a second
1 it would pick up some 0-entries as well (recall that rectangles
are closed under ``mix and match'').  Since there are $2^n$ 1's in the
matrix, every covering by monochromatic rectangles (even of just the
1's) has size $2^n$.
\begin{corollary}
The deterministic communication complexity of \eq
is at least $n$.\footnote{The 0s can be covered using another
  $2^n$ monochromatic rectangles, one per row (rectangles need not be
  ``contiguous''!).  This gives a lower bound of $n+1$.  The trivial
  upper has Alice sending her input to Bob and Bob announcing the
  answer, which is a $(n+1)$-bit protocol.  Analogous ``+1'' improvements 
  are possible for the other examples in this section.}
\end{corollary}
The exact same argument gives the same lower bound for the \gt
function.
\begin{corollary}
The deterministic communication complexity of \gt
is at least $n$.
\end{corollary}

We can generalize this argument as follows. A {\em fooling set} for a
function $f$ is a subset $F \sse \Inputs$ of inputs such that:
\begin{itemize}

\item [(i)] $f$ is constant on $F$;

\item [(ii)] for each distinct pair $(\bfx_1,\bfy_1),(\bfx_2,\bfy_2) \in F$,
  at least one of $(\bfx_1,\bfy_2),(\bfx_2,\bfy_1)$ has the opposite
  $f$-value.

\end{itemize}
Since rectangles are closed under the ``mix and match'' operation,
(i) and (ii) imply that every monochromatic rectangle contains at most
one element of $F$.
\begin{corollary}\label{cor:fool}
If $F$ is a fooling set for $f$, then the deterministic communication
complexity of $f$ is at least $\log_2 |F|$.
\end{corollary}
For \eq and \gt, we were effectively
using the fooling set $F = \{ (\bfx,\bfx) \,:\, \bfx \in \{0,1\}^n \}$.

The fooling set method is powerful enough to prove a strong lower
bound on the deterministic communication complexity of \disj.
\begin{corollary}\label{cor:disj_cover}
The deterministic communication complexity of \disj
is at least $n$.
\end{corollary}

\begin{proof}
Take $F = \{ (\bfx,\ones-\bfx) \,:\, \bfx \in \{0,1\}^n \}$ --- or in set
notation, $\{ (S,S^c) \,:\, S \sse \{1,2,\ldots,n\} \}$.  The set $F$ is
a fooling set --- it obviously consists only of ``yes'' inputs of
\disj, while for every $S \neq T$, either $S \cap T^c \neq
\emptyset$ or $T \cap S^c \neq \emptyset$ (or both).  See
Figure~\ref{f:disj}.  Since $|F|=2^n$, Corollary~\ref{cor:fool}
completes the proof.
\end{proof}

\begin{figure}
\centering
\includegraphics[width=.65\textwidth]{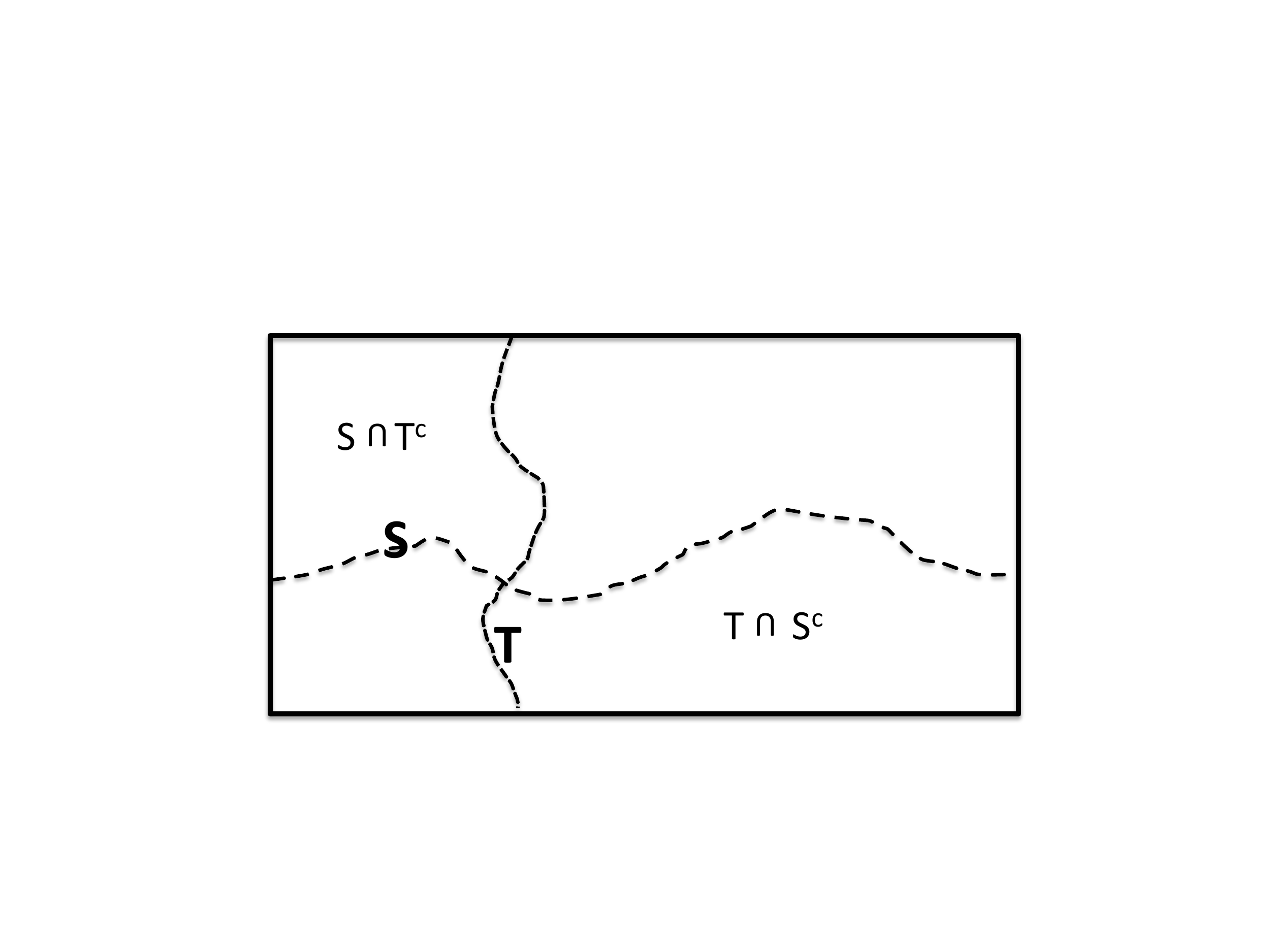}
\caption[If $S$ and $T$ are different sets, then either $S$ and $T^c$
or $T$ and $S^c$ are not disjoint]{If $S$ and $T$ are different sets, then either $S$ and $T^c$ or 
$T$ and $S^c$ are not disjoint.}\label{f:disj}
\end{figure}

\subsection{Take-Aways}

A key take-away point from this section is that, using covering
arguments, we can prove the lower bounds that we want on the
deterministic communication complexity of many functions of interest.  
These lower bounds apply also to nondeterministic protocols (discussed
next week) and randomized protocols with 1-sided error.

As with one-way communication complexity, 
proving stronger lower bounds that apply
also to randomized protocols with two-sided error is more challenging.
Since 
we're usually perfectly happy with a good randomized algorithm --- recall
the $F_2$ estimation algorithm from Section~\ref{s:f2} --- such lower
bounds are very relevant for algorithmic applications.  They are our
next topic.

\section{Randomized Protocols}

\subsection{Default Parameter Settings}

Our discussion of randomized one-way communication protocols in
Section~\ref{s:rand} remains equally relevant for general protocols.  Our
``default parameter settings'' for such protocols will be the same.

\textbf{Public coins.}  By default, we work with public-coin protocols,
where Alice and Bob have shared randomness in the form of
an infinite sequence of perfectly random bits
written on a blackboard in public view.  Such protocols are more
powerful than private-coin protocols, but not
by much (Theorem~\ref{t:newman}).  Recall that public-coin randomized
protocols are equivalent 
to distributions over deterministic protocols.

\textbf{Two-sided error.} We allow a protocol to error with constant
probability ($\tfrac{1}{3}$ by default), whether or not the correct
answer is ``1'' or ``0.''  This is the most permissive error model.

\textbf{Arbitrary constant error probability.}  Recall that all constant
error probabilities in $(0,\tfrac{1}{2})$ are the same --- changing
the error changes the randomized communication complexity by only a
constant factor (by the usual ``independent trials'' argument,
detailed in the exercises).
Thus for upper bounds, we'll be content to achieve
error 49\%; for lower bounds, it is enough to
rule out low-communication protocols with error \%1.

\textbf{Worst-case communication.}  We define the communication cost of a
randomized protocol as the maximum number of bits ever communicated,
over all choices of inputs and coin flips.  Measuring the expected
communication (over the protocol's coin flips) could reduce the
communication complexity of a problem, but only by a constant factor.

\subsection{Newman's Theorem: Public- vs.\ Private-Coin Protocols}

We mentioned a few times that, for our purposes, it usually won't
matter whether we consider public-coin or private-coin randomized
protocols.  What we meant is the following result.

\begin{theorem}[Newman's Theorem~\citeyearpar{N91}]\label{t:newman}
If there is a public-coin protocol for a function~$f$ with $n$-bit inputs
that has two-sided error~$1/3$ and communication cost~$c$, then there
is a private-coin protocol for the problem that has two-sided error
$1/3$ and communication cost $O(c + \log n)$.
\end{theorem}

Thus, for problems with public-coin randomized communication
complexity $\Omega(\log n)$, like most of the problems that we'll
study in this course, there is no difference 
between the communication complexity of
the public-coin and private-coin
variants (modulo constant factors).

An interesting exception is \eq.  Last lecture,
we gave a public-coin protocol --- one-way, even --- with constant
communication complexity.  Theorem~\ref{t:newman} only implies an
upper bound of $O(\log n)$ communication for private-coin protocols.
(One can also give such a private-coin protocol directly, see the
Exercises.)  There is also a matching lower bound of
$\Omega(\log n)$ for the private-coin communication complexity of 
\eq.  (This isn't very hard to prove, but we won't have
an occasion to do it.)  Thus public-coin protocols can save
$\Theta(\log n)$ bits of communication over private-coin protocols,
but no more. 

\vspace{.1in}
\noindent
\begin{prevproof}{Theorem}{t:newman}
Let $P$ denote a public-coin protocol with two-sided error $1/3$.
We begin with a thought experiment.  Fix an input $\inputs$, with
$\bfx,\bfy \in \{0,1\}^n$.  If we run $P$ on this input, a public string
$\bfr_1$ of random bits is consumed and the output of the protocol is
correct with probability at least $2/3$.
If we run it again, a second (independent) random string $\bfr_2$ is
consumed and another (independent) answer is given, again correct with
probability at least $2/3$.  After $t$ such trials and the consumption of
random strings $\bfr_1,\ldots,\bfr_t$, $P$ produces $t$ answers.  We expect
at least $2/3$ of these to be correct, and Chernoff bounds
(with $\delta = \Theta(1)$ and $\mu = \Theta(t)$) imply that at least
60\% of these answers are correct with probability at least $1 - \exp
\{ -\Theta(t) \}$.  

We continue the thought experiment by taking a Union Bound over the
$2^n \cdot 2^n = 2^{2n}$ choices of the input $\inputs$.  With
probability at least $1 - 2^{2n} \cdot \exp \{ -\Theta(t) \}$ over the
choice of $\bfr_1,\ldots,\bfr_t$, for every input $\inputs$, running the
protocol $P$ with these random 
strings yields at least $.6t$ (out of $t$) correct
answers.  In this event, the single sequence $\bfr_1,\ldots,\bfr_t$ of
random strings ``works'' simultaneously for all inputs $\inputs$.
Provided we take $t = cn$ with a large enough constant $c$,
this probability is positive.  In particular, such a set
$\bfr_1,\ldots,\bfr_t$ of random strings exist.

Here is the private-coin protocol.
\begin{itemize}

\item [(0)] Before receiving their inputs, Alice and Bob agree on a set
  of strings 
  $\bfr_1,\ldots,\bfr_t$ with the property that, for every input
  $\inputs$, running $P$ $t$ times with the random strings
  $\bfr_1,\ldots,\bfr_t$ yields at least 60\% correct answers.

\item [(1)] Alice picks an index $i \in \{1,2,\ldots,t\}$ uniformly at
  random and sends it to Bob.  This requires $\approx \log_2 t =
  \Theta(\log n)$ bit of communication (recall $t = \Theta(n)$).

\item [(2)] Alice and Bob simulate the private-coin protocol $P$ as if
  they had public coins given by $\bfr_i$.

\end{itemize}
By the defining property of the $\bfr_i$'s, this (private-coin) protocol
has error 40\%.  As usual, this can be reduced to $1/3$ via a constant
number of independent repetitions followed by taking the majority
answer.  The resulting communication cost is $O(c + \log n)$, as claimed.
\end{prevproof}

We stated and proved Theorem~\ref{t:newman} for general protocols, but
the exact same statement holds (with the same proof) for the one-way
protocols that we studied in Lectures~\ref{cha:data-stre-algor}--~\ref{cha:lower-bounds-compr}.

\subsection{Distributional Complexity}

Randomized protocols are significantly harder to reason about than
deterministic ones.  For example, we've seen that a deterministic
protocol can be thought of as a partition of the input space into
rectangles.  A randomized protocol is a distribution over such
partitions.  While a deterministic protocol
that computes a function $f$ induces only monochromatic rectangles,
this does not hold for randomized protocols (which can err with some
probability).

We can make our lives somewhat simpler by using Yao's Lemma
to translate distributional lower bounds for deterministic protocols
to worst-case lower bounds for randomized protocols.  Recall the lemma
from Lecture~\ref{cha:lower-bounds-one} (Lemma~\ref{l:yao}).

\begin{lemma}[\citealt{Y83}]
Let $D$ be a distribution over the space of inputs $\inputs$ to a
communication problem, and $\eps \in (0,\tfrac{1}{2})$.  Suppose that
every deterministic protocol $P$ with
\[
\prob[\inputs \sim D]{\text{$P$ wrong on $\inputs$}} \le \eps
\]
has communication cost at least $k$.  Then every (public-coin)
randomized protocol $R$ with (two-sided) error at most $\eps$ on every
input has communication cost at least $k$.
\end{lemma}

We proved Lemma~\ref{l:yao} in Lecture~\ref{cha:lower-bounds-one} for
one-way protocols, but
the same proof holds verbatim for general communication protocols.
Like in the one-way case, Lemma~\ref{l:yao} is a ``complete'' proof
technique --- whatever the true randomized communication complexity,
there is a hard distribution $D$ over inputs that can in principle be
used to prove it (recall the Exercises).  

Summarizing, proving lower bounds for randomized communication
complexity reduces to:
\begin{enumerate}

\item Figuring out a ``hard distribution'' $D$ over inputs.

\item Proving that every low-communication deterministic protocol 
  has large error w.r.t.\ inputs drawn from $D$.

\end{enumerate}
Of course, this is usually easier said than done.

\subsection{Case Study: \disj}\label{ss:disj_proof}

\subsubsection{Overview}

We now return to the \disj problem.  In
Lecture~\ref{cha:lower-bounds-one} we
proved that the {\em one-way} randomized communication complexity of
this problem is linear (Theorem~\ref{t:disj2}).
We did this by reducing 
\index to \disj --- the former is just a special case of
the latter, where one player has a singleton set (i.e., a standard
basis vector).  
We used Yao's Lemma (with $D$ the uniform distribution)
and a counting argument (about the volume of
small-radius balls in the Hamming cube, remember?) to prove that the
one-way randomized communication complexity of \index is $\Omega(n)$.
Unfortunately, for general communication protocols, the communication
complexity of \index is obviously $O(\log n)$ --- Bob can just send his
index $i$ to Alice using $\approx \log_2 n$ bits, and Alice can
compute the function.  So, it's back to the drawing board.

The following is a major and useful technical achievement.
\begin{theorem}[\citealt{KS92,R92}]\label{t:disj_2sided}
The randomized communication complexity of \disj is
$\Omega(n)$.
\end{theorem}
Theorem~\ref{t:disj_2sided} was originally proved in~\cite{KS92}; the
simplified proof in~\cite{R92} has been more influential.  More
recently, all the cool kids prove Theorem~\ref{t:disj_2sided} using
``information complexity'' arguments; see~\cite{B+02b}.

If you only remember one result from the entire field of communication
complexity, it should be Theorem~\ref{t:disj_2sided}.  The primary reason is
that the problem is unreasonably effective for proving lower bounds
for other algorithmic problems --- almost every subsequent lecture
will include a significant example.  Indeed, many algorithm designers
simply use Theorem~\ref{t:disj_2sided} as a ``black box'' to prove lower
bounds for other problems, without losing sleep over its
proof.\footnote{Similar to, for example, the PCP 
Theorem and the Parallel Repetition
Theorem in the context of hardness of
approximation (see e.g.~\cite{AL96}).}\footnote{There's no shame in this --- 
  life is short and there's lots of theorems that need proving.}
As a bonus, proofs of Theorem~\ref{t:disj_2sided} tend to showcase techniques
that are reusable in other contexts.  

For a trivial consequence of Theorem~\ref{t:disj_2sided} --- see future
lectures for less obvious ones --- let's return to the setting of
streaming algorithms.  Lectures~\ref{cha:data-stre-algor} and~\ref{cha:lower-bounds-one} considered only
one-pass algorithms.  In some contexts, like a telescope that
generates an exobyte of data per day, this is a hard constraint.  In
other settings, like database applications, a small constant number of
passes over the data might be feasible (as an overnight job, for
example).  Communication complexity lower bounds for one-way protocols
say nothing about two-pass algorithms, while those for general
protocols do.  Using Theorem~\ref{t:disj_2sided}, all of our $\Omega(m)$
space lower bounds for 1-pass algorithms become $\Omega(m/p)$ space
lower bounds for $p$-pass algorithms, via the same
reductions.\footnote{A $p$-pass space-$s$ streaming algorithm $S$ induces
  a communication protocol with $O(ps)$ communication, where Alice and
  Bob turn their inputs into data streams, repeatedly feed them into
  $S$, repeatedly sending the memory state of $S$ back and forth to
  continue the simulation.}  For example, we proved such lower bounds
for computing $F_{\infty}$, the highest frequency of an element, even
with randomization and approximation, and for computing $F_0$ or $F_2$
exactly, even with randomization.

So how would one go about proving Theorem~\ref{t:disj_2sided}?  Recall that
Yao's Lemma reduces the proof to exhibiting a hard distribution $D$
(a bit of dark art)
over inputs and proving that all low-communication deterministic
protocols have large error with respect to $D$ (a potentially tough
math problem).  We next discuss each of these steps in turn.

\subsubsection{Choosing a Hard Distribution}

The uniform distribution over inputs is not 
a hard distribution for \disj.  What is
the probability that a random input $\inputs$ satisfies $f\inputs =
1$?  Independently in each coordinate $i$, there is a 25\% probability
that $x_i = y_i = 1$.  Thus, $f\inputs = 1$ with probability
$(3/4)^n$.  This means that the zero-communication protocol that
always outputs ``not disjoint'' has low error with respect to this
distribution.
The moral is that a hard distribution $D$ must, at the very least,
have a constant probability of producing both ``yes'' and ``no''
instances.  

The next idea, motivated by the Birthday Paradox, is to
define $D$ such that each of Alice and Bob receive a random subset of
$\{1,2,\ldots,n\}$ of size $\approx \sqrt{n}$.  Elementary
calculations show that a random instance $\inputs$ from $D$ has
a constant probability of satisfying each of $f\inputs = 1$ and
$f\inputs = 0$.

An obvious issue with this approach is that there is a trivial
deterministic protocol that uses $O(\sqrt{n} \log n)$ communication
and has zero error: Alice (say) just sends her whole input to Bob by
describing each of her $\sqrt{n}$ elements explicitly by name ($\approx
\log_2 n$ bits each).  So there's no way to prove a linear
communication lower bound using this distribution.  
\cite{BFS86} prove that one can at least prove a
$\Omega(\sqrt{n})$ communication lower bound using this distribution,
which is already quite a non-trivial result (more on this below).
They also showed that for every product distribution $D$ --- meaning
whenever the random choices of $\bfx$ and of $\bfy$ are independent ---
there is a zero-error deterministic protocol that uses only
$O(\sqrt{n} \log n)$ bits of communication (see the
Exercises).\footnote{This does not imply that a linear lower bound is
  impossible.  The proof of the converse of Lemma~\ref{l:yao} --- that
  a tight lower bound on the randomized communication complexity 
of a problem can always
  be proved through a distributional lower bound for a suitable choice
  of $D$ --- generally makes use of distributions in which the choices
  of $\bfx$ and $\bfy$ are correlated.}

Summarizing, if we believe that \disj really requires
$\Omega(n)$ communication to solve via randomized protocols, then we need
to find a distribution $D$ that meets all of the following criteria.
\begin{enumerate}

\item There is a constant probability that $f\inputs = 1$ and that
  $f\inputs = 0$.  (Otherwise, a constant protocol works.)

\item Alice and Bob need to usually receive inputs that correspond to
  sets of size $\Omega(n)$.  (Otherwise, one player can explicitly
  communicate his or her set.)

\item The random inputs $\bfx$ and $\bfy$ are correlated.  (Otherwise, the
  upper bound from~\cite{BFS86} applies.)

\item It must be mathematically tractable to prove good lower bounds
  on the error of all deterministic communication protocols that use a
  sublinear amount of communication.

\end{enumerate}

\cite{R92} proposed a distribution that obviously satisfies
the first three properties and, less obviously, also satisfies the
fourth.  It is:
\begin{enumerate}

\item With probability 75\%:

\begin{enumerate}

\item $\inputs$ is chosen uniformly at random subject to:

\begin{enumerate}

\item $\bfx,\bfy$ each have exactly $n/4$ 1's;

\item there is no index $i \in \{1,2,\ldots,n\}$ with $x_i = y_i = 1$
  (so   $f\inputs = 1$).

\end{enumerate}

\end{enumerate}

\item With probability 25\%:

\begin{enumerate}

\item $\inputs$ is chosen uniformly at random subject to:

\begin{enumerate}

\item $\bfx,\bfy$ each have exactly $n/4$ 1's;

\item there is exactly one index $i \in \{1,2,\ldots,n\}$ with $x_i =
  y_i = 1$ (so   $f\inputs = 0$).

\end{enumerate}

\end{enumerate}

\end{enumerate}
Note that in both cases, the constraint on the number of indices $i$
with $x_i=y_i=0$ creates correlation between the choices of $\bfx$ and
$\bfy$.

\subsubsection{Proving Error Lower Bounds via Corruption Bounds}


Even if you're handed a hard distribution over inputs, there remains
the challenging task of proving a good error lower bound on
low-communication deterministic protocols.  
There are multiple methods for doing this, with the {\em corruption
  method} being the most successful one so far.  We outline this
method next.

At a high level, the corruption method is a natural extension of the
covering arguments of Section~\ref{s:det} to protocols that can err.
Recall that for deterministic protocols, the covering approach argues
that every covering of the matrix $M(f)$ of the function $f$ by
monochromatic rectangles requires a lot of rectangles.  In our
examples, we only bothered to argue about the $1$-inputs of
the function.\footnote{Since $f$ has only two outputs, it's almost
  without loss to pick a single output $z \in \{0,1\}$ of $f$ and
  lower bound only
  the number of monochromatic rectangles needed to cover all of the
  $z$'s.}  We'll do something similar here, weighted by the
distribution $D$ and allowing errors --- arguing that there's
significant mass on the 1-inputs of $f$, and that a lot of nearly
monochromatic rectangles are required to cover them all.

Precisely, suppose you have a distribution $D$ over the inputs of a
problem so that the ``1-mass'' of $D$, meaning $\prob[\inputs \sim
  D]{f(x,y)   = 1}$, is at least a constant, say~.5.
The plan is to prove two properties.
\begin{itemize}

\item [(1)] 
For every deterministic
protocol $P$ with error at most a sufficiently small constant $\eps$,
at least 25\% of the 1-mass of $D$ is contained in
  ``almost monochromatic 1-rectangles'' of $P$ (defined below).  
  We'll see below that this is easy to prove in general by an
  averaging argument.

\item [(2)] An almost monochromatic 1-rectangle contains at most
  $2^{-c}$ mass of the distribution $D$, where $c$ is as large as
  possible (ideally $c = \Omega(n)$).  This is the hard step, and the
  argument will be different for different functions $f$ and different
  input distributions $D$.

\end{itemize}
If we can establish~(1) and~(2), then we have a lower bound of
$\Omega(2^{-c})$ on the number of rectangles induced by $P$, which
proves that $P$ uses communication $\Omega(c)$.\footnote{Why call it the
  ``corruption method''?  Because the argument shows that, if a
  deterministic protocol has low communication, then most of its
  induced rectangles that contain 1-inputs are also ``corrupted'' by lots of
  0-inputs --- its rectangles are so big that~\eqref{eq:corrupt} fails.
  In turn, this implies large error.}

Here's the formal definition of an
{\em almost monochromatic 1-rectangle (AM1R) $R = A \times B$} of a
matrix $M(f)$ with respect to an input distribution $D$:
\begin{equation}\label{eq:corrupt}
\prob[\inputs \sim D]{\inputs \in R \text{ and } f\inputs = 0}
\le
8\eps \cdot \prob[\inputs \sim D]{\inputs \in R \text{ and }
  f\inputs = 1}.
\end{equation}

Here's why property~(1) is true in general.  Let $P$ be a
deterministic protocol with error at most $\eps$ with respect to $D$.
Since $P$ is deterministic, it partitions the matrix $M(f)$ into
rectangles, and in each rectangle, $P$'s output is constant.  Let
$R_1,\ldots,R_{\ell}$ denote the rectangles in which $P$ outputs ``1.''  

At least 50\% of the 1-mass of $D$ --- and hence at least 25\% of
$D$'s overall mass --- must be contained in $R_1,\ldots,R_{\ell}$.
For if not, on at least 25\% of the mass of $D$, $f\inputs = 1$ while
$P$ outputs ``0''.  This contradicts the assumption that $P$ has error
$\eps$ with respect to $D$ (provided $\eps < .25$).

Also, at least 50\% of the mass in $R_1,\ldots,R_{\ell}$ must lie in
AM1Rs.  For if not, using~\eqref{eq:corrupt} and the fact that the
total mass in $R_1,\ldots,R_{\ell}$ is at least .25, it would follow
that $D$ places more than $8\eps \cdot .125 = \eps$ mass on 0-inputs in
$R_1,\ldots,R_{\ell}$.  Since $P$ outputs ``1'' on all of these
inputs, this contradicts the assumption that $P$ has error at most
$\eps$ with respect to $D$.  This completes the proof of step~(1),
which applies
to every problem and every distribution $D$ over inputs with
1-mass at least .5.

Step~(2) is difficult and problem-specific.  
\cite{BFS86}, for their input distribution $D$ over \disj
inputs mentioned above, gave a proof of step~(2) with
$c = \Omega(\sqrt{n})$, thus giving an $\Omega(\sqrt{n})$ lower bound
on the randomized communication complexity of the problem.
\cite{R92} gave, for his input distribution, a proof of
step~(2) with $c = \Omega(n)$, implying the desired lower bound for
\disj.  
Sadly, we won't have time to talk about these and
subsequent proofs (as in~\cite{B+02b}); perhaps in a future course.

\chapter{Lower Bounds for the Extension Complexity of Polytopes}
\label{cha:lower-bounds-extens}

\section{Linear Programs, Polytopes, and Extended Formulations}\label{s:intro}

\subsection{Linear Programs for Combinatorial Optimization Problems}

You've probably seen some polynomial-time algorithms for the problem
of computing a maximum-weight matching of a bipartite
graph.\footnote{Recall that a graph is {\em bipartite} if its
  vertex set can be partitioned into
two sets $U$ and $V$ such that every edge has one
  endpoint in each of $U,V$.  Recall that a {\em matching} of a graph
  is a subset of edges that are pairwise disjoint.}  Many of these,
like the Kuhn-Tucker algorithm~\citep{K55}, are ``combinatorial
algorithms'' that operate directly on the graph.

{\em Linear programming} is also an effective tool for solving many
discrete optimization problems.  
For example, consider the following
linear programming relaxation of the maximum-weight bipartite matching
problem (for a weighted bipartite graph $G=(U,V,E,w)$):
\begin{equation}\label{eq:bm0}
\max \sum_{e \in E} w_ex_e
\end{equation}
subject to
\begin{equation}\label{eq:bm1}
\sum_{e \in \delta(v)} x_e \le 1
\end{equation}
for every vertex $v \in U \cup V$ (where $\delta(v)$ denotes the edges
incident to $v$) and
\begin{equation}\label{eq:bm2}
x_e \ge 0
\end{equation}
for every edge $e \in E$.

In this formulation, each decision variable $x_e$ is intended to
encode whether an edge $e$ is in the matching ($x_e = 1$) or not ($x_e
= 0$).  It is easy to verify that the vectors of $\zo^E$
that satisfy the constraints~\eqref{eq:bm1} and~\eqref{eq:bm2} are
precisely the characteristic 
vectors of the matchings of $G$, with the objective function value of
the solution to the linear program equal to the total weight of the
matching.  

Since every characteristic vector of a matching
satisfies~\eqref{eq:bm1} and~\eqref{eq:bm2}, and the set of feasible
solutions to the linear system defined by~\eqref{eq:bm1}
and~\eqref{eq:bm2} is convex, the convex 
hull of the characteristic vectors of matchings is contained in this
feasible region.\footnote{Recall that a set $S \sse \RR^n$ is {\em
    convex} if it is ``filled in,'' with $\lambda \bfx + (1-\lambda)\bfy
  \in S$ whenever $\bfx,\bfy \in S$ and $\lambda \in [0,1]$.  Recall that
  the {\em convex hull} of a point set $P \sse \RR^n$ is the smallest
  (i.e., intersection of all) convex set that contains it.
  Equivalently, it is the set of all finite convex combinations of points of
  $P$, where a convex combination has the form $\sum_{i=1}^p
  \lambda_i\bfx_i$ for non-negative $\lambda_i$'s summing to 1 and
  $\bfx_1,\ldots,\bfx_p \in P$.}
Also note that every characteristic vector $\bfx$ of a matching is a vertex\footnote{There is an unfortunate
  clash of terminology when talking about linear programming
  relaxations of combinatorial optimization problems: a ``vertex''
  might refer to a node of a graph or to a ``corner'' of a geometric
  set.}
of this feasible region --- since all feasible solutions have all
coordinates bounded by 0 and 1, the 0-1 vector $\bfx$ cannot be written as
a non-trivial convex combination of other feasible solutions.
The worry is does this feasible region contain anything other than
the convex hull of characteristic vectors of matchings?
Equivalently, does it have any vertices that are fractional, and hence
do not correspond to matchings?
(Note that
integrality is not explicitly enforced by~\eqref{eq:bm1}
or~\eqref{eq:bm2}.) 

A nice fact is that the vertices of the feasible region defined
by~\eqref{eq:bm1} and~\eqref{eq:bm2} are precisely the characteristic
vectors of matchings of~$G$.
This is equivalent to the Birkhoff-von Neumann theorem (see
Exercises).
There are algorithms that solve linear programs in polynomial time
(and output a vertex of the feasible region, see e.g.~\cite{GLS88}),
so this implies that the maximum-weight bipartite matching problem can be
solved efficiently using linear programming.

How about the more general problem of maximum-weight matching in
general (non-bipartite) graphs?  While the same linear
system~\eqref{eq:bm1} and~\eqref{eq:bm2} still contains the convex
hull of all characteristic vectors of matchings, and these
characteristic vectors are vertices of the feasible region, there are
also other, fractional, vertices.  To see this, consider the simplest
non-bipartite graph, a triangle.  Every matching contains at most 1
edge.  But assigning $x_e=\tfrac{1}{2}$ for each of the edges~$e$ yields a
fractional solution that satisfies~\eqref{eq:bm1} and~\eqref{eq:bm2}.
This solution clearly cannot be written as a convex combination of
characteristic vectors of matchings.

It is possible to add to~\eqref{eq:bm1}--\eqref{eq:bm2} additional
inequalities ---
``odd cycle inequalities'' stating that, for every odd cycle $C$ of
$G$, $\sum_{e \in C} x_e \le (|C|-1)/2$ --- 
so that the resulting smaller set of feasible solutions is precisely
the convex hull of the characteristic vectors of matchings.  
Unfortunately, many graphs have
an exponential number of odd cycles.  Is it possible to add only
a polynomial
number of inequalities instead?  Unfortunately not --- the convex
hull of the characteristic vectors of matchings can have
$2^{\Omega(n)}$ ``facets''~\citep{PE74}.\footnote{This linear
  programming formulation still leads to a polynomial-time
  algorithm, but using fairly heavy machinery --- the ``ellipsoid
  method''~\citep{K79} and a ``separation oracle'' for the odd cycle
  inequalities~\citep{PR82}.  There are also polynomial-time
  combinatorial algorithms for (weighted) non-bipartite matching,
  beginning with 
  \cite{edmonds}.}   We define facets more formally in
Section~\ref{ss:facets}, but intuitively they are the ``sides'' of a
polytope,\footnote{A {\em polytope} is just a high-dimensional polygon
  --- an intersection of halfspaces that is bounded or, equivalently,
  the convex hull of a finite set of points.} like the $2n$ sides of an
$n$-dimensional cube.  It is intuitively clear that a polytope with
$\ell$ facets needs $\ell$ inequalities to describe --- it's like
cleaving a shape out of marble, with each inequality contributing a
single cut.  We conclude that there is no linear program 
with variables $\{ x_e \}_{e \in E}$ of polynomial
size that captures the maximum-weight (non-bipartite) matching problem.

\subsection{Auxiliary Variables and Extended Formulations}\label{ss:ef}

The exponential lower bound above on the number of linear inequalities
needed to describe the convex hull of characteristic vectors of
matchings of a non-bipartite graph applies to linear systems in
$\RR^E$, with one dimension per edge.  The idea of an {\em extended
  formulation} is to add a polynomial number of auxiliary decision
variables, with the hope that radically fewer inequalities are needed to
describe the region of interest in the higher-dimensional space.

This idea might sound like grasping at straws, but sometimes it
actually works.  For example, 
fix a positive integer $n$, and represent a
permutation $\pi \in S_n$ by the $n$-vector
$\bfx_{\pi} = (\pi(1),\pi(2),\ldots,\pi(n))$, with all coordinates in
$\{1,2,\ldots,n\}$.  
The {\em permutahedron} is the convex hull of all $n!$ such vectors.
The permutahedron is known to have $2^{n/2}-2$ facets (see
e.g.~\cite{goemans}), so a polynomial-sized linear description would
seem out of reach. 

Suppose we add $n^2$ auxiliary variables, $y_{ij}$ for all $i,j \in
\{1,2,\ldots,n\}$.  The intent is for $y_{ij}$ to be a 0-1 variable
that indicates whether or not $\pi(i) = j$ --- in this case, the
$y_{ij}$'s are the entries of the $n \times n$ permutation matrix that 
corresponds to $\pi$.

We next add a set of constraints to enforce the desired semantics of
the $y_{ij}$'s (cf.,~\eqref{eq:bm1} and~\eqref{eq:bm2}):
\begin{equation}\label{eq:perm1}
\sum_{j=1}^n y_{ij} \le 1
\end{equation}
for $i=1,2,\ldots,n$;
\begin{equation}\label{eq:perm2}
\sum_{i=1}^n y_{ij} \le 1
\end{equation}
for $j=1,2,\ldots,n$; and
\begin{equation}\label{eq:perm3}
y_{ij} \ge 0
\end{equation}
for all $i,j \in \{1,2,\ldots,n\}$.
We also add constraints that enforce consistency between the permutation
encoded by the $x_i$'s and by the $y_{ij}$'s:
\begin{equation}\label{eq:perm4}
x_i = \sum_{j=1}^n jy_{ij}
\end{equation}
for all $i=1,2,\ldots,n$.

It is straightforward to check that the vectors $\bfy \in \{0,1\}^{n^2}$
that satisfy~\eqref{eq:perm1}--\eqref{eq:perm3}
are precisely the permutation matrices.  
For such a $\bfy$ corresponding to a permutation $\pi$, the
constraints~\eqref{eq:perm4} force the $x_i$'s to encode the same
permutation $\pi$.
Using again the Birkhoff-von
Neumann Theorem, every vector $\bfy \in \RR^{n^2}$
that satisfies~\eqref{eq:perm1}--\eqref{eq:perm3} is a convex
combination of permutation matrices (see Exercises).
Constraint~\eqref{eq:perm4} implies that the $x_i$'s encode the same
convex combination of permutations.  
Thus, if we take the set 
of solutions in $\RR^{n + n^2}$ 
that satisfy~\eqref{eq:perm1}--\eqref{eq:perm4} and project onto the
$x$-coordinates, we get exactly the permutahedron.  
This is what we mean by an {\em extended formulation} of a polytope.

To recap the remarkable trick we just pulled off: blowing up the
number of variables from $n$ to $n+n^2$ reduced the number of
inequalities needed from $2^{n/2}$ to $n^2 + 3n$.
This allows us to optimize a linear function over the permutahedron in
polynomial time.  Given a linear function (in the $x_i$'s), we
optimize it over the (polynomial-size) extended formulation, and
retain only the $x$-variables of the optimal solution.

Given the utility of polynomial-size extended formulations, we'd
obviously like to understand which problems have them.  
For example, does the non-bipartite matching problem admit such a
formulation? 
The goal of this lecture is to develop communication
complexity-based techniques for ruling out such polynomial-size
extended formulations.  We'll prove an impossibility result for the
``correlation polytope''~\citep{F+12}; similar (but much more involved)
arguments imply that every extended formulation of the non-bipartite
matching problem requires an exponential number of
inequalities~\citep{rothvoss}.

\begin{remark}[Geometric Intuition]
It may seem surprising that adding a relatively small number of
auxiliary variables can radically reduce the number of inequalities
needed to describe a set
--- described in reverse, that projecting onto a subset of variables can
massively blow up the number of sides.  It's hard to draw
(low-dimensional) pictures that illustrate this point.
If you play around with projections of
some three-dimensional polytopes onto
the plane, you'll observe that non-facets of the high-dimensional
polytope (edges) often become facets (again, edges) in the
low-dimensional projection.  Since the number of lower-dimensional
faces of a polytope can be much bigger than the number of facets ---
already in the 3-D cube, there are 12 edges and only 6 sides --- it
should be plausible that a projection could significantly increase the
number of facets.
\end{remark}


\section{Nondeterministic Communication Complexity}

The connection between extended formulations of polytopes and
communication complexity involves {\em nondeterministic} communication
complexity.  We studied this model implicitly in parts of
Lecture~\ref{cha:boot-camp-comm}; this section makes the model explicit.

Consider a function $f:X \times Y \rightarrow \zo$ and the
corresponding 0-1 matrix $M(f)$, with rows indexed by Alice's possible
inputs and columns indexed by Bob's possible inputs.  In Lecture~\ref{cha:boot-camp-comm}
we proved that if every covering of $M(f)$ by monochromatic
rectangles\footnote{Recall that a {\em rectangle} is a subset $S \sse
  X \times Y$ that has a product structure, meaning $S = A \times B$
  for some $A \sse X$ and $B \sse Y$.  Equivalently, $S$ is closed
  under ``mix and match:'' whenever $(\bfx_1,\bfy_1)$ and $(\bfx_2,\bfy_2)$ are
  in $S$, so are $(\bfx_1,\bfy_2)$ and $(\bfx_2,\bfy_1)$.  A rectangle is {\em
    monochromatic} (w.r.t.\ $f$) if it contains only 1-entries
  of $M(f)$ or only 0-entries of $M(f)$.  In these cases, we call it a
{\em 1-rectangle} or a {\em 0-rectangle}, respectively.}
requires at least $t$ rectangles, then the deterministic communication
complexity of $f$ is at least $\log_2 t$ (Theorem~\ref{t:part}).  
The reason is that every
communication protocol computing $f$ with communication cost $c$
induces a partition of $M(f)$ into at most $2^c$ monochromatic
rectangles, and partitions are a special case of coverings.  See also
Figure~\ref{f:cover2}.

\begin{figure}
\centering
\includegraphics[width=.2\textwidth]{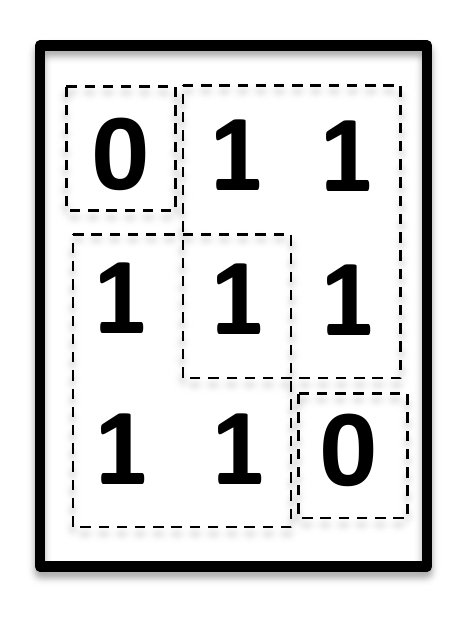}
\caption[A covering by four monochromatic rectangles that is not a
  partition]{A covering by four monochromatic rectangles that is not a
  partition.}\label{f:cover2}
\end{figure}

Communication complexity lower bounds that are proved through
coverings are actually much stronger than we've let on thus far ---
they apply also to {\em nondeterministic} protocols, which we define
next.

You presumably already have a feel for nondeterminism from your study
of the complexity class $NP$.  Recall that one way to define
$NP$ is as the problems for which membership can be verified
in polynomial time.
To see how an
analog might work with communication protocols, consider the
complement of the \eq problem, \noneq.  If a third party wanted to
convince Alice and Bob that their inputs $\bfx$ and $\bfy$ are different,
it would not be difficult: just specify an index $i \in
\{1,2,\ldots,n\}$ for which $x_i \neq y_i$.  Specifying an index
requires $\log_2 n$ bits, and specifying whether or not $x_i = 0$ and
$y_i = 1$ or $x_i = 1$ and $y_i = 0$ requires one additional bit.
Given such a specification, Alice and Bob can check the correctness of
this ``proof of non-equality'' without any communication.
If $\bfx \neq \bfy$, there is always a $(\log_2 + 1)$-bit proof that will
convince Alice and Bob of this fact; if $\bfx = \bfy$, then no such proof
will convince Alice and Bob otherwise.  This means that \noneq has
{\em nondeterministic communication complexity} at most $\log_2 n +
1$.

Coverings of $M(f)$ by monochromatic rectangles are closely related to
the nondeterministic communication complexity of $f$.  We first show
how coverings lead to nondeterministic protocols.  It's easiest to
formally define such protocols after the proof.

\begin{proposition}\label{prop:cover1}
Let $f: X \times Y \rightarrow \{0,1\}$ be a Boolean function and
$M(f)$ the corresponding matrix.
If there is a cover of the 1-entries of $M(f)$ by $t$ 1-rectangles,
then there is a nondeterministic protocol that verifies $f\inputs = 1$
with cost $\log_2 t$.
\end{proposition}

\begin{proof}
Let $R_1,\ldots,R_t$ denote a covering of the 1s of $M(f)$ by
1-rectangles.  Alice and Bob can agree to this covering in advance of
receiving their inputs.  Now consider the following scenario:
\begin{enumerate}

\item A {\em prover} --- a third party --- sees both inputs $\bfx$
{\em and} $\bfy$.
(This is 
the formal model used for nondeterministic protocols.)

\item The prover writes an index $i \in \{1,2,\ldots,t\}$ --- the name of
  a rectangle $R_i$ --- on a blackboard, in public view.  Since $R_i$
  is a rectangle, it can be written as $R_i = A_i \times B_i$ with
  $A_i \sse X$, $B_i \sse Y$.

\item 
Alice accepts if and only if $\bfx \in A_i$.

\item 
Bob accepts if and only if $\bfy \in B_i$.

\end{enumerate}
This protocol has the following properties:
\begin{enumerate}

\item If $f\inputs = 1$, then there exists a proof such that Alice and
  Bob both accept.  (Since $f\inputs = 1$, $\inputs \in R_i$ for some
  $i$, and Alice and Bob both accept if ``$i$'' is written on the
  blackboard.)

\item If $f\inputs = 0$, there is no proof that both Alice and Bob
  accept.  (Whatever index $i \in \{1,2,\ldots,t\}$ is written on the
  blackboard, since $f\inputs = 0$, either $\bfx \not\in R_i$ or $y
  \not\in R_i$, causing a rejection.)

\item The maximum length of a proof is $\log_2 t$.  (A proof is just
  an index $i \in \{1,2,\ldots,t\}$.)

\end{enumerate}
These three properties imply, by definition,
that the nondeterministic communication complexity of the function $f$
and the output~1 is at most $\log_2 t$.
\end{proof}

The proof of Proposition~\ref{prop:cover1} introduces
our formal model of nondeterministic communication complexity: Alice
and Bob are given a ``proof'' or ``advice string'' by a prover, which
can depend on both of their inputs; the communication cost is the
worst-case length of the proof; and a protocol is said to
compute an output $z \in \{0,1\}$ of a function $f$ if $f\inputs = z$
if and only if there exists proof such that both Alice and Bob accept.

With nondeterministic communication complexity, we speak about
both a function $f$ and an output $z \in \{0,1\}$.
For example, if $f$ is \eq, then we saw that the nondeterministic
communication complexity of $f$ and the output 0 is at most $\log_2 n
+1$.  Since it's not clear how to convince Alice and Bob that their
inputs {\em are} equal without specifying at least one bit for each of
the $n$ coordinates, one might expect the nondeterministic
communication complexity of $f$ and the output 1 to be roughly $n$.  (And
it is, as we'll see.)

We've defined nondeterministic protocols so that Alice and Bob never
speak, and only verify.  This is without loss of generality, since
given a protocol in which they do speak, one could modify it so that
the prover writes on the blackboard everything that they would have
said.  We encourage the reader to formalize an alternative definition
of nondeterministic protocols without a prover and in which Alice and
Bob speak nondeterministically, and to prove that this definition is
equivalent to the one we've given above (see Exercises).

Next we prove the converse of Proposition~\ref{prop:cover2}.
\begin{proposition}\label{prop:cover2}
If the nondeterministic communication complexity of the function $f$
and the output~1 is $c$, then there is a covering of the 1s of
$M(f)$ by $2^c$ 1-rectangles.
\end{proposition}

\begin{proof}
Let $\P$ denote a nondeterministic communication protocol for $f$ and
the output~1 with communication cost (i.e., maximum proof length) at
most $c$.
For a proof $\ell$, let $Z(\ell)$ denote the inputs $\inputs$ where
both Alice and Bob accept the proof.  We can write $Z(\ell) = A \times
B$, where $A$ is the set of inputs $\bfx \in
X$ of Alice where she accepts the proof $\ell$, and $B$ is the set of
inputs $\bfy \in Y$ of Bob where he accepts the proof.  By the
assumed correctness of $\P$, $f\inputs = 1$ for every $\inputs
\in Z(\ell)$.  That is, $Z(\ell)$ is a 1-rectangle.

By the first property of nondeterministic protocols, for every 1-input
$\inputs$ there is a proof such that both Alice and Bob accept.  That
is, $\cup_{\ell} Z(\ell)$ is precisely the set of 1-inputs of $f$ ---
a covering of the 1s of $M(f)$ by 1-rectangles.  Since the
communication cost of $\P$ is at most $c$, there are at most $2^c$
different proofs~$\ell$.
\end{proof}

Proposition~\ref{prop:cover2} implies that communication complexity
lower bounds derived from covering lower bounds apply to
nondeterministic protocols.
\begin{corollary}
If every covering of the 1s of $M(f)$ by 1-rectangles uses at least
$t$ rectangles, then the nondeterministic communication complexity of
$f$ is at least $\log_2 t$.
\end{corollary}
Thus our arguments in Lecture~\ref{cha:boot-camp-comm}, while simple, were even more
powerful than we realized --- they prove that the nondeterministic
communication complexity of \eq, \disj, and \gt (all with output~1) is
at least~$n$.  It's kind of amazing that these lower bounds can be
proved with so little work.

\section{Extended Formulations and Nondeterministic Communication
  Complexity}

What does communication complexity have to do with extended
formulations?  To forge a connection, we need to show that an extended
formulation with few inequalities is somehow useful for solving hard
communication problems.  
While this course includes a number of clever connections between
communication complexity and various computational models,  
this connection to extended formulations is perhaps the most
surprising and ingenious one of them all. 
Superficially, extended formulations with few inequalities can be
thought of as ``compressed descriptions'' of a polytope, and
communication complexity is generally useful for ruling out compressed
descriptions of various types.  
It is not at all obvious that this vague intuition can be turned into
a formal connection, let alone one that is useful for proving
non-trivial impossibility results.

\subsection{Faces and Facets}\label{ss:facets}

We discuss briefly some preliminaries about polytopes.
Let $P$ be a polytope in variables $\bfx \in \RR^n$.
By definition, an {\em extended formulation} of $P$ is a set of the
form
\[
Q = \{ \inputs \,:\, \Cx + \Dy \le \bfd \},
\]
where $\bfx$ and $\bfy$ are the original and auxiliary variables,
respectively, such that
\[
\{ \bfx \,:\, \exists \bfy \text{ s.t.\ } \inputs \in Q \} = P.
\]
This is, projecting $Q$ onto the original variables $\bfx$ yields the
original polytope $P$.  
The extended formulation of the permutahedron described in
Section~\ref{ss:ef} is a canonical example.
The {\em size} of an extended formulation is
the number of inequalities.\footnote{There is no need to keep track of
  the number of auxiliary variables --- there is no point in having
  an extended formulation of this type with more variables than
  inequalities (see Exercises).}

Recall that $\bfx \in P$ is a {\em vertex} if
it cannot be written as a non-trivial convex combination of other
points in $P$.  A {\em supporting hyperplane of $P$} is a vector $\bfa \in
\RR^n$ and scalar $b \in \RR$ such that $\bfa\bfx = b$ for all $\xP$.
Every supporting hyperplane $\bfa,b$ induces a {\em face} of $P$, defined
as $\{ \xP \,:\, \bfa\bfx = b \}$ --- the intersection of the boundaries of
$P$ and of the the halfspace defined by the supporting hyperplane.
(See Figure~\ref{f:suphyp}.)
Note that a face is generally induced by many different supporting
hyperplanes.  The empty set is considered a face.  
Note also that faces are nested --- in three dimensions, there
are vertices, edges, and sides.  
In general, if $f$ is a face of $P$, then the vertices of $f$ are
precisely the vertices of $P$ that are contained in $f$.

\begin{figure}
\centering
\includegraphics[width=.7\textwidth]{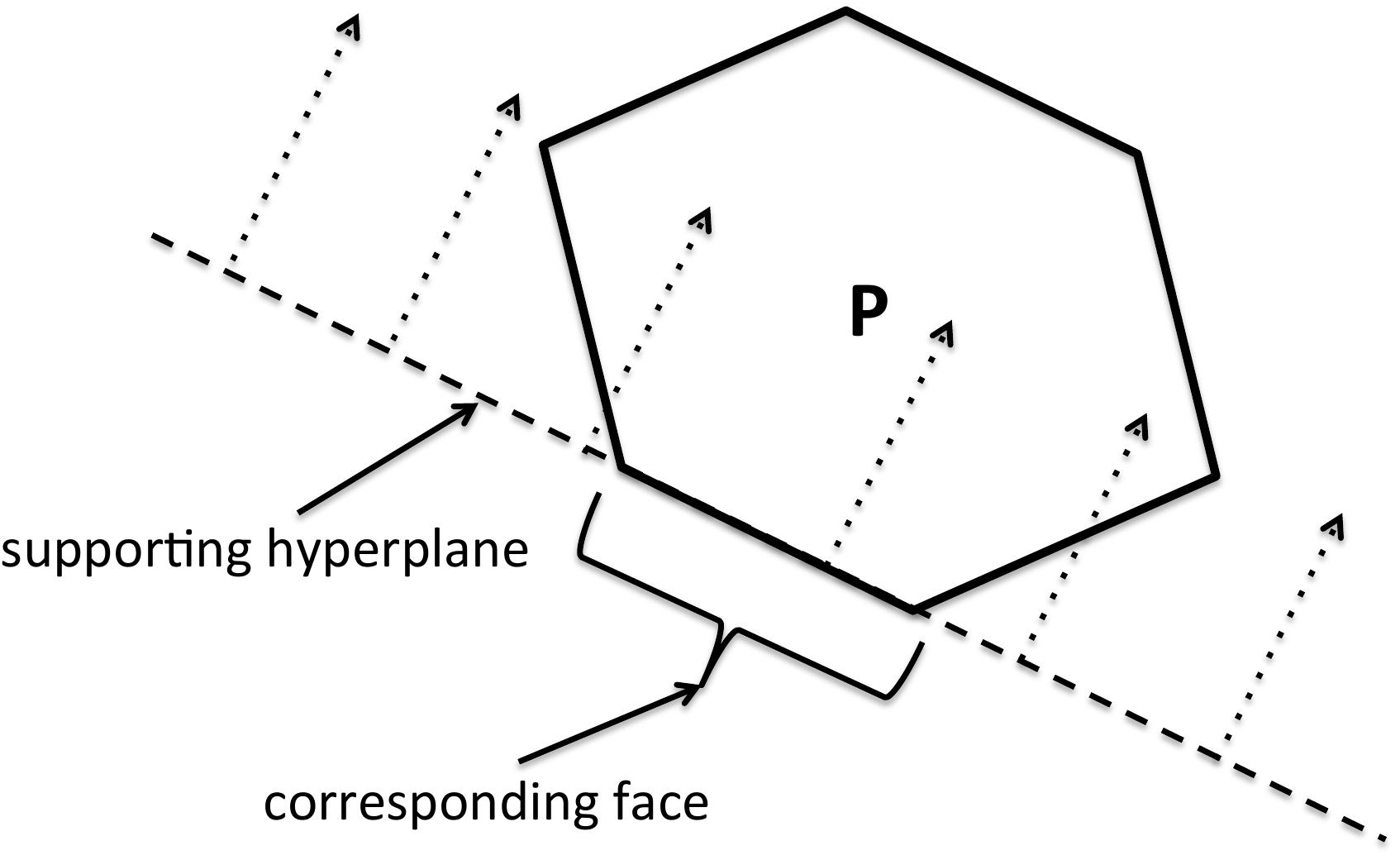}
\caption[A supporting hyperplane and the corresponding face]{A supporting hyperplane of a polytope~$P$ and the
  corresponding face of the polytope.}\label{f:suphyp}
\end{figure}

A {\em facet} of $P$ is a
maximal face --- a face that is not strictly contained in any other
face.  Provided $P$ has a non-empty interior, its facets are
$(n-1)$-dimensional. 

There are two different types of finite descriptions of a polytope,
and it is useful to go back and forth between them.  First, a polytope
$P$ equals the convex hull of its vertices.  Second, $P$ is the
intersection of the halfspaces that define its facets.\footnote{Proofs
  of all of these statements are elementary but outside the scope of
  this lecture; see e.g.~\cite{ziegler} for details.}

\subsection{Yannakakis's Lemma}

What good is a small extended formulation?
We next make up a contrived communication problem for which small
extended formulations are useful.  For a polytope $P$, in the
corresponding \fvp problem, Alice gets a face $f$ of $P$ (in the form
of a supporting hyperplane $\ab$) and Bob gets
a vertex $v$ of $P$.  The function $FV(f,v)$ is defined as~1 if $v$
does {\em not} belong to $f$, and 0 if $v \in f$.
Equivalently,
$FV(f,v) = 1$ if and only if $\av < b$, where $\ab$ is a supporting
hyperplane that induces~$f$.
Polytopes in $n$ dimensions generally have an exponential number of
faces and vertices.  Thus, trivial protocols for \fvp, where one party
reports their input to the other, can have communication cost
$\Omega(n)$.

A key result is the following.
\begin{lemma}[Yannakakis's Lemma~\citeyearpar{Y88}]\label{l:y}
If the polytope $P$ admits an extended formulation $Q$ with $r$
inequalities, then the nondeterministic communication complexity of
\fvp is at most $\log_2 r$.
\end{lemma}
That is, if we can prove a linear lower bound on the nondeterministic
communication complexity of the \fvp problem, then we have ruled out
subexponential-size extended formulations of $P$.

Sections~\ref{ss:pf1} and~\ref{ss:pf2} give two different proof
sketches of Lemma~\ref{l:y}.
These are roughly equivalent, with the first emphasizing
the geometric aspects
(following~\cite{lovasz}) and the second the
algebraic aspects (following~\cite{Y88}).  In Section~\ref{s:lb_polytope} we put
Lemma~\ref{l:y} to use and 
prove strong lower bounds for a concrete polytope.

Remarkably, \cite{Y88} did not give any applications of his
lemma --- the lower bounds for extended formulations 
in~\cite{Y88} are for ``symmetric'' formulations and proved via
direct arguments.  Lemma~\ref{l:y} was suggested by~\cite{Y88} as a
potentially useful tool for more general impossibility results, and
finally in the past five years (beginning with~\cite{F+12}) this
prophecy has come to pass.

\subsection{Proof Sketch of Lemma~\ref{l:y}: A Geometric
  Argument}\label{ss:pf1} 

Suppose $P$ admits an extended formulation 
$Q = \{ \inputs \,:\, \Cx + \Dy \le \bfd \}$ with only $r$
inequalities.
Both $P$ and $Q$ are known to Alice and Bob before the protocol
begins.
A first idea is for Alice, who is given a face $f$ of the original
polytope $P$, to tell Bob the name of the ``corresponding face''
of $Q$.  Bob can then check whether or not his ``corresponding
vertex'' belongs to the named face or not, thereby computing the
function.

Unfortunately, knowing that $Q$ is defined by $r$ inequalities only
implies that it has at most $r$ {\em facets} --- it can have a very
large number of faces.  Thus Alice can no more afford to write down an
arbitrary face of $Q$ than a face of $P$.

We use a third-party prover to name a suitable facet of $Q$ than
enables Alice and Bob to compute the \fvp function; since $Q$ has at
most $r$ facets, the protocol's communication cost is only $\log_2 r$,
as desired.

Suppose the prover wants to convince Alice and Bob that Bob's vertex
$v$ of $P$ does not belong to Alice's face $f$ of $P$.  If the prover
can name a facet $f^*$ of $Q$ such that:
\begin{itemize}

\item [(i)] there exists $\bfy_v$ such that $(v,\bfy_v) \not\in f^*$; and

\item [(ii)] for every $\inputs \in Q$ with $\bfx \in f$, $\inputs \in
  f^*$;

\end{itemize}
then this facet $f^*$ proves that $v \not\in f$.  Moreover, given $f^*$,
Alice and Bob can verify~(ii) and~(i), respectively, without any
communication.  

All that remains to prove is that, when $v \notin f$, there exists a
facet $f^*$ of $Q$ such that~(i) and~(ii) hold.  First consider
the inverse image of $f$ in $Q$, $\tilde{f} = \{ \inputs \in Q \,:\, \bfx
\in f \}$.  Similarly, define $\tilde{v} = \{ (v,\bfy) \in Q \}$.
Since $v \notin f$, $\tilde{f}$ and $\tilde{v}$ are disjoint subsets
of $Q$.  It is not difficult to prove that $\tilde{f}$ and
$\tilde{v}$, as inverse images of faces under a linear map, are faces
of $Q$ (exercise).  An intuitive but non-trivial fact is that every
face of a 
polytope is the intersection of the facets that contain it.\footnote{This
  follows from Farkas's Lemma, or equivalently the Separating
  Hyperplane Theorem.  See~\cite{ziegler} for details.}
Thus, for every vertex $v^*$ of $Q$ that is contained
in $\tilde{v}$ (and hence not in $\tilde{f}$) --- and since
$\tilde{v}$ is non-empty, there is at least one --- we can choose a
facet $f^*$ of $Q$ that contains $\tilde{f}$ (property~(ii)) but
excludes $v^*$ (property~(i)).  This concludes the proof sketch of
Lemma~\ref{l:y}.

\subsection{Proof Sketch of Lemma~\ref{l:y}: An Algebraic
  Argument}\label{ss:pf2} 

The next proof sketch of Lemma~\ref{l:y} is a bit longer but
introduces some of the most important concepts in the study of
extended formulations.

The {\em slack matrix} of a polytope $P$ has rows indexed by faces~$F$
and 
columns indexed by vertices~$V$.  We identify each face with a canonical
supporting hyperplane $\ab$.
Entry $S_{fv}$ of the slack matrix is defined as $b - \av$, where
$\ab$ is the supporting hyperplane corresponding to the face $f$.
Observe that all entries of $S$ are nonnegative.  
Define the {\em support} $\supp(S)$ of the slack matrix $S$ as the
$F \times V$
matrix with 1-entries wherever $S$ has positive entries, and 0-entries
wherever $S$ has 0-entries.  Observe that $\supp(S)$ is a property
only of the polytope~$P$, independent of the choices of the supporting
hyperplanes for the faces of $P$.  Observe also that $\supp(S)$ is
precisely the answer matrix for the \fvp problem for the
polytope $P$.

We next identify a sufficient condition for \fvp to have low
nondeterministic communication complexity; later we explain why
the existence of a small extended formulation implies this sufficient
condition.  Suppose the slack matrix
$S$ has {\em nonnegative rank} $r$, meaning it is possible to write $S
= TU$ with $T$ a $|F| \times r$ nonnegative matrix and $U$ a $r
\times |V|$ nonnegative matrix (Figure~\ref{f:nmf}).\footnote{This is
  called a {\em nonnegative matrix factorization}.  It is the analog
  of the singular value decomposition (SVD), but with the extra
  constraint that the factors are nonnegative matrices.  It obviously
  only makes sense to ask for such decompositions for nonnegative
  matrices (like $S$).}
Equivalently, suppose we can write $S$ as the sum of $r$ outer
products of nonnegative vectors (indexed by $F$ and $V$):
\begin{equation}\label{eq:nmf}
S = \sum_{j=1}^r \alpha_j \cdot \beta_j^T,
\end{equation}
where the $\alpha_j$'s correspond to the columns of $T$ and the
$\beta_j$'s to the rows of $U$.

\begin{figure}
\centering
\includegraphics[width=.8\textwidth]{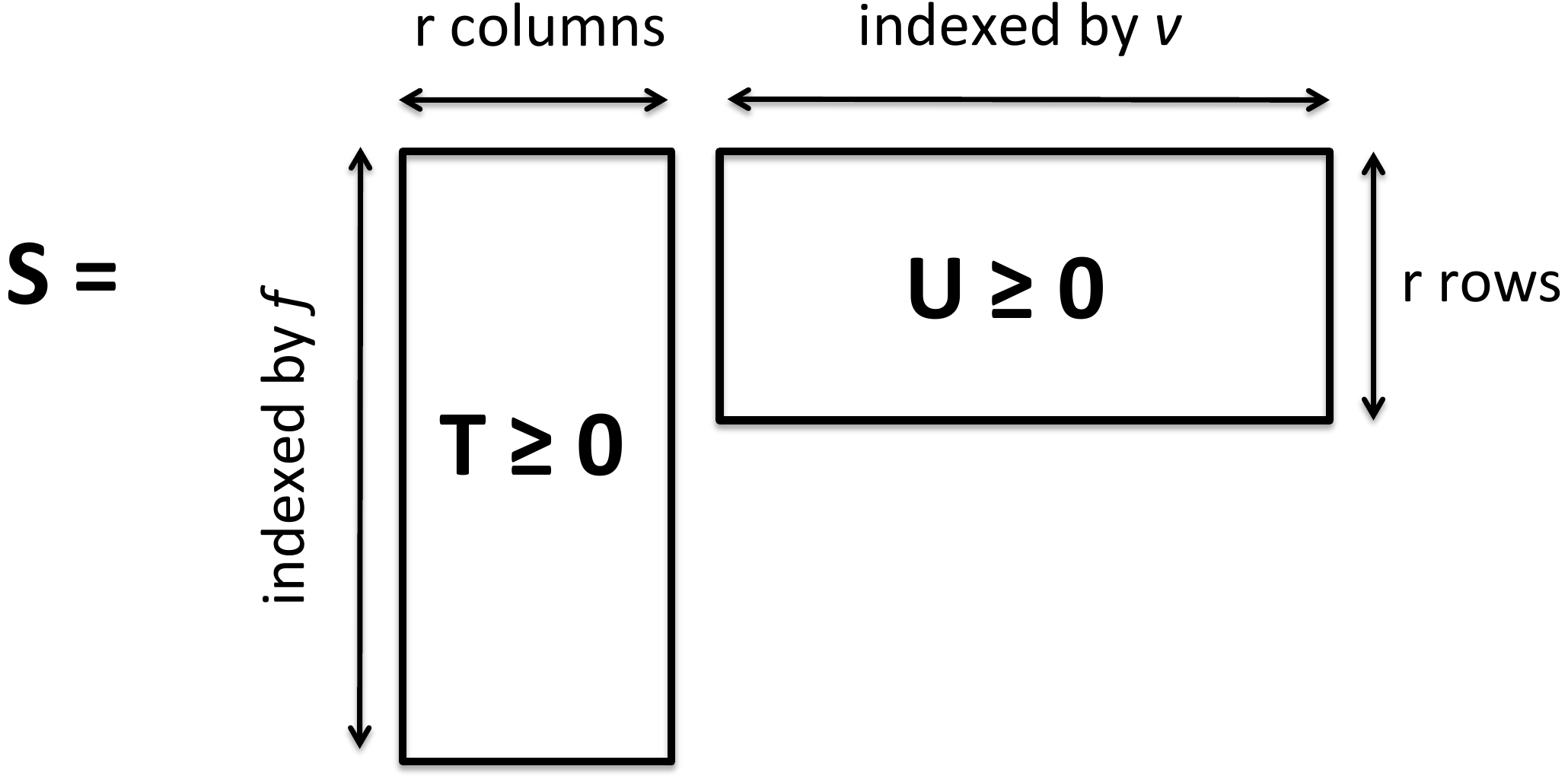}
\caption[Nonnegative matrix factorization]{A rank-$r$ factorization of the slack matrix~$S$ into
  nonnegative matrices $T$ and $U$.}\label{f:nmf}
\end{figure}

We claim that if the slack matrix $S$ of a polytope $P$ has
nonnegative rank $r$, then there is a nondeterministic communication
protocol for \fvp with cost at most $\log_2 r$.  
As usual, Alice and Bob can agree to the decomposition~\eqref{eq:nmf} in
advance.
A key observation is that, by inspection of~\eqref{eq:nmf}, $S_{fv} >
0$ if and only if 
there exists some $j \in \{1,2,\ldots,r\}$ with $\alpha_{fj},
\beta_{jv} > 0$.  (We are using here that everything is nonnegative
and so no cancellations are possible.)  Equivalently, the supports of
the outer products $\alpha_j \cdot \beta_j^T$ can be viewed as a
covering of the 1-entries of $\supp(S)$ by $r$ 1-rectangles.
Given this observation, the protocol for \fvp should be clear.
\begin{enumerate}

\item The prover announces an index $j \in \{1,2,\ldots,r\}$.

\item Alice accepts if and only if the $f$th component of $\alpha_j$
  is strictly positive.

\item Bob accepts if and only if the $v$th component of $\beta_j$
  is strictly positive.

\end{enumerate}
The communication cost of the protocol is clearly $\log_2 r$.  The key
observation above implies that there is a proof (i.e., an index $j \in
\{1,2,\ldots,r\}$) accepted by both Alice and Bob if and only if Bob's
vertex $v$ does not belong to Alice's face $f$.

It remains to prove that, whenever a polytope $P$ admits an extended
formulation with a small number of inequalities, its slack matrix
admits a low-rank nonnegative matrix factorization.\footnote{The
  converse also holds, and might well be the easier direction to
  anticipate.
See the Exercises for details.}  We'll show this by
exhibiting nonnegative $r$-vectors $\lambda_f$ (for all faces $f$ of
$P$) and $\mu_v$ (for all vertices $v$ of $P$) such that $S_{fv} =
\lambda_f^T\mu_v$ for all $f$ and $v$.  
In terms of Figure~\ref{f:nmf}, the $\lambda_f$'s and
$\mu_v$'s correspond to the rows of $T$ and columns of $U$,
respectively.

The next task is to understand better how an extended formulation $Q = \{
\inputs \,:\, \Cx + \Dy \le \bfd \}$ must be related to the original
polytope $P$.  Given that projecting $Q$ onto the variables $\bfx$
yields~$P$, it must be that every supporting hyperplane of $P$ is
logically implied by the inequalities that define $Q$.  To see one way
how this can happen, suppose there is a non-negative $r$-vector $\lambda
\in \RR^r_+$ with the following properties:
\begin{itemize}

\item [(P1)] $\lambda^T\C = \bfa^T$;

\item [(P2)] $\lambda^T\D = \zero$;

\item [(P3)] $\lambda^T\bfd = b$. 

\end{itemize}
(P1)--(P3) imply that, for every $\inputs$ in $Q$ (and so with
$\Cx+\Dy \le \bfd$), we have
\[
\underbrace{\lambda^T\C}_{=\bfa^T}\bfx + \underbrace{\lambda^T\D}_{=
  \zero}\bfy \le \underbrace{\lambda^T\bfd}_{= b}
\]
and hence $\axb$ (no matter what $\bfy$ is).

Nonnegative linear combinations $\lambda$ of the constraints of $Q$
that satisfy~(P1)--(P3) are one way in which the constraints of $Q$
imply constraints on the values of $\bfx$ in the projection of $Q$.  A
straightforward application of Farkas's Lemma (see e.g.~\cite{chvatal})
implies that such nonnegative linear combinations are the
{\em only way} in which the constraints of $Q$ imply constraints on
the projection of $Q$.\footnote{Farkas's Lemma is sometimes phrased as
the Separating Hyperplane
Theorem.  It can also be thought of as the feasibility version of
strong linear programming duality.}
Put differently, whenever $\axb$ is a
supporting hyperplane of $P$, there exists a nonnegative linear
combination $\lambda$ that proves it (i.e., that satisfies
(P1)--(P3)).  This clarifies what the extended formulation $Q$ really
accomplishes: ranging over all $\lambda \in \RR^r_+$ satisfying
(P2) generates all of the supporting hyperplanes $\ab$ of~$P$
(with $\bfa$ and $b$ arising as $\lambda^T\C$ and $\lambda^T\bfd$,
respectively).  

To define the promised $\lambda_f$'s and $\mu_v$'s, fix a face $f$ of
$P$ with supporting hyperplane $\axb$.  Since $Q$'s projection does not
include any points not in $P$, the
constraints of $Q$ imply this supporting hyperplane.
By the previous
paragraph, we can choose a nonnegative vector $\lambda_f$ so
that~(P1)--(P3) hold.

Now fix a vertex $v$ of $P$.  Since $Q$'s projection includes every
point of $P$, there exists a choice of $\bfy_v$ such that $(v,\bfy_v) \in
Q$.  Define $\mu_v \in \RR^r_+$ as the slack in $Q$'s constraints at
the point $(v,\bfy_v)$:
\[
\mu_v = \bfd - \Cv - \D\bfy_v.
\]
Since $(v,\bfy_v) \in Q$, $\mu_v$ is a nonnegative vector.

Finally, for every face $f$ of $P$ and vertex $v$ of $P$, we have
\[
\lambda^T_f\mu_v = \underbrace{\lambda^T_f\bfd}_{= b}
- \underbrace{\lambda^T_f\Cv}_{=\bfa^Tv} -
\underbrace{\lambda^T_f\D\bfy_v}_{=0} = b-\bfa^T\bfv = S_{fv},
\]
as desired.  This completes the second proof of Lemma~\ref{l:y}.

\section{A Lower Bound for the Correlation Polytope}\label{s:lb_polytope}

\subsection{Overview}

Lemma~\ref{l:y} reduces the task of proving lower bounds on the size
of extended formulations of a polytope $P$ to proving lower bounds on
the nondeterministic communication complexity of \fvp.  The case study
of the permutahedron (Section~\ref{ss:ef}) serves as a cautionary tale
here: the communication complexity of \fvp is surprisingly low
for some complex-seeming polytopes, so proving strong lower bounds,
when they exist, typically requires work and a detailed understanding
of the particular polytope of interest.

\cite{F+12} were the first to use Yannakakis's Lemma to
prove lower bounds on the size of extended formulations of interesting
polytopes.\footnote{This paper won the Best Paper Award at STOC '12.}
We follow the proof plan of~\cite{F+12}, which has two steps.
\begin{enumerate}

\item First, we exhibit a polytope that is tailor-made for proving a
  nondeterministic communication complexity lower bound on the
  corresponding \fvp problem, via a reduction from \disj.  We'll prove
  this step in full.

\item Second, we extend the consequent lower bound on the size of
  extended formulations to other problems, such as the Traveling
  Salesman Problem (TSP), via reductions.  These reductions are
  bread-and-butter $NP$-completeness-style reductions; see the
  Exercises for more details.

\end{enumerate}
This two-step plan does not seem sufficient to resolve the motivating
problem mentioned in Section~\ref{s:intro}, the non-bipartite matching
problem.  
For an $NP$-hard problem like TSP, 
we fully expect all extended formulations of the
convex hull of the characteristic vectors of solutions 
to be exponential; otherwise, we could use linear
programming to obtain a subexponential-time algorithm for the problem
(an unlikely result).
The non-bipartite matching problem is polynomial-time solvable,
so it's less clear what to expect.  \cite{rothvoss} proved that
every extended formulation of the convex hull of the perfect matchings
of the complete graph has exponential size.\footnote{This paper won
  the Best Paper Award at STOC '14.}  The techniques
in~\cite{rothvoss} are more sophisticated variations of the tools covered
in this lecture --- a reader of these notes is well-positioned to
learn them.

\subsection{Preliminaries}

We describe a polytope~$P$ for which it's relatively easy to prove
nondeterministic communication complexity lower bounds for the
corresponding \fvp problem.  The polytope was studied earlier for
other reasons~\citep{P91,dW03}.

Given a 0-1 $n$-bit vector $\bfx$, we consider the corresponding 
(symmetric and rank-1) outer product $\bfx\bfx^T$.  For example, if $\bfx =
10101$, then  
\[
\bfx\bfx^T =
\left(
\begin{array}{ccccc}
1 & 0 & 1 & 0 & 1\\
0 & 0 & 0 & 0 & 0\\
1 & 0 & 1 & 0 & 1\\
0 & 0 & 0 & 0 & 0\\
1 & 0 & 1 & 0 & 1
\end{array}
\right).
\]
For a positive integer $n$, we define \cor as the convex hull of all
$2^n$ such vectors $\bfx\bfx^T$ (ranging over $\bfx \in \zo^n$).  This is a
polytope in $\RR^{n^2}$, and its vertices are precisely the points
$\bfx\bfx^T$ with $\bfx \in \zo^n$.

Our goal is to prove the following result.
\begin{theorem}[\citealt{F+12}]\label{t:cor}
The \ndcc of \fvc is $\Omega(n)$.
\end{theorem}
This lower bound is clearly the best possible
(up to a constant factor), since Bob can communicate his vertex to
Alice using only $n$ bits (by specifying the appropriate $\bfx \in
\zo^n$).

Lemma~\ref{l:y} then implies that every extended formulation of 
the \cor polytope requires $2^{\Omega(n)}$ inequalities, no matter how
many auxiliary variables are added.  Note the dimension $d$ is
$\Theta(n^2)$, so this lower bound has the form
$2^{\Omega(\sqrt{d})}$.  

Elementary reductions (see the Exercises) translate
this extension complexity lower bound for the \cor polytope to a lower
bound of $2^{\Omega(\sqrt{n})}$ on the size of
extended formulations of
the convex hull of characteristic vectors of $n$-point traveling
salesman tours.

\subsection{Some Faces of the Correlation Polytope}

Next we establish a key connection between certain faces of the
correlation polytope and inputs to \disj.  Throughout, $n$ is a fixed
positive integer.

\begin{lemma}[\citealt{F+12}]\label{l:corfaces}
For every subset $S \sse \{1,2,\ldots,n\}$, there is a face $f_S$ of
\cor such that: for every $R \sse \{1,2,\ldots,n\}$ with
characteristic vector $\bfx_R$ and corresponding vertex $\bfv_R =
\bfx_R\bfx_R^T$ of \cor,
\[
\bfv_R \in f_S \quad \text{ if and only if } \quad |S \cap R| = 1.
\]
\end{lemma}
That is, among the faces of \cor are $2^n$ faces that encode the
``unique intersection property'' for each of the $2^n$ subsets $S$ of
$\{1,2,\ldots,n\}$.  Note that for a given $S$, the sets $R$ with $|S
\cap R|$ can be generated by (i) first picking a element of $S$; (ii)
picking a subset of $\{1,2,\ldots,n\} \sm S$.  Thus if $|S|=k$, there
are $k2^{n-k}$ sets $R$ with which it has a unique intersection.

Lemma~\ref{l:corfaces} is kind of amazing, but also not too hard to
prove.

\vspace{.1in}
\noindent
\begin{prevproof}{Lemma}{l:corfaces}
For every $S \sse \{1,2,\ldots,n\}$, we need to exhibit a supporting
hyperplane $\axb$ such that $\av_R = b$ if and only if $|S \cap R| =
1$, where $\bfv_R$ denotes $\bfx_R\bfx_R^T$ and $\bfx_R$ the characteristic
vector of $R \sse \{1,2,\ldots,n\}$.  

Fix $S \sse \{1,2,\ldots,n\}$.  We develop the appropriate supporting
hyperplane, in variables $\bfy \in \RR^{n^2}$, over several small steps.
\begin{enumerate}

\item For clarity, let's start in the wrong space, with variables $\bfz
  \in \RR^n$ 
  rather than $\bfy \in \RR^{n^2}$.  Here $\bfz$ is meant to encode the
  characteristic vector of a set $R \sse \{1,2,\ldots,n\}$.  One
  sensible inequality to start with is
\begin{equation}\label{eq:try1}
\sum_{i \in S} z_i -1 \ge 0.
\end{equation}
For example, if $S = \{1,3\}$, then this constraint reads $z_1 + z_3 -
1 \ge 0$.

The good news is that for 0-1 vectors $\bfx_R$, this inequality is
satisfied with equality if and only if $|S \cap R|=1$.  The bad news is
that it does not correspond to a supporting hyperplane: if 
$S$ and $R$ are disjoint,
then $\bfx_R$ violates the inequality.  How can we change the
constraint so that it holds with equality for $\bfx_R$ with $|S \cap R|
= 1$ and also valid for all $R$?

\item One crazy idea is to square the left-hand side of~\eqref{eq:try1}:
\begin{equation}\label{eq:try2}
\left( \sum_{i \in S} z_i -1 \right)^2 \ge 0.
\end{equation}
For example, if $S = \{1,3\}$, then the constraint reads (after
expanding) $z_1^2+z_3^2+2z_1z_3-2z_1-2z_3+1 \ge 0$.

The good news is that every 0-1 vector $\bfx_R$ satisfies this
inequality, and equality holds if and only if $|S \cap R|=1$.
The bad news is that the constraint is non-linear and hence does not
correspond to a supporting hyperplane.

\item The obvious next idea is to ``linearize'' the previous
  constraint.  Wherever the constraint has a $\bfz_i^2$ or a $\bfz_i$, we
  replace it by a variable $\bfy_{ii}$ (note these partially cancel out).
  Wherever the constraint has a 
  $2z_iz_j$ (and notice for $i \neq j$ these always come in pairs), we
  replace it by a $y_{ij} + y_{ji}$.   
Formally, the constraint now reads
\begin{equation}\label{eq:try3}
-\sum_{i \in S} y_{ii} + \sum_{i \neq j \in S} y_{ij} +1 \ge 0.
\end{equation}
Note that the new variable set is $\bfy \in \RR^{n^2}$.  
For example, if $S = \{1,3\}$, then the new constraint reads 
$y_{13}+y_{31}-y_{11}-y_{33} \ge -1$.

A first observation is that, for $\bfy$'s that are 0-1, symmetric, and
rank-1, 
with $\bfy = \bfz\bfz^T$ (hence $y_{ij} = z_i \cdot z_j$ for $i,j \in
\{1,2,\ldots,n\}$), the left-hand sides of~\eqref{eq:try2}
and~\eqref{eq:try3} are the same by definition.  Thus, for 
$\bfy = \bfx_R\bfx_R^T$ with $\bfx \in \zo^n$, $\bfy$
satisfies the (linear) inequality~\eqref{eq:try3}, and equality holds
if and only if~$|S \cap R| = 1$.

\end{enumerate}

We have shown that, for every $S \sse \{1,2,\ldots,n\}$, the linear
inequality~\eqref{eq:try3} is satisfied by every vector $\bfy \in
\RR^{n^2}$ of the form $\bfy = \bfx_R\bfx_R^T$ with $\bfx \in \zo^n$.  Since
\cor is by definition the convex hull of such vectors, every point of
\cor satisfies~\eqref{eq:try3}.  This inequality is therefore a
supporting hyperplane, and the face it induces contains precisely
those vertices of the form $\bfx_R\bfx_R^T$ with $|S \cap R| = 1$.  This
completes the proof.
\end{prevproof}


\subsection{\fvc and \udisj}

In the \fvc problem, Alice receives a face $f$ of \cor and Bob a
vertex $\bfv$ of \cor.  In the 1-inputs, $\bfv \not\in f$; in the
0-inputs, $\bfv \in f$.
Let's make the problem only easier by restricting Alice's possible
inputs to the $2^n$ faces (one per subset $S \sse \{1,2,\ldots,n\}$)
identified in Lemma~\ref{l:corfaces}.
In the corresponding matrix $M_U$ of this function, we can index the
rows by subsets $S$.  Since every vertex of \cor has the form $\bfy =
\bfx_R\bfx_R^T$ for $R \sse \{1,2,\ldots,n\}$, we can index the columns of
$M_U$ by
subsets $R$.  By Lemma~\ref{l:corfaces}, the entry $(S,R)$ of the
matrix $M_U$ is~1 if $|S \cap R| \neq 1$ and~0 of $|S \cap R|=1$.
That is, the 0-entries of $M_U$ correspond to pairs $(S,R)$ that
intersect in a unique element.

There is clearly a strong connection between the matrix $M_U$ above and
the analogous matrix $M_D$ for \disj.  They differ on entries $(S,R)$
with $|S \cap R| \ge 2$: these are 0-entries of $M_D$ but 1-entries of
$M_U$.  In other words, $M_U$ is the matrix corresponding to the
communication problem \nui: do the inputs $S$ and $R$ fail to have a
unique intersection?

The closely related \udisj problem is a ``promise'' version of \disj.
The task here is to distinguish between:
\begin{itemize}

\item [(1)] inputs $(S,R)$ of \disj with $|S \cap R| = 0$;

\item [(0)] inputs $(S,R)$ of \disj with $|S \cap R| = 1$.

\end{itemize}
For inputs that fall into neither case (with $|S \cap R| > 1$), the
protocol is off the hook --- any output is considered correct.
Since a protocol that solves \udisj has to do only less than one that
solves \nui, communication complexity lower bounds for former problem
apply immediate to the latter.

We summarize the discussion of this section in the following
proposition.
\begin{proposition}[\citealt{F+12}]\label{prop:udisj}
The \ndcc of \fvc is at least that of \udisj.
\end{proposition}

\subsection{A Lower Bound for \udisj}\label{ss:udisj}

\subsubsection{The Goal}

One final step remains in our proof of Theorem~\ref{t:cor}, and hence
of our lower bound on the size of extended formulations of the
correlation polytope.
\begin{theorem}[\citealt{F+12,KW15}]\label{t:udisj}
The \ndcc of \udisj is $\Omega(n)$.
\end{theorem}

\subsubsection{\disj Revisited}

As a warm-up, we revisit the standard \disj problem.
Recall that, in Lecture~\ref{cha:boot-camp-comm}, we proved that the
\ndcc of \disj is at least $n$ by
a fooling set argument (Corollary~\ref{cor:disj_cover}).  
Next we prove a slightly weaker lower bound,
via an argument that generalizes to \udisj.

The first claim is that, of the $2^n \times 2^n = 4^n$ possible inputs
of \disj, exactly $3^n$ of them are 1-inputs.  The reason is that the
following procedure, which makes $n$ 3-way choices, generates every
1-input exactly once: independently for each coordinate
$i=1,2,\ldots,n$, choose between the options (i) $x_i = y_i = 0$; (ii)
$x_i = 1$ and $y_i = 0$; and (iii) $x_i = 0$ and $y_i = 1$.

The second claim is that every 1-rectangle --- every subset $A$ of
rows of $M_D$ and $B$ of columns of $M_D$ such that $A \times B$
contains only 1-inputs --- has size at most $2^n$.  
To prove this, let $R = A \times B$ be a 1-rectangle.
We assert that, for every coordinate $i=1,2,\ldots,n$, either (i) $x_i
= 0$ for all $\bfx 
\in A$ or (ii) $y_i = 0$ for all $\bfy \in B$.  
That is, every coordinate has, for at least one of the two parties, a
``forced zero'' in $R$.
For if neither (i) nor
(ii) hold for a coordinate $i$, then since $R$ is a rectangle (and
hence closed under ``mix and match'') we can choose $\inputs \in R$
with $x_i = y_i = 1$; but this is a 0-input and $R$ is a 1-rectangle.
This assertion implies that the following procedure, which makes $n$
2-way choices, generates every 1-input of $R$ (and possibly other
inputs as well): independently for each coordinate $i=1,2,\ldots,n$,
set the forced zero ($x_i = 0$ in case~(i) or $y_i = 0$ in case~(ii))
and choose a bit for this coordinate in the other input.

These two
claims imply that every covering of the 1-inputs by 1-rectangles
requires at least $(3/2)^n$ rectangles.
Proposition~\ref{prop:cover2} then implies a lower bound of
$\Omega(n)$ on the \ndcc of \disj.

\subsubsection{Proof of Theorem~\ref{t:udisj}}

Recall that the 1-inputs $\inputs$ of \udisj are the same as those of
\disj (for each $i$, either $x_i = 0$, $y_i = 0$, or both).  
Thus, there are still exactly $3^n$ 1-inputs.  The
0-inputs $\inputs$ of \udisj are those with $x_i = y_i = 1$ in exactly
one coordinate $i$.  We call all other inputs, where the promise fails
to hold, {\em *-inputs}.
By a 1-rectangle, we now mean a rectangle with no
0-inputs (*-inputs are fine).  With this revised definition, it is
again true that every nondeterministic communication protocol that
solves \udisj using $c$ bits of communication induces a covering of
the 1-inputs by at most $2^c$ 1-rectangles.

\begin{lemma}\label{l:udisj}
Every 1-rectangle of \udisj contains at most $2^n$ 1-inputs.
\end{lemma}

As with the argument for \disj, Lemma~\ref{l:udisj} completes the
proof of Theorem~\ref{t:udisj}: since there are $3^n$ 1-inputs and at
most $2^n$ per 1-rectangle, every covering by 1-rectangles requires at
least $(3/2)^n$ rectangles.  This implies that the \ndcc of \udisj is
$\Omega(n)$.

Why is the proof of Lemma~\ref{l:udisj} harder than 
in Section~\ref{ss:udisj}?  We can no longer
easily argue that, in a rectangle $R = A \times B$, for each
coordinate $i$, either $x_i = 0$ for all $\bfx \in A$ or $y_i = 0$ for
all $\bfy \in B$.  Assuming the opposite no longer yields a contraction:
exhibiting $\bfx \in A$ and $\bfy \in B$ with $x_i = y_i = 1$ does not
necessarily contradict the fact that $R$ is a 1-rectangle, since $\inputs$
might be a *-input.

\vspace{.1in}
\noindent
\begin{prevproof}{Lemma}{l:udisj}
The proof is one of those slick inductions that you can't help but
sit back and admire.  

We claim, by induction on $k=0,1,2,\ldots,n$, that if $R = A \times B$
is a 1-rectangle for which all $\bfx \in A$ and $\bfy \in B$ have 0s in
their last $n-k$ coordinates, then the number 
of 1-inputs in $R$ is at most $2^k$.  The lemma is equivalent to the
case of $k=n$.  The base case $k=0$ holds, because in this case the
only possible input in $R$ is $(\zero,\zero)$.

For the inductive step, fix a 1-rectangle $R = A \times B$ in which
the last $n-k$ coordinates of all
$\bfx \in A$ and all $\bfy \in B$ are 0.  To
simplify notation, from here on we ignore the last $n-k$ coordinates
of all inputs (they play no role in the argument).

Intuitively, we need
to somehow ``zero out'' the $k$th coordinate of all inputs in $R$ so
that we can apply the inductive hypothesis.
This motivates focusing on the $k$th coordinate, and we'll often write
inputs $\bfx \in A$ and $\bfy \in B$ as $\bfx'a$ and $\bfy'b$, respectively,
with $\bfx',\bfy' \in \zo^{k-1}$ and $a,b \in \zo$.  (Recall we're ignoring
that last $n-k$ coordinates, which are now always zero.)

First observe that, whenever $\inputss$ is a 1-input, we cannot have
$a = b = 1$.  Also:
\begin{itemize}

\item [(*)] If $\inputss \in R$ is a 1-input, then $R$ cannot contain
  both the inputs 
$(\bfx'0,\bfy'1)$ and $(\bfx'1,\bfy'0)$.

\end{itemize}
For otherwise, $R$ would also contain the 0-input $(\bfx'1,\bfy'1)$,
contradicting that $R$ is a 1-rectangle. 
(Since $\inputss$ is a 1-input, the unique coordinate of $(\bfx'1,\bfy'1)$
with a 1 in both inputs is the $k$th coordinate.) 

The plan for the rest of the proof is to define two sets $S_1,S_2$ of
1-inputs --- not necessarily rectangles --- such that:
\begin{itemize}

\item [(P1)] the number of 1-inputs in $S_1$ and $S_2$ combined is at
  least that in $R$;

\item [(P2)] the inductive hypothesis applies to $\rect(S_1)$ and
  $\rect(S_2)$, where $\rect(S)$ denotes the smallest rectangle
  containing a set $S$ of inputs.\footnote{Equivalently, the closure of
  $S$ under the ``mix and match'' operation on pairs of inputs.
  Formally, $\rect(S) = X(S) \times Y(S)$, 
where $X(S) = \{ \bfx \,:\,  \inputs \in S \text{ for some } \bfy \}$
and $Y(S) = \{ \bfy \,:\,  \inputs \in S \text{ for some } \bfx \}$.}

\end{itemize}
If we can find sets $S_1,S_2$ with properties (P1),(P2), then we are
done: by the inductive hypothesis, the $\rect(S_i)$'s have at most
$2^{k-1}$ 1-inputs each, the $S_i$'s are only smaller, and hence
(by~(P1)) $R$ has at most $2^k$ 1-inputs, as required.

We define the sets in two steps, focusing first on property~(P1).
Recall that every 1-input $\inputs \in R$ has the form
$(\bfx'1,\bfy'0)$, $(\bfx'0,\bfy'1)$, or $(\bfx'0,\bfy'0)$.  We put all 1-inputs
of the first type into a set $S'_1$, and all 1-inputs of the second
type into a set $S'_2$.  When placing inputs of the third type, we
want to avoid putting two inputs of the form $\inputss$ with the same
$\bfx'$ and $\bfy'$ into the same set (this would create problems in the
inductive step).  So, for an input $(\bfx'0,\bfy'0) \in R$, we put it in
$S'_1$ if and only if the input $(\bfx'1,\bfy'0)$ was not already put in
$S'_1$; and we put it in
$S'_2$ if and only if the input $(\bfx'0,\bfy'1)$ was not already put in
$S'_2$.
Crucially, observation~(*) implies that $R$ cannot contain two
1-inputs of the form $(\bfx'1,\bfy'0)$ and $(\bfx'0,\bfy'1)$, so the 1-input
$(\bfx'0,\bfy'0)$ is placed in at least one of the sets $S'_1,S'_2$.  (It
is placed in both if $R$ contains neither $(\bfx'1,\bfy'0)$ nor
$(\bfx'0,\bfy'1)$.)  By construction, the sets $S'_1$ and $S'_2$ satisfy
property~(P1).  

We next make several observations about $S'_1$ and $S'_2$.  By
construction:
\begin{itemize}

\item [(**)] for each $i=1,2$ and $\bfx',\bfy' \in \zo^{k-1}$, 
there is at most one input of $S'_i$ of the form $\inputss$.

\end{itemize}
Also, since $S_1',S_2'$ are subsets of the rectangle $R$,
$\rect(S'_1),\rect(S'_2)$ are also subsets of $R$.  Since $R$ is a
1-rectangle, so are $\rect(S'_1),\rect(S'_2)$.  Also, since every
input $\inputs$ of $S'_i$ (and hence $\rect(S'_i)$) has $y_k = 0$ (for
$i=1)$ or $x_k = 0$ (for $i=2)$, the $k$th coordinate contributes
nothing to the intersection of any inputs of $\rect(S'_1)$ or
$\rect(S'_2)$.

Now obtain $S_i$ from $S'_i$ (for $i=1,2$) by zeroing out the $k$th
coordinate of all inputs.  Since the $S'_i$'s only contain 1-inputs,
the $S_i$'s only contain 1-inputs.
Since property~(**) implies that $|S_i| = |S'_i|$ for $i=1,2$, we
conclude that property~(P1) holds also for $S_1,S_2$.

Moving on to property~(P2),
since $\rect(S'_1),\rect(S'_2)$ contain no 0-inputs and contain only
inputs with no intersection in the $k$th coordinate,
$\rect(S_1),\rect(S_2)$ contain no 0-inputs.\footnote{The concern is that
zeroing out an input in the $k$th coordinate turns some *-input (with
intersection size~2) into a 0-input (with intersection size~1);
but since there were no intersections in the $k$th coordinate, anyways,
this can't happen.}
Finally, since all inputs of $S_1,S_2$ have zeroes in their final
$n-k+1$ coordinates, so do all inputs of $\rect(S_1),\rect(S_2)$.
The inductive hypothesis applies to $\rect(S_1)$ and $\rect(S_2)$, so
each of them has at most $2^{k-1}$ 1-inputs.  This implies the
inductive step and completes the proof.
\end{prevproof}

\chapter{Lower Bounds for Data Structures}
\label{cha:lower-bounds-data}

\section{Preamble}

Next we discuss how to use communication complexity to prove lower
bounds on the performance --- meaning space, query time, and
approximation --- of data structures.  Our case study will be the
high-dimensional approximate nearest neighbor problem.

There is a large literature on data structure lower bounds.
There are several different ways to use communication complexity to
prove such lower bounds, and we'll unfortunately only have time to
discuss one of them.  For example, we discuss only a static
data structure problem --- where the data structure can only be
queried, not modified --- and lower bounds for dynamic data structures
tend to use somewhat different techniques.
See~\cite{peter_survey} and~\cite{mihai_thesis} for some 
starting points for further reading.

We focus on the approximate nearest neighbor problem for a few
reasons: it is obviously a fundamental problem, that gets solved all
the time (in data mining, for example); there are some non-trivial
upper bounds; for certain parameter ranges, we have matching lower
bounds; and the techniques used to prove these lower bounds are
representative of work in the area --- asymmetric communication
complexity and reductions from the ``Lopsided Disjointness''
problem. 

\section{The Approximate Nearest Neighbor Problem}

In the {\em nearest neighbor problem}, the input is a set $S$ of $n$
points that lie in a metric space $(X,\ell)$.  Most commonly, the
metric space is Euclidean space ($\RR^d$ with the $\ell_2$ norm).
In these lectures, we'll focus on the Hamming cube, where $X =
\{0,1\}^d$ and $\ell$ is Hamming distance.  
On the upper bound side, the high-level ideas (related to ``locality
sensitive hashing (LSH)'') we use are also relevant for Euclidean
space and other natural metric spaces --- we'll get a glimpse of this
at the very end of the lecture.
On the lower bound side, you
won't be surprised to hear that the Hamming cube is the easiest metric
space to connect directly to the standard communication complexity
model.  Throughout the lecture, you'll want to think of $d$ as pretty
big --- say $d=\sqrt{n}$.

Returning to the general nearest neighbor problem, the goal is to
build a data structure $D$ (as a function of the point set $S \sse X$)
to prepare for all possible nearest neighbor queries.  Such a query
is a point $q \in X$, and the responsibility of the algorithm is to
use $D$ to return the point $p^* \in S$ that minimizes $\ell(p,q)$
over $p \in S$.  One extreme solution is to build no data structure at
all, and given $q$ to compute $p^*$ by brute force.  Assuming that
computing $\ell(p,q)$ takes $O(d)$ time, this query algorithm runs in
time $O(dn)$.  The other extreme is to pre-compute the answer to every
possible query $q$, and store the results in a look-up table.  For the
Hamming cube, this solution uses $\Theta(d2^d)$ space, and the query
time is that of one look-up in this table.  The exact nearest neighbor
problem is believed to suffer from the ``curse of dimensionality,''
meaning that a non-trivial query time (sublinear in $n$, say) requires
a data structure with space exponential in $d$.

There have been lots of exciting positive results for
the {\em $(1+\eps)$-approximate} version of the nearest neighbor
problem, where the query algorithm is only required to return a point $p$ with
$\ell(p,q) \le (1+\eps)\ell(p^*,q)$, where $p^*$ is the (exact)
nearest neighbor of $q$.  This is the problem we discuss in these
lectures.  You'll want to think of $\eps$ as a not-too-small constant,
perhaps~1 or~2.  
For many of the motivating applications of the nearest-neighbor
problem --- for example, the problem of detecting near-duplicate
documents (e.g., to filter search results) --- such approximate
solutions are still practically relevant.

\section{An Upper Bound: Biased Random Inner Products}\label{s:ub_ds}

In this section we present a non-trivial data structure for the
$(1+\eps)$-approximate nearest neighbor
problem in the Hamming cube.  The rough
idea is to hash the Hamming cube and then precompute the answer for
each of the hash table's buckets --- the trick is to make sure that
nearby points are likely to hash to the same bucket.
Section~\ref{s:lb_ds} proves a sense
in which this data structure is the best possible: no
data structure for the problem with equally fast query time uses
significantly less space.

\subsection{The Key Idea (via a Public-Coin Protocol)}

For the time being, we restrict attention to the decision version of
the $(1+\eps)$-nearest neighbor problem.  Here, the data structure
construction depends both on the point set $S \sse \{0,1\}^d$ and
on a given parameter $L \in \{0,1,2,\ldots,d\}$. 
Given a query $\bfq$,
the algorithm only has to decide correctly between the following
cases: 
\begin{enumerate}

\item There exists a point $p \in S$ with $\ell(p,q) \le L$.

\item $\ell(p,q) \ge (1+\eps)L$ for every point $p \in S$.

\end{enumerate}
If neither of these two cases applies, then the algorithm is off the
hook (either answer is regarded as correct).  We'll see in
Section~\ref{ss:full} how, using this solution as an ingredient, we can
build a data structure for the original version of the 
nearest neighbor problem.

Recall that {\em upper bounds} on communication complexity are always
suspect --- by design, the computational model is extremely powerful
so that lower bounds are as impressive as possible.  There are cases,
however, where designing a good communication protocol reveals the
key idea for a solution that is realizable in a reasonable
computational model.  Next is the biggest example of this that we'll
see in the course.

In the special case where $S$ contains only one point,
the decision version of the $(1+\eps)$-approximate nearest neighbor
problem resembles two other problems that we've studied before in
other contexts, one easy and one hard.
\begin{enumerate}

\item \eq.  Recall that when Alice and Bob just want to
  decide whether their inputs are the same or different ---
  equivalently, deciding between Hamming distance~0 and Hamming
  distance at least~1 --- there is an
  unreasonably effective (public-coin) randomized communication
  protocol for the problem.  Alice and Bob interpret the first $2n$
  public coins as two random $n$-bits strings $\bfr_1,\bfr_2$.  Alice
  sends the inner product modulo 2 of her input $\bfx$ with $\bfr_1$ and
  $\bfr_2$ 
  (2 bits) to Bob.  Bob accepts if and only if the two inner products
  modulo~2 of his input $\bfy$ with $\bfr_1,\bfr_2$ match those of Alice.  This
  protocol never rejects inputs with $\bfx=\bfy$, and accepts inputs
  with $\bfx \neq \bfy$ with probability $1/4$.

\item \gh.  Recall in this problem Alice and Bob want to
  decide between the cases where the Hamming distance
  between their inputs is at most $\tfrac{n}{2}-\sqrt{n}$, or at
  least $\tfrac{n}{2} + \sqrt{n}$.  In Theorem~\ref{t:gh} of
  Section~\ref{s:ghlb}, we proved that this
  problem is hard for one-way communication protocols (via a clever
  reduction from \index); it is also hard
  for general communication protocols~\citep{CR10,S12,V11}.  Note however that
  the gap between the two cases is very small, corresponding to $\eps
  \approx \tfrac{2}{\sqrt{n}}$.  In the decision version of
  the $(1+\eps)$-approximate nearest neighbor problem, we're assuming a
  constant-factor gap in Hamming distance between the two cases, so
  there's hope that the problem is easier.

\end{enumerate}

Consider now the communication problem where Alice and Bob want to
decide if the Hamming distance $\ell_H(\bfx,\bfy)$ between their inputs
$\bfx,\bfy \in \{0,1\}^d$
is at most $L$ or at least $(1+\eps)L$.  
We call this the \epsgh problem.
We analyze the following
protocol; we'll see shortly how to go from this protocol to a data
structure.
\begin{enumerate}

\item Alice and Bob use the public random coins to choose $s$ random
  strings $R = \{\bfr_1,\ldots,\bfr_s\} \in \{0,1\}^d$, where $d =
  \Theta(\eps^{-2})$.
The strings are {\em not} uniformly random: each entry is chosen
independently, with probability $1/2L$ of being equal to~1.

\item Alice sends the $s$ inner products (modulo 2) 
$\ip{\bfx}{\bfr_1},\ldots,\ip{\bfx}{\bfr_s}$ of her input and
  the random strings to Bob --- a ``hash value'' $h_R(\bfx) \in
  \{0,1\}^s$. 

\item Bob accepts if and only if the Hamming distance between the
  corresponding hash value $h_R(\bfy)$ of his input --- the $s$ inner
  products (modulo 2) of $\bfy$ with the random strings in $R$ --- differs from
  $h_R(\bfx)$ in only a ``small'' (TBD) number of coordinates.

\end{enumerate}
Intuitively, the goal is to modify our randomized communication
protocol for \eq
so that it continues to accept strings that are close to
being equal.  A natural way to do this is to bias the coefficient vectors
significantly more toward~0 than before.  For example, if $\bfx,\bfy$
differ in only a single bit, then choosing $\bfr$ uniformly at random
results in $\ip{\bfx}{\bfr} \not\equiv \ip{\bfy}{\bfr} \bmod 2$ with probability
$1/2$ (if and only if $\bfr$ has a~1 in the coordinate where $\bfx,\bfy$
differ).  With probability $1/2L$ of a~1 in each coordinate of $\bfr$, the 
probability that $\ip{\bfx}{\bfr} \not\equiv \ip{\bfy}{\bfr} \bmod 2$ is only
$1/2L$.  Unlike our \eq protocol, this protocol for \epsgh has
two-sided error.

For the analysis, it's useful to think of each random choice of a
coordinate $r_{ji}$ as occurring in two stages: in the first stage,
the coordinate is deemed {\em relevant} with probability $1/L$ and
{\em irrelevant} otherwise.  In stage~2, $r_{ji}$ is set to~0 if the
coordinate is irrelevant, and chosen uniformly from $\{0,1\}$ if the
coordinate is relevant.  We can therefore think of the protocol as:
(i) first choosing a subset of relevant coordinates; (ii) running the
old \eq protocol on these coordinates only.
With this perspective, we see that if
$\ell_H(\bfx,\bfy) = \Delta$, then
\begin{equation}\label{eq:gap}
\prob[\bfr_j]{\ip{\bfr_j}{\bfx} \not\equiv \ip{\bfr_j}{\bfy} \bmod 2} = 
\frac{1}{2}
\left(
\left(1 - \left( 1 - \frac{1}{L} \right)^{\Delta} \right)
\right)
\end{equation}
for every $\bfr_j \in R$.
In~\eqref{eq:gap}, the quantity inside the outer parentheses
is exactly the probability that
at least one of the $\Delta$ coordinates on which $\bfx,\bfy$ differ is
deemed relevant.  This is a necessary condition for the event $\ip{\bfr_j}{\bfx}
\not\equiv \ip{\bfr_j}{\bfy} \bmod 2$ and, in this case, the conditional
probability of 
this event is exactly~$\tfrac{1}{2}$ (as in the old \eq
protocol).

The probability in~\eqref{eq:gap} is an increasing function of
$\Delta$, as one would expect.  
Let $t$ denote the probability in~\eqref{eq:gap} when $\Delta = L$.
We're interested in how much bigger
this probability is when $\Delta$ is at least $(1+\eps)L$.
We can take the difference between these two cases and
bound it below 
using the fact that $1-x \in [e^{-2x},e^{-x}]$ for $x
\in [0,1]$:
\[
\tfrac{1}{2}
\left[
\underbrace{
\left(
1 - \frac{1}{L}
\right)^L
}_{\ge e^{-2}}
\left(
1 - 
\underbrace{
\left(
1 - \frac{1}{L}
\right)^{\eps L}
}_{\le e^{-\eps}}
\right)
\right] \ge 
\frac{1}{2\eps^2}\left(1 - e^{-\eps} \right) := h(\eps).
\]
Note that $h(\eps)$ is a constant, depending on $\eps$ only.
Thus, with $s$ random strings, if
$\ell_H(\bfx,\bfy) \le \Delta$ 
then we expect $ts$ of the random inner
products (modulo 2) to be different; if $\ell_H(\bfx,\bfy) \ge
(1+\eps)\Delta$, then we expect at 
least $(t+h(\eps))s$ of them to be
different.  A routine application of Chernoff bounds implies the
following.
\begin{corollary}\label{cor:protocol}
Define the hash function $h_R$ as in the communication protocol above.
If $s = \Omega(\tfrac{1}{\eps^2} \log \tfrac{1}{\delta})$, then with
probability at least $1-\delta$ over 
the choice of $R$:
\begin{itemize}

\item [(i)] If $\ell_H(\bfx,\bfy) \le L$, then $\ell_H(h(\bfx),h(\bfy)) \le
  (t+\tfrac{1}{2}h(\eps)))s$.

\item [(ii)] If $\ell_H(\bfx,\bfy) \ge (1+\eps)L$, then $\ell_H(h(\bfx),h(\bfy)) >
  (t+\tfrac{1}{2}h(\eps))s$.

\end{itemize}
\end{corollary}
We complete the communication protocol above by defining ``small'' in
the third step as $(t+\tfrac{1}{2}h(\eps))s$.  We conclude that the
\epsgh problem can be solved by a public-coin randomized protocol with
two-sided error and communication cost $\Theta(\eps^{-2})$.

\subsection{The Data Structure (Decision Version)}\label{ss:dec}

We now show how to translate the communication protocol above into a
data structure for the $(1+\eps)$-nearest neighbor problem.
For the moment, we continue to restrict to the decision version of the
problem, for an a priori known value of~$L \in \{0,1,2,\ldots,d\}$.
All we do is precompute the answers for all possible hash
values of a query (an ``inverted index'').

Given a point set $P$ of $n$ points in $\zo^d$, we choose a set $R$ of $s =
\Theta(\eps^{-2} \log n)$ random strings $\bfr_1,\ldots,\bfr_s$ according
to the distribution of the previous section (with a ``1'' chosen with
probability $1/2L$).  We again define the hash function $h_R:\zo^d
\rightarrow \zo^s$ by setting the $j$th coordinate of $h_R(\bfx)$ to
$\ip{\bfr_j}{\bfx} \bmod 2$.
We construct a table with $2^s = n^{\Theta(\eps^{-2})}$ buckets,
indexed by $s$-bit strings.\footnote{Note the frightening dependence
  of the space on $\tfrac{1}{\eps}$.  This is why we suggested
  thinking of $\eps$ as a not-too-small constant.}
Then, for each point $\bfp \in P$, we insert $\bfp$ into every bucket $\bfb
\in \zo^s$ for which $\ell_H(h_R(\bfp),\bfb) \le (t + \tfrac{1}{2}h(\eps))s$,
where $t$ is defined as in the previous section (as the probability
in~\eqref{eq:gap} with $\Delta = L$).  This preprocessing 
requires $n^{\Theta(\eps^{-2})}$ time.

With this data structure in hand, answering a query $\bfq \in \zo^d$ is
easy: just compute the hash value $h_R(\bfq)$ and return an arbitrary
point of the corresponding bucket, or ``none'' if this bucket is
empty.

For the analysis, think of an adversary who chooses a query point
$\bfq \in \zo^d$, and then we subsequently flip our coins and
build the above data
structure.  (This is the most common way to analyze hashing, with
the query independent of the random coin flips used to choose the hash
function.)  Choosing the hidden constant in the definition of $s$
appropriately and applying Corollary~\ref{cor:protocol} with $\delta =
\tfrac{1}{n^2}$, we find that, for every point $\bfp \in P$,
with probability at least
$1-\tfrac{1}{n^2}$, $\bfp$ is in $h(\bfq)$ (if
$\ell_H(\bfp,\bfq) \le L$) or is not in $h(\bfq)$ (if $\ell_h(\bfp,\bfq) \ge
  (1+\eps)L$).  Taking a Union Bound over the $n$ points of $P$, we
  find that the data structure correctly answers the query $\bfq$ with
  probability at least $1-\tfrac{1}{n}$.

Before describing the full data structure, let's take stock of what
we've accomplished thus far.  We've shown that, for every constant
$\eps > 0$, there is a data structure for the decision version of the
$(1+\eps)$-nearest neighbor problem that uses space
$n^{O(\eps^{-2})}$, answers a query with a single random access to
the data structure, and for every query is correct with high
probability.  Later in this lecture, we show a matching lower bound:
every (possibly randomized) data structure with equally good search
performance for the decision version of the $(1+\eps)$-nearest
neighbor problem has space $n^{\Omega(\eps^{-2})}$.   
Thus, smaller space can only be achieved by increasing the query time
(and there are ways to do this, see e.g.~\cite{indyk_survey}).

\subsection{The Data Structure (Full Version)}\label{ss:full}

The data structure of the previous section is an unsatisfactory
solution to the $(1+\eps)$-nearest neighbor problem in two respects:
\begin{enumerate}

\item In the real problem, there is no a priori known value of $L$.
  Intuitively, one would like to take $L$ equal to the actual
  nearest-neighbor distance of a query point $\bfq$, a quantity that
is different
  for different $\bfq$'s.

\item Even for the decision version, the data structure can answer
  some queries incorrectly.  Since the data structure only guarantees
  correctness with probability at least $1-\tfrac{1}{n}$ for each
  query $\bfq$, it might be wrong on as many as $2^d/n$ different
  queries.  Thus, an adversary that knows the data structure's coin
  flips can exhibit a query that the data structure gets wrong.

\end{enumerate}

The first fix is straightforward: just build $\approx d$ copies
of the data structure of Section~\ref{ss:dec}, one for each relevant
value of $L$.\footnote{For $L \ge d/(1+\eps)$, the data structure can
  just return an arbitrary point of $P$.  For $L=0$, when the data
  structure of Section~\ref{ss:dec} is not well defined, 
  a standard data structure for set membership, such as a
  perfect hash table~\citep{FKS}, can be used.}  Given a query $\bfq$,
the data structure now uses 
binary search over~$L$ to compute a $(1+\eps)$-nearest neighbor
of~$\bfq$; see the Exercises for details.  Answering a query thus
requires $O(\log d)$ lookups to the data structure.  This also
necessitates blowing up the number $s$ of random strings used in each
data structure by a $\Theta(\log \log d)$ factor --- this reduces the
failure probability of a given lookup by a $\log d$ factor,
enabling a Union Bound over $\log d$ times as many lookups as
before.\footnote{With some cleverness, this $\log \log d$ factor can
  be avoided -- see the Exercises.}

A draconian approach to the second problem is to again replicate the
data structure above $\Theta(d)$ times.  Each query $\bfq$ is asked in
all $\Theta(d)$ copies, and majority vote is used to determine the
final answer.  Since each copy is correct on each of the $2^d$
possible queries with probability at least $1-\tfrac{1}{n} >
\tfrac{2}{3}$, the majority vote is wrong on a given query with
probability at most inverse exponential in~$d$.  Taking a Union Bound
over the $2^d$ possible queries shows that, with high probability over
the coin flips used to construct the data structure, the data
structure answers every query correctly.  
Put differently, for almost all outcomes of the coin flips, not even
an adversary that knows the coin flip outcomes can produce a
query on which the data structure is incorrect.
This solution blows up both
the space used and the query time by a factor of $\Theta(d)$.

An alternative approach is to keep $\Theta(d)$ copies of the data
structure as above but, given a query~$\bfq$, to answer the query using
one of the $\Theta(d)$ copies chosen uniformly at random.  With this
solution, the space gets blown up by a factor of $\Theta(d)$ but the
query time is unaffected.  The correctness guarantee is now slightly
weaker.  With high probability over the coin flips used to construct
the data structure (in the preprocessing), the data structure
satisfies: for every query $\bfq \in \zo^d$, with probability at least
$1-\Theta(\tfrac{1}{n})$ over the coins flipped at query time, the data
structure answers the query correctly (why?).  
Equivalently, think of an adversary who is privy to the outcomes of the
coins used to construct the data structure, but not those used to
answer queries.  For most outcomes of the coins used in the
preprocessing phase, no matter what query the adversary suggests, the
data structure answers the query correctly with high probability.

To put everything together, 
the data structure for a fixed $L$ (from Section~\ref{ss:dec})
requires $n^{\Theta(\eps^{-2})}$ space, the first fix blows up the
  space by a factor of $\Theta(d \log d)$, and the second fix blows up
  the space by another factor of $d$.  
For the query time, with the alternative
    implementation of the second fix, answering a query involves
    $O(\log d)$ lookups into the data structure.  Each lookup involves
    computing a hash value, which in turn involves computing inner
    products (modulo 2) with $s = \Theta(\eps^{-2} \log n \log \log
    d)$ random strings.  Each such inner product requires $O(d)$ time.

Thus, the final scorecard for the data structure is:
\begin{itemize}

\item Space: $O(d^2 \log d) \cdot n^{\Theta(\eps^{-2})}$.  

\item Query time: $O(\eps^{-2} d \log n \log \log d)$.

\end{itemize}
For example, for the suggested parameter values of $d = n^c$ for a
constant $c \in (0,1)$ and $\eps$ a not-too-small constant, we obtain
a query time significantly better than the brute-force (exact)
solution (which is $\Theta(dn)$),
while using only a polynomial amount of space.

\section{Lower Bounds via Asymmetric Communication Complexity}\label{s:lb_ds}

We now turn our attention from upper bounds for the $(1+\eps)$-nearest
neighbor problem to lower bounds.  We do this in three steps.  
In Section~\ref{ss:cpm}, we introduce a model for proving lower bounds
on
time-space trade-offs in data structures --- the {\em cell probe model}.  
In Section~\ref{ss:asym}, we
explain how to deduce lower bounds in the cell probe model from a
certain type of communication complexity lower bound.
Section~\ref{ss:lb} applies this machinery to the
$(1+\eps)$-approximate nearest neighbor problem, and proves a sense in
which the data structure of Section~\ref{s:ub_ds} is optimal: every data
structure for the decision version of the problem that uses $O(1)$
lookups per query has space $n^{\Omega(\eps^{-2})}$.  Thus, no
polynomial-sized data structure can be both super-fast and
super-accurate for this problem.

\subsection{The Cell Probe Model}\label{ss:cpm}

\subsubsection{Motivation}

The most widely used model for proving data structure lower bounds is
the {\em cell probe model}, introduced by Yao --- two years after he
developed the foundations of communication complexity~\citep{Y79} ---
in the paper 
``Should Tables Be Sorted?''~\citep{Y81}.\footnote{Actually, what is
  now called the cell probe model is a bit stronger than the model
  proposed by~\cite{Y81}.}
The point of this model is to prove lower bounds for data structures
that allow random access.  To make the model as powerful as possible, 
and hence lower bounds for it as strong as possible (cf.,
our communication models), a 
random access to a data structure counts as 1 step in this model, no
matter how big the data structure is. 

\subsubsection{Some History}

To explain the title of the paper by~\cite{Y81}, suppose your job is
to store $k$
elements from a totally ordered universe of $n$
elements ($\{1,2,\ldots,n\}$,
say) in an array of length~$k$.  The goal is to minimize the worst-case
number of array accesses necessary to check whether or not a given
element is in the array, over all subsets of $k$ elements that might
be stored and over the $n$  possible queries.

To see that this problem is non-trivial, suppose $n=3$ and $k=2$.
One strategy is to store the pair of elements in sorted order, leading
to the three possible arrays in Figure~\ref{f:cpm}(a).  This yields a
worst-case query time of~2 array accesses.  To see this, suppose we
want to check whether or not~2 is in the array.  If we initially query
the first array element and find a ``1,'' or
if we initially query the second array element and find a
``3,'' then we can't be sure whether or not~2 is in the array.

\begin{figure}
  \centering
  \subbottom[Sorted Array]{%
    \includegraphics[width=0.4\textwidth]{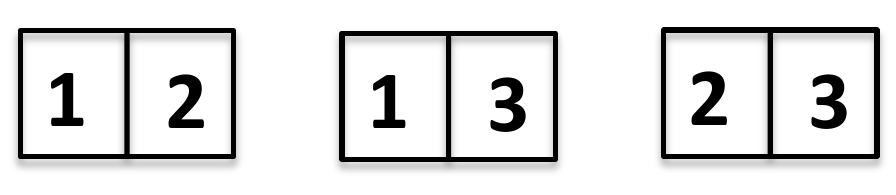}}
\qquad\qquad
  \subbottom[Unsorted Array]{%
    \includegraphics[width=0.4\textwidth]{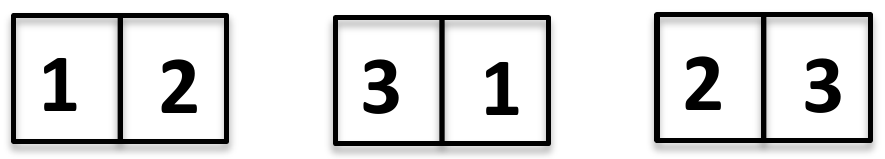}}
\caption[Should tables be sorted?]{Should tables be sorted?  For $n=3$ and $k=2$, it is suboptimal to
  store the array elements in sorted order.}
\label{f:cpm}
\end{figure}

Suppose, on the other hand, our storage strategy is as shown in
Figure~\ref{f:cpm}(b).
Whichever array entry is queried first, the answer uniquely determines
the other array entry.  Thus, storing the table in a non-sorted
fashion is necessary to achieve the optimal worst-case query time
of~1.  On the other hand, if $k=2$ and $n=4$, storing the elements in
sorted order (and using binary search to answer queries) is
optimal!\footnote{\cite{Y81} also uses Ramsey theory to prove
  that, provided the universe size is a sufficiently (really, really)
  large function of the array size, then binary search on a sorted
  array is optimal.  This result
assumes that no auxiliary storage is
  allowed, so solutions like perfect hashing~\citep{FKS} are
  ineligible.  If the universe size is not too much larger than
  the array size, then there are better solutions~\citep{FN93,GH10,GS12}, even
when there is no auxiliary storage.}

\subsubsection{Formal Model}

In the cell problem model, the goal is to encode a ``database'' $D$
in order to answer a set $Q$ of queries.  
The query set~$Q$ is known up front; the encoding scheme must work (in
the sense below) for all possible databases.

For example, in the $(1+\eps)$-approximate nearest neighbor problem, the 
database corresponds to the point set $P \sse \zo^d$, while the possible
queries $Q$ correspond to the elements of $\zo^d$.  Another canonical
example is the set membership problem: here, $D$ is a subset of a
universe~$U$, and each query $q \in Q$ asks ``is $i \in D$?'' for some
element $i \in U$.

A parameter of the cell probe model is the {\em word size} $w$; more
on this shortly.
Given this parameter, the design space consists of
the ways to encode databases $D$ as $s$ {\em cells} of $w$ bits each.
We can view such an encoding as an abstract data structure
representing the database, and
we view the number~$s$ of cells as the {\em space} used by
the encoding. 
To be a valid encoding, it must be possible to correctly answer every
query $q \in Q$ for the database~$D$ by reading enough bits of the
encoding. 
A query-answering algorithm accesses the encoding by specifying the
name of a cell; 
in return, the algorithm is given the contents of that cell.  Thus
every access to the encoding yields $w$ bits of information.
The {\em query time} of an algorithm (with respect to an
encoding scheme) is the maximum, over
databases $D$ (of a given size) and queries $q \in Q$, number
of accesses to the encoding used to answer a query.

For example, in the original array example from~\cite{Y81} mentioned
above, the word size $w$ is $\lceil \log_2 n \rceil$ --- just enough
to specify the name of an element.  The goal in~\cite{Y81} was, for
databases consisting of $k$ elements of the universe, to understand
when the minimum-possible query time 
is $\lceil \log_2 k \rceil$ 
under the constraint that the space is 
$k$.\footnote{The model in~\cite{Y81} was a bit more restrictive ---
  cells were required to contain names of elements in the database,
  rather than arbitrary $\lceil \log_2 n \rceil$-bit strings.}  

Most research on the cell-probe model seeks time-space trade-offs
with respect to a fixed value for the word size~$w$.
Most commonly, the word size is taken large enough so that a single element
of the database can be stored in a single cell, and ideally not too
much larger than this.  For nearest-neighbor-type problems involving
$n$-point sets in the $d$-dimensional hypercube, this guideline
suggests taking $w$ to be polynomial in $\max\{ d,\log_2 n\}$.  

For this choice of~$w$, the data structure in Section~\ref{ss:dec}
that solves the decision version of the $(1+\eps)$-approximate nearest
neighbor problem yields a (randomized) cell-probe encoding of point sets with
space $n^{\Theta(\eps^{-2})}$ and query time~1.  Cells of this
encoding correspond to all possible $s = \Theta(\eps^{-2} \log n)$-bit
hash values $h_R(\bfq)$ of a query $\bfq \in \zo^d$, and the
contents of a cell name an arbitrary point $\bfp \in P$ with hash value
$h_R(\bfp)$ sufficiently close (in Hamming distance in $\zo^s$) to that
of the cell's name (or ``NULL'' if no such $\bfp$ exists).
The rest of this lecture proves a matching lower bound in the cell-probe
model: constant query time can only be achieved by encodings (and
hence data structures) that use $n^{\Omega(\eps^{-2})}$ space.

\subsubsection{From Data Structures to Communication Protocols}

Our goal is to derive data structure lower bounds in the cell-probe
model from communication complexity lower bounds.
Thus, we need to extract low-communication protocols from good data
structures. Similar to our approach last lecture, we begin with a
contrived communication problem to forge an initial connection.  Later
we'll see how to prove lower bounds for the contrived problem via
reductions from other communication problems that we already
understand well.

Fix an instantiation of the cell probe model -- i.e., a set of
possible databases and possible queries.  For simplicity, we assume that
all queries are Boolean.
In the corresponding \qd problem, Alice gets a query $q$ and Bob gets
a database $D$.  (Note that in all natural examples, Bob's input is
{\em much} bigger than Alice's.)  The communication problem is to compute
the answer to~$q$ on the database $D$.

We made up the \qd problem so that the following lemma holds.
\begin{lemma}\label{l:qd}
Consider a set of databases and queries so that there is a cell-probe
encoding with word size $w$, space $s$, and query time $t$.  Then,
there is a communication protocol for the corresponding \qd problem
with communication at most
\[
\underbrace{t \log_2 s}_{\text{bits sent by Alice}} +
\underbrace{t w}_{\text{bits sent by Bob}}.
\]
\end{lemma}
The proof is the obvious simulation: Alice simulates the
query-answering algorithm, sending at most $\log_2 s$ bits to Bob specify
each cell requested by the algorithm, and Bob sends $w$ bits 
back to Alice to describe
the contents of each requested cell.  By assumption, they only need to
go back-and-forth at most $t$ times to identify the answer to Alice's
query $q$.

Lemma~\ref{l:qd} reduces the task of proving data structure lower
bounds to proving lower bounds on the communication cost of protocols
for the \qd problem.\footnote{The two-party communication model
  seems strictly stronger than the data structure design problem that
  it captures --- in a communication protocol, Bob can remember
  which queries Alice asked about previously, while a (static) data structure
  cannot.  An interesting open research direction is to find
  communication models and problems that more tightly capture 
data structure design problems,  thereby implying strong lower bounds.}

\subsection{Asymmetric Communication Complexity}\label{ss:asym}

Almost all data structure lower bounds derived from communication
complexity use {\em asymmetric} communication complexity.  This is
just a variant of the standard two-party model where we keep track of
the communication by Alice and by Bob separately.  The most common
motivation for doing this is when the two inputs have very different
sizes, like in the protocol used to prove Lemma~\ref{l:qd} above.

\subsubsection{Case Study: \textsc{Index}}\label{sss:index}

To get a feel for asymmetric communication complexity and lower bound
techniques for it, let's revisit an old friend, the \index problem.  In
addition to the application we saw earlier in the course, \index
arises naturally as the \qd problem corresponding to the membership
problem in data structures.

Recall that an input of \index gives Alice an index $i \in
\{1,2,\ldots,n\}$, specified using $\approx \log_2 n$ bits, and Bob a
subset $S \sse \{1,2,\ldots,n\}$, or equivalently an $n$-bit
vector.\footnote{We've reversed the roles of the players relative to
  the standard description we gave in
  Lectures~\ref{cha:data-stre-algor}--~\ref{cha:lower-bounds-one}.
  This reverse 
  version is the one corresponding to the \qd problem induced by the
  membership problem.}  In Lecture~\ref{cha:lower-bounds-one} we proved that the
communication complexity of \index is $\Omega(n)$ for one-way randomized
protocols with two-sided error (Theorem~\ref{t:index}) --- Bob must send almost his entire
input to Alice for Alice to have a good chance of computing her desired
index.  This lower bound clearly does not apply to general
communication protocols, since Alice can just send her $\log_2 n$-bit input
to Bob.  It is also easy to prove a matching lower bound on
deterministic and nondeterministic protocols (e.g., by a fooling set
argument).

We might expect a more refined
lower bound to hold: to solve \index, not only do the players have to send
at least $\log_2 n$ bits total, but more specifically {\em Alice} has
to send at least $\log_2 n$ bits to Bob.  Well not quite: Bob could
always send his entire input to Alice, using $n$ bits of
communication while freeing Alice to use essentially no communication.
Revising 
our ambition, we could hope to prove that in every \index protocol,
{\em either} (i) Alice has to communicate most of her input; or (ii)
Bob has to communicate most of his input.  The next result states that
this is indeed the case.

\begin{theorem}[\citealt{M+98}]\label{t:index_asym}
For every $\delta > 0$, there exists a constant $N = N(\delta)$ such
that, for every $n \ge N$ and every randomized communication
protocol with two-sided error
that solves \index with $n$-bit inputs, either:
\begin{itemize}

\item [(i)] in the worst case (over inputs and protocol randomness),
  Alice communicates at least $\delta \log_2 n$ bits; or

\item [(ii)] in the worst case (over inputs and protocol randomness),
  Bob communicates at least $n^{1-2\delta}$ bits.\footnote{From the
    proof, it will be evident that $n^{1-2\delta}$ can be replaced by
    $n^{1-c\delta}$ for any constant $c > 1$.}

\end{itemize}
\end{theorem}
Loosely speaking, Theorem~\ref{t:index_asym} states that 
the only way Alice can get away with sending
$o(\log n)$ bits of communication is if Bob sends at least
$n^{1-o(1)}$ bits of communication.

For simplicity, we'll prove Theorem~\ref{t:index_asym} only for
deterministic protocols.  The lower bound for randomized protocols
with two-sided error is very similar, just a little
messier (see~\cite{M+98}).

Conceptually, the proof of Theorem~\ref{t:index_asym} has the same flavor
as many of our previous lower bounds, and is based on covering-type
arguments.  The primary twist is that rather than keeping track only
of the size of monochromatic rectangles, we keep track of both the
height and width of such rectangles.  For example, we've seen in the
past that low-communication protocols imply the existence of a
large monochromatic rectangle --- if the players haven't had the
opportunity to speak much, then an outside observer hasn't had the
chance to eliminate many inputs as legitimate possibilities.  The
next lemma proves an analog of  
this, with the height and width of the monochromatic rectangle
parameterized by the communication used by Alice and Bob,
respectively.

\begin{lemma}[Richness Lemma~\citep{M+98}]\label{l:rich}
Let $f:X \times Y \rightarrow \{0,1\}$ be a Boolean function with
corresponding $X \times Y$ 0-1 matrix $M(f)$.  Assume that:
\begin{itemize}

\item [(1)] $M(f)$ has at least $v$ columns
that each have at least $u$ 1-inputs.\footnote{Such a matrix is 
  sometimes called {\em $(u,v)$-rich}.}

\item [(2)] There is a deterministic protocol that computes $f$ in
  which Alice and Bob always send at most $a$ and $b$ bits,
  respectively.\footnote{This is sometimes called an {\em
      $[a,b]$-protocol}.}

\end{itemize}
Then, $M(f)$ has a 1-rectangle $A \times B$ with $|A| \ge
\tfrac{u}{2^a}$ and $|B| \ge \tfrac{v}{2^{a+b}}$.
\end{lemma}

The proof of Lemma~\ref{l:rich} is a variation on the classic argument
that a protocol computing a function $f$ induces a partition of the
matrix $M(f)$ into monochromatic rectangles.  Let's recall the
inductive argument.  Let $\bfz$ be a transcript-so-far of the protocol,
and assume by induction that the inputs $\inputs$ that lead to $\bfz$
form a rectangle $A \times B$.  Assume that Alice speaks next (the
other case is symmetric).  Partition $A$ into $A_0,A_1$, with $A_\eta$
the inputs $\bfx \in A$ such Alice sends the bit $\eta$ next.  (As always,
this bit depends only on her input $\bfx$ and the transcript-so-far
$\bfz$.)  After Alice speaks, the inputs consistent with the resulting
transcript are either $A_0 \times B$ or $A_1 \times B$ --- either way,
a rectangle.  All inputs that generate the same final transcript $\bfz$
form a monochromatic rectangle --- since the protocol's output is
constant across these inputs and it computes the function $f$, $f$ is
also constant across these inputs.  

Now let's refine this argument to keep track of the dimensions
of the monochromatic rectangle, as a function of the number of times
that each of Alice and Bob speak.

\vspace{.1in}
\noindent
\begin{prevproof}{Lemma}{l:rich}
We proceed by induction on the number of steps of the protocol.  
Suppose the protocol has generated the transcript-so-far $\bfz$ and that
$A \times B$ is the rectangle of inputs consistent with this
transcript.  
Suppose that at least $c$ of the columns of $B$ have at least
$d$ 1-inputs in rows of $A$ (possibly with different rows for
different columns).

For the first case, suppose that Bob speaks next.  Partition $A \times B$
into $A \times B_0$ and $A \times B_1$, where $B_{\eta}$ are the
inputs $\bfy \in B$ such that (given the transcript-so-far $\bfz$) Bob sends
the bit~$\eta$.  At least one of the sets $B_0,B_1$ contains at least
$c/2$ columns that each contain at least $d$ 1-inputs in the rows of
$A$ (Figure~\ref{f:rich}(a)).

For the second case, suppose that Alice speaks.  Partition $A \times
B$ into $A_0 \times B$ and $A_1 \times B$.  It is not possible that
both 
(i) $A_0 \times B$ has strictly less that $c/2$
columns with $d/2$ or more 1-inputs in the rows of $A_0$ and
(ii) $A_1 \times B$ has strictly less that $c/2$
columns with $d/2$ or more 1-inputs in the rows of $A_1$.  For if
both~(i) and~(ii) held, then 
$A \times B$ would have less than $c$ columns with $d$ or
more 1-inputs in the rows of $A$, a contradiction
(Figure~\ref{f:rich}(b)).

\begin{figure}
  \centering
  \subbottom[When Bob Speaks]{%
    \includegraphics[width=0.25\textwidth]{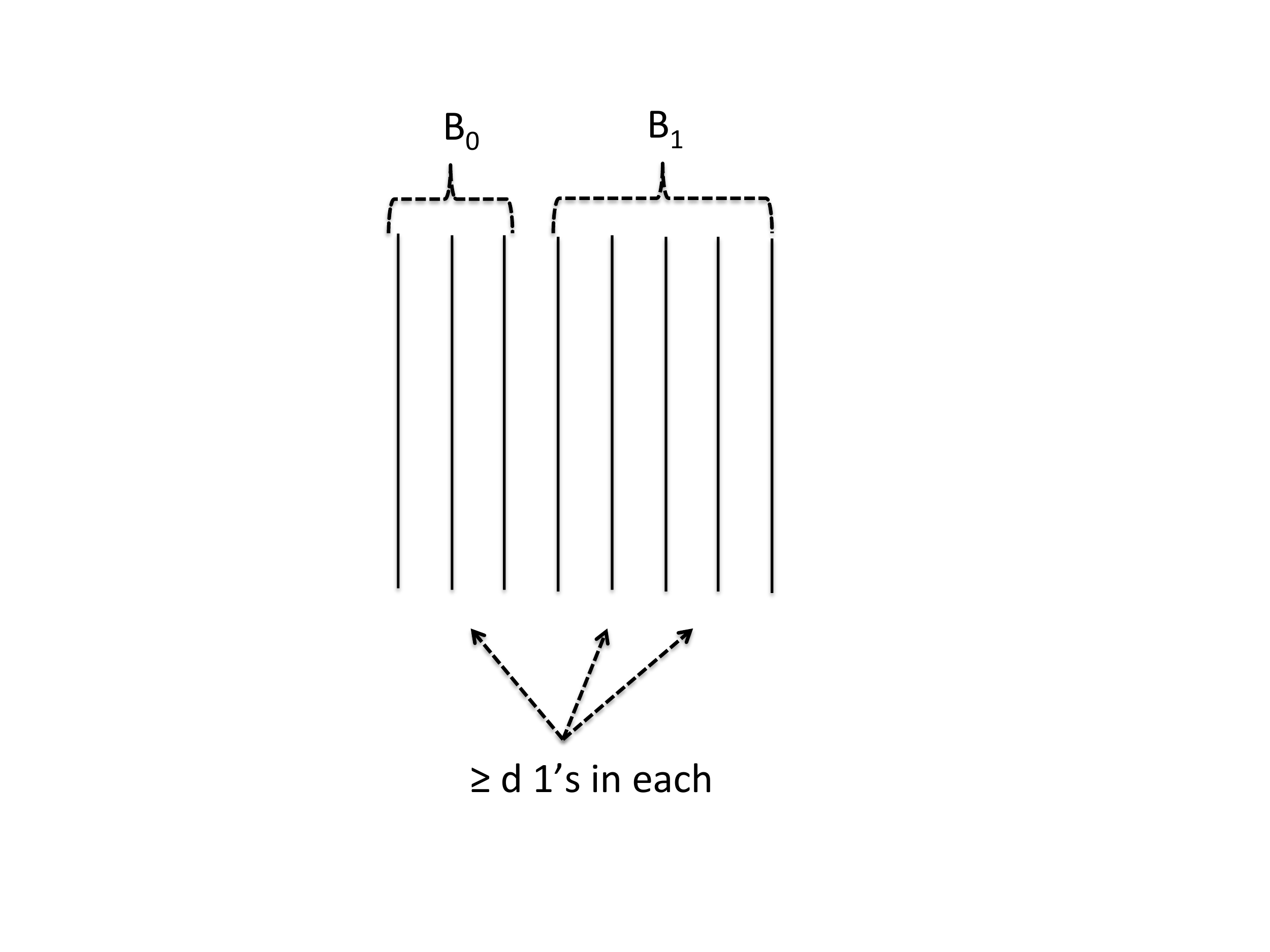}}
\qquad\qquad
  \subbottom[When Alice Speaks]{%
    \includegraphics[width=0.32\textwidth]{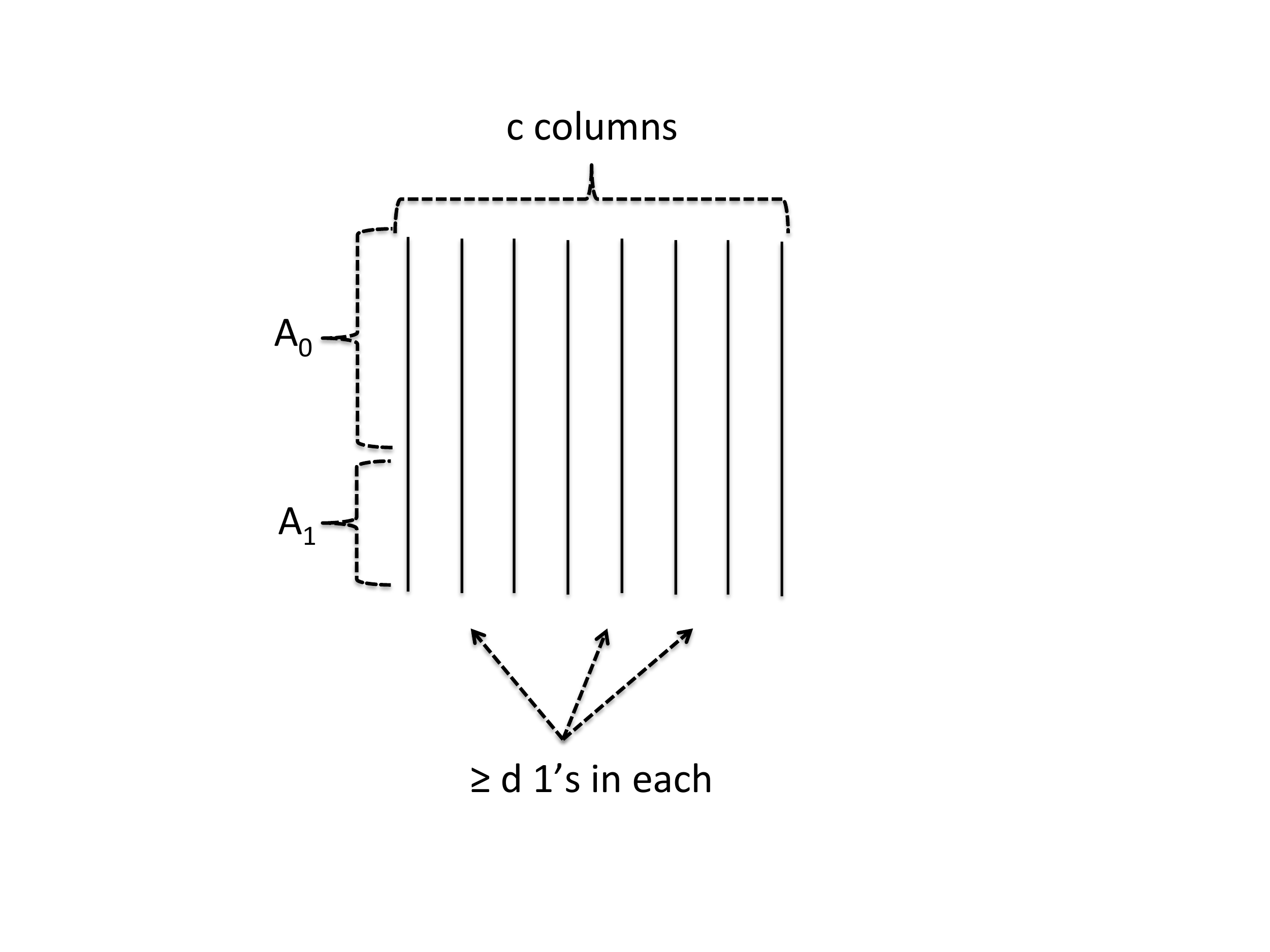}}
\caption[Proof of the Richness Lemma (Lemma~\ref{l:rich})]{Proof of the Richness Lemma (Lemma~\ref{l:rich}).  When Bob
  speaks, at least one of the corresponding subrectangles has at least
$c/2$ columns that each contain at least $d$ 1-inputs.
When Alice speaks, at least one of the corresponding
subrectangles has at least
$c/2$ columns that each contain at least $d/2$ 1-inputs.}
\label{f:rich}
\end{figure}

By induction, we conclude that there is a 1-input $\inputs$ such that,
at each point of the protocol's execution on $\inputs$ (with Alice and
Bob having sent $\alpha$ and $\beta$ bits so-far, respectively), the
current rectangle $A \times B$ of inputs consistent with the
protocol's execution has at least $v/2^{\alpha+\beta}$ columns (in $A$) that
each contain at least $u/2^{\alpha}$ 1-inputs (among the rows of $B$).  Since
the protocol terminates with a monochromatic rectangle of $M(f)$ and with
$\alpha \le a$, $\beta \le b$, the proof is complete.
\end{prevproof}

Lemma~\ref{l:rich} and some simple calculations prove
Theorem~\ref{t:index_asym}.

\vspace{.1in}
\noindent
\begin{prevproof}{Theorem}{t:index_asym}
We first observe that the matrix $M(\index)$ has a set of
$\binom{n}{n/2}$ columns that each contain $n/2$ 1-inputs:
choose the columns (i.e., inputs for Bob) that correspond to
subsets $S \sse \{1,2,\ldots,n\}$ of size $n/2$ and, for such a
column $S$, consider the rows (i.e., indices for Alice) that
correspond to the elements of $S$.

Now suppose for contradiction that there is a protocol that solves
\index in which Alice always sends at most $a = \delta \log_2 n$ bits and
Bob always sends at most $b = n^{1-2\delta}$ bits.  Invoking
Lemma~\ref{l:rich} proves that the matrix $M(\index)$ has a
1-rectangle of size at least 
\[
\frac{n}{2} \cdot \underbrace{\frac{1}{2^{a}}}_{=n^{-\delta}} \times
\underbrace{\binom{n}{n/2}}_{\approx 2^n/\sqrt{n}} \cdot
\underbrace{\frac{1}{2^{a+b}}}_{=n^{-\delta} \cdot 2^{-n^{1-2\delta}}}
  = \tfrac{1}{2}n^{1-\delta} \times c_2 2^{n-n^{1-2\delta}}, 
\]
where $c_2 > 0$ is a constant independent of~$n$.  (We're using
here that $n \ge N(\delta)$ is sufficiently large.)

On the other hand, how many columns can there be in a 1-rectangle with
$\tfrac{1}{2}n^{1-\delta}$ rows?  If these rows correspond to the set
$S \sse \{1,2,\ldots,n\}$ of indices, then every column of the
1-rectangle must correspond to a superset of $S$.  There are
\[
2^{n-|S|} = 2^{n - \tfrac{1}{2}n^{1-\delta}}
\]
of these.  But
\[
c_2 2^{n-n^{1-2\delta}} > 2^{n - \tfrac{1}{2}n^{1-\delta}},
\]
for sufficiently large~$n$, providing the desired contradiction.
\end{prevproof}

Does the asymmetric communication complexity lower bound in
Theorem~\ref{t:index_asym} have any interesting implications?
By Lemma~\ref{l:qd}, a data structure that supports membership queries
with query time~$t$, space $s$, and word size~$w$ induces a
communication protocol for \index in which Alice sends at
most $t \log_2 s$ bits and Bob sends at most $tw$ bits.
For example, suppose $t = \Theta(1)$ and $w$
at most poly-logarithmic in~$n$.  
Since Bob only sends $tw$ bits in
the induced protocol for \index, he certainly does not send
$n^{1-2\delta}$ bits.\footnote{Here $\delta \in (0,1)$ is a
constant and $n \ge N(\delta)$ is sufficiently large.
Using the version of Theorem~\ref{t:index_asym} with ``$n^{1-2\delta}$''
replaced by ``$n^{1-c\delta}$'' for an arbitrary constant $c > 1$, we
can take $\delta$ arbitrarily close to~1.}
Thus, Theorem~\ref{t:index_asym} implies that Alice must send at least
$\delta \log_2 n$ bits in the protocol.  This implies that
\[
t \log_2 s \ge \delta \log_2 n
\]
and hence $s \ge n^{\delta/t}$.  The good news is that this is a
polynomial lower bound for every constant~$t$.  The bad news is that
even for $t=1$, this argument will never prove a super-linear lower
bound.  We don't expect to prove a super-linear lower bound
in the particular case of the membership problem, since there is a
data structure for this problem with constant query time and linear
space (e.g., perfect hashing~\citep{FKS}).
For the $(1+\eps)$-approximate nearest neighbor problem, on the other
hand, we want to prove a lower bound of the form
$n^{\Omega(\eps^{-2})}$.  To obtain such a super-linear lower bound,
we need to reduce from a communication problem harder than \index ---
or rather, a communication problem in which Alice's input is bigger
than $\log_2 n$ and in which she still reveals almost her entire input
in every communication protocol induced by a constant-query data
structure.

\subsubsection{\textsc{$(k,\ell)$-Disjointness}}

A natural idea for modifying \index so that Alice's input is bigger
is to give Alice {\em multiple} indices; Bob's input remains an
$n$-bit vector.  The new question is whether or not for {\em at least 
  one} of Alice's indices, Bob's input is a~1.\footnote{This is
  reminiscent of 
  a ``direct sum,'' where Alice and Bob are given multiple
  instances of a communication problem and have to solve all of them.
  Direct sums are a fundamental part of communication
  complexity, but we won't have time to discuss them.}
This problem is essentially equivalent --- up to the details of how
Alice's input is encoded --- to \disj.

This section considers the special case of \disj where the sizes of
the sets given to Alice and Bob are restricted.  
If we follow the line of argument 
in the proof of Theorem~\ref{t:index_asym}, the best-case scenario 
is a space lower bound of
$2^{\Omega(a)}$, where $a$ is the length of Alice's input; see also
the proofs of Corollary~\ref{cor:simplelb} and Theorem~\ref{t:lb} at
the end of the lecture.
This is why the \index problem (where Alice's set is a 
singleton and $a = \log_2 n$)
cannot lead --- at least via Lemma~\ref{l:qd} --- to super-linear
data structure lower bounds.  The minimum $a$ necessary for the
desired space lower bound of $n^{\Omega(\eps^{-2})}$ is $\eps^{-2}
\log_2 n$.  This motivates considering instances of \disj in which Alice
receives a set of size $\eps^{-2}$.  Formally, we define \kldisj as
the communication problem in which Alice's input is a set of size $k$
(from a universe $U$) and Bob's input is a set of size $\ell$ (also
from $U$), and the goal is to determine whether or not the sets are
disjoint (a 1-input) or not (a 0-input).

We next extend the proof of Theorem~\ref{t:index_asym} to show the
following.
\begin{theorem}[\citealt{AIP06,M+98}]\label{t:kldisj}
For every $\eps,\delta > 0$ and every sufficiently large $n \ge
N(\eps,\delta)$, in every communication protocol that
solves \epsdisj with a universe of size $2n$, either:
\begin{itemize}

\item [(i)] Alice sends at least $\frac{\delta}{\eps^2} \log_2 n$
  bits; or

\item [(ii)] Bob sends at least $n^{1-2\delta}$ bits.

\end{itemize}
\end{theorem}
As with Theorem~\ref{t:index_asym}, 
we'll prove Theorem~\ref{t:kldisj} for
the special case of deterministic protocols.  The theorem also holds
for randomized protocols with two-sided error~\citep{AIP06}, and we'll
use this stronger result in Theorem~\ref{t:lb} below.  (Our upper bound
in Section~\ref{ss:dec} is randomized, so we really want a randomized
lower bound.)  The proof for randomized protocols
argues along the lines of the $\Omega(\sqrt{n})$ lower bound for the
standard version of \disj, and is not as hard as the stronger
$\Omega(n)$ lower bound (recall the discussion in
Section~\ref{ss:disj_proof}).

\vspace{.1in}
\noindent
\begin{prevproof}{Theorem}{t:kldisj}
Let $M$ denote the 0-1 matrix corresponding to the \epsdisj function.
Ranging over all subsets of $U$ of size $n$, and for a given such set
$S$, over all subsets of $U \sm S$ of size $\tfrac{1}{\eps^2}$, we see
that $M$ has at least $\binom{2n}{n}$ columns that each have at least
$\binom{n}{\eps^{-2}}$ 1-inputs.

Assume for contraction that there is a communication protocol for
\epsdisj such that neither~(i) nor~(ii) holds.
By the Richness Lemma (Lemma~\ref{l:rich}), there
exists a 1-rectangle $A \times B$ where
\begin{equation}\label{eq:A}
|A| = \binom{n}{\eps^{-2}} \cdot 2^{-\delta \tfrac{\log
    n}{\eps^2}} \ge
(\eps^2 n)^{\tfrac{1}{\eps^2}} \cdot n^{-\tfrac{\delta}{\eps^2}} =
\eps^{\tfrac{2}{\eps^2}} n^{\tfrac{1}{\eps^2}(1-\delta)}
\end{equation}
and
\begin{equation}\label{eq:B}
|B| = 
\underbrace{\binom{2n}{n}}_{\approx 2^{2n}/\sqrt{2n}} \cdot 2^{-\delta
  \tfrac{\log  n}{\eps^2}} \cdot 
2^{-n^{1-2\delta}} 
\ge 
2^{2n-n^{1-3\delta/2}},
\end{equation}
where in~\eqref{eq:B} we are using that $n$ is sufficiently large.

Since $A \times B$ is a rectangle, $S$ and $T$ are disjoint for every
choice of $S \in A$ and $T \in B$.  This implies that $\cup_{S \in A}
S$ and $\cup_{T \in B} T$ are disjoint sets.
Letting 
\[
s = \left| \cup_{S \in A} S \right|,
\]
we have
\begin{equation}\label{eq:A2}
|A| \le \binom{s}{\eps^{-2}} \le s^{1/\eps^2}.
\end{equation}
Combining~\eqref{eq:A} and~\eqref{eq:A2} implies that
\[
s \ge \eps^2 n^{1-\delta}.
\]
Since every subset $T \in B$ avoids the $s$ elements in $\cup_{S \in
  A} S$,
\begin{equation}\label{eq:B2}
|B| \le 2^{2n-s} \le 2^{2n-\eps^2n^{1-\delta}}.
\end{equation}
Inequalities~\eqref{eq:B} and~\eqref{eq:B2} furnish the desired
contradiction.
\end{prevproof}

The upshot is that, 
for the goal of proving a communication lower bound of
$\Omega(\eps^{-2} \log n)$ (for Alice, in a \qd problem)
and a consequent data structure space lower
bound of $n^{\Omega(\eps^{-2})}$,
\epsdisj is a promising candidate to reduce from.

\subsection[Lower Bound for the Approximate Nearest
  Neighbor Problem]{Lower Bound for the $(1+\eps)$-Approximate Nearest
  Neighbor Problem}\label{ss:lb}


The final major step is to show that \epsdisj, which is hard by
Theorem~\ref{t:kldisj}, reduces to the \qd problem for the decision
version of the $(1+\eps)$-nearest neighbor problem.

\subsubsection{A Simpler Lower Bound of $n^{\Omega(\eps^{-1})}$}

We begin with a simpler reduction that leads to a suboptimal but still 
interesting space lower bound of $n^{\Omega(\eps^{-1})}$.  In this
reduction, we'll reduce from \edisj rather than \epsdisj.  
Alice is given a $\tfrac{1}{\eps}$-set $S$ (from a universe $U$ of
size $2n$), which we need to map to a
nearest-neighbor query.  Bob is given an $n$-set $T \sse U$, which we
need to map to a point set.

Our first idea is to map the input to \edisj to a nearest-neighbor
query in the $2n$-dimensional hypercube $\zo^{2n}$.  
Alice performs the 
obvious mapping of her input, from the set $S$ to a query point
$\bfq$ that is the characteristic vector of $S$ (which lies in
$\zo^{2n}$).  Bob maps his input $T$ to the point set $P = \{ \bfe_i
\,:\, i \in T \}$, where $\bfe_i$ denotes the characteristic vector of
the singleton $\{ i \}$ (i.e., the $i$th standard basis vector).

If the sets $S$ and $T$ are disjoint, then the corresponding query $\bfq$
has Hamming distance $\tfrac{1}{\eps}+1$ from every point in the
corresponding point set $P$.  If $S$ and $T$ are not disjoint, then
there exists a point $\bfe_i \in P$ such that the Hamming distance
between $\bfq$ and $\bfe_i$ is $\tfrac{1}{\eps} - 1$.  Thus, the \edisj
problem reduces to the $(1+\eps)$-approximate nearest neighbor problem
in the $2n$-dimensional Hamming cube, where $2n$ is also the size of
the universe from which the point set is drawn.

We're not done, because extremely high-dimensional nearest neighbor
problems are 
not very interesting.  The convention in nearest neighbor
problems is to assume that the word size $w$ --- recall
Section~\ref{ss:cpm} --- is at least the dimension.  
When $d$ is at least the size of the universe from which points are drawn, 
an entire point set can be described using a single word!  This means
that our reduction so far cannot possibly yield an interesting lower
bound in the cell probe model.  We fix this issue by applying
dimension reduction --- just as in our upper bound in
Section~\ref{ss:dec} --- to the instances produced by the above
reduction.

Precisely, we can use the following embedding lemma.
\begin{lemma}[Embedding Lemma \#1]\label{l:emb1}
There exists a randomized function $f$ from $\zo^{2n}$ to $\zo^d$ 
with $d = \Theta(\tfrac{1}{\eps^2} \log n)$
and a
constant $\alpha > 0$ such that, 
for every
  set $P \sse \zo^{2n}$ of $n$ points and query $\bfq \in \zo^{2n}$
  produced by the reduction above, with probability at least
  $1-\tfrac{1}{n}$:
\begin{itemize}

\item [(1)] if the nearest-neighbor distance between $\bfq$ and $P$ is
  $\tfrac{1}{\eps} -1$, then the nearest-neighbor distance between
  $f(\bfq)$ and $f(P)$ is at most $\alpha$;

\item [(2)] if the nearest-neighbor distance between $\bfq$ and $P$ is
  $\tfrac{1}{\eps} +1$, then the nearest-neighbor distance between
  $f(\bfq)$ and $f(P)$ is at least $\alpha(1+h(\eps))$, where $h(\eps) >
  0$ is a constant depending on $\eps$ only.

\end{itemize}
\end{lemma}
Lemma~\ref{l:emb1} is an almost immediate consequence of
Corollary~\ref{cor:protocol} --- the map $f$ just takes $d =
\Theta(\eps^{-2} \log n)$ random inner products with $2n$-bit vectors,
where the probability of a ``1'' is roughly $\eps/2$.  We used this
idea in Section~\ref{ss:dec} for a data structure --- here we're using
it for a lower bound!

Composing our initial reduction with Lemma~\ref{l:emb1} yields the
following.
\begin{corollary}\label{cor:eps}
Every randomized asymmetric communication lower bound for \edisj
carries over to the \qd problem for the $(1+\eps)$-approximate nearest
neighbor problem in $d = \Omega(\eps^{-2} \log n)$ dimensions.
\end{corollary}

\begin{proof}
To recap our ideas, the reduction works as follows.  Given inputs to
\edisj, Alice interprets her input as a query and Bob interprets his
input as a point set (both in $\zo^{2n}$) as described at the
beginning of the section.  They use shared randomness to choose the
function $f$ of Lemma~\ref{l:emb1}
and use it to map their inputs to $\zo^d$ with $d =
\Theta(\eps^{-2} \log n)$.  They run the assumed protocol for the \qd
problem for the $(1+\eps)$-approximate nearest neighbor
problem in $d$ dimensions.
Provided the hidden constant in the definition of~$d$ is sufficiently large,
correctness (with high probability) is guaranteed by
Lemma~\ref{l:emb1}.  (Think of Lemma~\ref{l:emb1} as being invoked with
a parameter $\eps'$ satisfying $h(\eps') = \eps$.)
The amount of
communication used by the \edisj protocol is
identical to that of the \qd protocol.\footnote{The \edisj protocol
  is   randomized with two-sided error
even if the \qd protocol is deterministic.  This
  highlights our need for Theorem~\ref{t:kldisj} in its full
  generality.}
\end{proof}

Following the arguments of Section~\ref{sss:index} translates our
asymmetric communication complexity
lower bound (via Lemma~\ref{l:qd}) to a data
structure space lower bound.

\begin{corollary}\label{cor:simplelb}
Every data structure for the decision version of the
$(1+\eps)$-approximate nearest neighbors problem with query time $t =
\Theta(1)$ and word size $w = O(n^{1-\delta})$ for constant $\delta >
0$ uses space $s = n^{\Omega(\eps^{-1})}$.
\end{corollary}

\begin{proof}
Since $tw = O(n^{1-\delta})$, in the induced communication protocol
for the \qd problem (and hence \edisj, via Corollary~\ref{cor:eps}),
Bob sends a sublinear number of bits.  
Theorem~\ref{t:kldisj} then implies that Alice sends at least
$\Omega(\eps^{-1} \log n)$ bits, and so (by Lemma~\ref{l:qd}) we
have $t \log_2 s = \Omega(\eps^{-1} \log n)$.  Since $t = O(1)$,
this implies that $s = n^{\Omega(\eps^{-1})}$.
\end{proof}

\subsubsection{The $n^{\Omega(\eps^{-2})}$ Lower Bound}

The culmination of this lecture is the following.
\begin{theorem}\label{t:lb}
Every data structure for the decision version of the
$(1+\eps)$-approximate nearest neighbors problem with query time $t =
\Theta(1)$ and word size $w = O(n^{1-\delta})$ for constant $\delta >
0$ uses space $s = n^{\Omega(\eps^{-2})}$.
\end{theorem}

The proof is a refinement of the embedding arguments we used to prove
Corollary~\ref{cor:simplelb}.  In that proof, the reduction structure
was 
\[
S,T \sse U \mapsto \{0,1\}^{2n} \mapsto \{0,1\}^d,
\]
with inputs $S,T$ of \edisj mapped to the $2n$-dimensional Hamming
cube and then to the $d$-dimensional Hamming cube, with $d =
\Theta(\eps^{-2} \log n)$.  

The new plan is
\[
S,T \sse U \mapsto (\RR^{2n},\ell_2) \mapsto (\RR^{D},\ell_1)
\mapsto \zo^{D'} \mapsto \zo^d,
\]
where $d = \Theta(\eps^{-2} \log n)$ as before, and $D,D'$ can 
be very large.  Thus we map inputs $S,T$ of \epsdisj to
$2n$-dimensional Euclidean space (with the $\ell_2$ norm), which we
then map (preserving distances) to high-dimensional space with the
$\ell_1$ norm, then to the high-dimensional Hamming cube,
and finally to the $\Theta(\eps^{-2} \log n)$-dimensional Hamming cube
as before (via Lemma~\ref{l:emb1}).  The key insight is that switching
the initial embedding from the high-dimensional hypercube to
high-dimensional Euclidean space achieves a nearest
neighbor gap of $1\pm \eps$ even when Alice begins with a
$\tfrac{1}{\eps^2}$-set; the rest of the argument uses standard (if
non-trivial) techniques to eventually get back to a hypercube of
reasonable dimension.

To add detail to the important first step, consider inputs $S,T$ to
\epsdisj.
Alice maps her set $S$ to a query vector $\bfq$ that is $\eps$ times the
characteristic vector of~$S$, which we interpret as a point in
$2n$-dimensional Euclidean space.
Bob maps his input $T$ to the point set $P = \{ \bfe_i
\,:\, i \in T \}$, again in $2n$-dimensional Euclidean space, where
$\bfe_i$ denotes the $i$th standard basis vector.

First, suppose that $S$ and $T$ are disjoint.  Then, the $\ell_2$
distance between Alice's query~$\bfq$ and each point~$\bfe_i \in P$ is
\[
\sqrt{1 + \tfrac{1}{\eps^2} \cdot \eps^2} = \sqrt{2}.
\]
If $S$ and $T$ are not disjoint, then there exists a point $\bfe_i \in P$
such that the $\ell_2$ distance between $\bfq$ and $\bfe_i$ is:
\[
\sqrt{(1-\eps)^2 + \left(\tfrac{1}{\eps^2} - 1\right)\eps^2}
= \sqrt{2 - 2\eps} \le \sqrt{2} \left(1 - \tfrac{\eps}{2} \right).
\]
Thus, as promised, switching to the $\ell_2$ norm --- and tweaking
Alice's query -- allows us to get a $1\pm \Theta(\eps)$ gap in
nearest-neighbor distance between the ``yes'' and ``no'' instances of
\epsdisj.  This immediately yields (via Theorem~\ref{t:kldisj},
following the proof of Corollary~\ref{cor:simplelb}) 
lower bounds for the $(1+\eps)$-approximate nearest neighbor problem
in high-dimensional Euclidean space in the cell-probe model.
We can extend these lower bounds
to the hypercube through the following embedding lemma.

\begin{lemma}[Embedding Lemma \#2]\label{l:emb2}
For every $\delta > 0$
there exists a randomized function $f$ from $\RR^{2n}$ to $\zo^D$ (with
possibly large $D = D(\delta)$) such that, 
for every
  set $P \sse \zo^{2n}$ of $n$ points and query $\bfq \in \zo^{2n}$
  produced by the reduction above, with probability at least
  $1-\tfrac{1}{n}$,
\[
\ell_H(f(\bfp),f(\bfq)) \in (1 \pm \delta) \cdot
\ell_2(\bfp,\bfq)
\]
for every $\bfp \in P$.
\end{lemma}
Thus Lemma~\ref{l:emb2} says that one can re-represent a query~$\bfq$
and a set~$P$ of $n$ points in $\RR^{2n}$ in a high-dimensional
hypercube so that the nearest-neighbor distance --- $\ell_2$ distance
in the domain, Hamming distance in the range --- is approximately
preserved, with the approximation factor tending to~1 as the
number~$D$ of dimensions tends to infinity.  The Exercises outline the
proof, which combines two standard facts from the theory of metric
embeddings: 
\begin{enumerate}

\item ``$L_2$ embeds isometrically into $L_1$.''
For every $\delta > 0$ and dimension $D$
there exists a randomized function $f$ from $\RR^{D}$ to $\RR^{D'}$,
where $D'$ can depend on $D$ and $\delta$, such that,
for every set $P \sse \RR^D$ of $n$ points,
with probability at least $1-\tfrac{1}{n}$,
\[
\|f(\bfp)-f(\bfp')\|_1 \in (1 \pm \delta) \cdot \|\bfp-\bfp'\|_2
\]
for all $\bfp,\bfp' \in P$.

\item ``$L_1$ embeds isometrically into the 
(scaled) Hamming cube.''
For every $\delta > 0$,
there exists constants $M = M(\delta)$ and
$D'' = D''(D',\delta)$ 
and a function $g:\RR^{D'} \rightarrow \zo^{D''}$ such
that, for every set $P \sse \RR^{D'}$,
\[
d_H(g(\bfp,\bfp')) = M \cdot \|\bfp-\bfp'\|_1 \pm \delta
\]
for every $\bfp,\bfp' \in P$.

\end{enumerate}

With Lemma~\ref{l:emb2} in hand, we can prove Theorem~\ref{t:lb} by
following the argument in Corollary~\ref{cor:simplelb}.

\vspace{.1in}
\noindent
\begin{prevproof}{Theorem}{t:lb}
By Lemmas~\ref{l:emb1} and~\ref{l:emb2}, Alice and Bob can use a
communication protocol that solves the \qd problem for the decision
version of~$(1+\eps)$-nearest neighbors to solve the \epsdisj problem,
with no additional communication (only shared randomness, to pick the
random functions in Lemmas~\ref{l:emb1}
and~\ref{l:emb2}).\footnote{Strictly speaking, we're using a
  generalization of Lemma~\ref{l:emb1} (with the same proof) where the
  query and point set can lie in a hypercube of arbitrarily large
  dimension, not just~$2n$.}
Thus, the
(randomized) asymmetric communication lower bound for the latter
problem applies also to the former problem.

Since $tw = O(n^{1-\delta})$, in the induced communication protocol
for the \qd problem (and hence \epsdisj), Bob sends a sublinear number
of bits.  
Theorem~\ref{t:kldisj} then implies that Alice sends at least
$\Omega(\eps^{-2} \log_2 n)$ bits, and so (by Lemma~\ref{l:qd}) we
have $t \log_2 s = \Omega(\eps^{-2} \log_2 n)$.  Since $t = O(1)$,
this implies that $s = n^{\Omega(\eps^{-2})}$.
\end{prevproof}

\chapter{Lower Bounds in Algorithmic Game Theory}
\label{cha:lower-bounds-algor}

\section{Preamble}

This lecture explains some applications of communication complexity to
proving lower bounds in {\em algorithmic game theory (AGT)}, at the border
of computer science and economics.
In AGT, the natural description size
of an object is often exponential in a parameter of interest, 
and the goal is to perform
non-trivial computations in time polynomial in the parameter 
(i.e., logarithmic in the description size).  As we know, communication
complexity is a great tool for understanding when non-trivial
computations require looking at most of the input.

\section{The Welfare Maximization Problem}\label{s:wm}

The focus of this lecture is the following optimization problem, which
has been studied in AGT more than any other.
\begin{enumerate}

\item There are $k$ players.

\item There is a set $M$ of $m$ items.

\item Each player~$i$ has a {\em valuation} $v_i:2^M \rightarrow
  \RR_+$.  The number $v_i(T)$ indicates $i$'s value, or willingness
  to pay, for the items~$T \sse M$.
The valuation is the private input of player $i$ --- $i$ knows
  $v_i$ but none of the other $v_j$'s.  
We assume that $v_i(\emptyset) = 0$ and that the valuations are
{\em monotone}, meaning $v_i(S) \le v_i(T)$ whenever $S \sse T$.
To avoid bit complexity issues, we'll also assume that all of the
$v_i(T)$'s are integers with description length polynomial in $k$ and
$m$.

\end{enumerate}
Note that we may have more than two players --- more than just Alice
and Bob.  Also note that the description length of a player's valuation is
exponential in the number of items $m$.  

In the {\em welfare-maximization problem}, the goal is to partition
the items $M$ into sets $T_1,\ldots,T_k$ to maximize, at least
approximately, the welfare 
\[
\sum_{i=1}^k v_i(T_i),
\]
using communication polynomial in $n$ and $m$.  Note this amount of
communication is logarithmic in the sizes of the private inputs.

The main motivation for this problem is combinatorial auctions.
Already in the domain of government spectrum auctions, dozens of such
auctions have raised hundreds of billions of dollars of revenue.
They have also been used for other applications such as allocating
take-off and landing slots at airports.
For example, items could represent licenses for wireless spectrum ---
the right to use a certain frequency range in a certain geographic
area.  Players would then be wireless telecommunication companies.
The value $v_i(S)$ would be the amount of profit company~$i$ expects
to be able to extract from the licenses in~$S$.

Designing good combinatorial auctions requires careful attention to
``incentive issues,'' making the auctions as robust as possible to
strategic behavior by the (self-interested) participants.  
Incentives won't play much of a role in this lecture.
Our lower bounds for protocols in Section~\ref{s:lb_agt} apply even in the
ideal case where players are fully cooperative.
Our lower bounds for equilibria in Section~\ref{s:condpoa} effectively
apply no matter how incentive issues are resolved.

\section{Multi-Party Communication Complexity}\label{s:mpcc}

\subsection{The Model}

Welfare-maximization problems have an arbitrary number~$k$ of players,
so lower bounds for them follow most naturally from lower bounds for
{\em multi-party} communication protocols.  The extension from two to
many parties proceeds as one would expect, so we'll breeze through the
relevant points without much fuss.

Suppose we want to compute a Boolean function $f:\{0,1\}^{n_1} \times
  \{0,1\}^{n_2} \times \cdots \times \zo^{n_k} \rightarrow \zo$ that
  depends on the $k$ inputs $\bfx_1,\ldots,\bfx_k$.  We'll be interested 
  in the {\em number-in-hand (NIH)} model, where player~$i$ only knows
  $\bfx_i$.
What other model could there be, you ask?  There's also the stronger {\em
  number-on-forehead (NOF)} model, where player~$i$ knows everything
except $\bfx_i$.  (Hence the name --- imagine the players are sitting in a
circle.)  The NOF model is studied mostly for 
its connections to circuit complexity; it has few direct algorithmic
applications, so we won't discuss it in this course.  The NIH model is
the natural one for our purposes and, happily, it's also much easier
to prove strong lower bounds for it.

Deterministic protocols are defined as you would expect, with the
protocol specifying whose turn it is speak (as a function of the
protocol's transcript-so-far) and when the computation is complete.
We'll use the {\em blackboard model}, where we think of
the bits sent by each player as being written on a blackboard in
public view.\footnote{In the weaker {\em message-passing model},
  players communicate by point-to-point messages rather than via
  broadcast.}
Similarly, in a nondeterministic protocol, the prover writes a proof
on the blackboard, and the protocol accepts the input if and only if
all $k$ players accept the proof.

\subsection{The \mdisj Problem}

We need a problem that is hard for multi-party communication
protocols.  An obvious idea is to use an analog of \disj.  There is
some ambiguity about how to define a version of \disj for three or
more players.  For example, suppose there are three players, and
amongst the three possible pairings of them, two have disjoint sets
while the third have intersecting sets.  Should this count as a
``yes'' or ``no'' instance?  We'll skirt this issue by worrying only
about unambiguous inputs, that are either ``totally disjoint'' or
``totally intersecting.''

Formally, in the \mdisj problem, each of the $k$ players $i$ holds an
input $\bfx_i \in \zo^n$.  (Equivalently, a set $S_i \sse
\{1,2,\ldots,n\}$.)  
The task is to correctly identify inputs that
fall into one of the following two cases:
\begin{itemize}

\item [(1)] ``Totally disjoint,'' with $S_i \cap S_{i'} = \emptyset$ for
every $i \neq i'$.

\item [(0)] ``Totally intersecting,'' with $\cap_{i=1}^k S_i \neq
  \emptyset$.

\end{itemize}
When $k=2$, this is just \disj.  When $k > 2$, there are inputs that
are neither 1-inputs nor 0-inputs.  We let protocols off the hook on
such ambiguous inputs --- they can answer ``1'' or ``0'' with
impunity.

In the next section, we'll prove the following communication
complexity lower bound for \mdisj, credited to
Jaikumar Radhakrishnan and Venkatesh Srinivasan
in~\cite{N02}.
\begin{theorem}\label{t:mdisj}
The nondeterministic communication complexity of \mdisj, with $k$
players with $n$-bit inputs, is $\Omega(n/k)$.
\end{theorem}
The nondeterministic lower bound is for verifying a 1-input.  (It is
easy to verify a 0-input --- the prover just suggests the index of an
element $r$ in $\cap_{i=1}^k S_i$, the validity of which is easily
checked privately by each of the players.)  

In our application in Section~\ref{s:lb_agt}, we'll be interested in the
case where $k$ is much smaller than $n$, such as $k = \Theta(\log
n)$.  Intuition might suggest that the lower bound should be
$\Omega(n)$ rather than $\Omega(n/k)$, but this is incorrect --- a
slightly non-trivial argument shows that Theorem~\ref{t:mdisj} is
tight for nondeterministic protocols
(for all small enough $k$, like $k = O(\sqrt{n})$).  
See the Homework for details.
This factor-$k$ difference won't matter for our applications, however.

\subsection{Proof of Theorem~\ref{t:mdisj}}

The proof of Theorem~\ref{t:mdisj} has three steps, all of which are
generalizations of familiar arguments.

\vspace{.1in}
\noindent
\textbf{Step 1:}
{\em Every deterministic protocol with communication cost $c$ induces
  a partition of 
  $M(f)$ into at most $2^c$ monochromatic boxes.}
By ``$M(f)$,'' we mean the $k$-dimensional array in which the $i$th
dimension is indexed by the possible inputs of player~$i$, and an
array entry contains the value of the function $f$ on the
corresponding joint input.  By a ``box,'' we mean the $k$-dimensional
generalization of a rectangle --- a subset of inputs that can be
written as a product $A_1 \times A_2 \times \cdots \times A_k$.
By ``monochromatic,'' we mean a box that does not contain both a
1-input and a 0-input.  (Recall that for the \mdisj problem there are
also wildcard (``*'') inputs --- a monochromatic box can contain any
number of these.)

The proof of this step is the same as in the two-party case.  We just
run the protocol and keep track of the joint inputs that are
consistent with the transcript.  The box of all inputs is consistent
with the empty transcript, and the box structure is preserved
inductively: when player $i$ speaks, it narrows down the remaining
possibilities for the input $\bfx_i$, but has no effect on the possible
values of the other inputs.  Thus every transcript corresponds to a
box, with these boxes partitioning $M(f)$.
Since the protocol's output is constant over such a box and the protocol
computes $f$, all of the boxes it induces are monochromatic with
respect to $M(f)$.

Similarly, every nondeterministic protocol with communication cost $c$
(for verifying 1-inputs)
induces a cover of the 1-inputs of $M(f)$ by at most $2^c$
monochromatic boxes. 

\vspace{.1in}
\noindent
\textbf{Step 2:}
{\em The number of $1$-inputs in $M(f)$ is $(k+1)^n$.}
This step and the next are easy generalizations of our second proof 
of our nondeterministic communication complexity lower bounds for
\disj (from Section~\ref{ss:udisj}): first we lower bound the
number of 1-inputs,
then we
upper bound the number of 1-inputs that can coexist in a single
1-box.  In a $1$-input $\minputs$, for every 
coordinate $\ell$, at most one of the $k$ inputs has a 1 in the $\ell$th
coordinate.  This yields $k+1$ options for each of the $n$ coordinates,
thereby generating a total of $(k+1)^n$ 1-inputs.

\vspace{.1in}
\noindent
\textbf{Step 3:}
{\em The number of $1$-inputs in a monochromatic box is at most
  $k^n$.}  Let $B = A_1 \times A_2 \times \cdots \times A_k$ be a 1-box.
The key claim here is: for each coordinate $\ell=1,\ldots,n$, there is
a player $i \in \{1,\ldots,k\}$ such that, for every input $\bfx_i \in
A_i$, the $\ell$th coordinate of $\bfx_i$ is 0.  That is, to each
coordinate we can associate an ``ineligible player'' that, in this box,
never has a 1 in that coordinate.  This is easily seen by
contradiction: otherwise, there exists a coordinate $\ell$ such that,
for every player $i$, there is an input $\bfx_i \in A_i$ with a~1 in the
$\ell$th coordinate.  As a box, this means that $B$ contains the input
$\minputs$.  But this is a 0-input, contradicting the assumption that
$B$ is a 1-box.

The claim implies the stated upper bound.  Every 1-input of $B$ can be
generated by choosing, for each coordinate $\ell$, an assignment of at
most one ``1'' in this coordinate to one of the $k-1$ eligible players
for this coordinate.  With only $k$ choices per coordinate, there are
at most $k^n$ 1-inputs in the box~$B$.

\vspace{.1in}
\noindent
\textbf{Conclusion:}
Steps~2 and~3 imply that covering of the 1s of the $k$-dimensional
array of the \mdisj function requires at least $(1+\tfrac{1}{k})^n$
1-boxes.  By the discussion in Step~1, this implies a lower bound of
$n \log_2 (1 + \tfrac{1}{k}) = \Theta(n/k)$ on the nondeterministic
communication complexity of the \mdisj function (and output~1).  This
concludes the proof of Theorem~\ref{t:mdisj}.

\begin{remark}[Randomized Communication Complexity of
    \mdisj]\label{rem:mdisj}
Randomized protocols with two-sided error also require communication
$\Omega(n/k)$ to solve \mdisj~\citep{G09,CKS03}.\footnote{There is also a 
  far-from-obvious matching upper bound of $O(n/k)$~\citep{HW07,CKS03}.}
This generalizes the $\Omega(n)$ lower bound that we stated (but did
not prove) in Theorem~\ref{t:disj_2sided}, so naturally we're not
going to prove this
lower bound either.  Extending the lower bound for \disj to \mdisj
requires significant work, but it is a smaller step than proving from
scratch a
 linear lower bound for \disj~\citep{KS92,R92}.  This is
especially true if one settles for the weaker lower bound of
$\Omega(n/k^4)$~\citep{AMS96}, which is good enough for our purposes in
this lecture.
\end{remark}

\section{Lower Bounds for Approximate Welfare Maximization}\label{s:lb_agt}

\subsection{General Valuations}

We now put Theorem~\ref{t:mdisj} to work and prove that it is
impossible to obtain a non-trivial approximation of the general
welfare-maximization problem with a subexponential (in $m$) amount of
communication.  First, we observe that a $k$-approximation is trivial.
The protocol is to give the full set of items~$M$ to the player with
the largest $v_i(M)$.  This protocol can clearly be implemented with a
polynomial amount of
communication.  To prove the approximation guarantee, consider a partition
$T_1,\ldots,T_k$ of $M$ with the maximum-possible welfare~$W^*$.
There is a player $i$ with $v_i(T_i) \ge W^*/k$.  The welfare obtained
by our simple protocol is at least $v_i(M)$; since we assume that valuations
are monotone, this is at least $v_i(T_i) \ge W^*/k$.

To apply communication complexity, it is convenient to turn the
optimization problem of welfare maximization into a decision problem.
In the \wm[k] problem, 
the goal is to correctly identify inputs that
fall into one of the following two cases:
\begin{itemize}

\item [(1)] Every partition $(T_1,\ldots,T_k)$ of the items has
  welfare at most~1.

\item [(0)] There exists a partition $(T_1,\ldots,T_k)$ of the items
  with welfare at least $k$.

\end{itemize}
Clearly, communication lower bounds for \wm[k] apply more generally to
the problem of obtaining a better-than-$k$-approximation of the
maximum welfare.

We prove the following.
\begin{theorem}[\citealt{N02}]\label{t:wmk}
The communication complexity of \wm[k] is $\exp \{ \Omega(m/k^2) \}$.
\end{theorem}
Thus, if the number of items~$m$ is at least $k^{2+\eps}$ for some
$\eps > 0$, then the communication complexity of the \wm[k] problem is
exponential.  Because the proof is a reduction from \mdisj, the lower
bound applies to deterministic protocols, nondeterministic protocols
(for the output~1), and randomized protocols with two-sided error.

The proof of Theorem~\ref{t:wmk} relies on Theorem~\ref{t:mdisj} and a
combinatorial gadget.  We construct this gadget using the
probabilistic method.  As a thought experiment, consider $t$ random
partitions $P^1,\ldots,P^t$ of $M$, where $t$ is a parameter to be
defined later.  By a random partition $P^j = (P^j_1,\ldots,P^j_k)$, we
just mean that each of the $m$ items is assigned
to exactly one of the
$k$ players, independently and uniformly at random.

We are interested in the probability that two classes of different
partitions intersect: for all $i \neq i'$ and $j \neq \ell$, 
since the probability that a given item is assigned to $i$ in
$P^j$ and also to $i'$ in $P^{\ell}$ is $\tfrac{1}{k^2}$, we
have
\[
\prob{P^j_i \cap P^{\ell}_{i'} = \emptyset} = 
\left( 1 - \frac{1}{k^2} \right)^m \le e^{-m/k^2}.
\]
Taking a Union Bound over the $k$ choices for $i$ and $i'$ and the $t$
choices for $j$ and $\ell$, we have
\begin{equation}\label{eq:int}
\prob{\exists i \neq i', j \neq \ell \text{ s.t.\ } P^j_i \cap
  P^{\ell}_{i'} = \emptyset} \le k^2t^2e^{-m/k^2}.
\end{equation}
Call $P^1,\ldots,P^t$ an {\em intersecting family} if $P^j_i \cap
P^{\ell}_{i'} \neq \emptyset$ whenever $i \neq i'$, $j \neq \ell$.
By~\eqref{eq:int}, the probability that our random experiment fails to
produce an intersecting family is less than~1 provided $t <
\tfrac{1}{k}e^{m/2k^2}$.  The following lemma is immediate.
\begin{lemma}\label{l:int}
For every $m,k \ge 1$,
there exists an intersecting family of partitions $P^1,\ldots,P^t$
with $t = \exp \{ \Omega(m/k^2) \}$.
\end{lemma}

A simple combination of Theorem~\ref{t:mdisj} and Lemma~\ref{l:int}
proves Theorem~\ref{t:wmk}.

\vspace{.1in}
\noindent
\begin{prevproof}{Theorem}{t:wmk}
The proof is a reduction from \mdisj.
Fix $k$ and $m$.  (To be interesting, $m$ should be significantly
bigger than $k^2$.)  
Let $(S_1,\ldots,S_k)$ denote an input to \mdisj
with $t$-bit inputs, where $t = \exp \{ \Omega(m/k^2) \}$ is the same
value as in Lemma~\ref{l:int}.  We can assume that the players have
coordinated in advance on an intersecting family of $t$ partitions of a
set~$M$ of $m$ items.  Each player~$i$ uses this family and its input
$S_i$ to form the following valuation:
\[
v_i(T) = 
\left \{
\begin{array}{cl}
1 & \text{if $T \supseteq P^j_i$ for some $j \in S_i$}\\
0 & \text{otherwise.}
\end{array}
\right.
\]
That is, player~$i$ is either happy (value~1) or unhappy (value~0),
and is happy if and only if it receives all of the items in the
corresponding class $P^j_i$ of some partition $P^j$ with index~$j$
belonging to its input to \mdisj.
The valuations $v_1,\ldots,v_k$ define an input to \wm[k].
Forming this input requires no communication between the players.

Consider the case where the input to \mdisj is a 1-input, with $S_i
\cap S_{i'} = \emptyset$ for every $i \neq i'$.  We claim that the
induced input to \wm[k] is a 1-input, with maximum welfare at most~1.
To see this, consider a partition $(T_1,\ldots,T_k)$ in which some
player~$i$ is happy (with $v_i(T_i) = 1$).  For some $j \in S_i$,
player $i$ receives all the items in $P^j_i$.  Since $j \not\in S_{i'}$
for every $i' \neq i$, the only way to make a second player~$i'$ happy
is to give it all the items in $P^{\ell}_{i'}$ in some other partition
$P^{\ell}$ with $\ell \in S_{i'}$ (and hence $\ell \neq j$).  Since
$P^1,\ldots,P^t$ is an intersecting family, this is impossible ---
$P^j_i$ and $P^{\ell}_{i'}$ overlap for every $\ell \neq j$.

When the input to \mdisj is a 0-input, with an element $r$ in the
mutual intersection $\cap_{i=1}^k S_i$, we claim that the induced
input to \wm[k] is a 0-input, with maximum welfare at least~$k$.  This
is easy to see: for $i=1,2,\ldots,k$, assign the items of $P^r_i$ to
player~$i$.  Since $r \in S_i$ for every $i$, this makes all $k$
players happy.

This reduction shows that a (deterministic, nondeterministic, or
randomized) protocol for \wm[k] yields one for \mdisj (with $t$-bit
inputs) with the same communication.  We conclude that the
communication complexity of \wm[k] is $\Omega(t/k) = \exp \{
\Omega(m/k^2) \}$.
\end{prevproof}

\subsection{Subadditive Valuations}

To an algorithms person, Theorem~\ref{t:wmk} is depressing,
as it rules out any non-trivial positive results.
A natural idea is to seek positive results by imposing additional
structure on players' valuations.  Many such restrictions have been
studied.  We consider here the case of {\em subadditive} valuations,
where each $v_i$ satisfies $v_i(S \cup T) \le v_i(S) + v_i(T)$ for every
pair $S,T \sse M$.  

Our reduction in Theorem~\ref{t:wmk} immediately yields a weaker
inapproximability result for welfare maximization with subadditive
valuations.  Formally, 
define the \wm[2] problem as that of
identifying inputs that
fall into one of the following two cases:
\begin{itemize}

\item [(1)] Every partition $(T_1,\ldots,T_k)$ of the items has
  welfare at most~$k+1$.

\item [(0)] There exists a partition $(T_1,\ldots,T_k)$ of the items
  with welfare at least $2k$.

\end{itemize}
Communication lower bounds for \wm[2] apply to
the problem of obtaining a better-than-$2$-approximation of the
maximum welfare.

\begin{corollary}[\citealt{DNS05}]\label{t:wm2}
The communication complexity of \wm[2] is\\ $\exp \{ \Omega(m/k^2) \}$,
even when all players have subadditive valuations.
\end{corollary}

\begin{proof}
Picking up where the reduction in the proof of Theorem~\ref{t:wm2}
left off, every player~$i$ adds~1 to its valuation
for every non-empty set of items.
Thus, the previously 0-1 valuations become 0-1-2
valuations that are only~0 for the empty set.  Such functions always
satisfy the subadditivity condition ($v_i(S \cup T) \le v_i(S) +
v_i(T)$).  1-inputs and 0-inputs of \mdisj now become 1-inputs and
0-inputs of \wm[2], respectively.  The communication complexity lower
bound follows.
\end{proof}

There is also a quite non-trivial matching upper bound of~2 for
deterministic, polynomial-communication protocols~\citep{F06}.

\section{Lower Bounds for Equilibria}\label{s:condpoa}

The lower bounds of the previous section show that
every protocol for the welfare-maximization problem that interacts
with the players and then explicitly
computes an allocation has either a bad approximation ratio or 
high communication cost.  Over the past five years, many researchers have
aimed to shift the work from the protocol to the players, by analyzing the
equilibria of simple auctions.  Can such equilibria bypass the
communication complexity lower bounds proved in Section~\ref{s:lb_agt}?
The answer is not obvious, because equilibria are defined
non-constructively, and not through a low-communication
protocol.\footnote{This question was bothering your instructor back in
  CS364B (Winter '14) --- hence, Theorem~\ref{t:condpoa}.}

\subsection{Game Theory}

Next we give the world's briefest-ever game theory tutorial.
See e.g.~\cite{SLB}, or the instructor's CS364A lecture notes, for a more
proper introduction.  We'll be brief because the details of these
concepts do not play a first-order role in the arguments below.

\subsubsection{Games}

A {\em (finite, normal-form) game} is specified by:
\begin{enumerate}

\item A finite set of $k \ge 2$ {\em players}.

\item For each player $i$, a finite {\em action set} $A_i$.

\item For each player $i$, a {\em utility function} $u_i(\actions)$
  that maps an action profile $\actions \in A_1 \times \cdots \times
  A_k$ to a real number.  The utility of a player generally depends
  not only on its action, but also those chosen by the other players.

\end{enumerate}
For example, in ``Rock-Paper-Scissors (RPS),'' there are two players,
each with three actions.  
A natural choice of utility functions is depicted in Figure~\ref{f:rps}.

\begin{figure}
\centering
\begin{tabular}{|c|c|c|c|}
  \hline
   & Rock & Paper & Scissors \\ \hline
  Rock & 0,0 & -1,1 & 1,-1 \\
  Paper & 1,-1 & 0,0 & -1,1 \\
  Scissors & -1,1 & 1,-1 & 0,0 \\
  \hline
\end{tabular}
\caption[Player utilities in Rock-Paper-Scissors]{Player utilities in Rock-Paper-Scissors.  The pair of numbers in a matrix entry denote the utilities of the row and column players, respectively, in a given outcome.}\label{f:rps}
\end{figure}

For a more complex and relevant example of a game, consider {\em
  simultaneous first-price auctions (S1As)}.  There are $k$ players.
An action $\actioni$ of a player~$i$ constitutes a bid $b_{ij}$
on each item $j$ of a set $M$ of $m$ items.\footnote{To keep the game
  finite, let's agree that each bid has to be an integer between 0 and
  some known upper bound $B$.}
In a S1A, each item is sold separately in parallel using a
``first-price auction'' --- 
the item is awarded to the highest bidder, and the price is whatever
that player bid.\footnote{You may have also heard of the {\em Vickrey}
  or {\em second-price} auction, where the winner does not pay their
  own bid, but rather the highest bid by someone else (the
  second-highest overall).  We'll stick with S1As for simplicity, but
  similar results are known for simultaneous second-price
  auctions, as well.}
To specify the utility functions, we assume that each player $i$ has a
valuation $v_i$ as in Section~\ref{s:wm}.  We define
\[
u_i(\actions) = \underbrace{v_i(S_i)}_{\text{value of items won}} -
\underbrace{\sum_{j \in S_i} b_{ij}}_{\text{price paid for them}},
\]
where $S_i$ denotes the items on which $i$ is the highest bidder
(given the bids of $\actions$).\footnote{Break ties in an arbitrary
  but consistent way.}  Note that the utility of a bidder
depends both on its own action and those of the other bidders.
Having specified the players, their actions, and their utility
functions, we see that an S1A is an example of a game.

\subsubsection{Equilibria}

Given a game, how should one reason about it?  The standard approach
is to define some notion of ``equilibrium'' and then study the
equilibrium outcomes.  There are many useful notions of equilibria
(see e.g.~the instructor's CS364A notes); for simplicity, we'll stick
here with the most common notion, (mixed) Nash
equilibria.\footnote{For the auction settings we study, ``Bayes-Nash
  equilibria'' are more relevant.  These generalize Nash equilibria,
  so our lower bounds immediately apply to them.}

A {\em mixed strategy} for a player~$i$ is a probability distribution
over its actions --- for example, the uniform distribution over
Rock/Paper/Scissors.  
A {\em Nash equilibrium} is a collection
$\sigma_1,\ldots,\sigma_k$ of mixed strategies, one per player, so
that each player is performing a ``best response'' to the others.
To explain, adopt the perspective of player~$i$.  We think of~$i$ as
knowing the mixed strategies $\sigmasmi$ used by the other $k-1$
players (but not their coin flips).  Thus, player~$i$ can compute the
expected payoff of each action $\actioni \in \Actioni$, where the
expectation assumes that the other $k-1$ players randomly and
independently select actions from their mixed strategies.  Every
action that maximizes $i$'s expected utility is a {\em best response}
to $\sigmasmi$.  Similarly, every probability distribution over best
responses is again a best response (and these exhaust the best
responses).
For example, in Rock-Paper-Scissors, both players playing the uniform
distribution yields a Nash equilibrium.  (Every action of a player has
expected utility~0 w.r.t.\ the mixed strategy of the other player, so
everything is a best response.)

Nash proved the following.
\begin{theorem}[\citealt{N50}]\label{t:nash}
In every finite game, there is at least one Nash equilibrium.
\end{theorem}
Theorem~\ref{t:nash} can be derived from, and is essentially
equivalent to, Brouwer's Fixed-Point Theorem.  Note that
a game can have a large number of \nes --- if you're trying
to meet a friend in New York City, with actions equal to
intersections, then every intersection corresponds to a Nash
equilibrium.

An {\em \ene} is the relaxation of a \ne in which no player can
increase its expected utility by more than $\eps$ by switching to a
different strategy.  Note that the set of \enes is nondecreasing with
$\eps$.  Such approximate Nash equilibria seem crucial to the lower
bound in Theorem~\ref{t:condpoa}, below.

\subsubsection{The Price of Anarchy}

So how good are the equilibria of various games, such as S1As?  To
answer this question, we use an analog of the approximation ratio,
adapted for equilibria.  Given a game (like an S1A) and a nonnegative
maximization objective function on the outcomes (like welfare), the
{\em price of 
  anarchy (POA)}~\citep{KP99} is defined as the ratio between the
objective function value of an optimal solution, and that of the worst
equilibrium.\footnote{Recall that games generally have multiple
  equilibria.  Ideally, we'd like an approximation guarantee that
  applies to {\em all} equilibria --- this is the point of the POA.}
If the equilibrium involves randomization, as with mixed 
strategies, then we consider its expected objective function value.

The POA of a game and a maximization objective function is always at
least~1.  It is common to identify ``good performance'' of a system
with strategic participants as having a POA close to~1.\footnote{An
  important issue, outside the scope of these notes, is the
  plausibility of a system reaching an 
  equilibrium.  A natural solution is to relax the notion of
  equilibrium enough so that it become ``relatively easy'' to reach an
  equilibrium.  See e.g.~the instructor's CS364A notes for much more on
  this point.}

For example,
the equilibria of S1As are surprisingly good in fairly general
settings.
\begin{theorem}[\citealt{FFGL13}]\label{t:ffgl}
In every S1A with subadditive bidder valuations, the POA is at most~2.
\end{theorem}
Theorem~\ref{t:ffgl} is non-trivial and we won't prove it here (see
the paper or the instructor's CS364B notes for a proof).
This result is particularly impressive because achieving an
approximation factor of~2 for the welfare-maximization problem with
subadditive bidder valuations by any means (other than brute-force
search) is not easy (see~\cite{F06}).

A recent result shows that the analysis of~\cite{FFGL13} is tight.
\begin{theorem}[\citealt{CKST14}]\label{t:ckst14}
The worst-case POA of S1As with subadditive bidder valuations is at
least~2.
\end{theorem}
The proof of Theorem~\ref{t:ckst14} is an ingenious explicit
construction --- the authors exhibit a choice of subadditive bidder
valuations and a Nash equilibrium of the corresponding S1A so that the
welfare of this equilibrium is only half of the maximum possible.
One reason that proving results like Theorem~\ref{t:ckst14} is
challenging is that it can be difficult to solve for a (bad) equilibrium
of a complex game like a S1A.

\subsection[POA Lower Bounds from Communication
  Complexity]{Price-of-Anarchy Lower Bounds from Communication Complexity}

Theorem~\ref{t:ffgl} motivates an obvious question: can we do better?
Theorem~\ref{t:ckst14} implies that the analysis in~\cite{FFGL13}
cannot be improved, but can we reduce the POA by considering a
different auction?  Ideally, the auction would still be ``reasonably
simple'' in some sense.  Alternatively, perhaps no ``simple'' auction
could be better than S1As?  If this is the case, it's not clear how to
prove it directly --- proving lower bounds via explicit constructions
auction-by-auction does not seem feasible.

Perhaps it's a clue that the POA upper bound of~2 for S1As
(Theorem~\ref{t:ffgl}) gets stuck at the same threshold for which
there is a lower bound for protocols that use polynomial communication
(Theorem~\ref{t:wm2}).  It's not clear, however, that a lower bound
for low-communication protocols has anything to do with equilibria.
In the spirit of the other reductions that we've seen in this course,
can we extract a low-communication protocol from an equilibrium?

\begin{theorem}[\citealt{R14}]\label{t:condpoa}
Fix a class $\V$ of possible bidder valuations.
Suppose there exists no nondeterministic protocol with 
subexponential (in $m$) communication for the 1-inputs of the
following promise version of the welfare-maximization problem
with bidder valuations in~$\V$: 
\begin{itemize}

\item [(1)] Every allocation has welfare at most $W^*/\alpha$.

\item [(0)] There exists an allocation with welfare at least $W^*$.

\end{itemize}
Let $\eps$ be bounded below by some inverse polynomial function of $n$
and $m$.
Then, for every auction with sub-doubly-exponential (in $m$) actions
per player, the worst-case POA of \enes with bidder valuations in $\V$
is at least $\alpha$.
\end{theorem}
Theorem~\ref{t:condpoa} says that lower bounds for nondeterministic
protocols carry over to all ``sufficiently simple'' auctions, where
``simplicity'' is measured by the number of actions available
to each player.  These POA lower bounds follow from communication
complexity lower bounds, and do not require any new explicit
constructions.

To get a feel for the simplicity constraint, note that
S1As with integral bids between 0 and~$B$ have $(B+1)^m$ actions per
player --- singly exponential in~$m$.  On the other hand, in a
``direct-revelation'' auction, where each bidder is allowed to
submit a bid on each bundle $S \sse M$ of items, each player has a
doubly-exponential (in $m$) number of actions.\footnote{Equilibria can
  achieve the optimal welfare in direct-revelation mechanisms, so the
  bound in Theorem~\ref{t:condpoa} on the number of actions is
  necessary.  See the Exercises for further details.}

The POA lower bound promised by Theorem~\ref{t:condpoa} is
only for \ene; since the POA is a worst-case measure and the set of
\enes is nondecreasing with $\eps$, this is weaker than a lower bound for
exact \nes.  It is an open question whether or not
Theorem~\ref{t:condpoa} holds also for the POA of exact \nes.
Arguably, Theorem~\ref{t:condpoa} is good enough for all
practical purposes --- a POA upper bound that holds for exact \nes and
does not hold (at least approximately) for \enes with very small $\eps$
is too brittle to be meaningful.

Theorem~\ref{t:condpoa} has a number of interesting corollaries.
First, since S1As have only a singly-exponential (in $m$) 
number of actions per
player, Theorem~\ref{t:condpoa} applies to them.  Thus, 
combining it with Theorem~\ref{t:wm2} recovers the POA lower bound of
Theorem~\ref{t:ckst14} --- modulo the exact vs.\ approximate \nes issue
--- and shows the optimality of the upper bound in
Theorem~\ref{t:ffgl} without an explicit construction.
More interestingly, this POA lower bound of~2 (for subadditive bidder
valuations) applies not only to S1As, but more generally to all
auctions in which each player has a sub-doubly-exponential number of
actions.  Thus, S1As are in fact {\em optimal} among the class of
all such auctions when bidders have subadditive valuations
(w.r.t.\ the worst-case POA of \enes).

We can also combine Theorem~\ref{t:condpoa} with Theorem~\ref{t:wmk}
to prove that no ``simple'' auction gives a non-trivial (better than
$k$-) approximation for general bidder valuation.  Thus with general
valuations, complexity is essential to any auction format that offers
good equilibrium guarantees.

\subsection{Proof of Theorem~\ref{t:condpoa}}

Presumably, the proof of Theorem~\ref{t:condpoa} extracts a
low-communication protocol from a good POA bound.  The hypothesis of
Theorem~\ref{t:condpoa} offers the clue that we should be looking to
construct a nondeterministic protocol.  So what could we use an
all-powerful prover for?  We'll see that a good role for the prover is
to suggest a \ne to the players.

Unfortunately, it's too expensive for the prover to even write down
the description of a \ne, even in S1As.  Recall that a mixed strategy is a
distribution over actions, and that each player has an exponential (in
$m$) number of actions available in a S1A.  Specifying a \ne thus
requires an exponential number of probabilities.  To circumvent this
issue, we resort to \enes, which are guaranteed to exist even if we
restrict ourselves to distributions with small descriptions.

\begin{lemma}[\citealt{LMM03}]\label{l:lmm}
For every $\eps > 0$ and every game with $k$ players with action sets
$A_1,\ldots,A_k$, there exists an \ene with description length
polynomial in $k$, $\log (\max_{i=1}^k |A_i|)$, and $\tfrac{1}{\eps}$.
\end{lemma}

We give the high-level idea of the proof of Lemma~\ref{l:lmm}; see the
Exercises for details.
\begin{enumerate}

\item Let $(\sigma_1,\ldots,\sigma_k)$ be a \ne.  (One exists, by
  Nash's Theorem.)

\item Run $T$ independent trials of the following experiment: draw
  actions $a^t_1 \sim \sigma_1,\ldots,a^t_k \sim \sigma_k$ for the
  $k$ players independently, according to their mixed strategies in
  the \ne.

\item For each $i$, define $\hat{\sigma}_i$ as the empirical
  distribution of the $a^t_i$'s.  (With the probability of~$a_i$ in
  $\hat{\sigma}_i$ equal to the fraction of trials in which $i$ played
  $a_i$.)

\item Use Chernoff bounds to prove that, if $T$ is at least a
  sufficiently large polynomial in $k$, $\log (\max_{i=1}^k |A_i|)$, and
  $\tfrac{1}{\eps}$, then with high probability
  $(\hat{\sigma}_1,\ldots,\hat{\sigma}_k)$ is an \ene.  Note that the
    natural description length of
    $(\hat{\sigma}_1,\ldots,\hat{\sigma}_k)$ --- for example, just by
      listing all of the sampled actions --- is polynomial in $n$,
$\log(\max_{i=1}^k |A_i|)$, and~$\tfrac{1}{\eps}$.

\end{enumerate}
The intuition is that, for $T$ sufficiently large, expectations with
respect to $\sigma_i$ and with respect to $\hat{\sigma}_{i}$ should
be roughly the same.  Since there are $|A_i|$ relevant expectations
per player (the expected utility of each of its actions) and Chernoff
bounds give deviation probabilities that have an inverse exponential
form, we might expect a $\log |A_i|$ dependence to show up in the
number of trials.

We now proceed to the proof of Theorem~\ref{t:condpoa}.

\vspace{.1in}
\noindent
\begin{prevproof}{Theorem}{t:condpoa}
Fix an auction with at most $A$ actions per player, and a value for
$\eps = \Omega(1/\poly(k,m))$.
Assume that, no matter what the bidder valuations
$v_1,\ldots,v_k \in \V$ are, the POA of \enes of the auction is at most
$\rho < \alpha$.  We will show that $A$ must be doubly-exponential in
$m$.

Consider the following nondeterministic protocol for computing a
1-input of the welfare-maximization problem --- for convincing the $k$
players that every allocation has welfare at most $W^*/\alpha$.  See
also Figure~\ref{f:condpoa}.  The
prover writes on a publicly visible blackboard an \ene
$(\sigma_1,\ldots,\sigma_k)$ of the auction,
with description length polynomial in $k$, $\log A$, and
$\tfrac{1}{\eps} = O(\poly(k,m))$ as guaranteed by Lemma~\ref{l:lmm}.
The prover also writes down the expected welfare contribution
$\expect{v_i(S)}$ of each bidder~$i$ in this equilibrium.

\begin{figure}
\centering
\includegraphics[width=.6\textwidth]{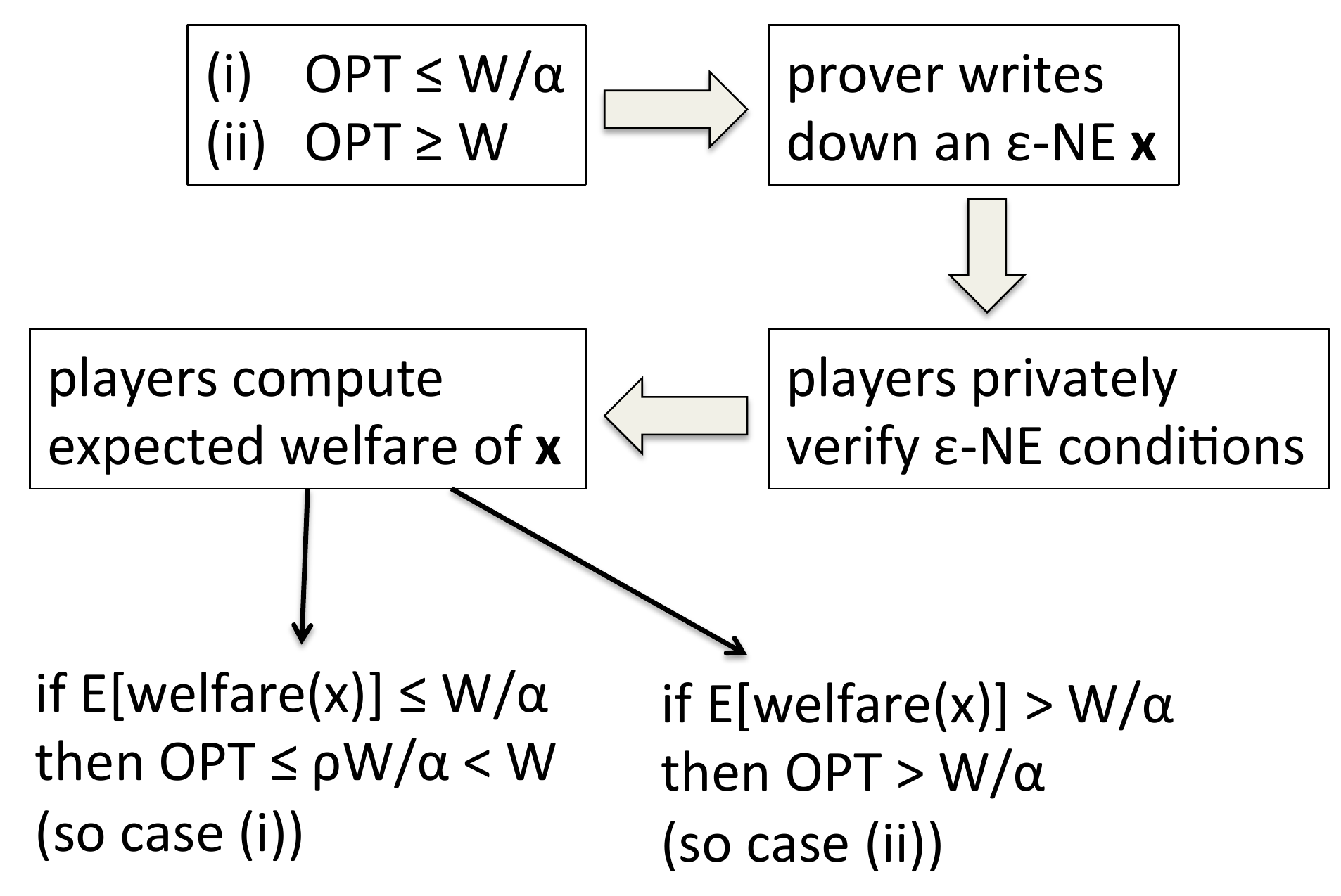}
\caption[Proof of Theorem~\ref{t:condpoa}]{Proof of Theorem~\ref{t:condpoa}.  How to extract a
  low-communication nondeterministic protocol from a good
  price-of-anarchy bound.}\label{f:condpoa}
\end{figure}

Given this advice, each player~$i$ verifies that $\sigma_i$ is indeed
an $\eps$-best response to the other $\sigma_j$'s and that its
expected welfare is as claimed when all players play the mixed strategies
$\sigma_1,\ldots,\sigma_k$.  Crucially, player~$i$ is fully equipped
to perform both of these checks without any communication --- it knows
its valuation $v_i$ (and hence its utility in each outcome of the
game) and the mixed strategies used by all players, and this is all
that is needed to verify the \ene conditions that apply to it and to
compute its expected contribution to the welfare.\footnote{These
  computations may take a super-polynomial amount of time, but they
  do not contribute to the protocol's cost.}  Player~$i$
accepts if and only if the prover's advice passes these two tests, and
if the expected welfare of the equilibrium is at most~$W^*/\alpha$.

For the protocol correctness, consider first the case of a 1-input,
where every allocation has welfare at most $W^*/\alpha$.  If the prover
writes down the description of an arbitrary \ene and the appropriate
expected contributions to the social welfare, then all of the players
will accept (the expected welfare is obviously at most $W^*/\alpha$).
We also need to argue that, for the case of a 0-input --- where
some allocation has welfare at least $W^*$ --- there is no proof that causes
all of the players to accept.  We can assume that the prover writes
down an \ene and its correct expected welfare~$W$, since otherwise at
least one player will reject.  Since the maximum-possible welfare is
at least $W^*$ and (by assumption) the POA of \enes is at most $\rho <
\alpha$, the expected welfare of the given \ene must satisfy $W \ge
W^*/\rho > W/\alpha$.  Since the players will reject such a proof, we
conclude that the protocol is correct.  
Our assumption then implies that the protocol has communication cost
exponential in~$m$.
Since the cost of the protocol is polynomial in $k$, $m$, and $\log A$, 
$A$ must be doubly exponential in $m$.
\end{prevproof}
Conceptually, the proof of Theorem~\ref{t:condpoa} argues
that, when the POA of \enes is small, every \ene provides a
privately verifiable proof of a good upper bound
on the maximum-possible welfare.  When such upper bounds require large
communication, the equilibrium description length (and hence the
number of actions) must be large.

\subsection{An Open Question}

While Theorems~\ref{t:wm2}, \ref{t:ffgl}, and~\ref{t:condpoa}
pin down the best-possible POA achievable by simple auctions with
subadditive bidder valuations, there are still open questions for
other valuation classes.  For example, a valuation $v_i$ is {\em
  submodular} if it satisfies
\[
v_i(T \cup \{j\}) - v_i(T) \le v_i(S \cup \{j\}) - v_i(S)
\]
for every $S \sse T \subset M$ and $j \notin T$.  This is a
``diminishing returns'' condition for set functions.  Every 
submodular function is also subadditive, so welfare-maximization with
the former valuations is only easier than with the latter.

The worst-case POA of S1As is exactly $\tfrac{e}{e-1} \approx 1.58$
when bidders have submodular valuations.
The upper bound was proved in~\cite{ST13}, the lower bound
in~\cite{CKST14}.   It is an open question whether or not there is a
simple auction with a smaller worst-case POA.  The best lower bound
known --- for nondeterministic protocols and hence, by
Theorem~\ref{t:condpoa}, for the POA of \enes of simple auctions --- is
$\tfrac{2e}{2e-1} \approx 1.23$.  Intriguingly, there is an upper
bound (slightly) better than $\tfrac{e}{e-1}$ for 
polynomial-communication
protocols~\citep{FV06} --- can this better upper bound also be realized
as the POA of a simple auction?  What is the best-possible
approximation guarantee, either for polynomial-communication protocols
or for the POA of simple auctions?

\chapter{Lower Bounds in Property Testing}
\label{cha:lower-bounds-prop}

\section{Property Testing}

We begin 
in this section with a brief introduction to the field of property testing.
Section~\ref{s:blr} explains the famous example of ``linearity testing.''
Section~\ref{s:ub_pt} gives upper bounds for the canonical problem of
``monotonicity testing,'' and Section~\ref{s:lb_pt} shows how to derive
property testing lower bounds from communication complexity lower
bounds.\footnote{Somewhat amazingly, this connection was only discovered
  in 2011~\citep{BBM11}, even though the connection is simple and
  property testing is a relatively mature field.}
These lower bounds will follow from our existing communication
complexity toolbox (specifically, \disj); no new results are
required.

Let $D$ and $R$ be a finite domain and range, respectively.  In this
lecture, $D$ will always be $\{0,1\}^n$, while $R$ might or might not
  be $\zo$.  A {\em property} is simply a set $\P$ of functions from
  $D$ to $R$.  Examples we have in mind include:
\begin{enumerate}

\item Linearity, where $\P$ is the set of linear functions 
(with $R$ a field and $D$ a vector space over $R$).

\item Monotonicity, where $\P$ is the set of monotone functions (with
  $D$ and $R$ being partially ordered sets).

\item Various graph properties, like bipartiteness (with functions
  corresponding to characteristic vectors of edge sets, with respect to a fixed vertex set).

\item And so on.  The property testing literature is vast.
  See~\cite{ron_survey} for a starting point.

\end{enumerate}

In the standard property testing model, one has ``black-box access''
to a function $f:D \rightarrow R$.  That is, one can only learn about
$f$ by supplying an argument $x \in D$ and receiving the function's
output  $f(x) \in R$.  The goal is to test membership in $\P$ by
querying $f$ as few times as possible.  Since the goal is to use a
small number of queries (much smaller than $|D|$), there is no hope of
testing membership exactly.  For example, 
suppose you derive $f$ from your
favorite monotone function by changing its value at a single point to
introduce a non-monotonicity.  There is little hope of detecting this
monotonicity violation with a small number of queries.  We therefore
consider a relaxed ``promise'' version of the membership problem.

Formally, we say that a function $f$ is {\em $\eps$-far} from the
property $\P$ if, for every $g \in \P$, $f$ and $g$ differ in at least
$\eps|D|$ entries.  Viewing functions as vectors indexed by
$D$ with coordinates in $R$, this definition says that $f$ has
distance at least $\eps|D|$ from its nearest neighbor in~$\P$ (under
the Hamming metric).  Equivalently, repairing $f$ so that it belongs
to $\P$ would require changing at least an $\eps$ fraction of its
values.  A function $f$ is {\em $\eps$-close} to $\P$ if it is not
$\eps$-far --- if it can be turned into a function in $\P$ by
modifying its values on strictly less than $\eps|D|$ entries.

The property testing goal is to query a function $f$ a small number of
times and then decide if:
\begin{enumerate}

\item $f \in \P$; or

\item $f$ is $\eps$-far from $\P$.

\end{enumerate}
If neither of these two conditions applies to $f$, then the tester is
off the hook --- any declaration is treated as correct.

A {\em tester} specifies a sequence of queries to the unknown function
$f$, and a declaration of either ``$\in \P$'' or ``$\eps$-far from
$\P$'' at its conclusion.
Interesting property testing results almost always require
randomization.  Thus, we allow the tester to be randomized, and allow
it to err with probability at most~$1/3$.  As with communication
protocols, testers come in various flavors.  
{\em One-sided error} means that functions in $\P$ are accepted with
probability~1, with no false negative allowed.  Testers with {\em
  two-sided error} are allowed both false positives and false
negatives (with probability at most~$1/3$, on every input that
satisfies the promise).
Testers can be {\em non-adaptive}, meaning that they flip all
their coins and specify all their queries up front, or {\em adaptive},
with queries chosen as a function of the answers to previously asked
queries.
For upper bounds, we prefer the weakest model of non-adaptive testers
with 1-sided error.
Often (though not always) in property testing, neither adaptivity
nor two-sided error leads to more efficient testers.
Lower bounds can be much more difficult to prove for adaptive testers
with two-sided error, however.

For a given choice of a class of testers, the {\em query complexity}
of a property $\P$ is the minimum (over testers) worst-case (over
inputs) number of queries used by a tester that solves the testing
problem for $\P$.  The best-case scenario is that the query complexity
of a property is a function of $\eps$ only; sometimes it depends on
the size of $D$ or $R$ as well.

\section{Example: The BLR Linearity Test}\label{s:blr}

The unofficial beginning of the field of property testing
is~\cite{BLR93}.  (For the official beginning, see~\cite{RS96}
and~\cite{GGR98}.)
The setting is $D = \{0,1\}^n$ and $R = \{0,1\}$, and the property is
the set of {\em linear functions}, meaning functions $f$ such that
$f(\bfx+\bfy) = f(\bfx) + f(\bfy)$ (over $\FF_2$) for all $\bfx,\bfy \in
\zo^n$.\footnote{Equivalently, these are the functions that can be
  written as $f(\bfx) = 
\sum_{i=1}^n a_ix_i$ for some $a_1,\ldots,a_n \in \zo$.}
The {\em BLR linearity test} is the following:
\begin{enumerate}

\item Repeat $t = \Theta(\tfrac{1}{\eps})$ times:

\begin{enumerate}

\item Pick $\bfx,\bfy \in \zo^n$ uniformly at random.

\item If $f(\bfx+\bfy) \neq f(\bfx)+f(\bfy)$ (over $\FF_2$), then REJECT.

\end{enumerate}

\item ACCEPT.

\end{enumerate}

It is clear that if $f$ is linear, then the BLR linearity test accepts
it with probability~1.  That is, the test has one-sided error.  The
test is also non-adaptive --- the $t$ random choices of $\bfx$ and $\bfy$
can all be made up front.  The non-trivial statement is that only
functions that are close to linear pass the test with large
probability.
\begin{theorem}[\citealt{BLR93}]\label{t:blr}
If the BLR linearity test accepts a function $f$ with probability
greater than $\tfrac{1}{3}$, then $f$ is $\eps$-close to the set of
linear functions.
\end{theorem}
The modern and slick proof of Theorem~\ref{t:blr} uses Fourier
analysis --- indeed, the elegance of this proof serves as convincing
motivation for the more general study of Boolean functions from a
Fourier-analytic perspective.  See Chapter~1 of~\cite{odonnell} for a
good exposition.  There are also more direct proofs of
Theorem~\ref{t:blr}, as in~\cite{BLR93}.  None of these proofs are
overly long, but we'll spend our time on monotonicity testing
instead.  We mention the BLR test for the following reasons:
\begin{enumerate}

\item If you only remember one property testing result, 
  Theorem~\ref{t:blr} and the BLR linearity test would be a good one.

\item The BLR test is the thin end of the wedge in constructions of
  probabilistically checkable proofs (PCPs).  Recall that a language is in
  $NP$ if membership can be efficiently verified --- for example,
  verifying an alleged satisfying assignment to a SAT formula is easy
  to do in polynomial time.  The point of a PCP is to rewrite such a
  proof of membership so that it can be probabilistically verified
  after reading only a constant number of bits.  The BLR test does
  exactly this for the special case of linearity testing --- for
  proofs where ``correctness'' is equated with being the
  truth table of a linear function.
The BLR test effectively means that one can assume without loss of
generality that a proof encodes a linear function --- the BLR test
can be used as a preprocessing step to reject alleged proofs that are
not close to a linear function.  Subsequent testing steps can then
focus on whether or not the encoded linear function is close to a
subset of linear functions of interest.

\item Theorem~\ref{t:blr} highlights a consistent theme in property
  testing --- establishing connections between ``global'' and
  ``local'' properties of a function.  Saying that a function $f$ is
  $\eps$-far from a property $\P$ refers to the entire domain $D$ and
  in this sense asserts a ``global violation'' of  the property.
Property testers work well when there are ubiquitous ``local violations''
of the property.  Theorem~\ref{t:blr} proves that, for the property of
linearity, a global violation necessarily implies lots of local
violations.  We give a full proof of such a ``global to local''
statement for monotonicity testing in the next section.

\end{enumerate}

\section{Monotonicity Testing: Upper Bounds}\label{s:ub_pt}

The problem of monotonicity testing was introduced in~\cite{G+98} and
is one of the central problems in the field.  We discuss the Boolean case,
where there have been several breakthroughs in just the past few
months, in Sections~\ref{ss:bool} and~\ref{ss:recent}.  We discuss the
case of larger 
ranges, where communication complexity has been used to prove strong
lower bounds, in Section~\ref{ss:gen}.

\subsection{The Boolean Case}\label{ss:bool}

In this section, we take $D = \zo^n$ and $R = \zo$.
For $b \in \zo$ and $\xmi \in \zo^{n-1}$,
we use the notation $(b,\xmi)$ to denote a vector of $\zo^n$ in which
the $i$th bit is $b$ and the other $n-1$ bits are $\xmi$.
A function $f:\zo^n \rightarrow \zo$ is {\em monotone} if flipping a
coordinate of an input from~0 to~1 can only increase the function's
output:
\[
f(0,\xmi) \le f(1,\xmi)
\]
for every $i \in \{1,2,\ldots,n\}$ and $\xmi \in \zo^{n-1}$.

It will be useful to visualize the domain $\zo^n$ as the
$n$-dimensional hypercube; see also Figure~\ref{f:cube}.
This graph has $2^n$ vertices and
$n2^{n-1}$ edges.  An edge can be uniquely specified by a coordinate
$i$ and vector $\xmi \in \zo^{n-1}$ --- the edge's endpoints are then
$(0,\xmi)$ and $(1,\xmi)$.  By the {\em $i$th slice} of the hypercube,
we mean the $2^{n-1}$ edges for which the endpoints differ (only) in the
$i$th coordinate.  The $n$ slices form a partition of the edge set of
the hypercube, and each slice is a perfect matching of the hypercube's
vertices.  A function $\zo^n \rightarrow \zo$ can be visualized as a
binary labeling of the vertices of the hypercube.

\begin{figure}
\centering
\includegraphics[width=.3\textwidth]{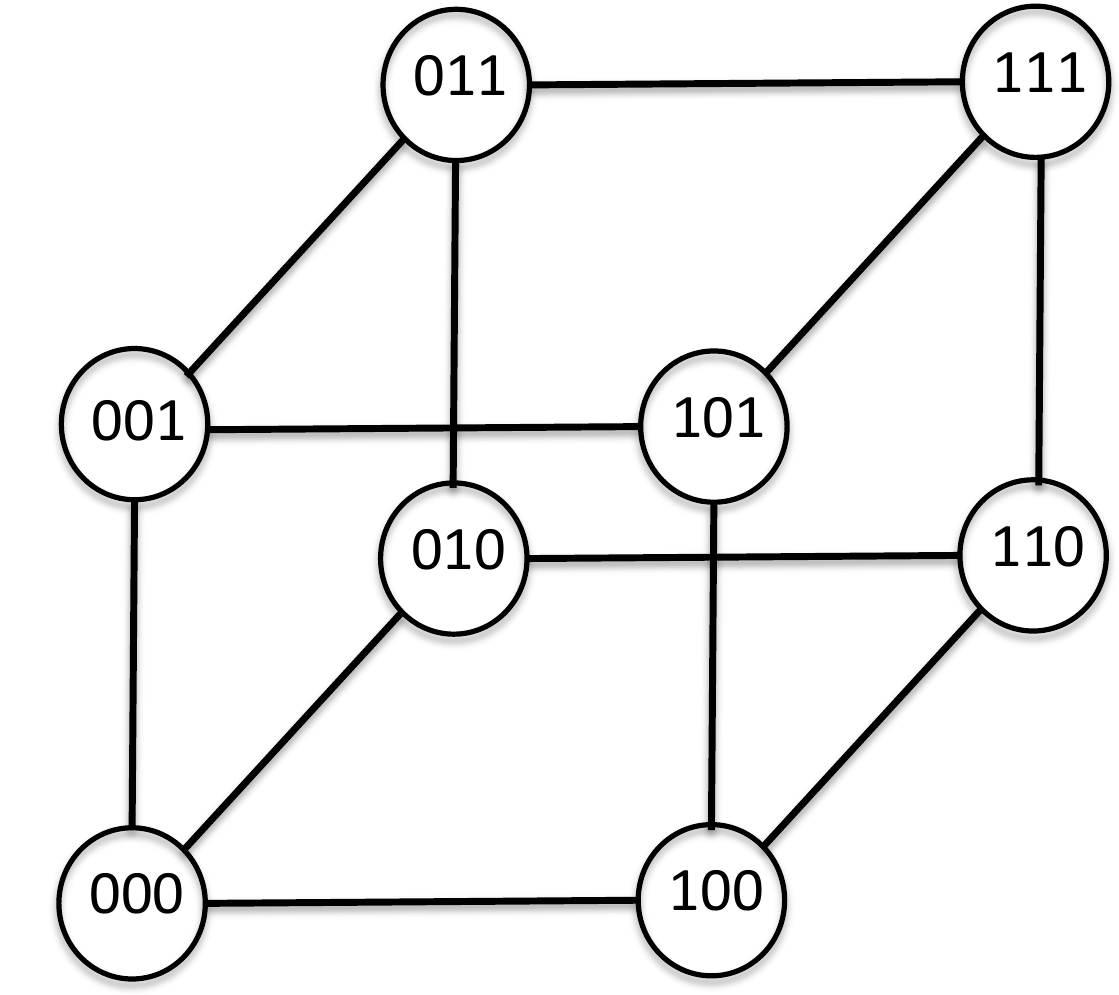}
\caption[$\{0,1\}^n$ as an $n$-dimensional
hypercube]{$\{0,1\}^n$ can be visualized as an $n$-dimensional hypercube.}\label{f:cube}
\end{figure}

We consider the following {\em edge tester}, which picks random
edges of the hypercube and rejects if it ever finds a monotonicity
violation across one of the chosen edges.
\begin{enumerate}

\item Repeat $t$ times:

\begin{enumerate}

\item Pick $i \in \{1,2,\ldots,n\}$ and $\xmi \in \zo^{n-1}$
uniformly at random.

\item If $f(0,\xmi) > f(1,\xmi)$ then REJECT.

\end{enumerate}

\item ACCEPT.

\end{enumerate}
Like the BLR test, it is clear that the edge tester has 1-sided error
(no false negatives) and is non-adaptive.  The non-trivial part is to
understand the probability of rejecting a function that is $\eps$-far
from monotone --- how many trials $t$ are necessary and sufficient for
a rejection probability of at least~$2/3$?  Conceptually, how
pervasive are the local failures of monotonicity for a function that
is $\eps$-far from monotone?

The bad news is that, in contrast to the BLR linearity test, taking
$t$ to be a constant (depending only on $\eps$) is not good enough.
The good news is that we can take $t$ to be only logarithmic in the size of
the domain.
\begin{theorem}[\citealt{G+98}]\label{t:ub}
For $t=\Theta(\tfrac{n}{\eps})$, the edge tester rejects every
function that is $\eps$-far from monotone with probability at least
2/3.
\end{theorem}

\begin{proof}
A simple calculation shows that it is enough to prove that a single
random trial of the edge test rejects a function that is $\eps$-far from
monotone with probability at least $\tfrac{\eps}{n}$.

Fix an arbitrary function $f$.  There are two quantities that we need
to relate to each other --- the rejection probability of~$f$, and the
distance between $f$ and the set of monotone functions.  We do this by
relating both quantities to the sizes of the following sets: for
$i=1,2,\ldots,n$, define
\begin{equation}\label{eq:a}
|A_i| = \{ \xmi \in \zo^{n-1} \,:\, f(0,\xmi) > f(1,\xmi) \}.
\end{equation}
That is, $A_i$ is the edges of the $i$th slice of the hypercube across
which $f$ violates monotonicity.  By the definition of the edge
tester, the probability that a single trial rejects~$f$ is exactly
\begin{equation}\label{eq:ub1}
\underbrace{\sum_{i=1}^n |A_i|}_{\text{\# of
    violations}}/\underbrace{n2^{n-1}}_{\text{\# of edges}}.
\end{equation}

Next, we upper bound the distance between $f$ and the set of monotone
functions, implying that the only way in which the $|A_i|$'s
(and hence the rejection probability) can be small is if $f$ is close to
a monotone function.  To upper bound the distance, all we need to do
is exhibit a single monotone function close to $f$.  Our plan is to
transform $f$ into a monotone function, coordinate by coordinate,
tracking the number of changes that we make along the way.  The next
claim controls what happens when we ``monotonize'' a single
coordinate.

\vspace{.1in}
\noindent
\noindent
\textbf{Key Claim: } Let $i \in \{1,2,\ldots,n\}$ be a coordinate.
Obtain $f'$ from $f$ by, for each violated edge $((0,\xmi),(1,\xmi))
\in A_i$ of the $i$th slice, swapping the values of $f$ on its
endpoints (Figure~\ref{f:swap}).  That is, set $f'(0,\xmi) = 0$ and
$f'(1,\xmi) = 1$.  (This operation is well defined because the edges
of $A_i$ are disjoint.)  For every coordinate $j=1,2,\ldots,n$, $f'$
has no more monotonicity violations in the $j$th slice than does $f$.

\vspace{.1in}
\noindent
\noindent
\textbf{Proof of Key Claim: } The claim is clearly true for $j=i$: by
construction, the swapping operation fixes all of the monotonicity
violations in the $i$th slice, without introducing any new violations
in the $i$th slice.  The interesting case is when $j \neq i$, since
new monotonicity violations can be introduced (cf., Figure~\ref{f:swap}).
The claim asserts that the overall number of violations cannot
increase (cf., Figure~\ref{f:swap}).

\begin{figure}
  \centering
  \subbottom[Fixing the first slice]{%
    \includegraphics[width=0.7\textwidth]{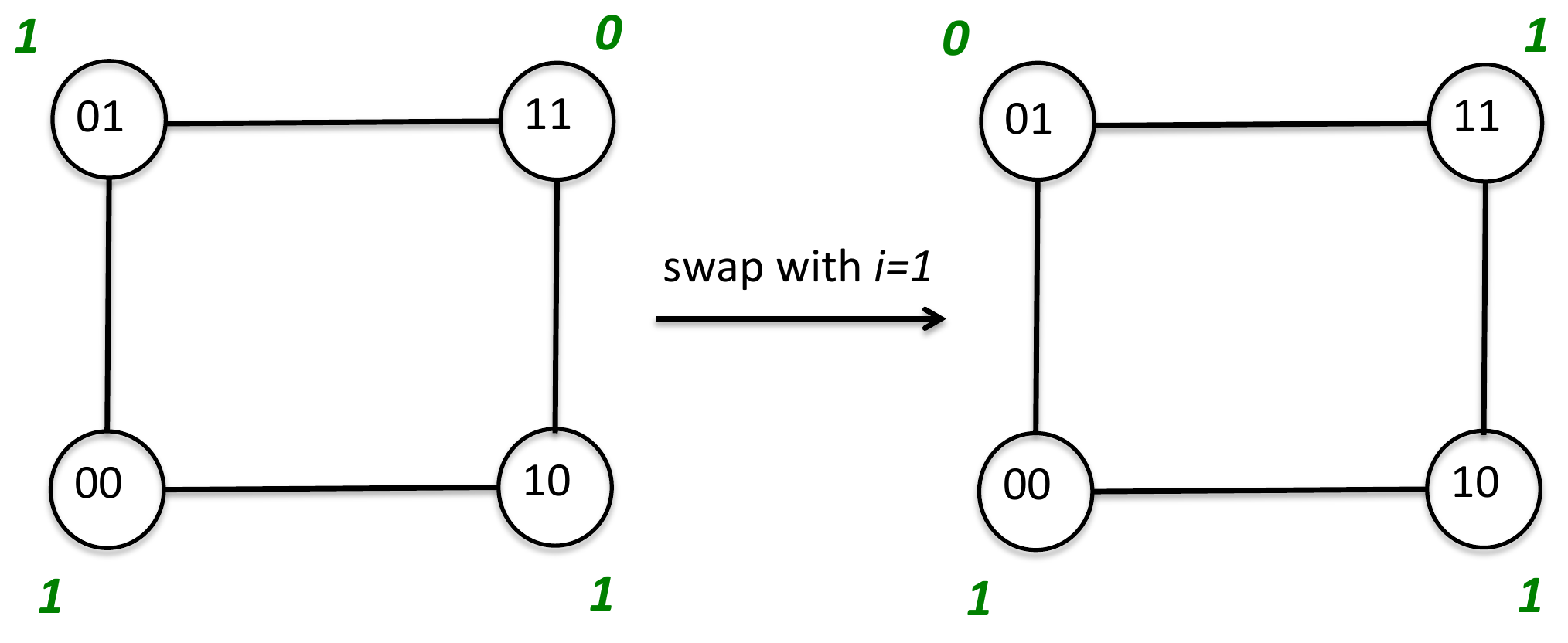}}
\qquad\qquad
  \subbottom[Fixing the second slice]{%
    \includegraphics[width=0.7\textwidth]{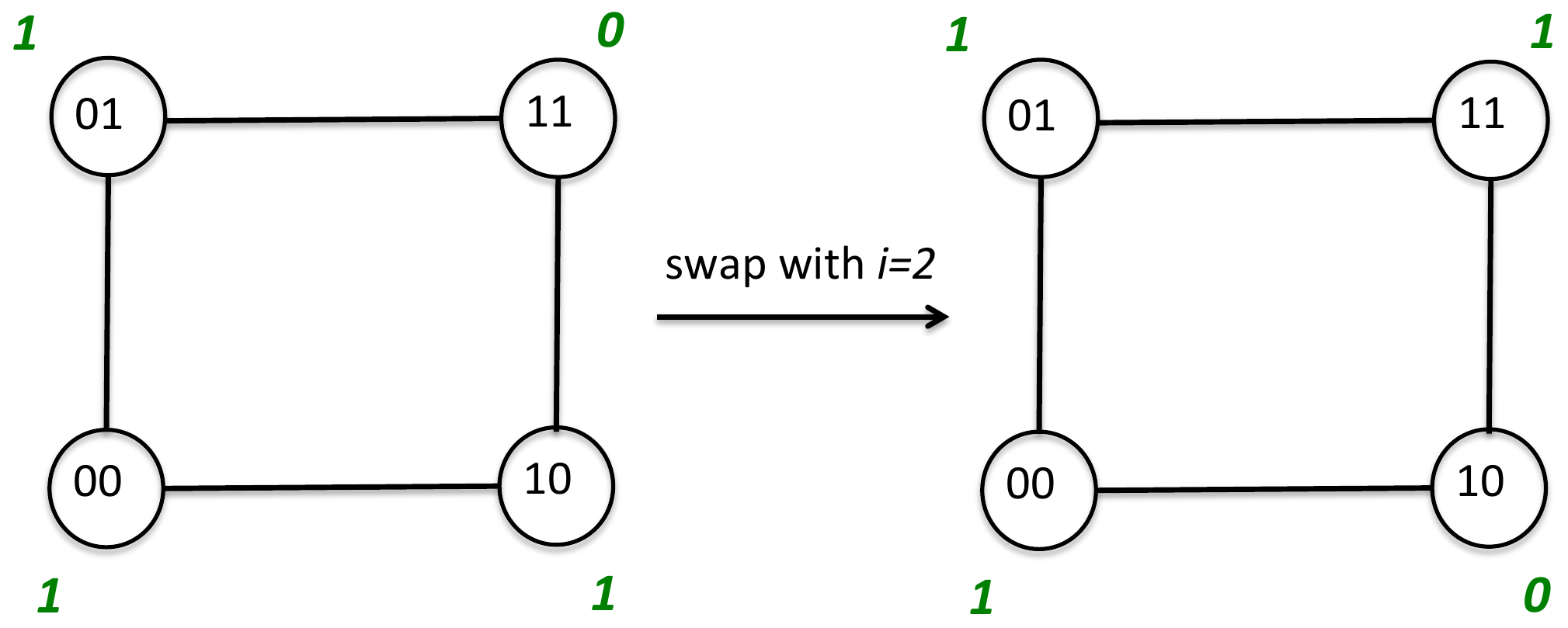}}
\caption[Swapping values to eliminate the monotonicity
  violations in the $i$th slice]{Swapping values to eliminate the monotonicity
  violations in the $i$th slice.}
\label{f:swap}
\end{figure}


We partition the edges of the $j$th slice into edge pairs as
follows.  We use $\xmj^0$ to denote an assignment to the $n-1$
coordinates other than $j$ in which the $i$th coordinate is 0, and
$\xmj^1$ the corresponding assignment in which the $i$th coordinate
is flipped to~1.  For a choice of~$\xmj^0$, we can
consider the ``square'' formed by the vertices $(0,\xmj^0)$,
$(0,\xmj^1)$, $(1,\xmj^0)$, and $(1,\xmj^1)$; see
Figure~\ref{f:square}.  The edges $((0,\xmj^0),(1,\xmj^0))$ and
$((0,\xmj^1),(1,\xmj^1))$ belong to the $j$th slice, and ranging over 
the $2^{n-2}$ choices for $\xmj^0$ --- one binary choice per
coordinate other than~$i$ and~$j$ --- generates each such edge exactly
once.

\begin{figure}
\centering
\includegraphics[width=.8\textwidth]{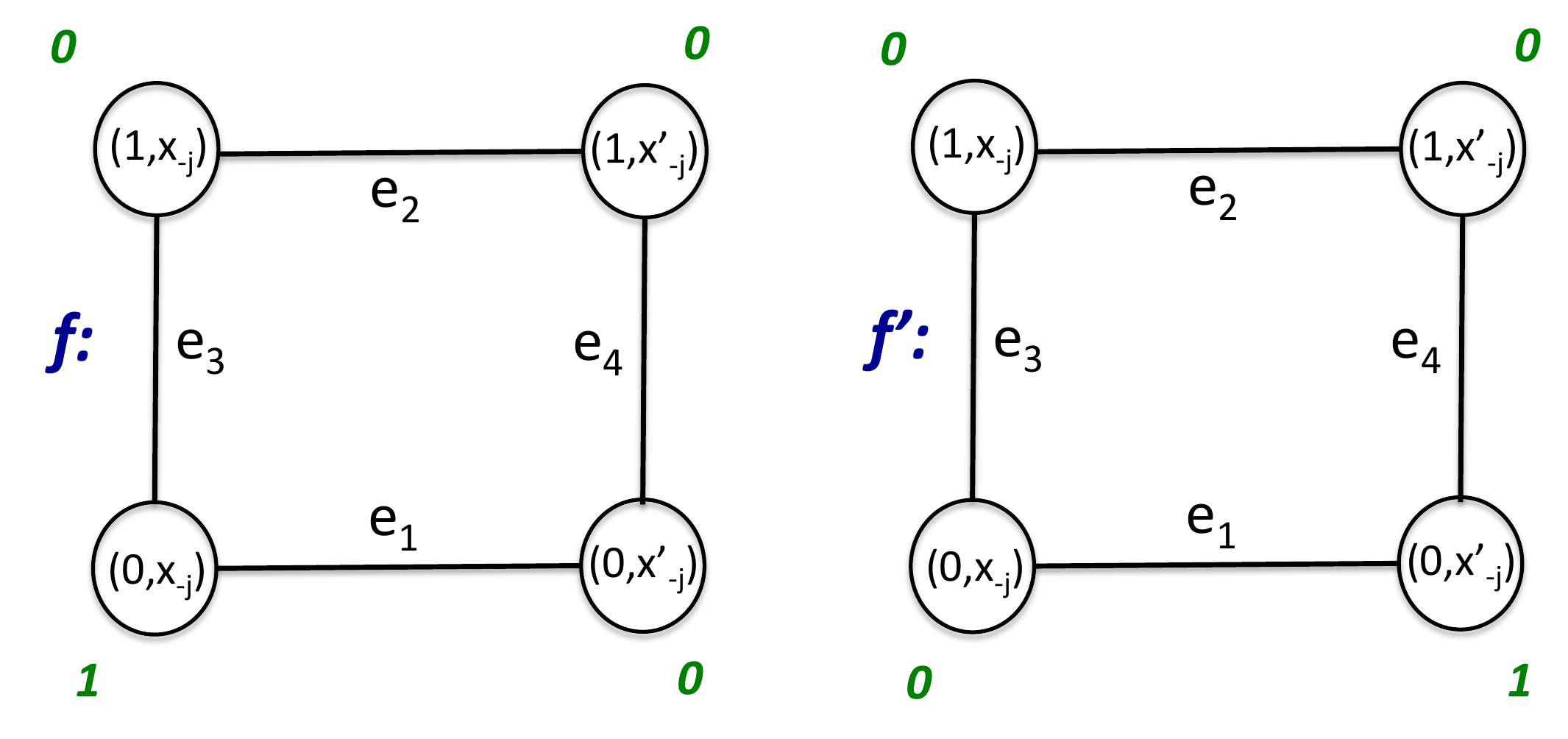}
\caption[Tracking the number of monotonicity violations]{The number of monotonicity violations on edges
$e_3$ and $e_4$ is at least as large under~$f$ as under~$f'$.}
\label{f:square}
\end{figure}

Fix a choice of $\xmj^0$, and label the edges of the corresponding
square $e_1,e_2,e_3,e_4$ as in Figure~\ref{f:square}.  A simple case
analysis shows that the number of monotonicity violations on edges
$e_3$ and $e_4$ is at least as large under~$f$ as under~$f'$.
If neither $e_1$ nor $e_2$ was violated under~$f$, then $f'$ agrees
with $f$ on this square and the total number of monotonicity violations
is obviously the same.  If both $e_1$ and $e_2$ were violated under
$f$, then values were swapped along both these edges; hence $e_3$
(respectively, $e_4$) is violated under $f'$ if and only if 
$e_4$ (respectively, $e_3$) was violated under $f$.
Next, suppose the endpoints of $e_1$ had their values swapped, while
the endpoints of $e_2$ did not.  
This implies that $f(0,\xmj^0) = 1$ and $f(0,\xmj^1) = 0$, and hence
$f'(0,\xmj^0) = 0$ and $f(0,\xmj^1) = 1$.  
If the endpoints $(1,\xmj^0)$ and $(1,\xmj^1)$ of $e_2$ have the
values~0 and~1 (under both $f$ and $f'$), then the number of
monotonicity violations on $e_3$ and $e_4$ drops from~1 to~0.  The
same is true if their values are~0 and~0.  If their values are~1
and~1, then the monotonicity violation on edge $e_4$ under $f$
moves to one on edge $e_3$ under $f'$, but the number of violations
remains the same.  The final set of cases, when the endpoints of $e_2$
have their values swapped while the endpoints of~$e_1$ do not, is
similar.\footnote{Suppose we corrected only one endpoint of an edge to
  fix a monotonicity violation, rather than swapping the endpoint
  values.  Would the proof still go through?}

Summing over all such squares --- all choices of $\xmj^0$ --- we
conclude that the number of monotonicity violations in the $j$th slice
can only decrease.
{$\blacksquare$\vskip
\belowdisplayskip}

Now consider turning a function $f$ into a monotone function $g$ by
doing a single pass through the coordinates, fixing all monotonicity
violations in a given coordinate via swaps as in the Key Claim.
This process terminates with a monotone function: immediately after
coordinate $i$ is treated, there are no monotonicity violations in the
$i$th slice by construction; and by the Key Claim, fixing future
coordinates does not break this property.  
The Key Claim also implies that, in the iteration where this procedure
processes the $i$th coordinate, the number of monotonicity violations 
that need fixing is at most the number
$|A_i|$ of monotonicity violations in this slice under the
original function $f$.  Since the procedure makes two modifications to
$f$ for each monotonicity violation that it fixes (the two endpoints
of an edge), we conclude that $f$ can be made monotone by changing at
most
$2 \sum_{i=1}^n |A_i|$
of its values.  If $f$ is $\eps$-far from monotone, then 
$2\sum_{i=1}^n |A_i| \ge \eps2^n$.  Plugging this into~\eqref{eq:ub1},
we find that a single trial of the edge tester rejects such an $f$
with probability at least
\[
\frac{\tfrac{1}{2} \eps 2^n}{n2^{n-1}} = \frac{\eps}{n},
\]
as claimed.
\end{proof}

\subsection{Recent Progress for the Boolean Case}\label{ss:recent}

An obvious question is whether or not we can improve over the
query upper bound in Theorem~\ref{t:ub}.  The analysis in
Theorem~\ref{t:ub} of the edge tester is tight up to a constant factor
(see Exercises), so an improvement would have to come from a different
tester.  There was no progress on this problem for 15 years, but
recently there has been a series of breakthroughs on the problem.
\cite{CS13} gave the first improved upper bounds, of
$\tilde{O}(n^{7/8}/\eps^{3/2})$.\footnote{The notation
  $\tilde{O}(\cdot)$ suppresses logarithmic factors.}
A year later, \cite{CST14} gave an upper bound of
$\tilde{O}(n^{5/6}/\eps^{4})$.
Just a couple of months ago, \cite{KMS15} gave a bound of
$\tilde{O}(\sqrt{n}/\eps^2)$.
All of these improved upper bounds are for {\em path testers}.  The
idea is to sample a random monotone path from the hypercube (checking
for a violation on its endpoints), rather than a random edge.  One way
to do this is: pick a random point $\bfx \in \zo^n$; pick a random
number $z$ between 0 and the number of zeroes of $\bfx$ (from some
distribution); and obtain $\bfy$ from $\bfx$ by choosing at random $z$ of
$\bfx$'s 0-coordinates and flipping them to~1.  Given that a function
that is $\eps$-far from monotone must have lots of violated edges (by
Theorem~\ref{t:ub}), it is plausible that path testers, which aspire
to check many edges at once, could be more effective than edge
testers.  The issue is that just because a path contains one or more 
violated edges does not imply that the path's endpoints 
will reveal a monotonicity violation.  Analyzing path testers seems
substantially more complicated than the edge
tester~\citep{CS13,CST14,KMS15}.  Note that path testers are non-adaptive
and have 1-sided error.

There have also been recent breakthroughs on the lower bound side.  It
has been known for some time that all non-adaptive testers with
1-sided error require $\Omega(\sqrt{n})$ queries~\citep{F+02}; see also
the Exercises.  For non-adaptive testers with two-sided error,
\cite{CST14} proved a lower bound
of~$\tilde{\Omega}(n^{1/5})$ and 
\cite{CDST14} improve this to $\Omega(n^{(1/2)-c})$ for every
constant $c > 0$.  Because the gap in query complexity between
adaptive and non-adaptive testers can only be exponential (see
Exercises), these lower bounds also imply that adaptive testers (with
two-sided error) require $\Omega(\log n)$ queries.  The gap between
$\tilde{O}(\sqrt{n})$ and $\Omega(\log n)$ for adaptive testers remains open;
most researchers think that adaptivity cannot help and that the upper
bound is the correct answer.

An interesting open question is whether or not 
communication complexity is useful for proving interesting lower
bounds for the monotonicity testing of Boolean functions.\footnote{At the
very least, some of the techniques we've learned in previous lectures
are useful.  The
  arguments in~\cite{CST14} and~\cite{CDST14} use an analog of Yao's Lemma (Lemma~\ref{l:yao}))
  to switch from randomized to distributional lower bounds.
The hard part is then to come up with a distribution over both monotone
functions and functions $\eps$-far from monotone such that no
deterministic tester can reliably distinguish between the two cases
using few queries to the function.}
We'll see in
Section~\ref{s:lb_pt} that it is useful for proving lower bounds in the
case where the range is relatively large.

\subsection{Larger Ranges}\label{ss:gen}

In this section we study monotonicity testing with the usual domain
$D=\{0,1\}^n$ but with a range $R$ that is an arbitrary finite,
totally ordered set.  Some of our analysis for the Boolean case
continues to apply.  For example, the edge tester continues to be a
well-defined tester with 1-sided error.  Returning to the proof of
Theorem~\ref{t:ub}, we can again define each $A_i$ as the set of
monotonicity violations --- meaning $f(0,\xmi) > f(1,\xmi)$ --- along
edges in the $i$th slice.  The rejection probability again equals
the quantity in~\eqref{eq:ub1}.

We need to revisit the major step of the proof of Theorem~\ref{t:ub},
which for Boolean functions gives an upper bound of $2\sum_{i=1}^n
|A_i|$ on the distance from a function $f$ to the set of monotone
functions.  
One idea is to again
do a single pass through the coordinates, swapping
the function values of the endpoints of the edges in the current slice
that have monotonicity violations.
In contrast to the Boolean case, this idea does not always result 
in a monotone function (see Exercises).

We can extend the argument to general finite ranges $R$ by doing
multiple passes over 
the coordinates.  
The simplest approach uses one pass over the
coordinates, fixing all monotonicity violations that involve a vertex
$\bfx$ with $f(\bfx) = 0$; a second pass, fixing all monotonicity
violations that involve a vertex~$\bfx$ with $f(\bfx) = 1$; and so on.
Formalizing this argument yields a bound of $2|R| \sum_{i=1}^n |A_i|$ on
the distance between $f$ and the set of monotone functions, which
gives a query bound of $O(n|R|/\eps)$~\cite{G+98}.

A divide-and-conquer approach gives a better upped bound~\cite{D+99}.
Assume without loss of generality (relabeling if necessary) that $R =
\{0,1,\ldots,r-1\}$, and also (by padding) that $r=2^k$ for a positive
integer $k$.
The first pass over the coordinates fixes all monotonicity violations
that involve values that differ in their most significant bit
--- one value that is less than $\tfrac{r}{2}$ and one value
that is at least $\tfrac{r}{2}$.  The second pass fixes all
monotonicity violations involving two values that differ in their
second-highest-order bit.  And so on.  The Exercises ask you to prove
that this idea can be made precise and show that the distance between
$f$ and the set of monotone functions is at most $2 \log_2 |R|
\sum_{i=1}^n |A_i|$.  This implies an upper bound of
$O(\tfrac{n}{\eps} \log |R|)$ on the number of queries used by the
edge tester for the case of general finite ranges.
The next section shows a lower bound of $\Omega(n/\eps)$ when $|R| =
\Omega(\sqrt{n})$; in these cases, this upper bound is the best
possible, up to the $\log R$ factor.\footnote{It is an
  open question to reduce the dependence on $|R|$.  Since we can
  assume that $|R| \le 2^n$ (why?), any sub-quadratic upper
  bound~$o(n^2)$ would constitute an improvement.}

\section{Monotonicity Testing: Lower Bounds}\label{s:lb_pt}

\subsection{Lower Bound for General Ranges}

This section uses communication complexity to prove a lower bound on
the query complexity of testing monotonicity for sufficiently large
ranges. 
\begin{theorem}[\cite{BBM11}]\label{t:bbm}
For large enough ranges $R$ and $\eps = \tfrac{1}{8}$, every
(adaptive) monotonicity tester with two-sided error uses $\Omega(n)$
queries.
\end{theorem}

Note that Theorem~\ref{t:bbm} separates the case of a general range
$R$ from the case of a Boolean range, where $\tilde{O}(\sqrt{n})$ queries are
enough~\cite{KMS15}.
With the right communication complexity tools, Theorem~\ref{t:bbm} is
not very hard to prove.  Simultaneously with~\cite{BBM11},
Bri\"et et al.~\cite{B+12} gave a non-trivial proof from scratch of a similar
lower bound, but it applies only to non-adaptive testers with 1-sided
error.  Communication complexity 
techniques naturally lead to lower bounds
for adaptive testers with two-sided error.

As always, the first thing to try is a reduction from \disj, with the
query complexity somehow translating to the communication cost.
At first this might seem weird --- there's only one ``player'' in
property testing, so where do Alice and Bob come from?  But as we've seen
over and over again, starting with our applications to streaming lower
bounds, it can be useful to invent two parties just for the sake of
standing on the shoulders of communication complexity lower bounds.
To implement this, we need to show how a low-query tester for
monotonicity leads to a low-communication protocol for \disj.

It's convenient to reduce from a ``promise'' version of \disj that is
just as hard as the general case.  In the \udisj problem, the goal is
to distinguish between inputs where Alice and Bob have sets $A$ and
$B$ with $A \cap B = \emptyset$, and inputs where $|A \cap B| = 1$.
On inputs that satisfy neither property, any output is considered
correct.  The \udisj problem showed up a couple of times in
previous lectures; let's review them.
At the conclusion of our 
lecture on the extension complexity of polytopes,
we proved that the nondeterministic
communication complexity of the problem is $\Omega(n)$ using a
covering argument with a clever inductive proof
(Theorem~\ref{t:udisj}). 
In our boot camp
(Section~\ref{ss:disj_proof}), we discussed the high-level approach of
Razborov's proof that every randomized protocol for \disj with
two-sided error requires $\Omega(n)$ communication.  Since the hard
probability 
distribution in this proof makes use only of inputs with intersection
size 0 or 1, the lower bound applies also to the \udisj problem.

Key to the proof of Theorem~\ref{t:bbm} is the following lemma.
\begin{lemma}\label{l:bbm}
Fix sets $A,B \sse U = \{1,2,\ldots,n\}$.  Define the function
$h_{AB}:2^U \rightarrow \ZZ$ by
\begin{equation}\label{eq:h}
h_{AB}(S) = 2|S| + (-1)^{|S \cap A|} + (-1)^{|S \cap B|}.
\end{equation}
Then:
\begin{itemize}

\item [(i)] If $A \cap B = \emptyset$, then $h$ is monotone.

\item [(ii)] If $|A \cap B| = 1$, then $h$ is $\tfrac{1}{8}$-far from
  monotone. 

\end{itemize}
\end{lemma}

We'll prove the lemma shortly; let's first see how to use it to prove
Theorem~\ref{t:bbm}.  Let $Q$ be a tester that distinguishes between
monotone functions from $\zo^n$ to $R$ and functions that are
$\tfrac{1}{8}$-far from monotone.  We proceed to construct a 
(public-coin randomized) protocol for the \udisj problem.

Suppose Alice and Bob have sets $A,B \sse \{1,2,\ldots,n\}$.  The idea
is for both parties to run local copies of the tester $Q$ to test the
function $h_{AB}$, communicating with each other as needed to carry
out these simulations.  In more detail, Alice and Bob first use the
public coins to agree on a random string to be used with the tester
$Q$.  Given this shared random string, $Q$ is deterministic.  Alice
and Bob then simulate local copies of $Q$ query-by-query:
\begin{enumerate}

\item Until $Q$ halts:

\begin{enumerate}

\item Let $S \sse \{1,2,\ldots,n\}$ be the next query that $Q$ asks
  about the function $h_{AB}$.\footnote{As usual, we're not distinguishing
    between subsets of $\{1,2,\ldots,n\}$ and their characteristic
    vectors.} 

\item Alice sends $(-1)^{|S \cap A|}$ to Bob.

\item Bob sends $(-1)^{|S \cap B|}$ to Alice.

\item Both Alice and Bob evaluate the function $h_{AB}$ at $S$, and
  give the result to their respective local copies of $Q$.

\end{enumerate}

\item Alice (or Bob) declares ``disjoint'' if $Q$ accepts the function
  $h_{AB}$, and ``not disjoint'' otherwise.

\end{enumerate}
We first observe that the protocol is well defined.  Since Alice and
Bob use the same random string and simulate $Q$ in lockstep, both
parties know the (same) relevant query $S$ to $h_{AB}$ in every
iteration, and thus are positioned to send the relevant bits
($(-1)^{|S \cap A|}$ and $(-1)^{|S \cap B|}$) to each other.
Given these bits, they are able to evaluate $h_{AB}$ at the point $S$
(even though Alice doesn't know $B$ and Bob doesn't know $A$).

The communication cost of this protocol is twice the number of queries
used by the tester $Q$, and it doesn't matter if $Q$ is adaptive or
not.  Correctness of the protocol follows immediately
from Lemma~\ref{l:bbm}, with the error of the protocol the same as that of
the tester~$Q$.  Because every randomized protocol (with two-sided
error) for \udisj has communication complexity $\Omega(n)$, we
conclude that every (possibly adaptive) tester~$Q$ with two-sided
error requires $\Omega(n)$ queries for monotonicity testing.  This
completes the proof of Theorem~\ref{t:bbm}.

\vspace{.1in}
\noindent
\begin{prevproof}{Lemma}{l:bbm}
For part~(i), assume that $A \cap B = \emptyset$ and consider any set
$S \sse \{1,2,\ldots,n\}$ and $i \notin S$.  
Because $A$ and $B$ are disjoint, $i$ does not belong to at least one
of $A$ or $B$.
Recalling~\eqref{eq:h}, in the expression $h_{AB}(S \cup \{i\}) -
h_{AB}(S)$, the difference between the first terms is~2, the
difference in either the second terms (if $i \notin A$) or in the third
terms (if $i \notin B$) is zero, and the difference in the remaining
terms is at least~-2.  Thus, $h_{AB}(S \cup \{i\}) -               
h_{AB}(S) \ge 0$ for all $S$ and $i \notin S$, and $h_{AB}$ is
monotone.

For part~(ii), let $A \cap B = \{i\}$.
For all $S \sse \{1,2,\ldots,n\} \sm \{i\}$
such that $|S \cap A|$ and $|S \cap B|$ are both even,
$h_{AB}(S \cup \{i\}) - h_{AB}(S) = -2$.
If we choose such an $S$ uniformly at random, then
$\prob{|S \cap A| \text{ is even}}$ is~1 (if $A = \{i\}$) or~$\tfrac{1}{2}$
(if $A$ has additional elements, using the Principle of Deferred
Decisions).  Similarly, $\prob{|S \cap B| \text{ is even}} \ge
\tfrac{1}{2}$.
Since no potential element of $S \sse \{1,2,\ldots,n\} \sm \{i\}$ is
a member of both $A$ and $B$, these two events are independent and
hence 
$\prob{|S \cap A|, |S \cap B| \text{ are both even}} \ge \tfrac{1}{4}$.
Thus, for at least $\tfrac{1}{4} \cdot 2^{n-1} = 2^n/8$ choices of
$S$, $h_{AB}(S \cup \{i\}) < h_{AB}(S)$.  Since all of these
monotonicity violations involve different values of $h_{AB}$ --- in
the language of the proof of Theorem~\ref{t:ub}, they are all edges of
the $i$th slice of the hypercube --- fixing all of them requires
changing $h_{AB}$ at $2^n/8$ values.  We conclude that $h_{AB}$ is
$\tfrac{1}{8}$-far from a monotone function.
\end{prevproof}

\subsection{Extension to Smaller Ranges}

Recalling the definition~\eqref{eq:h} of the function $h_{AB}$, we see
that the proof of Theorem~\ref{t:bbm} establishes a query complexity
lower bound of $\Omega(n)$ provided the range~$R$ has
size~$\Omega(n)$.  It is not difficult to extend the lower bound to
ranges of size $\Omega(\sqrt{n})$.  The trick is to consider a
``truncated'' version of $h_{AB}$, 
call it $h'_{AB}$, where values of $h_{AB}$
less than $n - c\sqrt{n}$ are rounded up to $n-c\sqrt{n}$ and values
more than $n + c\sqrt{n}$ are rounded down to $n+c\sqrt{n}$.
(Here~$c$ is a sufficiently large constant.)  The range of $h'_{AB}$
has size $\Theta(\sqrt{n})$ for all $A,B \sse \{1,2,\ldots,n\}$.

We claim that Lemma~\ref{l:bbm} still holds for $h'_{AB}$, with the
``$\tfrac{1}{8}$'' in case~(ii) replaced by ``$\tfrac{1}{16}$;''
the new version of Theorem~\ref{t:bbm} then follows.  Checking that
case~(i) in 
Lemma~\ref{l:bbm} still holds is easy: truncating a monotone function
yields another monotone function.  For case~(ii), it is enough to show
that $h_{AB}$ and $h'_{AB}$ differ in at most a $\tfrac{1}{16}$ 
fraction of
their entries; since Hamming distance satisfies the triangle
inequality, this implies that $h'_{AB}$ must be $\tfrac{1}{16}$-far
from the set of monotone functions.  Finally, consider choosing $S
\sse \{1,2,\ldots,n\}$ uniformly at random: up to an ignorable
additive term in $\{-2,-1,0,1,2\}$, the value of $h_{AB}$ lies in $n
\pm c\sqrt{n}$ with probability at least $\tfrac{15}{16}$,
provided~$c$ is a sufficiently large constant (by Chebyshev's
inequality). 
This implies that $h_{AB}$ and $h'_{AB}$ agree on all but a
$\tfrac{1}{16}$ fraction of the domain, completing the proof.

For even smaller ranges $R$, the argument above can be augmented by a
padding argument to prove a query complexity lower bound of
$\Omega(|R|^2)$; see the Exercises.

\section{A General Approach}

It should be clear from the proof of Theorem~\ref{t:bbm} that its
method of deriving property testing lower bounds from communication
complexity lower bounds is general, and not particular to the problem
of testing monotonicity.  The general template for deriving lower
bounds for testing a property $\P$ is:
\begin{enumerate}

\item Map inputs $\inputs$ of a communication problem $\Pi$ with
  communication complexity at least $c$ to a function $h_{\inputs}$
  such that:

\begin{enumerate}

\item 1-inputs $\inputs$ of $\Pi$ map to functions $h_{\inputs}$
  that belong to $\P$;

\item 0-inputs $\inputs$ of $\Pi$ map to functions $h_{\inputs}$
  that are $\eps$-far from $\P$.

\end{enumerate}

\item Devise a communication protocol for evaluating $h_{\inputs}$
  that has cost $d$.  (In the proof of Theorem~\ref{t:bbm}, $d=2$.)

\end{enumerate}
Via the simulation argument in the proof of Theorem~\ref{t:bbm},
instantiating this template yields a query complexity lower bound of
$c/d$ for testing the property $\P$.\footnote{There is an analogous
  argument that uses one-way communication complexity lower bounds to
  derive query complexity lower bounds for {\em non-adaptive} testers;
  see the Exercises.}

There are a number of other applications of this template to various
property testing problems, such as testing if a function admits a
small representation (as a sparse polynomial, as a small decision
tree, etc.).  See~\cite{BBM11,G13} for several examples.

A large chunk of the property testing literature is about testing
graph properties~\cite{GGR98}.  
An interesting open question is if communication complexity can be
used to prove strong lower bounds for such problems.



\end{document}
